\documentclass[12pt,oneside]{amsart}
\usepackage[latin9]{inputenc}
\usepackage[a4paper]{geometry}
\usepackage{color}
\usepackage{verbatim}
\usepackage{enumitem}
\usepackage{amstext}
\usepackage{amsthm}
\usepackage{amssymb}
\usepackage{stmaryrd}
\usepackage{stackrel}
\usepackage{graphicx}
\usepackage{esint}
\usepackage{xargs}[2008/03/08]

\makeatletter

\providecommand{\tabularnewline}{\\}

\numberwithin{equation}{section}
\numberwithin{figure}{section}
\theoremstyle{plain}
\newtheorem{thm}{\protect\theoremname}[section]
  \theoremstyle{plain}
  \newtheorem{cor}[thm]{\protect\corollaryname}
  \theoremstyle{remark}
  \newtheorem{rem}[thm]{\protect\remarkname}
  \theoremstyle{definition}
  \newtheorem{defn}[thm]{\protect\definitionname}
  \theoremstyle{plain}
  \newtheorem{lem}[thm]{\protect\lemmaname}
  \theoremstyle{plain}
  \newtheorem{prop}[thm]{\protect\propositionname}
  \theoremstyle{definition}
  \newtheorem{example}[thm]{\protect\examplename}

\@ifundefined{date}{}{\date{}}
\usepackage{amsthm}
\usepackage{longtable}

\DeclareSymbolFont{extraup}{U}{zavm}{m}{n}
\DeclareMathSymbol{\varheart}{\mathalpha}{extraup}{86}
\DeclareMathSymbol{\vardiamond}{\mathalpha}{extraup}{87} 
\renewcommand{\textendash}{--}
\renewcommand{\textemdash}{---}

\makeatother

  \providecommand{\corollaryname}{Corollary}
  \providecommand{\definitionname}{Definition}
  \providecommand{\examplename}{Example}
  \providecommand{\lemmaname}{Lemma}
  \providecommand{\propositionname}{Proposition}
  \providecommand{\remarkname}{Remark}
\providecommand{\theoremname}{Theorem}

\begin{document}
\global\long\def\C{\mathbb{C}}
\global\long\def\Cd{\C^{\delta}}
\global\long\def\Cprim{\C^{\delta,\circ}}
\global\long\def\Cdual{\C^{\delta,\bullet}}
\global\long\def\Od{\Omega^{\delta}}
\global\long\def\Oprim{\Omega^{\delta,\circ}}
\global\long\def\Odual{\Omega^{\delta,\bullet}}

\global\long\def\en{\mathcal{\varepsilon}}
\global\long\def\wone{\widehat{w}}
\global\long\def\wtwo{w}
\global\long\def\wo{\wone}
 \global\long\def\wt{\wtwo}
\global\long\def\cZ{\mathcal{Z}}

\global\long\def\tfixed{wired}
\global\long\def\tfree{free}
\global\long\def\tplus{plus}
\global\long\def\tminus{minus}

\global\long\def\fixed{\{\mathtt{\tfixed}\}}
\global\long\def\free{\{\mathtt{\tfree}\}}
\global\long\def\plus{\{\mathtt{\tplus}\} }
\global\long\def\minus{\left\{  \mathtt{\tminus}\right\}  }

\global\long\def\cvr{\mathrm{\varpi}}
\global\long\def\Ocvr{\Od_{\cvr}}
\global\long\def\Ccvr{\Cd_{\cvr}}
\global\long\def\double#1{\cvr(#1)}

\global\long\def\Cgr{\mathcal{C}}
\global\long\def\sCgr{\mathit{c}}

\global\long\def\Pf{\mathrm{Pf}}

\global\long\def\vv{v_{1},\dots,v_{n}}
\global\long\def\uu{u_{1},\dots,u_{m}}
\global\long\def\svv{\sigma_{v_{1}}\ldots\sigma_{v_{n}}}
\global\long\def\muu{\mu_{u_{1}}\dots\mu_{u_{m}}}

\global\long\def\Cutz#1#2{\Cgr_{\{#2\}}(#1)}

\global\long\def\offdiag#1{\mathcal{\mathring{C}}\vphantom{\mathcal{C}}{}_{\{\mathrm{diag}\}}^{#1}(\Ocvr)}
\global\long\def\withdiag#1{\overline{\mathcal{C}}\vphantom{\mathcal{C}}{}_{\{\mathrm{diag}\}}^{#1}(\Ocvr)}

\global\long\def\pesm{\psi^{[\eta]}\!,\en,\mu,\sigma}
\global\long\def\bcdd{\mathcal{\mathcal{B}}_{\mathrm{mono}}^{\delta}}
\global\long\def\sfix#1{\sigma_{\mathrm{fix}}^{#1}}

\global\long\def\any{\diamond}
\global\long\def\anyother{\triangleright}

\global\long\def\Opunc{\Omega^{\boxcircle}}
\global\long\def\Onopunc{\Omega^{\boxempty}}
\global\long\def\Oother{\Omega'}
\global\long\def\cvrother{\cvr'}
\global\long\def\Sother{S'}
\global\long\def\etaother{\eta'}

\global\long\def\const{\mathrm{const}}

\global\long\def\normLoc#1{\beta_{#1}^{\delta}}

\global\long\def\P{\mathsf{\mathbb{P}}}
 \global\long\def\E{\mathsf{\mathbb{E}}}
 \global\long\def\sF{\mathcal{F}}
 \global\long\def\ind{\mathbb{I}}

\global\long\def\R{\mathbb{R}}
 \global\long\def\Z{\mathbb{Z}}
 \global\long\def\N{\mathbb{N}}
 \global\long\def\Q{\mathbb{Q}}

\global\long\def\C{\mathbb{C}}
 \global\long\def\Rsphere{\overline{\C}}
 \global\long\def\re{\Re\mathfrak{e}}
 \global\long\def\im{\Im\mathfrak{m}}
 \global\long\def\arg{\mathrm{arg}}
 \global\long\def\i{\mathfrak{i}}
\global\long\def\eps{\varepsilon}
\global\long\def\lamb{\lambda}
\global\long\def\lambb{\bar{\lambda}}

\global\long\def\D{\mathbb{D}}
 \global\long\def\H{\mathbb{H}}

\global\long\def\dist{\mathrm{dist}}
 \global\long\def\reg{\mathrm{reg}}

\global\long\def\half{\frac{1}{2}}
 \global\long\def\sgn{\mathrm{sgn}}

\global\long\def\bdry{\partial}
 \global\long\def\cl#1{\overline{#1}}

\global\long\def\diam{\mathrm{diam}}
\global\long\def\corr#1{\overline{#1}}
\global\long\def\Corr#1#2{\E_{#1}(#2)}

\global\long\def\corr#1#2#3{\langle#1\rangle_{#2,#3}}
\global\long\def\pa{\partial}
\global\long\def\tto#1{\stackrel{#1}{\longrightarrow}}

\global\long\def\res{\text{res}}

\global\long\def\u{u}
\global\long\def\v{v}
 \global\long\def\z{z}
\global\long\def\mod{\;\mathrm{mod\;}}
\global\long\def\wind{\mathrm{w}}
\global\long\def\vz{z^{\bullet}}
\global\long\def\fz{z^{\circ}}
\global\long\def\CF{\mathfrak{C}}
\global\long\def\RPF{\text{I}}

\global\long\def\FFS#1#2#3{F_{#2}^{#3}(#1)}
\global\long\def\ds#1{\eta_{#1}}
\global\long\def\dbar{\overline{\partial}}
\global\long\def\dual#1{\left(#1\right)^{*}}

\global\long\def\bcond{\mathcal{B}}
\global\long\def\bcd{\mathcal{B}^{\delta}}
\global\long\def\Eod{\Od_{+}}
\global\long\def\lapv{\Delta^{\circ}}
\global\long\def\lapf{\Delta^{\bullet}}
\global\long\def\GammaR{\Gamma_{\R}}
\global\long\def\GammaiR{\Gamma_{i\R}}

\global\long\def\ccor#1{\langle#1\rangle}
\global\long\def\bar#1{\overline{#1}}
\global\long\def\anypsi{\psi^{*}}
\global\long\def\Op{\mathcal{O}}
\global\long\def\crossing{c}
\global\long\def\Fdual{F^{\mathrm{dual}}}

\global\long\def\zz{z_{1},z_{2}}
\newcommandx\norm[1][usedefault, addprefix=\global, 1=]{n_{#1}}
\newcommandx\tang[1][usedefault, addprefix=\global, 1=]{\tau_{#1}}
\global\long\def\Ocvrc{\Omega_{\cvr}}
\global\long\def\crad{\text{crad}}
\global\long\def\feta{f^{[\eta]}}
\global\long\def\T{\mathbb{T}}
\global\long\def\reg{\sharp}
\global\long\def\regg{*}
\global\long\def\coefA{\mathcal{A}}

\global\long\def\pbar#1{#1^{\star}}
\global\long\def\fdag{f^{\star}}
\global\long\def\CorrO#1{\langle#1\rangle_{\Omega}}
\global\long\def\formL{\mathcal{L}}
\global\long\def\appe{\approx_{\eps}}
\global\long\def\jayhat{\hat{j}}
\global\long\def\Csigma{C_{\sigma}}
\global\long\def\Ceps{C_{\eps}}
\global\long\def\Cpsi{C_{\psi}}
\global\long\def\Cmu{C_{\mu}}
\global\long\def\Cmu{C_{\mu}}
\global\long\def\umax{u_{\text{max}}}
\global\long\def\psistar#1{\psi^{\star}{}_{\hskip-5pt #1}}
\global\long\def\sqr#1{(#1)^{\frac{1}{2}}}

\global\long\def\scvr{\sigma_{\cvr}}
\global\long\def\Ob{\mathcal{O}}
 \global\long\def\Obs#1{\mathcal{O}\left[#1\right]}
\global\long\def\wcgr{\Cgr^{\diamondsuit}}
\global\long\def\bcgr{\Cgr^{\vardiamond}}
\global\long\def\dzmone#1#2{P_{#2}\left(#1\right)}
\global\long\def\dmsqrt#1#2{Q_{#2}\left(#1\right)}
\global\long\def\dsqrt#1#2{R_{#2}\left(#1\right)}
 \global\long\def\proj#1#2{\mathrm{Pr}_{#1}\left(#2\right)}
\global\long\def\projj#1#2{#1\re\left(\bar{#1}#2\right)}
\global\long\def\placket{\mathcal{S}}
\global\long\def\sflat{s_{\flat}}
\global\long\def\rother#1#2{\tilde{R}_{#1}(#2)}

\global\long\def\ssharp{c^{\sharp}}
\global\long\def\sflat{c^{\flat}}
\global\long\def\sany{c^{\diamond}}
\global\long\def\Ann{\mathbb{A}}
\global\long\def\spIn{\sigma_{\mathrm{in}}}
\global\long\def\spOut{\sigma_{\mathrm{out}}}
\global\long\def\S{\text{\ensuremath{\mathbb{S}}}}
\global\long\def\FK{\text{\text{FK}}}
\global\long\def\SOd{\text{\ensuremath{\hat{\Omega}^{\delta}}}}
\global\long\def\ArcHat{\text{\ensuremath{\hat{\gamma}^{\delta}}}}
\global\long\def\CCHat{\text{\ensuremath{\hat{\gamma}_{\text{in}}}}}
\global\long\def\SOmega{\text{\ensuremath{\hat{\Omega}}}}
\global\long\def\AOC{\text{\ensuremath{\gamma_{1}}}}
\global\long\def\ANC{\text{\ensuremath{\gamma_{2}}}}
\global\long\def\Neigh{\text{\ensuremath{\tilde{\Omega}}}}
\global\long\def\ArcNeig{\text{\ensuremath{\tilde{\gamma}}}}
\global\long\def\CCNeig{\text{\ensuremath{\tilde{\gamma}_{\text{in}}}}}
\global\long\def\nextc#1{#1_{\text{next}}}
\global\long\def\CCHatd{\text{\ensuremath{\hat{\gamma}_{\text{in}}^{\delta}}}}
\global\long\def\ProbeF{F_{\text{Ref}}^{\delta}}
\global\long\def\ProbeFc{f_{\text{Ref}}}
\global\long\def\ProbeFct{\tilde{f}_{\text{FK}}}
\global\long\def\NormC{C}
\global\long\def\Square{Q}
\global\long\def\Side#1{l_{#1}}
\global\long\def\normAny{\beta_{\delta}}
\global\long\def\normPf#1{\beta_{\delta}^{(#1)}}
\global\long\def\intFsquare#1#2{(\delta)\int^{#2}\im\left(#1dz\right)}
\global\long\def\intFQ#1#2{(\delta)\int_{#2}\im\left(#1dz\right)}
\global\long\def\eventA{A}
\global\long\def\eventAh{\hat{A}}
\global\long\def\eventCut{\mathcal{E}}
\global\long\def\disSign{\mathfrak{S}}
\global\long\def\constCC{K}

\title{Correlations of primary fields in the critical Ising model}

\author{D. Chelkak, C. Hongler, and K. Izyurov}
\begin{abstract}
We prove convergence of renormalized correlations of primary fields
\textendash{} i. e., spins, disorders, fermions and energy densities
\textendash{} in the scaling limit of the critical Ising model in
arbitrary finitely connected domains, with fixed (plus or minus) or
free boundary conditions, or mixture thereof. We describe the limits
of correlations in terms of solutions of Riemann boundary value problems,
and prove their conformal covariance. Moreover, we prove fusion rules,
or operator product expansions, which describe asymptotics of the
scaling limits of the correlations as some of the points collide together.
We give explicit formulae for correlations in the case of simply-connected
and doubly-connected domains. Our presentation is self-contained,
and the proofs are simplified as compared to the previous work where
particular cases are treated.
\end{abstract}

\thanks{K. I. is partially supported by Academy of Finland via academy project
``Critical phenomena in dimension 2'' and CoE ``Randomness and
structure''. C. H. is supported by ERC Starting grant ``CONSTAMIS''. }

\maketitle
\tableofcontents{}

\section{Introduction}

\allowdisplaybreaks

\subsection{Background}

The two-dimensional Ising model plays a very special role in statistical
mechanics and in mathematical physics in general. It was introduced
almost a century ago by W. Lenz in an attempt to understand one of
the most basic properties of metals, namely, the transition from ferromagnetic
to paramagnetic behavior above certain critical temperature, known
as Curie point. Being the first successful model of a phase transition,
it has become a central object of study in statistical mechanics,
thanks to the following features:
\begin{itemize}
\item it possesses a great deal of exact solvability, in the sense that
many relevant quantities can be computed explicitly. Among such quantities
are the value of the critical temperature for various lattices, free
energy per lattice site, spontaneous magnetization in sub-critical
regime in the full plane, two-point correlation function;
\item on the physics side, not only it turns out to be an adequate model
for the Curie point phase transition, it also has features representative
enough for second order phase transitions in general;
\item it has an extremely reach mathematical structure, and served as a
motivation for important advances in combinatorics, probability, algebra
and representations theory, analysis, integrable systems, theory of
isomonodromic deformations of linear differential equations, and other
fields of mathematics.
\end{itemize}
Indeed, the Ising model is the simplest and the ``most solvable''
of many similar models of planar statistical mechanics. It has the
simplest possible state space (each lattice site has only two possible
states, $\pm1$), the simplest possible interaction (nearest-neighbor
interactions only), and the probability measure is given by the Gibbs-Boltzmann
distribution. 

On the physics side, statistical mechanics has seen major breakthrough
about thirty five years ago, with the development of Conformal Field
Theory (CFT) of A.~Belavin, A. Polyakov and A. Zamolodchikov \cite{BPZ,Yellow_book}.
The starting point of the theory was the postulate that at criticality,
many lattice models of 2D statistical mechanics have conformally invariant
scaling limits. To any such limit, a representation of an infinite-dimensional
Lie algebra, called the Virasoro algebra, was associated. Representation
theory then allowed one to classify possible conformally invariant
theories, and identify theories corresponding to particular lattice
models. Moreover, conformal field theories turned out to be exactly
solvable in quite a strong sense: it was possible to compute scaling
exponents and correlation functions. The impact of Conformal Field
Theory on both physics and mathematics is immense and by no means
limited by the applications to statistical mechanics. 

The critical scaling limit of the Ising model was one of the prime
examples in the new theory as well, being described by a minimal CFT
of central charge $\frac{1}{2}$. Already in \cite{BPZ}, as an example
of an application of the new theory, the four-point correlation function
of the critical Ising model in the full plane was computed. Thus,
the conformal covariance hypothesis leads to a much more detailed
picture even in such an ``exactly solvable'' case as that of the
Ising model.

On the mathematical level, understanding of the conformal invariance
of scaling limits in statistical mechanics requires working in general
domains, which was not tackled until early 2000's. In 2000, Schramm
\cite{Schramm_SLE} proposed the celebrated Stochastic Loewner evolution
(SLE) as a way to describe lattice models in terms of random geometric
shapes. More precisely, the SLE is a conformally invariant random
process on a planar domain whose trace is a simple, self-touching,
or space-filling curve. It was conjectured, and in a number of cases
proven, that these processes are scaling limits of random curves,
or interfaces, arising in lattice models. This is the case for the
Ising model, where SLE$_{3}$ and SLE$_{\frac{16}{3}}$ appear in
the scaling limit. In fact, Schramm's SLE describes a single random
curve, but it was later used to define more elaborate objects, Conformal
Loop Ensembles (CLE), that describe the scaling limits of all the
interfaces in the model at once, and thus are, in some sense, ``full
scaling limits''. In the case of the Ising model, CLE$_{3}$ is the
limit in distribution of the collection of all interfaces in the low-temperature
expansion of the model, that is, of all dual edges separating the
vertices carrying $+1$ spin from the vertices carrying $-1$ spin
\cite{BenoistHongler}.

About the same time, R. Kenyon, S. Smirnov and others put forth the
discrete complex analysis as a powerful tool to actually prove existence
and conformal invariance of the scaling limits. This led, in particular,
to proofs of conformal invariance of height function in the dimer
model \cite{Kenyon_conf_inv}, conformal invariance of critical percolation
\cite{Smirnov_percolation}, and of the Ising model \cite{Smirnov_Towards,Smirnov_Ising,ChelkakSmirnov2};
in the latter two cases, convergence of interfaces to SLE was established
\cite{Camia_Newman_Percolation,ChelkakSmirnov_et_al}. 

In the past decade, a number of results was obtained in the direction
of rigorously establishing the predictions of CFT regarding the scaling
limits of lattice fields. In \cite{hongler_thesis}, the case of energy
correlations with locally monochromatic boundary conditions was treated.
In \cite{ChelkakHonglerIzyurov}, the scaling limits of spin correlations
in simply-connected domains with plus boundary conditions were computed;
these results were later extended to a number of other boundary conditions
\cite{ChelkakIzyurov,IzyurovFree}. In \cite{gheissari2019ising},
probabilities to find general ``local patterns'' of spins were considered.
In a number of works \cite{HonglerViklundKytolaCFT,ameen2020slit,ChelkakGlazmanSmirnov},
various aspects of the Ising CFT were elucidated on the lattice level. 

In the present paper, we complement this program by proving a general
results concerning scaling limits of all possible correlations of
\emph{primary fields} in the Ising model: \emph{spins, energies, fermions
and disorders}. Primary fields are principal building blocks of Conformal
fields theories, and their correlations possess the simplest possible
conformal covariance rules. More general correlations can be treated
by combining our results with \cite{HonglerViklundKytolaCFT}.

Our main results, Theorems \ref{thm: intro_2} and \ref{thm: Intro_3},
establish that a general correlation of the four fields listed above,
when renormalized by a suitable power of the mesh size, has a conformally
covariant scaling limit. The scaling limit, the ``continuous correlation
function'', is described in terms of solutions to special boundary
value problems for analytic functions; large portion of our paper
is devoted to systematic description of the ``continuous correlation
functions'' arising in the scaling limit. In the end of the day,
we provide an algorithmic procedure that allows to obtain completely
explicit results for \emph{any} such correlation in a simply connected
domain (with boundary condition given by any number of $+$, $-$,
or free boundary arcs), and we obtain new exact results in doubly
connected domains. Since correlations in the Ising model with suitable
boundary conditions are SLE martingales, our convergence theorems
yield new results about convergence of discrete interfaces to SLE
curves.

As another contribution towards establishing rigorously the CFT structure,
we prove \emph{fusion rules} for continuous correlation functions.
This is a set of asymptotic formulae describing behavior of correlation
functions as marked points merge together. Our proof works in arbitrary
multiply connected domains and hence does not use explicit formulae
for the correlations; instead, we are using analytic characterization
of the continuous correlation functions and uniqueness properties
of the boundary value problems. 

The proofs, as in the previous previous work \cite{hongler_thesis,HonglerSmirnov,ChelkakHonglerIzyurov}
which treated particular cases, relies on systematically leveraging
the discrete holomorphicity properties of fermionic observables. Our
exposition, however, is self-contained, and contains many simplifications
compared to those papers. First, we simplify the combinatorial part
by avoiding the use of contour representations altogether, instead
defining fermionic operators by the order-disorder formalism on Kadanoff
and Ceva \cite{kadanoff1971determination}. Second, there are several
simplifications in the discrete analysis part. The main of those comes
from the companion paper \cite{ChelkakIzyurovMahfouf} devoted to
the generalization of the analysis on the square grid to Baxter's
$Z$-invariant model on isoradial graphs: the Cauchy-type formula
for spinor observables given in Lemma \ref{lem: Cauchy_spinors},
that allows us to derive the necessary results for convergence of
discrete logarithmic derivatives of spin correlations in a very simple
way. Thus, we only use two explicitly constructed full plane discrete
analytic functions, the analogs of $z^{-1}$ and $z^{-\frac{1}{2}}$,
and we only need the interplay between their asymptotic expansions
and their values near the singularities. We follow the methods of
\cite{Dubedat} to construct both. A similar idea is used to simplify
the analysis of the behavior of s-holomorphic function near rough
boundaries, see Lemma \ref{lem: Clements_clever_lemma}. 

Recently, there's been a lot of work towards extending the results
on conformal invariance in the Ising model to more general setups.
S. C.~Park \cite{park2018massive,Park2021Fermionic} partially extended
convergence of spin and fermionic correlations to the massive scaling
limit setup. Another recent development is universality, which for
fermionic observables was established in \cite{ChelkakSmirnov2} within
the class of isoradial graphs with critical weights. This result was
recently extended to spin and energy correlations \cite{ChelkakIzyurovMahfouf}
and the results and methods of this paper provede a further extension
to all mixed correlations. A further direction towards the full universality
was suggested in \cite{chelkak2017ICM} via so-called s-embeddings;
the frist results on the convergence for fermionic observables in
this setup appeared in \cite{chelkak2020ising}. In our exposition
we strove to isolate the key ``blocks'' in the proofs, so that the
results would be directly extendable once the corresponding blocks
are provided in the relevant setup. 

\subsection{Main results: convergence of correlations}

\label{subsec: Main-results}

In the present paper, we discuss general results about convergence
of correlations in the critical Ising model on planar domains. The
observables included in our setup are spins, boundary spins, energies,
disorder variables and fermions. We work with arbitrary finitely connected
domains, with boundaries divided into finitely many arcs carrying
free or fixed boundary conditions. 

Let $\Omega$ be a bounded finitely connected domain, with no single-point
boundary components. The domain $\Omega$ will be equipped with a
subdivision of its boundary into two subsets $\fixed$ and $\free$,
each comprised of finitely many arcs. We view this subdivision as
\emph{boundary conditions}. Assume that $\Omega^{\delta}$ is a sequence
approximations to $\Omega$ by discrete (square lattice) domains of
mesh size $\delta$, and $\pa\Od$ is also subdivided into subsets
$\fixed$ and $\free$ that approximate the corresponding subsets
of $\pa\Omega$. Consider the critical Ising model on $\Od$ with
free boundary conditions on $\free$, and \emph{monochromatic along
boundary components} on $\fixed$, meaning that if $\gamma$ is a
connected component of $\pa\Od$, then all the spins on the $\gamma\cap\fixed$
are conditioned to be the same. While arguably not the most natural,
these boundary conditions are the easiest to deal with in the multiply
connected setup; we call them the \emph{standard boundary conditions}.
As we discuss below, other boundary conditions can be reduced to the
standard ones. 

Our first result will concern the \emph{spin-disorder correlations}
in the critical Ising model with standard boundary conditions; see
Section \ref{subsec: Ising_definitions} for the definitions. The
\emph{spin} variable $\sigma_{v}$ is indexed by a vertex of the domain
$\Od$, and is equal simply to the value of the spin at that vertex.
The disorder variables are defined as follows. Given a subset $\gamma$
of interior edged of the dual domain $\Odual$, we define the disorder
variable with respect to that subset by
\[
\mu_{\gamma}=\exp\left[-2\beta\sum_{(\v\v')\cap\gamma\neq\emptyset}\sigma_{\v}\sigma_{\v'}\right],
\]
where the sum is over all the edges $(vv')$ of the original domain
that intersect edges of $\gamma$. The object we are interested in
is the correlation 
\[
\E(\mu_{\gamma}\sigma_{v_{1}},\ldots\sigma_{v_{n}})=\E_{\Od}(\mu_{\gamma}\sigma_{v_{1}},\ldots\sigma_{v_{n}}).
\]
It is explained in Section \ref{sec: combinatorics} that up to sign,
this quantity only depends on 
\[
\{u_{1},\dots,u_{m}\}:=\pa\gamma\mod2,
\]
which is the set of vertices of $\Odual$ (i. e., centers of faces
of $\Od$) incident to an odd number of edges in $\gamma$. The sign,
in its turn, depends only on the homology class of $\gamma$ in $\Od\setminus\{v_{1},\dots,v_{n}\}$
modulo 2. Thus, we will use the notation 
\[
\E_{\Od}(\mu_{u_{1}}\cdot\dots\cdot\mu_{u_{m}}\sigma_{v_{1}},\ldots\sigma_{v_{n}}):=\E_{\Od}(\mu_{\gamma}\sigma_{v_{1}}\cdot\ldots\cdot\sigma_{v_{n}})
\]
where $\gamma$ is such that $\{u_{1},\dots,u_{m}\}=\pa\gamma\mod2.$
As it stands, this expression depends on the choice of $\gamma$,
and thus it is only defined up to sign. However, once this sign is
chosen for one configuration $u_{1},\dots,u_{m},v_{1},\dots v_{n}$,
there is a natural way to extend this choice as the points move around
in the lattice. After this extension, the expression $\E_{\Od}(\mu_{u_{1}}\cdot\dots\cdot\mu_{u_{m}}\sigma_{v_{1}},\ldots\sigma_{v_{n}})$
becomes a function on the universal cover of $(\Od)^{n,m}:=(\Od)^{\times n}\times\left((\Od)^{*}\right)^{\times m}$
that picks a $-1$ sign every time any of $v_{i}$ winds around a
$u_{j}$; in other words, it becomes a well-defined function on $(\Od)^{n,m}$
after multiplication by $\prod_{i,j}(v_{i}-u_{j})^{\frac{1}{2}}.$
Such functions, defined on a double cover of a certain graph or manifold
(in this case $(\Od)^{n,m}$) and changing sign between sheets, will
be called \emph{spinors}. We refer the reader to Section \ref{sec: combinatorics}
for a detailed discussion. 

Our first convergence result states that a spin-disorder correlation,
when properly renormalized, has a scaling limit. Following the physics
literature tradition, we use the $\ccor{\dots}$ notation to denote
this limit; for the purpose of this paper, this is just a \emph{function}
or a \emph{spinor} in the corresponding variables.
\begin{thm}
\label{thm: Intro_1}One has, as $\delta\to0$,
\begin{equation}
\delta^{-\frac{n+m}{8}}\E_{\Od}(\sigma_{v_{1}}\cdot\ldots\cdot\sigma_{v_{n}}\mu_{u_{1}}\cdot\ldots\cdot\mu_{u_{m}})=C_{\sigma}^{n}\cdot C_{\mu}^{m}\cdot\ccor{\sigma_{v_{1}},\ldots,\sigma_{v_{n}}\mu_{u_{1}}\ldots\mu_{u_{m}}}_{\Omega}+o(1),\label{eq: thm_1_intro}
\end{equation}
where $o(1)$ is uniform with respect to all configuration of points
$v_{1},\dots,v_{n}$, $\u_{1},\dots,u_{m}$ at a definite distance
from $\pa\Omega$ and from each other. 
\end{thm}

When $m$ is odd, we define the correlations in both sides of (\ref{eq: thm_1_intro})
to be identically zero; when $n$ is odd, these correlations also
vanish because of spin flip symmetry; otherwise, they are (generically)
non-trivial. The constants $\Csigma,\Cmu$ are given by (\ref{eq: constants}).
The quantity $\ccor{\sigma_{v_{1}},\ldots,\sigma_{v_{n}}\mu_{u_{1}}\ldots\mu_{u_{m}}}_{\Omega}$
is a real-valued spinor on the set of $(n+m)$-tuples of distinct
points in $\Omega$, which becomes a well-defined function when multiplied
by $\prod_{i,j}(v_{i}-u_{j})^{\frac{1}{2}}.$ Moreover, it is \emph{conformally
covariant}, that is, if $\varphi:\Omega\to\hat{\Omega}$ is a conformal
map from $\Omega$ onto $\hat{\Omega}$, then 
\[
\ccor{\sigma_{v_{1}}\ldots\sigma_{v_{n}}\mu_{u_{1}}\ldots\mu_{u_{m}}}_{\Omega}=\left(\prod_{i=1}^{n}|\varphi'(v_{i})|\right)^{\frac{1}{8}}\left(\prod_{i=1}^{m}|\varphi'(u_{j})|\right)^{\frac{1}{8}}\ccor{\sigma_{\varphi(v_{1})}\ldots\sigma_{\varphi(v_{n})}\mu_{\varphi(u_{1})}\ldots\mu_{\varphi(u_{m})}}_{\hat{\Omega}}.
\]

The spinors $\ccor{\sigma_{v_{1}},\ldots,\sigma_{v_{n}}\mu_{u_{1}}\ldots\mu_{u_{m}}}_{\Omega}$
are defined as special values of solutions to certain boundary value
problems. \textcolor{black}{In the case $\Omega$ is simply-connected
or doubly connected, they can be written down explicitly in terms
of algebraic and elliptic functions respectively, see Section \ref{sec:Explicit-formulae}.}

We also consider the situation when some of the points $v_{i}$ and
$u_{j}$ are grouped into pairs and placed at one lattice step from
each other. Given an edge $e$ of $\Od$, we define the \emph{energy
observable} as 
\[
\en_{e}:=\sqrt{2}\left(\sigma_{e_{+}^{\circ}}\sigma_{e_{-}^{\circ}}-\frac{1}{\sqrt{2}}\right),
\]
where $e_{+}^{\circ}$ and $e_{-}^{\circ}$ are two vertices incident
to $e$. The (lattice-dependent) constant $\frac{1}{\sqrt{2}}$ is
the limit of $\E_{\Od}(\sigma_{e_{+}^{\circ}}\sigma_{e_{-}^{\circ}})$
for $e$ in the bulk, i. e., the thermodynamic limit, and the factor
of $\sqrt{2}$ is introduced for convenience. Alternatively, the energy
observable $\en_{e}$ can be expressed in terms of disorder variables
at neighboring faces: 
\[
\en_{e}:=\sqrt{2}\left(\frac{1}{\sqrt{2}}-\mu_{e_{+}^{\bullet}}\mu_{e_{-}^{\bullet}}\right).
\]

Another observable is obtained by placing $v_{i}$ and $u_{j}$ next
to each other. If $z$ is a \emph{corner} of the lattice, that is,
a midpoint of a segment connecting a vertex to an incident dual vertex,
we denote the former by $z^{\circ}$ and the latter by $z^{\bullet}$;
the set of all corners of $\Od$ will be denoted by $\Cgr(\Od)$.
We would like to study the asymptotics of the expression of the form
\begin{equation}
\E_{\Od}(\sigma_{v_{1}}\cdot\dots\cdot\sigma_{v_{n}}\sigma_{z_{1}^{\circ}}\cdot\dots\cdot\sigma_{z_{k}^{\circ}}\mu_{u_{1}}\cdot\dots\cdot\mu_{u_{m}}\mu_{z_{1}^{\bullet}}\cdot\dots\cdot\mu_{z_{k}^{\bullet}})\label{eq: spin_disorder_ferm}
\end{equation}
when the points $v_{1},\dots,v_{n}\in\Od$, $u_{1},\dots,u_{m}\in\Odual$,
$\z_{1},\dots,z_{k}\in\Cgr(\Od)$ are at definite distance of each
other. This asymptotics will depend on the orientation of the corners
$z_{i}$ in the lattice. Moreover, note the complicated spinor structure
of (\ref{eq: spin_disorder_ferm}): provided the marked points are
separated, it changes its sign every time a corner $z_{i}$ winds
around a vertex, face, or edge, because of the winding of $z_{i}^{\bullet}$
around $z_{i}^{\circ}$. To keep track of the necessary information,
we introduce the Dirac spinor 
\begin{equation}
\ds z:=e^{\frac{i\pi}{4}}\delta^{\frac{1}{2}}\left(\vz-\fz\right)^{-\frac{1}{2}},\label{eq: def_eta_discrete}
\end{equation}
which is the spinor on the lattice $\Cgr(\Od)$ ramified at every
face of that lattice. The (real) fermionic observable is then defined
as the formal expression 
\[
\psi_{z}^{[\eta_{z}]}=\sigma_{z^{\circ}}\mu_{z^{\bullet}},
\]
interpreted as follows. Every time an expression of the form 
\begin{equation}
\E_{\Od}(\mu_{u_{1}}\cdot\dots\cdot\mu_{u_{m}}\sigma_{v_{1}}\cdot\dots\cdot\sigma_{v_{n}}\psi_{z_{1}}^{[\eta_{z_{1}}]}\cdot\dots\cdot\psi_{z_{k}}^{[\eta_{z_{k}}]})\label{eq: FSD}
\end{equation}
is written, it should be understood as equal to (\ref{eq: spin_disorder_ferm}),
with natural sign conventions as the marked point move around in the
lattice. Note that each $\eta_{z_{i}}$ is determined by $z_{i}$
up to sign; as $z_{i}$ moves, this sign changes according to (\ref{eq: def_eta_discrete}).

Now, let $\Op(\sigma,\mu,\en,\psi)$ %
{} stand for any expression of the form 
\begin{equation}
\sigma_{v_{1}}\cdot\ldots\cdot\sigma_{v_{n}}\mu_{u_{1}}\cdot\ldots\cdot\mu_{u_{m}}\en_{e_{1}}\cdot\ldots\cdot\en_{e_{s}}\psi_{z_{1}}^{[\eta_{z_{1}}]}\cdot\ldots\cdot\psi_{z_{k}}^{[\eta_{z_{k}}]}.\label{eq: op_any}
\end{equation}
 We have the following generalization of Theorem 1: 
\begin{thm}
\label{thm: intro_2}As $\delta\to0$, one has 
\begin{equation}
\delta^{-\Delta}\E_{\Od}(\Op(\sigma,\mu,\epsilon,\psi))=C\cdot\ccor{\Op(\sigma,\mu,\epsilon,\psi)}_{\Omega}+o(1),\label{eq: thm_2_intro}
\end{equation}
where $o(1)$ is uniform in the positions of marked points $v_{1},\dots,z_{k}$
at a definite distance from $\pa\Omega$ and from each other. The
exponent $\Delta$ is given by 
\[
\Delta=\frac{n}{8}+\frac{m}{8}+s+\frac{k}{2},
\]
and the constant is $C=C_{\sigma}^{n}C_{\mu}^{m}C_{\en}^{s}C_{\psi}^{k}$,
see (\ref{eq: constants}).
\end{thm}

Let us comment a bit more on the nature of $\ccor{\Op(\sigma,\mu,\epsilon,\psi)}_{\Omega}$.
It is a real multi-valued function of the positions of the marked
points $v_{1},\dots,z_{k}$ and of $k$ complex numbers $\eta_{1},\dots,\eta_{k}$
(the upper indices of $\psi$'s). Its spinor structure is the same
as that of 
\begin{equation}
\prod_{i,j=1}^{m,n}(u_{i}-v_{j})^{\frac{1}{2}}\prod_{i,j=1}^{m,k}(u_{i}-z_{j})^{\frac{1}{2}}\prod_{i,j=1}^{n,k}(v_{i}-z_{j})^{\frac{1}{2}},\label{eq: RS_all_around_all}
\end{equation}
that is, it picks a $-1$ sign every time a spin and a disorder, a
spin and a fermion, or a disorder and a fermion, go around each other.
Although in the discrete, the Dirac spinor (\ref{eq: def_eta_discrete})
only has eight possible values, the continuous correlation function
$\ccor{\Op(\sigma,\mu,\epsilon,\psi)}_{\Omega}$ is defined for arbitrary
complex values of $\eta_{i}$, and is, in fact, real linear in each
$\eta_{i}$; in particular, it changes the sign whenever $\eta_{i}$
does. Moreover, it happens to be anti-symmetric in $z_{1},\dots,z_{k}$
and symmetric in $v_{1},\dots,v_{n}$, in $u_{1},\dots,u_{m}$, and
in $e_{1},\dots,e_{s}$. Also, it it conformally covariant: if $\varphi$
is a conformal map from $\Omega$ onto $\Omega'$, then 
\[
\ccor{\Op(\sigma,\mu,\epsilon,\psi)}_{\Omega}=\CF\cdot\ccor{\Op(\hat{\sigma},\hat{\mu},\hat{\epsilon},\hat{\psi})}_{\hat{\Omega}},
\]
where the conformal factor $\CF$ is given by 
\[
\CF=\prod_{i=1}^{m}|\varphi'(u_{i})|^{\frac{1}{8}}\prod_{i=1}^{n}|\varphi'(v_{i})|^{\frac{1}{8}}\prod_{i=1}^{s}|\varphi'(e_{i})|,
\]
and $\Op(\hat{\sigma},\hat{\mu},\hat{\epsilon},\hat{\psi})$ is obtained
from $\Op(\sigma,\mu,\epsilon,\psi)$ by replacing every $v_{1},\dots,z_{k}$
with $\varphi(v_{1}),\dots,\varphi(z_{k})$, respectively, and each
$\eta_{i}$ by $\varphi'(z_{i})^{\frac{1}{2}}\eta_{i}.$ 

The lattice-dependent constant in Theorems \ref{thm: Intro_1} and
\ref{thm: intro_2} are given explicitly by 
\begin{equation}
\Cpsi=\left(\frac{2}{\pi}\right)^{\frac{1}{2}};\quad\Csigma=\Cmu=2^{\frac{1}{6}}e^{\frac{3}{2}\zeta'(-1)},\quad\Ceps=\frac{2}{\pi};\label{eq: constants}
\end{equation}
notice that in our convention, the mesh size $\delta$ stands for
half the diagonal of a lattice face.

Theorem \ref{thm: intro_2} concerns the Ising model with standard
boundary conditions. It is also possible to handle, perhaps, more
natural \emph{plus-minus-free} boundary conditions. For simplicity,
we formulate the result for the correlation $\E_{\Od}(\Op(\sigma,\en))$,
that is, we do not consider any disorders or fermions.

Assume that the boundary conditions on $\Omega$ are given by a subdivision
of $\pa\Omega$ into three subsets $\plus$, $\minus$ and $\free$,
each comprised of finitely many non-degenerate boundary arcs; we denote
such a subdivision by $\bcond$. The boundary of $\Od$ is equipped
with a similar subdivision $\bcd$ approximating $\bcond$, and we
consider the Ising model in $\Od$ with spins conditioned to be $+1$
on $\plus$ and $-1$ on $\minus$. 
\begin{thm}
\label{thm: Intro_3}We have, as $\delta\to0$, 
\begin{equation}
\delta^{-\Delta}\E_{\Od,\bcond^{\delta}}(\Op(\sigma,\epsilon))=C\cdot\ccor{\Op(\sigma,\epsilon)}_{\Omega,\bcond}+o(1),\label{eq: Intro_thm_3}
\end{equation}
uniformly with respect to positions of marked points at a definite
distance from each other and from $\pa\Omega$. Here 
\[
\Delta=\frac{n}{8}+s,
\]
where $n$ is the number of spins and $s$ is the number of energies
in $\Op(\sigma,\en)$. Moreover, if $\Omega$ is not simply connected,
and $\bcond_{1}$ and $\bcond_{2}$ are both boundary conditions on
$\Omega$ with the same set $\free$ and with the same collection
of boundary points separating $\plus$ from \textup{$\minus$}, and
$\bcond_{1_{,}2}^{\delta}$ are boundary conditions in $\Od$ approximating
$\bcond_{1,2}$, then the ratio of the partition functions has a limit
\[
\frac{Z(\Od,\bcd_{1})}{Z(\Od,\bcd_{2})}\to\RPF(\Omega,\bcond_{1},\bcond_{2}),
\]
 where $I(\Omega,\bcond_{1},\bcond_{2})$ is conformally invariant.
\end{thm}

Note that for simply-connected domains, the second claim of the theorem
would be vacuous, since $\bcond_{1}$ and $\bcond_{2}$ could only
differ by a global spin flip. In the simplest interesting example,
$\Omega$ is a doubly connected domain, $\bcond_{1}$ is $\plus$
everywhere, and $\bcond_{2}$ is $\plus$ on one boundary component
and $\minus$ on the other. We compute explicitly $I(\Omega,\bcond_{1},\bcond_{2})$
for this case in Theorem \ref{thm: ann_explicit}.

In fact, the results and techniques of this paper allow one to extend
the scope of theorems Theorems \ref{thm: Intro_1} \textendash{} \ref{thm: Intro_3},
in particular, as follows:
\begin{itemize}
\item it is possible to include disorders and fermionic observables in Theorem
\ref{thm: Intro_3}; the only additional technicality is that now,
the boundary conditions must be prescribed on the boundary of the
double cover of $\Omega$ ramified at each $u_{1},\dots,u_{m},z_{1}^{\bullet},\dots,z_{k}^{\bullet}$.
In other words, the result will depend on the homology class of the
disorder set $\gamma$ modulo 2, and the dependence will not be limited
to a sign change. 
\item we can also add spins on the regular parts of the free boundary arcs.
The exponent $\Delta$ in the statements of Theorem \ref{thm: intro_2}
and Theorem \ref{thm: Intro_3} will pick additional $\frac{1}{2}$
per each boundary spin.
\item we also can place disorders on free boundary arcs or at regular parts
of fixed boundary arcs. The disorders at free boundary arcs have conformal
weight $0$, and the disorders at regular parts of fixed boundary
arcs have conformal weight $\frac{1}{2}$.
\end{itemize}
To lighten the presentation, we decided not to include these into
the main results;%
{} we hope that the the interested reader will be able to find out how
to adapt our methods to these cases.

\subsection{Fusion rules for continuous correlations}

Our next topic is \emph{fusion rules, }also known as \emph{operator
product expansions} (OPE), describing the asymptotic behavior of continuous
correlation functions as some of the points merge together. Although
in simply-connected and double connected domains they can, at least
in principle, be checked explicitly, it is of interest to obtain a
more robust, conceptual proof, that has an added benefit of working
in multiply-connected domains:
\begin{thm}
The continuous correlation functions obtained as limits in Theorems
\ref{thm: Intro_1} \textendash{} \ref{thm: Intro_3} satisfy the
fusion rules (or operator product expansion) as predicted by Conformal
Field theory. See Section \ref{sec: ccor_fusion} for the precises
statements.
\end{thm}

The complete list of rules is given in Theorems \ref{thm: fusion rules-1},
\ref{thm: fusion rules-2}, \ref{thm: fusion rules-3} below; to give
one example, we prove that, as $\wo\to\wt$ away from other marked
points, one has 
\[
\CorrO{\Op\sigma_{\wone}\sigma_{\wtwo}}=|\wone-\wtwo|^{-\frac{1}{4}}\left(\CorrO{\Op}+\frac{1}{2}|\wone-\wtwo|\CorrO{\Op\en_{\wtwo}}+o(|\wone-\wtwo|)\right),
\]
where, as before, $\Op$ stands for any expression of the form (\ref{eq: op_any}). 

\subsection{Corollaries for SLE}

As direct application of the convergence results in the present paper,
we obtain new results on convergence of interfaces in the Ising model
in multiply connected domains. Let $\Omega$ be a multiply connected
domain with boundary conditions $\bcond$ given by subdivision of
the boundary into finitely many $\plus$, $\minus$, and $\free$
arcs. Consider all the boundary points $b_{1},\dots,b_{\hat{q}}$
separating a $\plus$ arc from a $\minus$ one, and add to them arbitrarily
one of the endpoints for each $\free$ boundary act that separates
a $\plus$ arc from a $\minus$ one. Let $b_{1},\dots,b_{q}$ be the
resulting collection of points. If $\pa\Omega$ is analytic near $b_{1},\dots,b_{q}$,
we associate to these data the \emph{SLE partition function} \emph{
\begin{equation}
Z(\Omega,\bcond):=\ccor{\psi_{b_{1}}\dots\psi_{b_{\hat{q}}}\psi_{b_{\hat{q}+1}}^{\flat}\dots\psi_{b_{q}}^{\flat}}_{\Omega}\cdot\left(\sum_{\bcond'}I(\Omega,\bcond',\bcond)\right)^{-1}.\label{eq: SLE_partition_function}
\end{equation}
}Here the first factor is the limit of (normalized) continuous fermion
correlation at bulk points tending to $b_{1},\dots,b_{q}$ (see Definition
\ref{def: corr_flat}); $I(\Omega,\bcond,\bcond')$ is as in Theorem
\ref{thm: Intro_3}, and the sum is over all boundary conditions $\bcond'$
that can be obtained from $\bcond$ by flipping all the spins along
one or several connected components of $\pa\Omega$. The second term
can be also written as $\sum_{S}\alpha_{S}^{\bcond}\corr{\sigma_{S}}{\Omega}{\tilde{\bcond}}$,
see Section \ref{subsec: Def_corr_cont} for definitions.
\begin{cor}
\label{cor: SLE}Let $\Omega$ be a finitely connected domain with
boundary conditions $\bcond$ as above, and such that $\hat{q}>0$,
and let $\Od,\bcond^{\delta}$ be a discrete approximation to $(\Omega,\bcond).$
Then, the interface between pluses and minuses in $\Od$ starting
at $b_{1}^{\delta}$ converges locally to $SLE_{3}$ in $\Omega$
with partition function $Z(\Omega,\bcond)$ given by (\ref{eq: SLE_partition_function}).
\end{cor}

The statement of the last Corollary means that the driving process
of the chordal Loewner chain describing the interface after a mapping
to an analytic subdomain of the upper half-plane satisfies $b_{1}(t)=\sqrt{3}B_{t}+3\partial_{b_{1}(t)}\log(Z(g_{t}(\Omega),g_{t}(\bcond))$;
we refer the reader to \cite{IzyurovMconn} for details. The result
of Corollary \ref{cor: SLE} is known for simply connected domains,
see \cite{IzyurovFree}. In \cite{IzyurovMconn}, multiply connected
domains were considered; however, only \emph{locally monochromatic}
boundary conditions were treated, which are $\fixed$ boundary conditions
with disorder insertions. The simplest new examples covered by Corollary
\ref{cor: SLE} are doubly connected domains with Dobrushin boundary
conditions on one boundary component and $\plus$ or $\free$ on the
other one. We give explicit expression for the partition function
in those cases in Example \ref{exa: annulus-SLE-PF}.

In \cite{IzyurovMconn,IzyurovFree}, regularity assumptions were imposed
on $\pa\Omega$ near $b_{1},\dots,b_{\hat{q}}$. This is a disadvantage
from the multiple SLE point of view, since when several interfaces
are grown from $b_{1},\dots,b_{\hat{q}}$ simultaneously or interchangeably,
the regularity is lost. Here we observe that these assumptions are
superfluous, so, we do not impose them in Corollary \ref{cor: SLE}.
The idea is that it is enough to prove convergence of martingale observable
in $\Od\setminus\gamma_{t}$ under \emph{some normalization} that
does not depend on $\gamma_{t}$, and this can be achieved by constructing
a suitable ratio of observables in different domains. A similar idea
was used in \cite{Hongler_Kytola}.

\subsection{Outline of the proofs and the organization of the paper}

Section \ref{sec: combinatorics} is devoted to the combinatorics
of the Ising model in the discrete. After introducing notation, we
define the disorder variables and discuss their spinor properties.
Then, we introduce the fermionic spinor observable and establish its
basic properties, namely, the discrete holomorphicity and the boundary
conditions. We then introduce the multi-point observable and discuss
its properties - mainly, the Pfaffian identity. Finally, we express
all the (discrete) correlations featuring in Theorems \ref{thm: Intro_1}
\textendash{} \ref{thm: Intro_3} in terms of spin correlations and
fermionic observables, possibly with fermions adjacent to spins or
to the boundary.

In Section \ref{sec: ccor}, we introduce the Riemann boundary value
problem that feature in the definition of continuous fermionic observables.
We collect their properties and develop a convenient formalism to
formulate all the necessary convergence results for fermionic observables,
and to give the definitions of the correlations in the continuum. 

In Section \ref{sec: proofs_of_convergence_theorems}, we prove convergence
theorems for fermionic observables that were formulated in Section
\ref{sec: ccor}. First, using discrete holomorphicity with respect
to the position of fermions and the uniqueness in the Riemann boundary
value problem, we prove ``convergence in the bulk'', i. e., when
all marked points are at a definite distance from each other and the
boundary. Second, we prove that if a discrete holomorphic function
has a ``discrete singularity'' of a type that arises from the insertion
of spin, disorder, of another fermion, and converges uniformly on
compact subsets of a punctured neighborhood of the singularity, then
its value next to the singularity must have a particular asymptotics
governed by the expansion terms of the limit around the singularity.
We also prove corresponding results for boundary points. 

In Section \ref{sec: corr-continuum}, we give general definitions
of the continuous correlations functions, i. e., the quantities appearing
in the right-rand sides of Theorems \ref{thm: Intro_1} \textendash{}
\ref{thm: Intro_3}, and prove those theorems. To a large extend,
these theorems become corollaries of the convergence theorems for
fermionic observables, formulated in Section \ref{sec: ccor} and
proven in Section \ref{sec: proofs_of_convergence_theorems}. The
only missing ingredient is a probabilistic argument allowing one to
fix the normalization of spin correlation.

In Section \ref{sec: ccor_fusion}, we formulate and prove fusion
rules. 

In Section \ref{sec:Explicit-formulae}, we explicitly solve the Riemann
boundary value problems for the continuous fermionic observable in
the upper half-plane and in annuli, thus providing explicit expressions
for continuous correlation functions in simply- and doubly connected
domains.

In the Appendix, we prove some technical results used in the paper.
Those results are mostly contained in the literature, but not it the
form that is most suitable for our purposes, or not with the simplest
proofs. We also prove Corollary \ref{cor: SLE}. %
.

\subsection*{Acknowledgements}

K. I. is partially supported by Academy of Finland via academy project
``Critical phenomena in dimension 2'' and Center of Excellence ``Randomness
and structure''. C. H. is supported by the ERC Starting grant CONSTAMIS.

\newpage{}

\section{Combinatorics of the critical Ising model on $\protect\Z^{2}$}

\label{sec: combinatorics}

\subsection{Discrete domains, double covers and the Ising model}

\label{subsec: Ising_definitions} We start by fixing the notation
and recalling basic definitions. A \emph{discrete domain} $\Od$ of
mesh size $\delta$ is a finite connected subset of the square lattice
\[
\Cd=\Cprim:=\sqrt{2}e^{\frac{i\pi}{4}}\delta\cdot\Z^{2}.
\]
This lattice is the standard square lattice, rotated by $45^{\circ}$
and scaled, so that $\delta$ denotes half the diagonal of a lattice
face; we use the notation $\Cprim$ in rare situations when we want
to emphasize that we consider the primal lattice. By the \emph{dual
domain} $\Odual$ of $\Od=\Oprim$ we mean the set of all vertices
of the dual lattice 
\[
\Cdual:=\Cd+\delta
\]
that are incident to at least one vertex in $\Od$. By the \emph{boundary}
$\partial\Od$ of $\Od$, we mean the set of edges of $\Cdual$ that
separate a vertex $v\in\Od$ from a vertex $v'\in\Cd\setminus\Od$.
We consider each domain $\Od$ to be equipped with a partition of
$\pa\Od$ into two subsets $\fixed$ and $\free$, which we understand
as boundary conditions and omit from the notation. Throughout this
section, the domain $\Od$ and the boundary conditions will be fixed. 

Below we assume that $\pa\Od$ is a disjoint collection of simple
loops (drawn on the dual lattice), and, moreover, that no connected
component of $\Cd\setminus\Od$ can be disconnected by deleting one
edge. These assumptions are made for the sole purpose of a notational
convenience when extending domains and can be lifted by the cost of
heavier definitions.

We denote by $\free^{\bullet}\subset\Odual$ the set of all dual vertices
incident to the edges of $\free$, and by $\fixed^{\circ}$ the set
of all vertices of $\Cd\setminus\Od$ incident to the edges of $\fixed$.
By a \emph{free boundary arc}, we mean a connected component of $\free$.

The \emph{Ising model with standard boundary conditions} on $\Od$
is a random assignment of spins $\sigma_{v}=\pm1$ to the vertices
of $\Cd$, subject to the condition that $\sigma_{v}$ is constant
on each connected component of the complement $\Cd\setminus\Od$.
The probability of a spin configuration $\sigma$ is given by 
\[
\P[\sigma]=\cZ^{-1}\exp\biggl[\beta\sum_{(vv')}\sigma_{v}\sigma_{v'}\biggr],
\]
where the sum is over all pairs $v\sim v'$ of nearest-neighbor vertices
of $\Cd$ such that at least one of $v,v'$ belongs to $\Od$, and
$(vv')\cap\free=\emptyset$\textit{. }As usual, $\cZ$ stands for
the partition function of the model: $\cZ:=\sum_{\sigma}\exp\big[\beta\sum_{(vv')}\sigma_{v}\sigma_{v'}\big]$.
Throughout this paper, we will be dealing with the Ising model at
its critical point, that is, we set 
\begin{equation}
\beta:=\beta_{\mathrm{crit}}=\tfrac{1}{2}\log\big(\sqrt{2}+1\big).\label{eq: beta-crit}
\end{equation}
Note that the above probability measure is invariant under the global
spin flip $\sigma\mapsto-\sigma.$ 

A \emph{double cover} $\cvr$ of a discrete domain $\Omega^{\delta}$
is a two-to-one mapping 
\[
\cvr:\Ocvr\to\Od
\]
that locally preserves the graph structure (both edges and non-edges).
Given a double cover $\cvr$ and a point $z\in\Ocvr$, we denote by
$z^{*}$ another fiber over the same point of $\Od$. Note that a
double cover of a discrete domain is determined up to isomorphism
by the set of \textit{faces }at which it is ramified (``branching
points''). We therefore introduce a notation 
\[
\double{\uu}
\]
for the double cover ramified over faces $u_{1},\ldots,u_{m}\in\Odual$.
We do allow repetitions and use the natural convention that if a face
appears an odd (resp., even) number of times among $\uu$, then $\double{\uu}$
is ramified (resp., is not ramified) over this face. If $\cvr'=\double{u'_{1},\dots,u'_{m'}}$
and $\cvr''=\double{u''_{1},\dots,u''_{m''}}$, we introduce the notation
\[
\cvr'\cdot\cvr'':=\double{u'_{1},\dots,u'_{m'},u''_{1},\dots,u''_{m''}}.
\]

The following ``hands-on'' way to view double covers will be useful.
A subset $\gamma$ of edges of $\Odual$ is called a \textit{branch
cut} for $\cvr$ if $(\pa\gamma\mod2$) coincides with the set of
ramification points of $\cvr$. A branch cut induces an equivalence
relation on vertices of $\Ocvr$, with two vertices being ``on the
same sheet'' if some (equivalently, any) nearest-neighbor lattice
path connecting them in $\Ocvr$ intersects $\gamma$ an even number
of times. Given a branch cut, a vertex of $\Ocvr$ can be encoded
by a vertex of $\Od$ with an additional bit of information prescribing
the sheet, see Fig. \ref{fig: branch_cut}

\begin{figure}
\includegraphics[width=0.8\textwidth]{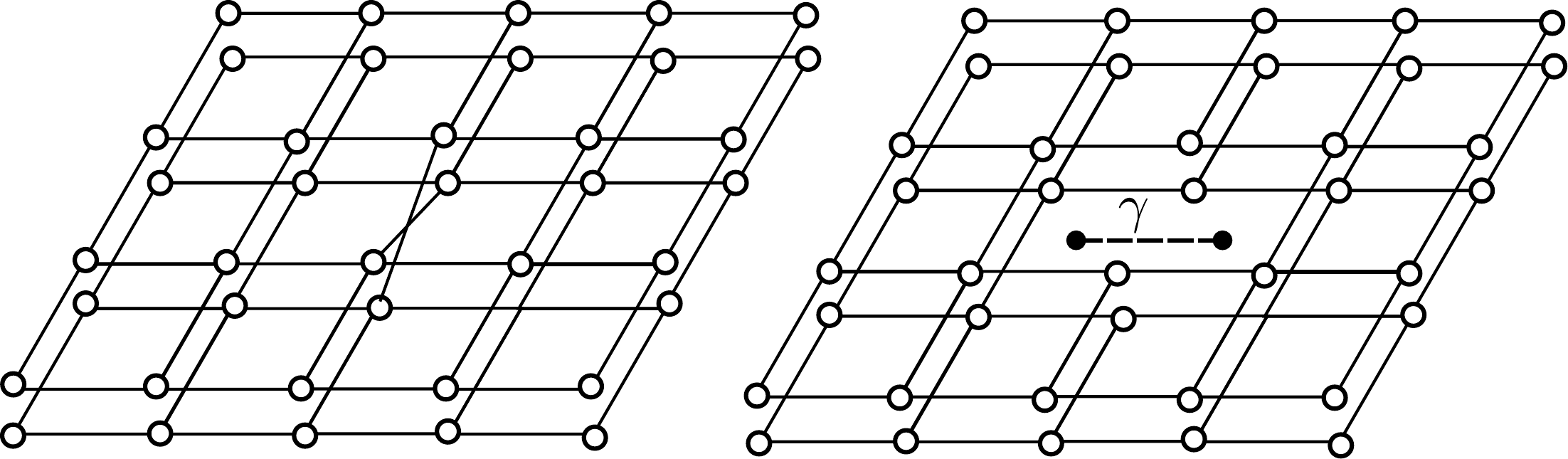}\caption{An example of a double cover of a $4\times5$ rectangle ramified over
two central faces and its decomposition into two sheets with a branch
cut.\label{fig: branch_cut}}
\end{figure}

The \emph{corner graph} $\Cgr(\Cd)$ of $\Cd$ is the square lattice\emph{
\[
\Cgr(\Cd):=\tfrac{1}{2}(\Cd+\delta)
\]
}whose vertices are corners (i. e., midpoints of segments connecting
a vertex of $\Cprim$ to an incident vertex of $\Cdual$). We denote
by $\Cgr(\Od)$ the subgraph of $\Cgr(\Cd)$ containing all the corners
incident to the vertices in $\Od$ or edges of $\fixed.$ By $\Cgr$($\Od_{\cvr}$)
we mean the double cover of $\Cgr(\Od)$ ramified at the ramification
points of $\cvr$. Given $z\in\Cgr(\Cd),$ we define $\fz\in\Cprim$
and $\vz\in\Cdual$ to be the vertex (respectively, the face) of $\Cd=\Cprim$
incident to $z$.

A \emph{$\cvr$-spinor} on $\Od$ is a function $f:\Ocvr\to\C$ such
that $f(v)=-f(v^{*})$ for all $v\in\Ocvr$. For example, 
\begin{equation}
f(v):=\sqrt{(v-u_{1})(v-u_{2})},\label{eq: x_spinor-example}
\end{equation}
where $v\in\Cd$ and $u_{1},u_{2}\in\Cdual$, is a $\double{u_{1},u_{2}}$-spinor
on $\Cd$. Another example is given by the Dirac spinor $\eta_{z}$,
see (\ref{eq: def_eta_discrete}), which is a spinor on the double
cover of $\Cgr(\Cd)$ ramified over every face of $\Cgr(\Cd)$ (i.e.,
over all vertices, faces and edges of $\Cd$).

We will often work with ``multi-valued functions'' defined on the
(Cartesian products of the) graphs introduced above. By a multi-valued
function of vertices $\mbox{\ensuremath{v_{1}},\ensuremath{\dots},\ensuremath{v_{n}\in\Oprim}}$,
faces $\mbox{\ensuremath{u_{1}},\ensuremath{\dots},\ensuremath{u_{m}\in\Odual}}$
and corners $\mbox{\ensuremath{z_{1}},\ensuremath{\dots},\ensuremath{z_{k}\in\Cgr}(\ensuremath{\Od})}$
of a discrete domain $\Od$, we mean a function defined on the universal
cover of (a subset of) the set 
\begin{equation}
(\Od)_{\circ,\bullet,\sCgr}^{n,m,k}:=(\Oprim)^{\times n}\times(\Odual)^{\times m}\times\Cgr(\Od)^{\times k},\label{eq: direct_product}
\end{equation}
endowed with the direct product graph structure (sometimes we also
use the shorthand notation $(\Od)_{\circ,\bullet}^{n,m}:=(\Od)_{\circ,\bullet,\sCgr}^{n,m,0}$
etc). In other words, to define a multi-valued function, one has to
prescribe its value at some point of (\ref{eq: direct_product}) and
define how it changes as one of its arguments moves to a neighboring
position on the corresponding graph. In fact, all our multi-valued
functions will be ``double-valued'', with two values differing only
by the sign. 
\begin{rem}
\label{rem: dbl-covers-general} It is worth noting that all the multi-valued
functions on $(\Od)_{\circ,\bullet,\sCgr}^{n,m,k}$ that we discuss
below can/should be viewed as spinors on relevant double covers of
their domains of definition. However, we want to emphasize that the
identifications of such double covers are not always trivial and include,
in particular, the passage from symmetric correlation functions (spin-disorders)
to anti-symmetric ones (fermions), see Lemma \ref{lem: anti-symmetry}.
\end{rem}

\subsection{The spin-disorder correlations }

\label{subsec: spin-disorder} We extend the definition of the Ising
model to double covers of $\Od$ as follows. Let $u_{1},\dots,u_{m}\in\Odual$
and $\cvr=\double{u_{1},\dots,u_{m}}$. 
\begin{defn}
The $\cvr$-Ising model on $\Od$ is a random $\cvr$-spinor $\sigma:\Ccvr\to\left\{ \pm1\right\} $,
subject to the condition that the spins $\sigma_{v}$ are constant
on each connected component of $\Ccvr\setminus\Ocvr$, with the probability
measure given by 
\[
\P_{\cvr}(\sigma)=\cZ_{\cvr}^{-1}\exp\biggl[\frac{\beta}{2}\sum_{(vv')}\sigma_{\v}\sigma_{\v'}\biggr],
\]
where the sum is taken over all pairs $v\sim v'$ of nearest-neighbor
vertices of $\Ccvr$ such that at least one of $v,v'$ belongs to
$\Ocvr$, and $(vv')\cap\free=\emptyset$; the additional factor $\frac{1}{2}$
is added to handle the two equal contributions $\sigma_{v}\sigma_{v'}=\sigma_{v^{*}}\sigma_{v'^{*}}$
coming from the two sheets of $\Ocvr$. As above, $\cZ_{\cvr}$ is
the partition function that makes the total mass equal to one. 
\end{defn}

\begin{rem}
Note that one recovers the original Ising model on $\Od$ if the double
cover $\cvr$ is trivial (i.e., if $\Ocvr$ is a disjoint union of
two copies of $\Od$). 
\end{rem}

Naturally, the spin correlations $\E_{\cvr}\left(\sigma_{\v_{1}}\ldots\sigma_{\v_{n}}\right)$
in the $\cvr$-Ising model on $\Ocvr$ are also $\cvr$-spinors, in
each of $\v_{q}$. It will be also useful to rewrite those correlations
in terms of the usual Ising model on $\Od$.
\begin{lem}
\label{lem: branch cut}Let $\gamma$ be any branch cut for $\cvr$,
and assume that $\v_{1},\ldots,\v_{n}$ are one the same sheet of
$\Ocvr$ with respect to $\gamma$. Then, 
\[
\E_{\cvr}\left(\sigma_{\v_{1}}\ldots\sigma_{\v_{n}}\right)\ =\ \cZ\cZ_{\cvr}^{-1}\cdot\E\left(\mu_{\gamma}\sigma_{\cvr(v_{1})}\ldots\sigma_{\cvr(v_{n})}\right),
\]
 where 
\[
\mu_{\gamma}:=\exp\biggl[-2\beta\sum_{(\v\v'):\thinspace(\v\v')\cap\gamma\neq\emptyset}\sigma_{\v}\sigma_{\v'}\biggr]
\]
 and the sum is taken over the subset of the same set of edges $(vv')$
of $\Cd$ as usual: at least one of $v,v'$ belongs to $\Od$ and
$(vv')\cap\free=\emptyset$.
\end{lem}

\begin{proof}
Note that, once a branch cut $\gamma$ is chosen, there is a bijection
between $\cvr$-spinors $\sigma$ and \emph{functions} $\sigma:\Cd\to\left\{ \pm1\right\} $,
given by the restriction of a spinor to a sheet. This bijection preserves
weights, except that the contribution of the edges intersecting $\gamma$
to the energy comes with the opposite sign. Therefore,
\[
\begin{aligned}\cZ_{\cvr}\cdot\E_{\cvr}\left(\sigma_{\v_{1}}\ldots\sigma_{\v_{n}}\right)\  & =\sum_{\sigma:\thinspace\cvr\operatorname-\mathrm{spinors}}\sigma_{\v_{1}}\dots\sigma_{\v_{n}}e^{\frac{\beta}{2}\sum_{(vv')}\sigma_{\v}\sigma_{\v'}}\\
 & =\sum_{\sigma:\thinspace\mathrm{functions}}\sigma_{\cvr(\v_{1})}\dots\sigma_{\cvr(v_{n})}e^{\beta\sum_{(vv')}\sigma_{\v}\sigma_{\v'}\:-\ 2\beta\sum_{(\v\v'):\thinspace(vv')\cap\gamma\neq\emptyset}\sigma_{v}\sigma_{\v'}}\\
 & =\ \cZ\cdot\E\left(\mu_{\gamma}\sigma_{\cvr(v_{1})}\ldots\sigma_{\cvr(v_{n})}\right),
\end{aligned}
\]
where the summation goes over the edges $(vv')$ of the double cover
$\Ccvr$ in the first line and the edges of $\Cd$ in the second;
this is why the additional factor $\frac{1}{2}$ (dis)appears.
\end{proof}
\begin{rem}
\label{rem: cover_auto}Every double cover has an automorphism $z\mapsto z^{*}$.
Our choice of boundary conditions ensures that all the correlations
we consider, e.g., $\E_{\cvr}\left(\sigma_{\v_{1}}\ldots\sigma_{\v_{n}}\right)$,
are invariant under this automorphism. In particular, all such correlations
vanish if $n$ is odd.
\end{rem}

So far we have defined the correlation functions
\begin{equation}
G_{\double{\uu}}(\vv):=\E_{\double{\uu}}\left(\svv\right)\label{eq: x_G-cvr-def}
\end{equation}
for each double cover $\cvr=\double{u_{1},\ldots,u_{m}}$ of $\Od$
separately. In other words, at the moment (\ref{eq: x_G-cvr-def})
is a collection of $\cvr$-spinors (in each of $v_{q}$), which are
invariant under the automorphisms $\left(v_{1},\ldots,v_{n}\right)\mapsto\left(v_{1}^{*},\ldots,v_{n}^{*}\right)$
of the corresponding double covers. We now explain how such a collection
can be combined into a \emph{single} multi-valued function of $\vv\in\Od=\Oprim$
and $\uu\in\Odual$. To start with, let $\tilde{v}_{1},\ldots,\tilde{\v}_{n}\in\Ocvr$
be any lifts of $v_{1},\ldots,v_{n}$ to the double cover $\cvr=\double{u_{1},\ldots,u_{m}}$
of $\Od$. Note that there is a canonical way to extend $G_{\cvr}\left(\tilde{v}_{1},\ldots,\tilde{\v}_{n}\right)$
along any nearest-neighbor path in $\left(\Cd\right)_{\circ,\bullet}^{n,m}$.
Namely, if a step of the path consists in moving one of the vertices
$v_{q}$ (in $\Od$), then one just lifts this move to the double
cover $\Ocvr$. If a face $u_{p}$ moves to a neighboring position
$u'_{p}$, then both $\cvr=\double{u_{1},\ldots,u_{p},\ldots,u_{m}}$
and $\cvr'=\double{u,\ldots,u'_{p},\ldots,u_{m}}$ can be also viewed
as isomorphic double covers of $\Cd\setminus\{(u_{p}u'_{p})^{\circ}\}$,
where $(u_{p}u'_{p})^{\circ}$ denotes the edge of $\Cd=\Cprim$ that
intersects $(u_{p}u'_{p})$. Through these identifications and the
isomorphism, the points $\tilde{v}_{1},\ldots,\tilde{\v}_{n}\in\Ocvr$
can be pulled to $v'_{1},\ldots,v'_{n}\in\Od_{\cvr'}$. Moreover,
the symmetry of the correlation functions under the global spin flip
ensures that $G_{\cvr'}\left(v'_{1},\ldots,v'_{n}\right)$ does not
depend on the choice of the isomorphism.

Another way to phrase this is as follows: if a branch cut $\gamma$
is chosen for $\cvr$ and the lift $\tilde{v}_{1},\ldots,\tilde{\v}_{n}$
of $v_{1},\ldots,v_{n}$ to $\Ocvr$ is prescribed by assigning a
label $1$ or $2$ (the sheet number) to each $v_{q}$, then we define
the lift $v'_{1},\ldots,v'_{n}$ to $\Od_{\cvr'}$ by taking the symmetric
difference of $\gamma$ and the edge $(u_{p}u'_{p})$ as the new branch
cut $\gamma'$, and keeping the labeling the same.
\begin{lem}
\label{lem: spinor-in-all-variables}Let $n$ be even. If $G_{\double{\uu}}\left(\vv\right)$
is extended along a closed loop in $(\Cd)_{\circ,\bullet}^{n,m}$
by the rule explained above, then its starting and final values differ
by the sign 
\begin{equation}
(-1)^{\sum_{p=1}^{m}\sum_{q=1}^{n}\wind(\v_{q}-\u_{p})},\label{eq: global sign}
\end{equation}
where $\wind(v_{q}-u_{p})$ is the winding number of the vector $v_{q}-u_{p}$
along the loop (i. e., the number of times it rotates around the origin).
\end{lem}

\begin{proof}
Clearly, the sign change is entirely determined by the loop (more
precisely, by how the points $v_{1},\ldots,v_{n}$ change sheets after
tracing the loop) and hence it is independent of the spinors $G_{\cvr}$.
On the other hand, for the particular choice $G_{\cvr}:=\prod_{p=1}^{m}\prod_{q=1}^{n}\sqrt{v_{q}-u_{p}}$
the statement is obvious.
\end{proof}
It follows from Lemma \ref{lem: branch cut} that, up to the sign,
the quantity 
\begin{equation}
\E_{\Od}(\mu_{\gamma}\svv)=\cZ_{\cvr}\cZ^{-1}\cdot\E_{\double{\uu}}\left(\svv\right)\label{eq: E_sigma_mu_gamma}
\end{equation}
only depends on $\{\uu\}:=\pa\gamma\mod2$. We therefore adopt the
following notation.
\begin{defn}
We denote by
\[
\E_{\Od}(\muu\svv)
\]
the quantity (\ref{eq: E_sigma_mu_gamma}), where $\{\uu\}:=\pa\gamma\mod2$,
understood as a multi-valued function of $\uu,\vv$ according to Lemma
\ref{lem: spinor-in-all-variables}. 
\end{defn}

\subsection{The fermionic observable}

\label{subsec: fermions-discrete} We now introduce the key tool for
the study of the nearest-neighbor Ising model in 2D: the (two-point)
fermionic observable. Though it does not look as natural as spin (or
spin-disorder) correlations, it has a huge advantage of satisfying
a very simple propagation equation (e.g., see \cite[ and references therein]{ChelkakCimasoniKassel}),
which makes it amenable to the analysis in various setups. In the
context of our paper, this equation can be equivalently formulated
as the so-called \emph{s-holomorphicity} property, see Definition
\ref{def: s-hol} and Lemma \ref{lem: s-hol} below.
\begin{defn}
\label{def: obs_discrete} Let $n$ be even. We define the two-point
fermionic observable, as a multi-valued function of $n$ vertices
$\vv\in\Cd$ and two corners $z_{1},z_{2}\in\Cgr(\Od)$, by the formula
\[
F_{\Od,\double{\vv}}(z_{1,}z_{2}):=\ds{z_{1}}\ds{z_{2}}\frac{\E_{\Od}(\mu_{\vz_{1}}\mu_{\vz_{2}}\sigma_{\fz_{1}}\sigma_{\fz_{2}}\svv)}{\E_{\Od}(\svv)},
\]
 where $\eta_{z}$ is the Dirac spinor on $\Cgr(\Cd)$ given by (\ref{eq: def_eta_discrete}). 
\end{defn}

Until the end of this section, we will assume that $\Od$ and $\vv$
are fixed. We abbreviate $\cvr:=\double{\vv}$, viewed as a double
cover of $\Cgr(\Od),$ and omit $\Od$ and $\cvr$ from the notation,
thus $F(z_{1},z_{2})=F_{\Od,\double{\vv}}(z_{1},z_{2})$.
\begin{lem}
\label{lem: Obs_spinor}For each $z_{1}\in\Cgr(\Od),$ the multi-valued
function $F(z_{1},\cdot)$ defined above is a spinor on $\Cgr(\Od)$
ramified over $\vz_{1},\fz_{1},\vv$, viewed as faces of $\Cgr(\Cd)$.
\end{lem}

\begin{proof}
It follows directly from Lemma \ref{lem: spinor-in-all-variables}
that the expectation 
\[
\E(\mu_{\vz_{1}}\mu_{\vz_{2}}\sigma_{\fz_{1}}\sigma_{\fz_{2}}\svv)
\]
is a spinor in $z_{2}$ ramified over each face of $\Cgr(\Od)$ (because
of the braiding of $\vz_{2}$ around $\fz_{2}$), except at those
corresponding to $\vz_{1},\fz_{1},\vv$. Taking into account that
the spinor $\ds{z_{2}}$ is ramified at every face of the corner graph,
we get the result.
\end{proof}
\begin{defn}
\label{def: s-hol}Assume that $\cvr_{\sCgr}$ is a double cover of
$\Cgr(\Od)$, such that all ramification points of $\cvr_{\sCgr}$
are either vertices of $\Cprim$ or vertices of $\Cdual$ (but not
edges of $\Cd$). Let $e$ be an edge of $\Cd.$ We say that a $\cvr_{\sCgr}$-spinor
$F(\cdot)$ on $\Cgr(\Od)$ is \emph{s-holomorphic at $e$} if the
following two conditions are satisfied: 

\begin{itemize}
\item for each corner $z$ incident to $e$, one has 
\begin{equation}
F(z)\in\eta_{z}\R;\label{eq: phases_condition}
\end{equation}
\item if $z_{N,W,S,E}$ are the four corners of $\Cgr(\Cd_{\cvr_{\sCgr}})$
adjacent to $e$, in this cyclic order, then 
\begin{equation}
F(z_{N})+F(z_{S})=F(z_{E})+F(z_{W}).\label{eq: s-hol}
\end{equation}
\end{itemize}
We say that a $\cvr$-spinor $F$ is \emph{s-holomorphic in $\Od$}
if it is s-holomorphic at all edges of $\Od$ and at all edges of
$\Cd$ intersecting $\fixed$.
\end{defn}

\begin{rem}
\label{rem:s-hol-projections}Note that $\eta_{z_{N}}\R=i\eta_{z_{S}}\R$
and $\eta_{z_{E}}\R=i\eta_{z_{W}}\R$. Hence, assuming the condition
(\ref{eq: phases_condition}), the equation (\ref{eq: s-hol}) is
equivalent to the existence of a complex number $F(e)$ such that
\[
F(z)=\proj{\eta_{z}}{F(e)},\quad z=z_{N},z_{W},z_{S},z_{E},
\]
where $\proj{\eta}{F(e)}:=\eta\re[\eta^{-1}F(e)]$ denotes the orthogonal
projection of $F(e)$ onto the line $\eta\R$ in the complex plane.
In \cite{Smirnov_Ising,ChelkakSmirnov2,hongler_thesis,HonglerSmirnov,HonglerViklundKytolaCFT,ChelkakICMP},
the authors use this observation to extend the domain of definition
of an s-holomorphic spinor from the set of corners $\Cgr(\Od)$ to
the edges of $\Od$. Though such an extension is useful for the analysis
of the critical Ising model on less regular planar graphs, we do not
need it when working on $\Cd$.
\end{rem}

\begin{lem}
\label{lem: s-hol} For each $z_{1}\in\Cgr(\Od)$, the spinor $\bar{\eta}_{z_{1}}F(z_{1},\cdot)$
is s-holomorphic in $\Od$.
\end{lem}

\begin{proof}
First, note that (\ref{eq: phases_condition}) is satisfied by the
definition of $F(z_{1},z)$. To prove (\ref{eq: s-hol}), let $e$
be the edge in question and denote by $e^{\bullet}$ the corresponding
dual edge. We move the corner $z$ along the clockwise route $z_{N}\to z_{E}\to z_{S}\to z_{W}$
around the midpoint of $e$, assuming for the concreteness that $z_{N}$
shares a vertex of $\Cd=\Cprim$ with $z_{W}$ (and a vertex of $\Cdual$
with $z_{E})$, and keep track of the multi-valued functions involved
in the definition of $F_{\cvr}(z_{1},z).$ Denote $\fz_{NW}:=\fz_{N}=\fz_{W}$,
$\fz_{SE}:=\fz_{S}=\fz_{E}$ and $\vz_{SW}:=\vz_{S}=\vz_{W}$, $\vz_{NE}:=\vz_{N}=\vz_{E}$.
First, we have 
\[
\eta_{z_{E}}=e^{-\frac{i\pi}{4}}\eta_{z_{N}};\quad\eta_{z_{S}}=-i\eta_{z_{N}};\quad\eta_{z_{W}}=e^{-\frac{3i\pi}{4}}\eta_{z_{N}}.
\]
Second, write 
\[
F(z_{1},z_{N}):=\eta_{z_{1}}\eta_{z_{N}}\frac{\E_{\Od}(\mu_{\gamma}\sigma_{\fz_{NW}}\sigma_{\fz_{1}}\svv)}{\E_{\Od}(\svv)},
\]
where $\gamma$ is a branch cut for $\double{\vz_{1},\vz_{NE}}$ that
does not contain $e^{\bullet}=(\vz_{NE}\vz_{SW})$. According to the
discussion preceding Lemma \ref{lem: spinor-in-all-variables}, as
$z$ moves from $z_{E}$ to $z_{S}$, we must append $\gamma$ with
$e^{\bullet}$, and further, as $z$ moves from $z_{S}$ to $z_{W}$,
there is a sign change because $z$ crosses the new branch cut $\gamma\cup e^{\bullet}$.
Putting everything together, one gets 
\begin{multline*}
\bar{\eta}_{z_{1}}\bar{\eta}_{z_{N}}\E(\svv)\cdot\big(F(z_{1},z_{N})-F(z_{1},z_{E})+F(z_{1},z_{S})-F(z_{1},z_{W})\big)\\
=\ \E\big(\mu_{\gamma}\sigma_{\fz_{NW}}\sigma_{\fz_{1}}\svv\cdot(1-e^{-\frac{i\pi}{4}}\xi-ie^{-2\beta\xi}\xi+e^{-\frac{3i\pi}{4}}e^{-2\beta\xi})\big),
\end{multline*}
where $\xi:=\sigma_{\fz_{NW}}\sigma_{\fz_{SE}}=\pm1$ and hence $e^{-2\beta\xi}=\cosh(\beta)-\xi\sinh(2\beta)=\sqrt{2}-\xi$
since the value $\beta=\beta_{\mathrm{crit}}$ is given by (\ref{eq: beta-crit}).
Therefore, 
\[
1-e^{-\frac{i\pi}{4}}\xi-ie^{-2\beta\xi}\xi+e^{-\frac{3i\pi}{4}}e^{-2\beta\xi}=(1+i+e^{-\frac{3i\pi}{4}}\sqrt{2})+(-e^{-\frac{i\pi}{4}}-i\sqrt{2}-e^{-\frac{3i\pi}{4}})\xi=0
\]
and the s-holomorphicity condition (\ref{eq: s-hol}) follows.
\end{proof}
\begin{rem}
\label{Remark_discrte_singularity} Lemma \ref{lem: Obs_spinor} suggests
that $F(\cdot):=F(z_{1},\cdot)$ is a $\double{\fz_{1},\vz_{1}}\cdot\cvr$
- spinor on $\Cgr(\Od)$. However, it is more useful to view it rather
as a $\cvr$-spinor with a ``discrete singularity'' at the corner
$z_{1}$. To this end, we introduce a graph $\Cutz{\Ocvr}{z_{1}},$
which is $\Cgr(\Ocvr)$ cut along the segment $\gamma_{z_{1}}:=[\vz_{1}\fz_{1}]$.
This cut splits the vertex $z_{1}$of $\Ocvr$ into two vertices $z_{1}^{+}$
and $z_{1}^{-}$, each of degree two, on the left (respectively, on
the right) side of $\gamma_{z_{1}}$ as seen from $\vz_{1}$. Similarly,
the vertex $z_{1}^{*}$ on the other sheet of $\cvr$ is split into
$z_{1}^{+,*}$ and $z_{1}^{-,*}$. Each $\double{\fz_{1},\vz_{1}}\cdot\cvr$
- spinor $F(\cdot)$ can be \emph{equivalently} viewed as a $\cvr$-spinor
on $\Cutz{\Od_{\cvr}}{z_{1}}$ satisfying the condition 
\[
F\left(z_{1}^{+}\right)=-F\left(z_{1}^{-}\right)=-F\left(z_{1}^{+,*}\right)=F\left(z_{1}^{-,*}\right).
\]
If $\cvr$ is trivial, this is simply the restriction of $F$ to a
sheet defined by the branch cut $\gamma_{z_{1}}.$ Note that the s-holomorphicity
is not affected by this passage from $\Cgr(\Ocvr)$ to $\Cutz{\Ocvr}{z_{1}}$.
\end{rem}

\begin{figure}
\includegraphics[width=0.4\textwidth]{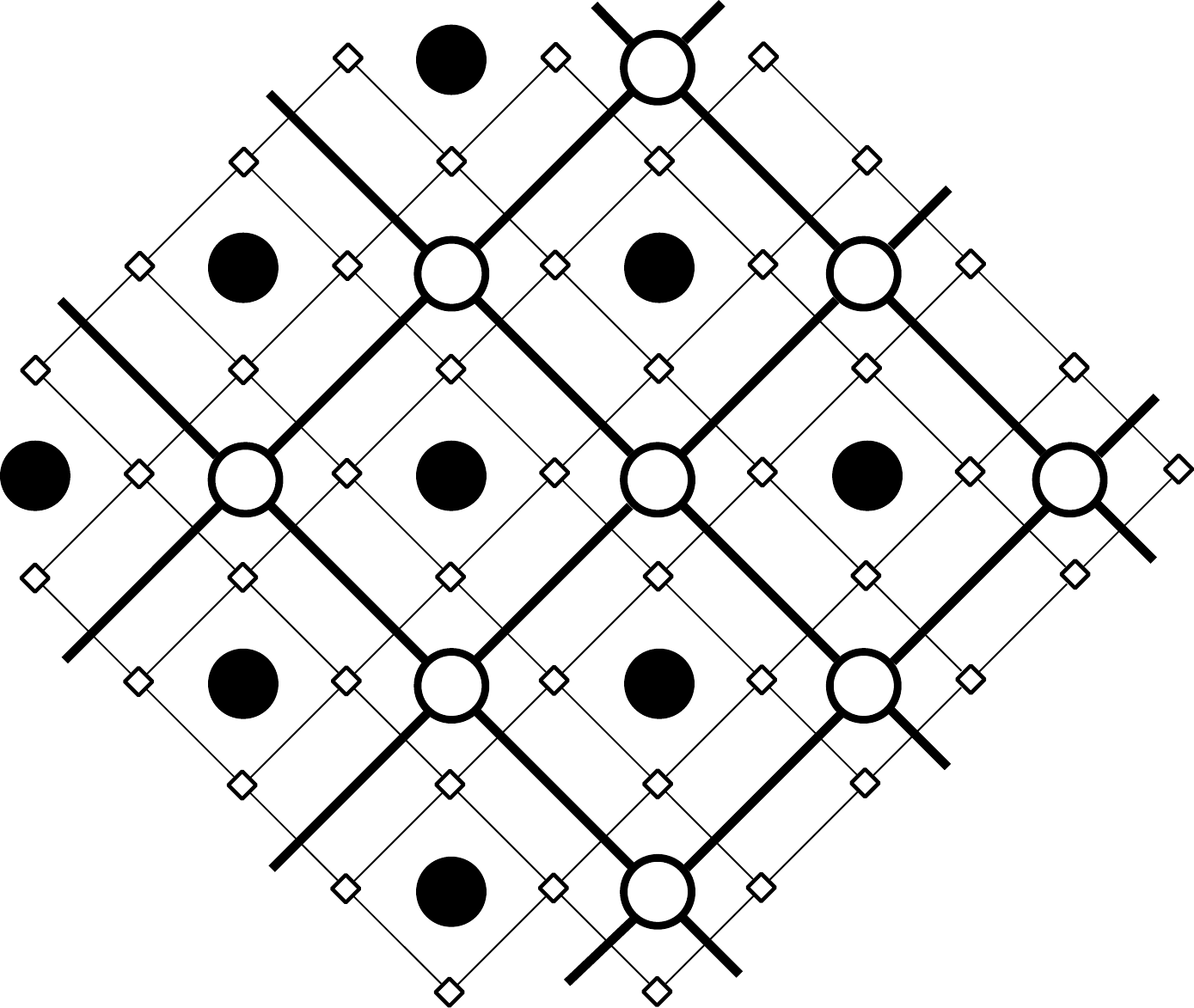}\qquad{}\includegraphics[width=0.4\textwidth]{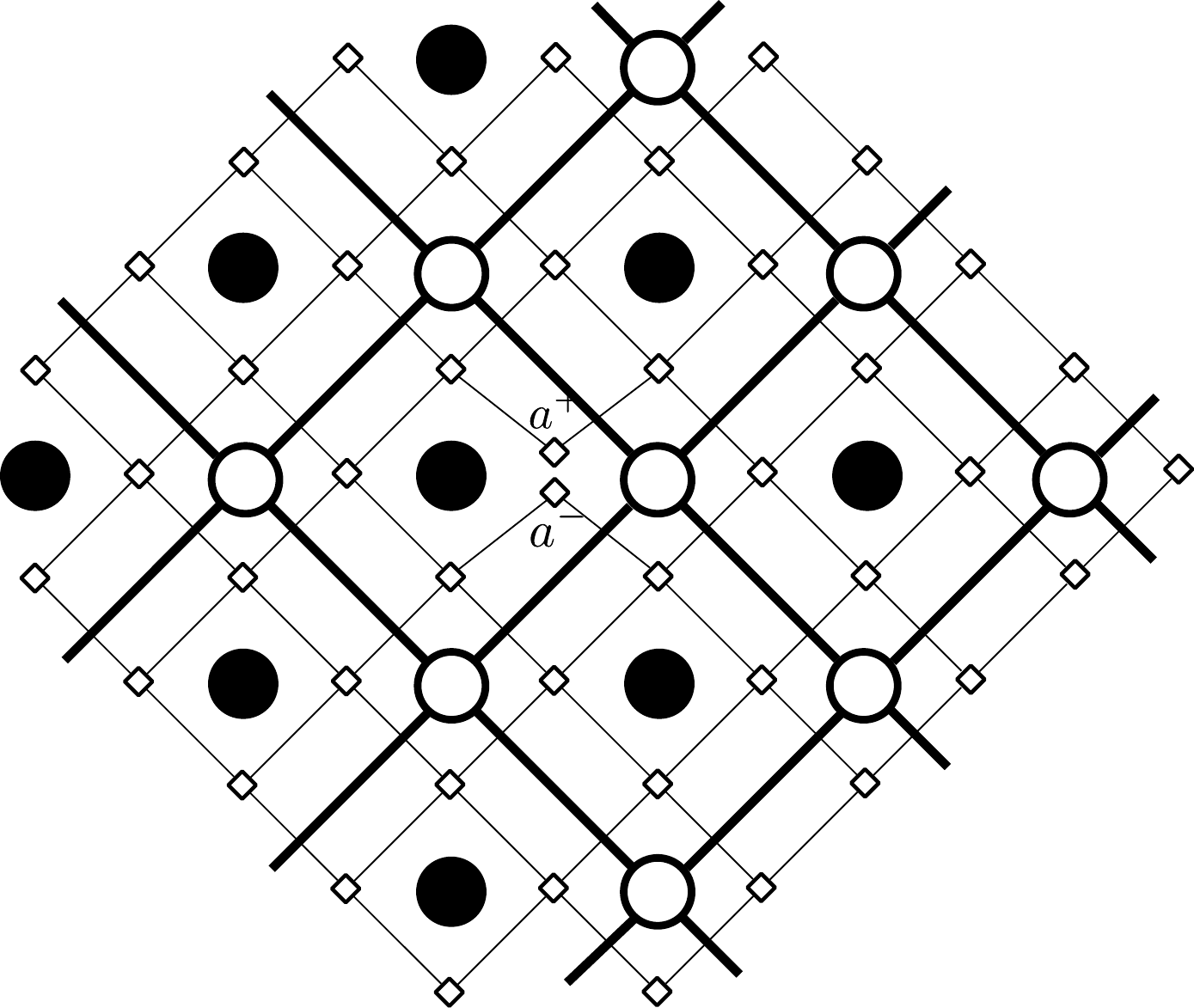}\caption{A piece of the lattice $\protect\Od$,with its dual domain $\protect\Odual$
and the corner lattice lattice $\protect\Cgr(\protect\Od)$. On the
right: the modified corner lattice $\protect\Cutz{\protect\Od}a$
with a ``discrete singularity'' at $a$. Any closed path on $\protect\Cutz{\protect\Od}a$
either encircles both $a^{\bullet}$ and $a^{\circ}$, or none; hence
$F_{\protect\Od,\protect\cvr}(a,\cdot)$ is locally single-valued
near $a$ on $\protect\Cutz{\protect\Od_{\protect\cvr}}a$ as long
as $\protect\cvr$ is not ramified at $a^{\circ}$.}
\end{figure}

In addition to the s-holomorphicity, in the setup described above,
two-point fermionic observables satisfy the following boundary conditions,
which we also call the standard ones, see Lemma \ref{lem: bc_discrete}
for the proof.
\begin{defn}
\label{def: Standard_BC} We say that a $\cvr$-spinor $F(\cdot)$
defined in $\Cgr(\Od)$ satisfies the \textit{standard boundary conditions}
if the equation 
\begin{equation}
\bar{\eta}_{z_{a}}\cdot F(z_{a})=\bar{\eta}_{z_{b}}\cdot F(z_{b}).\label{eq:bcond}
\end{equation}
holds for each pair of the corners $z_{a},z_{b}\in\Cgr(\Od_{\cvr})$
satisfying one of the following: 

\begin{enumerate}
\item $\vz_{a}=\vz_{b}\in\Odual_{\cvr}$ and $\fz_{a},\fz_{b}\not\in\Od_{\cvr}$
are on (the lift of) the same wired boundary arc; 
\item $\fz_{a}=\fz_{b}\in\Od_{\cvr}$ and $\vz_{a},\vz_{b}\in\Odual_{\cvr}$
are on (the lift of) the same free boundary arc;
\item $\vz_{a},\vz_{b}\in\Odual_{\cvr}$ are two endpoints of (the lift
of) a free boundary arc $\nu$, and $\fz_{a},\fz_{b}\not\in\Od_{\cvr}$
are on (the lifts of) two wired boundary arcs on the same boundary
component. 
\end{enumerate}
The signs of $\ds{z_{a}}$ and $\ds{z_{b}}$ in (\ref{eq:bcond})
are assumed to be related as follows: in the case (1), by rotating
$z_{a}$ towards $z_{b}$ around $z_{a}^{\bullet}=\vz_{b}$ outside
$\Od$; in the case (2) by rotating $z_{a}$ towards $z_{b}$ inside
$\Od$; and in the case (3) by the extension along the path on $\Cgr(\Od)$
running along the free boundary arc \emph{outside} $\Od$.
\end{defn}

\begin{lem}
\label{lem: bc_discrete}The observable $F(z_{1},\cdot)$ satisfies
the standard boundary conditions.
\end{lem}

\begin{proof}
In all three cases $\sigma_{\fz_{a}}=\sigma_{\fz_{b}}$ by definition.
Moreover, the only difference in the values of $F$ at $z_{a}$ and
$z_{a}$ actually comes from the difference between $\ds{z_{a}}$
and $\ds{z_{b}}$. Indeed, in the case (1) this is obvious since $\vz_{a}=\vz_{b}$.
In the cases (2) and (3) one can drag the ramification point from
$\vz_{a}$ to $\vz_{b}$ along the corresponding free boundary arc.
The correlation will not change since the edges of the free arcs are
excluded from the expression defining $\mu_{\gamma}$.
\end{proof}
The following proposition and especially its continuous counterpart
\textemdash{} Proposition \ref{prop: Uniqueness_continuous} \textemdash{}
will be crucial for the whole analysis. They assert that standard
boundary conditions define a well-posed boundary value problem for
s-holomorphic spinors. 
\begin{prop}
\label{prop: discrete-uniqueness} Let $n$ be even, $\cvr=\double{\vv}$,
$\vv\in\Od\cup\fixed^{\circ}$. Assume that $G$ is a $\cvr$-spinor
on $\Cgr(\Ocvr)$ which is s-holomorphic everywhere in $\Od$, i.e.
at all edges of $\Od$ and all edges of $\Cd$ intersecting $\fixed$.
If $G$ satisfies the standard boundary conditions as in Definition
\ref{def: Standard_BC}, then $G=0$ everywhere in $\Ocvr$.
\end{prop}

\begin{proof}
The proof is postponed until Section \ref{sec: proofs_of_convergence_theorems}.
\end{proof}
\begin{rem}
We stress that in Proposition \ref{prop: discrete-uniqueness} the
double cover $\cvr$ is allowed to be ramified over \textit{vertices
of $\Od$ only}. A similar uniqueness result holds for double covers
ramified over \emph{dual vertices only}. However, if both primal and
dual vertices are allowed in $\cvr$, then the uniqueness fails: the
observable $F_{\cvr}(z_{1},\dots)$ itself gives an example of a nontrivial
s-holomorphic $\double{\vz_{1},\fz_{1}}\cdot\cvr$ - spinor satisfying
the standard boundary conditions. Similarly, the conclusion may fail
if one of the ramification points $v_{q}$ is placed inside a boundary
component with purely free boundary conditions.
\end{rem}

\subsection{Multi-point fermionic observables and the Pfaffian formula.}

\label{subsec: multi-fermions-discrete}

The multi-point fermionic observables are defined very similarly to
the two-point ones discussed above. 
\begin{defn}
Let $k$ and $n$ be even, and let, as above, $\cvr=\double{\vv}$.
We put 
\[
F_{\Od,\cvr}(z_{1,},\dots,z_{k}):=\ds{z_{1}}\ldots\ds{z_{k}}\frac{\E(\mu_{\vz_{1}}\dots\mu_{\vz_{k}}\sigma_{\fz_{1}}\dots\sigma_{\fz_{k}}\svv)}{\E(\svv)},
\]
viewed as a multi-valued function of $\vv\in\Cd$ and $z_{1},\ldots,z_{k}\in\Cgr(\Od)$.
\end{defn}

Throughout this section, we assume that $\Od$ and $\cvr:=\double{\vv}$
are fixed and omit them from notation, thus $F(z_{1,},\dots,z_{k})=F_{\Od,\cvr}(z_{1,},\dots,z_{k})$.
We first fix also $z_{1},\dots,z_{k-1}$ and discuss the properties
of $F$ as a function of its last argument. 
\begin{prop}
The quantity 
\[
\bar{\eta}_{z_{1}}\cdot\dots\cdot\bar{\eta}_{z_{k-1}}F(z_{1},\dots z_{k-1},\,\cdot\,)
\]
is a $\double{z_{1}^{\bullet},\dots,z_{k-1}^{\bullet},z_{1}^{\circ},\dots,z_{k-1}^{\circ}}\cdot\cvr$
- spinor, which is s-holomorphic in $\Od$ and satisfies the standard
boundary conditions as in Definition~\ref{def: Standard_BC}.
\end{prop}

\begin{proof}
The same proofs as in the case $k=2$, see Lemma \ref{lem: s-hol}
and Lemma \ref{lem: bc_discrete}.
\end{proof}
\begin{rem}
Similarly to Remark \ref{Remark_discrte_singularity}, we can view
$F(\cdot):=F(z_{1},\dots z_{k-1},\thinspace\cdot\thinspace)$ as a
$\cvr$-spinor with ``discrete singularities'' at $z_{1},\dots,z_{k-1},$
i.e., as being defined on the graph $\Cutz{\Od_{\cvr}}{z_{1},\dots,z_{k-1}}$
obtained by cutting $\Cgr(\Ocvr)$ along the segments $[z_{1}^{\bullet},z_{1}^{\circ}],\dots,[z_{k-1}^{\bullet},z_{k-1}^{\circ}]$,
subject to the condition 
\[
F\left(z_{j}^{+}\right)=-F\left(z_{j}^{-}\right)=-F\left(z_{j}^{+,*}\right)=F\left(z_{j}^{-,*}\right),\quad j=1,\dots,k-1.
\]
\end{rem}

Next, we allow \emph{all} points $z_{1},\dots,z_{k}$ to move simultaneously.
\begin{lem}
\label{lem: anti-symmetry}The restriction of the multi-valued function
\textup{$F(z_{1},\dots,z_{k})$} to the set \textup{
\begin{equation}
\offdiag k:=(\Cgr(\Od_{\cvr}))^{\times k}\setminus\cup_{i\neq j}\{z_{i}=z_{j}\;\mathrm{or}\;z_{i}=z_{j}^{*}\}\label{eq: offdiag}
\end{equation}
}of all $k$-tuples of \emph{distinct} corners of the double cover
\textup{$\Cgr(\Od_{\cvr})$} of the corner graph, equipped with the
product graph structure, is a well defined (up to a global choice
of the sign) anti-symmetric function. Moreover, it is a $\cvr$-spinor
in each $z_{j}$. 
\end{lem}

\begin{proof}
Let $(z_{1},\dots,z_{k})=z(0)\sim z(1)\sim\dots\sim z(N)=(z_{1},\dots,z_{k})$
be a nearest-neighbor loop in the graph (\ref{eq: offdiag}). By Lemma
\ref{lem: spinor-in-all-variables}, to figure our the sign picked
up by $F$ along this loop, we have to count the winding numbers of
each $z_{i}^{\bullet}(t)$ with respect to each $z_{j}^{\circ}(t)$
or $v_{q}$ modulo two, as $t$ runs from $0$ to $N$. The contribution
from braiding of $z_{i}^{\bullet}(t)$ around $z_{i}^{\circ}(t)$
is canceled by $\eta_{z_{i}(t)}$. Since we, in particular, assume
that each $z_{i}(t)$ makes a loop in $\Cgr(\Od_{\cvr})$ (i.e., arrives
to the same sheet it started from), no sign is picked up from braiding
of $z_{i}^{\bullet}(t)$ around $v_{1},\dots,v_{n}$. Finally, since
the moving segments $(z_{i}^{\bullet}(t),z_{i}^{\circ}(t))$ and $(z_{j}^{\bullet}(t),z_{j}^{\circ}(t))$
never intersect for $i\neq j$, the winding number of $z_{i}^{\bullet}(t)-z_{j}^{\circ}(t)$
along the loop is the same as that of $z_{i}^{\circ}(t)-z_{j}^{\bullet}(t)$,
so their contributions cancel each other. Therefore, $F$ is indeed
well-defined on $\offdiag k$. The claim that $F$ is a $\cvr$-spinor
in each $z_{i}$ follows readily from similar considerations.

It remains to check the anti-symmetry of $F$ in $z_{1},\ldots,z_{k}$.
For that, it suffices to show that it picks the required sign for
\emph{some} $k$-tuple of pairwise distinct corners $z_{1},\ldots,z_{k}$
of $\Ocvr$ and \emph{some} nearest-neighbor path on (\ref{eq: offdiag})
exchanging $\cvr(z_{i})$ and $\cvr(z_{j})$ and keeping the other
points fixed. This can be easily done by choosing $z_{i}$ and $z_{j}$
so that $z_{i}^{\bullet}=z_{j}^{\bullet}$ and exchanging them in
four moves by rotating around this dual vertex. 
\end{proof}
\begin{rem}
\label{rem: withdiag} Instead of simply excluding the diagonal in
(\ref{eq: offdiag}) one can also consider a multi-dimensional analogue
$\withdiag k$ of the graphs $\Cutz{\Ocvr}{z_{1}}$ discussed in the
previous section. Formally, the \emph{vertices} of $\withdiag k$
are $k$-tuples of corners of $\Ocvr$ with an additional agreement
that if, say, $\cvr(z_{1}),\ldots,\cvr(z_{s})$ match and differ from
$\cvr(z_{s+1}),\ldots,\cvr(z_{k})$, then a linear ordering $\prec$
of these corners is prescribed: for each $i\ne j$ in $\{1,\ldots,s\},$
either $\cvr(z_{i})\prec\cvr(z_{j})$ or $\cvr(z_{j})\prec\cvr(z_{i})$.
One can view $\prec$ as an artificial counter-clockwise ordering
of $\cvr(z_{1}),\ldots,\cvr(z_{s})$, as seen from their common vertex
of $\Odual$.

In the situation described above only the four \emph{edges} corresponding
to two relevant moves $z_{\mathrm{max}}\mapsto(z_{\mathrm{max}})'$
of the maximal (with respect to $\prec$) element of $\{z_{1,}\ldots,z_{s}\}$
and two moves of the minimal element $z_{\mathrm{min}}$ are kept
in $\withdiag k$, with an agreement that, after such a move, $(z_{\mathrm{max}})'$
becomes the minimal element in the updated ordering of corners sharing
the same position on $\Od$ (and similarly for $z_{\mathrm{min}}$,
which becomes maximal at its new position). The proof of Lemma \ref{lem: anti-symmetry}
then extends to $\withdiag k$ provided that one keeps track of the
ordering $\prec$ when permuting $z_{1},\ldots,z_{n}$; we leave the
details to the interested reader.

In particular, in the simplest case $k=2$, the graph $\withdiag k$
contains \emph{two} vertices corresponding to each of the pairs $(z,z)$,
$z\in\Ocvr$, which are mapped to each other by the mapping $(z_{1},z_{2})\mapsto(z_{2},z_{1})$.
For a fixed $z_{1}=z$, these vertices can be viewed as $(z,z^{\pm})\in\{z\}\times\Cutz{\Ocvr}z$,
they become $(z^{\mp},z)\in\Cutz{\Ocvr}z\times\{z\}$ if one fixes
$z_{2}=z$; a similar consideration applies to pairs $(z,z^{*})$,
$z\in\Ocvr$.
\end{rem}

\begin{defn}
\label{def: obs_signs}We adopt the following recursive convention
to fix the global sign of $F(z_{1},\ldots,z_{k})$: for $k\ge4$ and
($z_{1},\dots z_{k-1})\in\offdiag{k-1}$, put 
\begin{equation}
F(z_{1},\ldots,z_{k-1},z_{k-1}^{\pm}):=\pm\eta_{k-1}^{2}F(z_{1},\ldots,z_{k-2})\label{eq: sign_convention}
\end{equation}
and put $F(z,z^{\pm}):=\pm\ds z^{2}$ for $k=2$ and $z\in\Ocvr$.
\end{defn}

Strictly speaking, we choose one\emph{ particular} configuration ($z_{1},\dots z_{k-1})\in\offdiag{k-1}$
and fix the sign of $F(z_{1},\dots,z_{k-1},z_{k-1}^{+})$ by (\ref{eq: sign_convention})
for this configuration. Then, it is straightforward to check that
(\ref{eq: sign_convention}) is preserved when one moves the marked
points locally, and thus the same relations holds for \emph{all} $(k-1)$-tuples
of distinct corners $z_{1},\dots,z_{k-1}\in\Ocvr$.
\begin{prop}
\label{prop: pfaff_discrete} For all $k$-tuples $(z_{1},\ldots,z_{k})\in\offdiag k$\textup{
we have} 
\begin{equation}
F(z_{1},\ldots,z_{k})=\Pf\left[F(z_{i},z_{j})\right]_{i,j=1}^{k}.\label{eq: Pfaff_discrete}
\end{equation}
\end{prop}

\begin{proof}
The formula is proven by induction in $k$. There is nothing to prove
if $k=2$. For $k\ge4$, consider the following $\cvr$-spinor on
$\Cutz{\Ocvr}{z_{1},\ldots,z_{k}}$: 
\[
G(\cdot)=F(z_{1},\ldots,z_{k-1},\thinspace\cdot\thinspace)-\sum_{j=1}^{k-1}(-1)^{j+1}F(z_{1},\ldots,\hat{z}_{j},\ldots,z_{k-1})\cdot F(z_{j},\thinspace\cdot\thinspace).
\]
and check that $\bar{\eta}_{z_{1}}\dots\bar{\eta}_{z_{k-1}}G$ satisfies
the conditions of Proposition \ref{prop: discrete-uniqueness}. Indeed,
each term can be viewed as an s-holomorphic $\cvr$-spinor with ``discrete
singularities'' at $z_{1},\ldots,z_{k-1}$ satisfying the standard
boundary conditions. Moreover, we have
\[
\begin{aligned}F(z_{1},\ldots,z_{k-1},z_{i}^{\pm}) & =(-1)^{i+1}F(z_{1},\ldots\hat{z}_{j},\ldots,z_{k-1},z_{i},z_{i}^{\pm})\\
 & =\pm(-1)^{i+1}\eta_{z_{i}}^{2}F(z_{1},\ldots,\hat{z}_{i},\ldots,z_{k-1})\\
 & =(-1)^{i+1}F(z_{1},\ldots,\hat{z}_{i},\ldots,z_{k-1})F(z_{i},z_{i}^{\pm}),
\end{aligned}
\]
where we used the anti-symmetry (Lemma \ref{lem: anti-symmetry})
in the first line and the definition (\ref{eq: sign_convention})
in the second and the third lines. Hence, $G(z_{i}^{-})=G(z_{i}^{+})$
for all $i=1,\ldots,n$, which means that we can glue $z_{i}^{+}$
and $z_{i}^{-}$ together, recovering a $\cvr$-spinor on the whole
graph $\Cgr(\Ocvr)$. Therefore, $G\equiv0$ due to Proposition \ref{prop: discrete-uniqueness}.
The result follows from a well-known recursive formula for the Pfaffian.
\end{proof}
\begin{rem}
Proposition \ref{prop: pfaff_discrete} can be generalized to all
$k$-tuples $(z_{1},\ldots,z_{k})\in\withdiag k$ provided that the
diagonal entries $F(z_{i},z_{i})$ are replaced by zeros. Though we
do not use this more complicated graph $\withdiag k$ in our paper,
a careful reader could notice that $\offdiag k$ might be empty or
have nasty connectivity properties for huge values of $k$ while $\withdiag k$
is well defined for all $k$. We also refer the interested reader
to \textcolor{red}{\cite{ChelkakCimasoniKassel},} where the Pfaffian
structure of fermionic observables is discussed from a combinatorial
perspective.
\end{rem}

\subsection{General correlations via spin and fermionic ones.}

\label{subsec: general_corr_in_terms_of_fermionic}In this section
we explain how to express all correlations of the form 
\[
\E(\Op(\pesm))
\]
featuring in Theorem \ref{thm: intro_2} in terms of
\begin{itemize}
\item multi-point spin correlations and
\item multi-point fermionic observables, which can be further expressed
via two-point observables by the Pfaffian formula (\ref{eq: Pfaff_discrete}).
\end{itemize}
Assume that we are given a collection of vertices $\vv$, faces $\uu$,
edges $e_{1},\dots,e_{s}$ and corners $z_{1},\dots,z_{k}$ of $\Od$.
For each $j=1,\ldots,s$ and $p=1,\ldots,m$, let
\begin{itemize}
\item $z_{k+2j-1},z_{k+2j}$ be (any pair of) two opposite corners adjacent
to the edge $e_{j}$;
\item $z_{k+2s+p}$ be an (arbitrarily chosen) corner adjacent to $u_{p}$
and $v_{n+p}:=\fz_{k+2s+p}$.
\end{itemize}
Let $N:=n+m$ and $K:=k+2s+m$, note that the correlation $\E(\Op(\pesm))$
vanishes (due to the $z\mapsto z^{*}$ symmetry) unless $K$ is even. 
\begin{lem}
\label{lem: corr_to_obs_for_THM_2}The following equality of multi-valued
functions holds: 
\begin{equation}
\eta_{z_{1}}\ldots\eta_{z_{K}}\E_{\Od}(\Op(\pesm))\ =\ F_{\double{v_{1},\dots,v_{N}}}(z_{1},\dots,z_{K})\cdot\E(\sigma_{v_{1}}\ldots\sigma_{v_{N}}),\label{eq: corr_to_obs_2}
\end{equation}
up to a global choice of the sign.
\end{lem}

\begin{proof}
By definition, the right-hand side of (\ref{eq: corr_to_obs_2}) is
equal to 
\begin{equation}
\eta_{z_{1}}\ldots\eta_{z_{K}}\E(\mu_{\vz_{1}}\ldots\mu_{\vz_{K}}\sigma_{\fz_{1}}\dots\sigma_{\fz_{K}}\sigma_{v_{1}}\dots\sigma_{v_{N}}).\label{eq: corr_to_obs_interm}
\end{equation}
Note that the spins at $v_{n+1}=z_{k+1}^{\text{\ensuremath{\circ}}},\dots,v_{n+m}=z_{k+m}^{\circ}$
appear in this expression twice, and thus cancel out. Further, given
an edge $e$ of $\Od$, denote by $\gamma_{e}$ a subset $\{e^{\bullet}$\}
of edges of $\Odual$ consisting of the single edge $e^{\bullet}$
dual to $e$. Let $e_{\pm}\in\Od$ be the endpoints of $e$ and $\xi_{e}:=\sigma_{e_{+}}\sigma_{e_{-}}$.
For $\beta=\beta_{\mathrm{crit}},$ we have
\begin{equation}
\mu_{\gamma_{e}}\sigma_{e_{+}}\sigma_{e_{-}}=e^{-2\beta\xi_{e}}\xi_{e}=\cosh(2\beta)\cdot\xi_{e}-\sinh(2\beta)=\sqrt{2}\xi_{e}-1=\en_{e}.\label{eq: en_spin_disorder}
\end{equation}
Thus, choosing the set $\gamma$ of dual edges implicit in the definition
of the disorder variable in the right-hand side of (\ref{lem: corr_to_obs_for_THM_2})
to be a disjoint union $\gamma=\widetilde{\gamma}\sqcup\gamma_{e_{1}}\sqcup\dots\sqcup\gamma_{e_{s}}$,
where $\pa\widetilde{\gamma}=\{\vz_{1},\dots,\vz_{k},\uu\}$, one
sees that 
\[
\E(\mu_{\vz_{1}}\ldots\mu_{\vz_{K}}\sigma_{\fz_{1}}\dots\sigma_{\fz_{K}}\sigma_{v_{1}}\dots\sigma_{v_{N}})\ =\ \E(\mu_{\widetilde{\gamma}}\,\en_{e_{1}}\ldots\en_{e_{s}}\sigma_{z_{1}^{\circ}}\ldots\sigma_{z_{k}^{\circ}}\sigma_{v_{1}}\dots\sigma_{v_{n}}).
\]
Therefore, (\ref{eq: corr_to_obs_interm}) is equal to the left-hand
side of (\ref{eq: corr_to_obs_2}).
\end{proof}
Lemma \ref{lem: corr_to_obs_for_THM_2} is instrumental in the proof
of Theorem \ref{thm: intro_2}. We now formulate a similar statement
suitable for Theorem \ref{thm: Intro_3}. Let the domain $\Od$ be
equipped with boundary conditions $\bcd$, that is, with a subdivision
of its boundary into three subsets $\plus$, $\minus$ and $\free$.
For convenience, we assign the $-1$ spin to each \emph{free} boundary
arc, and denote by $b_{1},\dots,b_{q}$ the vertices of $\Cdual$
separating a plus boundary arc from a minus one, with the aforementioned
assignment made on free ones. We number $b_{1},\dots,b_{q}$ in such
a way that points at each boundary component are listed consecutively
and the index increases when one tracks this boundary component so
that $\Od$ stays at the left.

We now introduce auxiliary boundary conditions, which we denote by
$\bcdd$ and defined as follow: if $w_{1},w_{2}$ are two vertices
of $\Cd\setminus\Od$ adjacent to the same connected component of
$\pa\Od$, then we require that
\begin{itemize}
\item $\sigma_{w_{1}}=\sigma_{w_{2}}$ if an even number of points $b_{1},\dots,b_{q}$
lie on the boundary arc between $w_{1}$ and $w_{2}$ (traced so that
the domain \textbf{$\Od$} stays at the left);
\item $\sigma_{w_{1}}=-\sigma_{w_{2}}$ otherwise.
\end{itemize}
In other words, the spins are required to be \emph{locally monochromatic},
i.e., to stay locally constant along each boundary component except
the required changes at each of $b_{j}$.

Given $s$ edges $e_{1},\dots,e_{s}$, $n$ vertices $\vv$, and the
boundary conditions $\bcd$, let
\begin{itemize}
\item $z_{i}$ denote any corner adjacent both to $b_{i}=\vz_{i}\in\Odual$
and to a vertex of $\Cd\setminus\Od$,
\item $z_{q+2j-1}$ and $z_{q+2j}$ be two opposite corners adjacent to
the edge $e_{j}$, such that $\vz_{q+2j-1}-\fz_{q+2j-1}>0$ (and hence
$\vz_{q+2j}-\fz_{q+2j}<0$), 
\end{itemize}
where $i=1,\ldots,q$ and $j=1,\dots,s$. As above, we denote $K:=q+2s$. 
\begin{lem}
\label{lem: corr_to_obs_for_Thm_3}The following identity holds:
\begin{equation}
\E_{\Od}^{\bcdd}(\Op(\en,\sigma))\ =\ \frac{F_{\double{\vv}}(z_{1},\dots,z_{K})}{F(z_{1},\dots,z_{q})}\,\E_{\Od}(\svv),\label{eq: Corr_to_obs_2}
\end{equation}
where $F_{\double{\vv}}(z_{1},\dots,z_{K})$ and $F(z_{1},\dots,z_{q})$
are understood according to Definition \ref{def: obs_signs}.
\end{lem}

\begin{proof}
Taking into account the definition of $F$ and the computation (\ref{eq: en_spin_disorder}),
we can write the right-hand side of (\ref{eq: Corr_to_obs_2}) as
\begin{equation}
\pm\eta_{z_{q+1}}\ldots\eta_{z_{K}}\frac{\E(\en_{e_{1}}\ldots\en_{e_{s}}\mu_{\vz_{1}}\dots\mu_{\vz_{q}}\sigma_{\fz_{1}}\ldots\sigma_{\fz_{q}}\svv)}{\E(\mu_{\vz_{1}}\dots\mu_{\vz_{q}}\sigma_{\fz_{1}}\ldots\sigma_{\fz_{q}})},\label{eq: Lem_Thm_3_intermediate}
\end{equation}
where both expectations are taken with standard boundary conditions.
Moreover, under the convention (\ref{eq: sign_convention}), the sign
in the above expression is $+$ provided that $\eta_{z_{q+2j-1}}\eta_{z_{q+2j}}=1$
for $j=1,\ldots,s$, and that the expression $\mu_{\vz_{1}}\dots\mu_{\vz_{q}}=\mu_{b_{1}}\ldots\mu_{b_{q}}=\mu_{\gamma}$
is defined using the same set $\gamma\subset\partial\Od$ of dual
edges running along the boundary of $\Od$. 

Since each boundary component of $\Od$ contains an even number of
the points $b_{1},\dots,b_{q}$, we have $\sigma_{\fz_{1}}\ldots\sigma_{\fz_{q}}=1$
under the standard boundary conditions. Due to Lemma \ref{lem: branch cut},
the expression (\ref{eq: Lem_Thm_3_intermediate}) simplifies to 
\[
\frac{\E(\mu_{\gamma}\en_{e_{1}}\ldots\en_{e_{s}}\svv)}{\E(\mu_{\gamma})}\ =\ \E_{\double{b_{1,}\ldots,b_{q}}}(\en_{e_{1}}\ldots\en_{e_{s}}\svv).
\]
It remains to notice that restricting a $\double{b_{1,}\ldots,b_{q}}$-spinor
$\sigma$ (with standard boundary conditions) to a sheet defined by
$\gamma$ is a measure preserving bijection to the set of spin configurations
with boundary conditions $\bcdd.$ 
\end{proof}
Finally, we explain how to pass from the boundary conditions $\bcdd$
to $\bcd$ by placing auxiliary spin variables on fixed boundary arcs.
For each boundary component that has a non-empty $\fixed$ part, pick
a vertex on the corresponding part of $\fixed^{\circ}$; denote these
vertices by $v_{n+1},\dots,v_{n+d}$. Given a subset $S\subset\{v_{n+1},\dots,v_{n+d}\},$
let $\sigma_{S}:=\prod_{v\in S}\sigma_{v}$. Viewing each $\sigma_{v_{n+1}},\dots,\sigma_{v_{n+d}}$
as an independent variable with values $\pm1$, it is well known that
the functions $\sigma_{S}$ form an orthonormal basis of $L^{2}(\{\pm1\}^{d})$,
called the Fourier\textendash Walsh basis. In particular, for each
prescribed values $\sfix{}=(\sfix{(1)},\ldots,\sfix{(d)})\in\{\pm1\}^{d}$,
one can easily expand the corresponding indicator function in this
basis: 
\begin{equation}
\ind(\sigma_{v_{n+1}}=\sfix{(1)},\dots,\sigma_{v_{n+d}}=\sfix{(d)})\ =\ \sum_{S}\alpha_{S}(\sfix{})\sigma_{S}.\label{eq: FW_indicator}
\end{equation}
Therefore, one has 
\begin{equation}
\E_{\Od}^{\bcd}(\Op(\en,\sigma))\ =\ \frac{\sum_{S}\alpha_{S}(\sfix{})\E_{\Od}^{\bcdd}(\Op(\en,\sigma)\sigma_{S})}{\sum_{S}\alpha_{S}(\sfix{})\E_{\Od}^{\bcdd}(\sigma_{S})},\label{eq: FW_correlation}
\end{equation}
and, indeed, the proof of Theorem \ref{thm: Intro_3} can be reduced
to the computation of the scaling limits of $\E_{\Od}^{\bcdd}(\Op(\en,\sigma)\sigma_{S})$
and $\E_{\Od}^{\bcdd}(\sigma_{S})$ for all subsets $S\subset\{v_{n+1},\dots,v_{n+d}\}$.
Note that a half of these correlations always vanish due to the parity
reasons.

\newpage{}

\section{Fermionic observables in continuum: definitions and convergence theorems}

\label{sec: ccor} In this section we introduce continuous counterparts
of fermionic observables discussed in Section \ref{subsec: fermions-discrete}
and formulate the relevant convergence results. Recall that all the
correlation functions $\E_{\Od}(\Op(\pesm))$ and $\E_{\Od}^{\bcd}(\Op(\en,\sigma))$
that we discuss in our paper can be eventually expressed via the spin
correlations and these discrete fermionic observables evaluated either
in the bulk of $\Od$ or near one of the spins, or at the boundary
of $\Od$; in the latter case points separating free and wired boundary
arcs play a special role. For the sake of the reader, we now list
the variants of the continuous fermionic observable that we progressively
introduce below:
\begin{itemize}
\item $\feta_{\Omega,\cvr}(a,z)$: Definition \ref{def: feta} (holomorphic
in $z$);
\item $f_{\Omega,\cvr}^{[\eta_{1},\eta_{2}]}(z_{1},z_{2})$: Lemma \ref{lem: antisym_cont}
(real-valued); 
\item $f_{\Omega,\cvr}(z_{1},z_{2})$ and $\fdag_{\Omega,\cvr}(z_{1},z_{2})$:
Lemma \ref{lem: spinor_hol_antyhol} (both functions are holomorphic
in $z_{2}$, $f_{\Omega,\cvr}$ is holomorphic in $z_{1}$, $\fdag_{\Omega,\cvr}$
is anti-holomorphic in $z_{1}$);
\item $f_{\Omega,\cvr}^{[\eta,\reg]}(a,v):$ Definition \ref{def: sharp}
(real-valued, $v$ is a ramification point of $\cvr$);
\item $f_{\Omega,\cvr}^{[\eta,\flat]}(a,b)$: Definition \ref{def: flat}
(real-valued, $b\in\partial\Omega$ separates $\free$ and $\fixed$);
\item $f_{\Omega,\cvr}^{[\reg]}(v,z)$ and $f_{\Omega,\cvr}^{[\flat]}(b,z)$:
Lemma \ref{lem: f_star_holom} (holomorphic in $z$; $v$ and $b$
are as above); 
\item $f_{\Omega,\cvr}^{[\reg,\any]}(v,\cdot)$ and $f_{\Omega,\cvr}^{[\flat,\any]}(b,\cdot)$,
$\any\in\{\eta,\reg,\flat\}$: Definition \ref{def: f-any-any} (real-valued).
\end{itemize}
Theorem \ref{thm: conv_bulk_bulk} claims the convergence of discrete
fermionic observables to $f_{\Omega,\cvr}^{[\eta_{1},\eta_{2}]}(z_{1},z_{2})$
when both points are in the bulk of $\Omega$. Theorem \ref{thm: Convergence_singular}
provides a unifying result on the convergence to $f_{\Omega,\cvr}^{[\triangleleft,\triangleright]}$
for all possible pairs of superscripts $\triangleleft,\triangleright\in\{\eta,\reg,\flat\}$.
At the end of the section we also introduce a quantity
\begin{itemize}
\item $\coefA_{\Omega,\cvr}(v)$: Definition \ref{def: coefA} (second coefficient
in the expansion of $f_{\Omega,\cvr}^{[\reg]}(v,z)$ as $z\to v$)
\end{itemize}
and formulate the corresponding convergence theorem (Theorem \ref{thm: Conv_both_near_spin}).
The coefficient $\coefA_{\Omega,\cvr}(v)$ is of a crucial importance
for the definition of spin correlations in continuum, this discussion
is postponed until Section \ref{sec: corr-continuum}. The proofs
of Theorems \ref{thm: conv_bulk_bulk}, \ref{thm: Convergence_singular}
and \ref{thm: Conv_both_near_spin} (as well as those of Theorem \ref{thm: Convergence_bulk_near}
and Lemma \ref{lem: Clements_clever_lemma}) are discussed in Section
\ref{sec: proofs_of_convergence_theorems}.

\subsection{Domains and their convergence}

\label{subsec: domains_convergence}Let $\Omega\subset\C$ be a (possibly,
punctured) bounded finitely connected domain equipped with boundary
conditions represented by a pair of disjoint open subsets $\fixed,\free\subset\partial\Omega$
such that $\pa\Omega\setminus(\fixed\cup\free)$ consists of a finite
number of points; the punctures (i.e., single-point boundary components)
are \emph{not} allowed to belong neither to $\fixed$ nor to $\free$. 

We say that $\Omega$ is \textit{nice} if $\pa\Omega$ consists of
finitely many disjoint analytic Jordan curves or points. By Koebe's
extension of Riemann mapping theorem, any finitely connected domain
$\Omega$ is conformally equivalent to a nice domain $\Omega'$; indeed,
$\Omega'$ can be taken to be a circular domain. 

Formally speaking, below we need to distinguish $\Omega$ and another
domain obtained by adding back to $\Omega$ all its punctures. In
practice, we believe that no confusion arises and the meaning of the
symbol $\Omega$ is always clear from the context. Still, in very
rare situations when the difference is of importance, we will use
the symbol $\Opunc:=\Omega$ for the former (punctured) and $\Onopunc$
for the latter (punctures filled) domains.

In Theorems \ref{thm: Intro_1} \textendash{} \ref{thm: Intro_3},
as well as in all the convergence results discussed below, we fix
$\Onopunc$ and a sequence of its discrete approximations $\Od$,
in the following sense.
\begin{defn}
\label{def: domain-conv}Let $\Od$ be a sequence of discrete domains,
equipped with boundary conditions $\pa\Od=\fixed\cup\free$, and let
$\Omega$ be a domain as above. We say that $\Od$ \emph{approximate}
$\Omega$ as $\delta\to0$ if the following holds: 
\begin{itemize}
\item eventually, the number of boundary components of $\Od$ is the same
as the number of boundary components in $\Onopunc$ and the numbers
of wired and free arcs of $\Od$ and those of $\Omega$ match for
each boundary component;
\item $\Od$ approximates $\Onopunc$ in the Carathéodory sense;
\item the wired and the free boundary arcs of $\Od$ approximate those of
$\Onopunc$.
\end{itemize}
\end{defn}

The only reason to introduce single-point boundary components of $\Omega$
is that we will work with double covers of planar domains. Recall
that, in the discrete setting, the corner graph $\Cgr(\Od)$ is often
equipped with a double cover, say $\cvr=\double{v_{1},\dots,v_{n}}$,
where the points $v_{1},\dots,v_{n}$ may (and typically do) belong
to $\Od.$ 

We will always assume that $\{\vv\}\cap\Od=\{\vv\}\cap\Onopunc$ (in
other words, that $v_{q}$ do not belong to possible fjords of $\Od$),
the remaining points $v_{q}\not\in\Od$ belong to the corresponding
connected components of $\C\setminus\Omega$, and that $\Opunc$=$\Onopunc\setminus\{\vv\}$.
In such a situation, $\double{\vv}$ can be viewed both as a double
cover of $\Od$ and that of $\Opunc$.

In our paper, all the convergence statements are understood in the
following sense.
\begin{defn}
\label{def: conv-in-the-bulk}Let discrete domains $\Od$ approximate
$\Omega$ in the sense of Definition \ref{def: domain-conv}. If $A$
is a quantity that depends on $\Od$ and a collection of marked points
in $\Od$ (e. g., vertices $\vv$, corners $z_{1},\dots,z_{k}$, etc),
and $B$ is an expression depending on $\Omega$ and a similar collection
of marked points, we say that 
\[
A=B+o(1)\quad\text{as }\mbox{\ensuremath{\delta\to}0}
\]
\emph{uniformly over marked points in the bulk of $\Omega$ and away
from each other} if for each compact subset $K$ of $\Omega$ and
each $\eps>0$, the supremum of $|A-B|$ over all possible positions
of marked points in $K$ and at distance at least $\eps$ from each
other tends to zero as $\delta\to0$.
\end{defn}

\subsection{Continuous counterpart of the standard boundary conditions for spinors }

In the following definition, we consider relatively open subsets $S\subset\pa\Ocvrc$
of the boundary of a double cover of $\Omega$. In applications, $S$
will be typically the whole $\pa\Ocvrc$, except, maybe, one or several
points. Through the rest of the paper, we denote 
\[
\lamb:=e^{-\frac{i\pi}{4}}.
\]
If $\Omega$ is a nice domain and $\zeta\in\pa\Omega$ (or $\zeta\in\pa\Ocvrc$,
where $\cvr$ is a double cover of $\Omega$), we denote by $\tau_{\zeta}$
the unit tangent vector to $\pa\Omega$ at $\zeta$, oriented so that
$\Omega$ is on its left, and viewed as a complex number. 
\begin{defn}
\label{def: BC_continuous}Let $\Omega$ be a nice domain, $\cvr$
its double cover, $f$ a holomorphic $\cvr$-spinor in $\Omega$,
and $S\subset\pa\Ocvrc$ a relatively open set . We say that $f$
satisfies the \textit{standard boundary conditions} on $S$ if the
following holds true:
\end{defn}

\begin{itemize}
\item If $\zeta\in S\cap\fixed$, then $f$ extends continuously to $\zeta,$
and 
\begin{equation}
\im\big[\tang[\zeta]^{\frac{1}{2}}\cdot f(\zeta)]=0.\label{eq: bc_fixed}
\end{equation}
\item If $\zeta\in S\cap\free$, then $f$ extends continuously to $\zeta,$
and 
\begin{equation}
\re\big[\tang[\zeta]^{\frac{1}{2}}\cdot f(\zeta)\big]=0.\label{eq: bc_free}
\end{equation}
\item If $\{v\}\subset S$ is a puncture of $\Onopunc$ (i.e., a single-point
boundary component of $\Omega$), then there exists $c\in\R$ such
that 
\begin{equation}
f(z)=\frac{\lambb c}{\sqrt{z-v}}+O(1)\ \ \text{as}\ \ z\to v.\label{eq: bc_branchpoint}
\end{equation}
\item If $\nu$ is a free boundary arc such that $\bar{\nu}\subset\text{Int}\,S$,
and $b_{1},b_{2}$ are the endpoints of $\nu$, following in this
order when tracing $\nu$ in the direction of $\tau$, then 
\begin{equation}
f(z)=\frac{c_{p}}{\sqrt{z-b_{p}}}+O(1)\ \ \text{as}\ \ z\to b_{p},\ p=1,2,\quad\text{and}\quad c_{2}=ic_{1}\in\R,\label{eq: bc_free_arc}
\end{equation}
where $f(z)$ as $z\to b_{p}$, $p=1,2$, denote the restriction of
$f$ to the \emph{same} sheet of $\Ocvrc$ over a simply-connected
neighborhood of $\nu$ inside $\Omega$, and the branch of $\sqrt{z-b}$
is chosen continuously for $z$ in this neighborhood and\textbf{ $b\in\nu$};
cf. Remark \ref{rem: free-arcs-cond}. 
\end{itemize}
\begin{rem}
\label{rem: free-arcs-cond} To illustrate the sign condition in (\ref{eq: bc_free_arc}),
let us first note that (\ref{eq: bc_fixed}), (\ref{eq: bc_free})
always imply that $c_{1}\in i\R$ and $c_{2}\in\R$ provided that
the asymptotics in (\ref{eq: bc_free_arc}) hold. Still, the equation
$c_{1}=-ic_{2}$ means more than $|c_{1}|=|c_{2}|$. E.g., let $[b_{1},b_{2}]\subset\R$
and consider the function $f(z):=\left((z-b_{1})(z-b_{2})\right)^{-\frac{1}{2}}$
defined in a neighborhood of the segment $[b_{1},b_{2}]$ in the upper
half-plane $\H$. If one views $(b_{1},b_{2})$ as a free arc on the
boundary of $\Omega\subset\H$, then both (\ref{eq: bc_fixed}) and
(\ref{eq: bc_free}) are satisfied near $[b_{1},b_{2}]$ but (\ref{eq: bc_free_arc})
is \emph{not}: $c_{2}=1/\sqrt{b_{2}-b_{1}}=-ic_{1}.$
\end{rem}

Note that the standard boundary conditions are real-linear, that is,
if $\cvr$-spinors $f_{1,2}$ both satisfy the standard boundary conditions
on $S$, then so does 
\[
\alpha_{1}f_{1}+\alpha_{2}f_{2}
\]
 for any $\alpha_{1,2}\in\R$. Also, the standard boundary conditions
are local, except for (\ref{eq: bc_free_arc}). Finally, if they hold
on $S$, they also hold on $S^{*}$; therefore, we may speak of the
standard boundary conditions on $S\subset\pa\Omega$ rather than on
$S\subset\pa\Ocvrc$.

A simple, but fundamental property of the standard boundary conditions
is the conformal covariance.
\begin{lem}
\label{lem: SBC_conformal}Let $\Omega,\Oother$ be two nice domains,
$\varphi:\Omega\to\Oother$ a conformal isomorphism, and $\cvrother$
denote a double cover of $\Oother$ isomorphic to $\cvr$. Let $f$
be a $\cvrother$-spinor in $\Oother$ satisfying the standard boundary
conditions on $\Sother\subset\Oother$. Then, $f(\varphi(z))\cdot\varphi'(z)^{\frac{1}{2}}$
is a $\cvr$-spinor in $\Omega$ satisfying the standard boundary
conditions on $S=\varphi^{-1}(\Sother)$. 
\end{lem}

\begin{proof}
Note that $\varphi'(\cdot)^{\frac{1}{2}}$ is a well-defined \textit{function}.
Indeed, since $\varphi'$ does not vanish, its square root is well
defined on any simply connected sub-domain of $\Omega$, whence we
only have to check that the monodromy of $\varphi'(\cdot)^{\frac{1}{2}}$
around each boundary component is trivial. For $z\in\pa\Omega$, we
have 
\begin{equation}
\tau_{\varphi(z)}=\frac{\varphi'(z)}{|\varphi'(z)|}\tau_{z}.\label{eq: tau_CC}
\end{equation}
As $z$ goes around a boundary component, the winding number of both
$\tau_{z}$ and $\tau_{\varphi(z)}$ around zero is odd, hence, the
winding number of $\varphi'(z)$ is even and $(\varphi'(z))^{\frac{1}{2}}$
does not change the sign. This shows that $f(\varphi(z))\varphi'(z)^{\frac{1}{2}}$
is indeed a $\cvr$-spinor. The boundary conditions (\ref{eq: bc_fixed})
and (\ref{eq: bc_free}) follow readily from (\ref{eq: tau_CC}),
and the boundary conditions $(\ref{eq: bc_branchpoint})$ and (\ref{eq: bc_free_arc})
follow by Taylor expansions.
\end{proof}
Lemma \ref{lem: SBC_conformal} ensures the consistency of the following
definition:
\begin{defn}
\label{def: BC_cont_general}If $\Omega$ is a finitely connected
domain and $f$ is a spinor in $\Omega$, we say that $f$ satisfies
the standard boundary conditions on $S\subset\pa\Omega$ if for some
(equivalently, for any) conformal isomorphism $\varphi$ from a nice
domain $\Oother$ to $\Omega$, the spinor $f(\varphi(\cdot))\varphi'(\cdot)^{\frac{1}{2}}$
satisfies the standard boundary conditions on $\varphi^{-1}(S)$.
\end{defn}

Given a spinor $f$ in $\Omega$, we can define $h=\im\int f^{2}dz$,
which is a harmonic function on the universal cover of $\Omega$ whose
values on any two sheets of the universal cover differ by a constant.
The standard boundary conditions imply somewhat more robust Dirichlet
boundary conditions for the function $h$, as shown in the following
Proposition. 
\begin{prop}
\label{prop: f_to_h}Suppose $f$ satisfies the standard boundary
conditions on a subset $S$ of the boundary, and let $h=\im\int f^{2}dz$.
Then, the following holds true: 

\begin{enumerate}
\item The function $h$ is bounded near each point of $S$ which is not
a puncture of $\Omega$;
\item If $\nu\subset S$ is an arc, then $h$ is constant on $\nu\cap\fixed$
and $\pa_{i\tau}h\geq0$ on $\nu\cap\fixed$;
\item If $\nu\subset S\cap\free$ is an arc, then $h$ is constant on $\nu$
and $\pa_{i\tau}h\leq0$ on $\nu$;
\item If $\{v\}\subset S$ is a puncture of $\Omega$, then there exists
a number $\beta\geq0$ such that 
\begin{equation}
h(z)=\beta\log|z-v|+O(1),\quad z\to v.\label{eq: h-near-puncture}
\end{equation}
\end{enumerate}
\end{prop}

\begin{proof}
First, note that (\ref{eq: bc_fixed}) (respectively, (\ref{eq: bc_free}))
is equivalent to $\pa_{\tau}h=0$ and $\pa_{i\tau}h\geq0$ on $S\cap\fixed$
(respectively, $\pa_{\tau}h=0$ and $\pa_{i\tau}h\leq0$ on $S\cap\free$).
Now, the asymptotics in (\ref{eq: bc_free_arc}) imply that the function
$h$ is bounded near the endpoints of free arcs and the identity (\ref{eq: bc_free_arc})
implies that the jumps of $h$ at the two endpoints of each free have
equal absolute values and opposite signs. Finally, (4) is simply a
restatement of (\ref{eq: bc_branchpoint}).
\end{proof}
\begin{rem}
\label{rem: bc_h_to_bc_f}Let $h$ be a harmonic function on a nice
domain $\Omega$ (possibly multi-valued with a single-valued gradient)
satisfying (1)\textendash (4) above. In general, $f=\sqrt{2\pa_{z}h}$
needs not even be a well-defined spinor in $\Omega$ as the gradient
of $h$ might have odd-degree zeros. However, it follows directly
form the proof that if $f=\sqrt{2\pa_{z}h}$, for some reason, \emph{does}
happen to be a $\cvr$-spinor in $\Omega$, then $f$ satisfies standard
boundary conditions on $S$, except that $c_{2}=ic_{1}$ in (\ref{eq: bc_free_arc})
should be replaced by $c_{2}=\pm ic_{1}$; cf. Remark \ref{rem: free-arcs-cond}.
The sign consistency in (\ref{eq: bc_free_arc}) and the spinor property
of $f$ have no simple formulations in terms of $h$ itself. 
\end{rem}

\begin{rem}
\label{rem: h_well_defined}Note that the conditions above imply,
in particular, that if $S$ contains every boundary component, except,
maybe, one of them, then $h$ is a well defined \emph{function} in
$\Omega$ since its additive monodromy around all but one boundary
components vanishes and the last one can be written as the sum of
all others.
\end{rem}

The following Propositions gives a crucial uniqueness result for spinors
satisfying the standard boundary conditions everywhere on $\pa\Omega$. 
\begin{prop}
\label{prop: Uniqueness_continuous}Let $\Omega$, $\cvr$, $f$ be,
respectively, a finitely connected domain, its double cover, and a
holomorphic $\cvr$-spinor in $\Omega$ satisfying the standard boundary
conditions on the \emph{entire} boundary $\pa\Omega$. Assume furthermore
that $\cvr$ is not ramified over those connected components of $\C\setminus\Omega$
whose boundary is fully contained in $\free\subset\pa\Omega$. Then,
$f\equiv0$. 
\end{prop}

\begin{rem}
Let us stress that the sign consistency in (\ref{eq: bc_free_arc})
cannot be disposed of, as the example given in Remark \ref{rem: free-arcs-cond}
shows. The assumption on the double cover $\cvr$ not being ramified
over purely free boundary components is also essential. 
\end{rem}

\begin{proof}
By conformal covariance, we may assume that $\Omega$ is nice. Consider
$h=\im\int f^{2}dz$, which is a harmonic function in $\Omega$, well
defined up to an additive constant. It follows from Proposition \ref{prop: f_to_h}
that the function $h$ is bounded from \emph{above} on $\pa\Omega$
and the maximum principle implies that $h$ attains its maximum at
a boundary point $\zeta\in\pa\Omega$. If $\{\zeta\}$ is a puncture
(single-point boundary component) of $\Omega$, then one should have
$\beta=0$ in (\ref{eq: h-near-puncture}), which would mean that
$h$ is in fact harmonic at $z$ and hence $h\equiv\const$. Also,
$z\in\fixed$ also implies $h\equiv\const$ since $\pa_{i\tau}h\ge0$
on wired arcs, see Proposition \ref{prop: f_to_h}. Hence, we may
assume that $\zeta\in\free$. 

Let $\nu$ denote the free arc containing $\zeta$, recall that $h$
is constant on $\nu$ by Proposition \ref{prop: f_to_h}. Note that
$f$ cannot vanish on $\nu$ unless $h\equiv\const$: if \textbf{$f(\zeta)=0$}
for some $\zeta\in\nu$ but $f\not\equiv0$, then the values of $h$
in a vicinity of $\zeta$ cannot be bounded from above by $h(\zeta)$.
Therefore, the function $if(\zeta)\sqrt{\tau_{\zeta}}\in\R$ locally
does not change its sign on $\nu$. If $\nu$ comprises an entire
boundary component, then this shows that $f(\zeta)$ picks up a $-1$
factor as $\zeta$ winds around $\nu$, contradicting the assumption
that $\cvr$ is not ramified over such a boundary component. Finally,
let $\nu$ be a free arc with endpoints $b_{1,2}$. By conformal invariance,
we may assume that the boundary component containing $\nu$ is the
real line and that $\gamma=(b_{1},b_{2})$. The fact that the function
$if$ does not change sign on $(b_{1},b_{2})$ implies that one should
have $c_{2}=-ic_{1}$ in (\ref{eq: bc_free_arc}); cf. Remark \ref{rem: free-arcs-cond}.
This could only be possible if $c_{1}=c_{2}=0$, which would mean
that $h$ also attains the same maximal value on the nearby wired
arcs, the case already discussed above.
\end{proof}

\subsection{Fermionic observables in the bulk of $\Omega$}

The goal of this section is to define the continuous counterparts
of discrete two-point fermionic observables $F_{\Od,\double{\vv}}(z_{1},\z_{2})$
discussed in Section \ref{subsec: fermions-discrete}, in the situation
when both points $z_{1},z_{2}$ stay in the bulk of $\Omega$ and
at definite distance from $\vv$ and from each other. The convergence
statement in this situation is given by Theorem \ref{thm: conv_bulk_bulk}.
\begin{defn}
\label{def: feta} Let $\Omega$ be a finitely connected domain, equipped
with boundary conditions and a double cover $\cvr$, which is not
ramified over purely free boundary components of $\text{\ensuremath{\Omega}}$,
and let $a\in\Ocvrc$ and $\eta\in\C$. We define the continuous observable 

\[
\feta(a,z)=\feta_{\Omega,\cvr}(a,z)
\]
to be the unique holomorphic $\cvr$-spinor in $\Omega\setminus\{\cvr(a)\}$
that satisfies the standard boundary conditions everywhere on $\pa\Ocvrc\setminus\{a,a^{*}\}$
and the expansion

\begin{equation}
\feta_{\Omega,\cvr}(a,z)=\frac{\bar{\eta}}{\z-a}+O(1)\ \ \text{as}\ \ z\to a.\label{eq: f_residue}
\end{equation}
\end{defn}

\begin{rem}
\label{rem: cont_obs_existence_and_uniqueness}The uniqueness asserted
in this definition follows from Proposition \ref{prop: Uniqueness_continuous}:
if there were two such spinors, then their difference would be holomorphic
at $a,a^{*}$ and would satisfy the standard boundary conditions on
the entire $\pa\Ocvrc$. The existence will follow from the proof
of Theorem \ref{thm: conv_bulk_bulk} (postponed until Section \ref{sec: proofs_of_convergence_theorems}).
Note that since the standard boundary conditions are real-linear,
$f^{[\eta]}$ is a real-linear function of $\eta$, that is, if $\alpha,\beta\in\R$,
then 
\[
f^{[\alpha\eta_{1}+\beta\eta_{2}]}=\alpha f^{[\eta_{1}]}+\beta f^{[\eta_{2}]}.
\]
\end{rem}

\begin{lem}
\label{lem: Conformal-covariance-observable} If $\varphi:\Omega\to\Omega'$
is a conformal isomorphism and $\cvrother$ is isomorphic to $\cvr$,
then 
\begin{equation}
f_{\Omega,\cvr}^{[\eta]}(a,z)=f_{\Oother,\cvrother}^{[\etaother]}(\varphi(a),\varphi(z))\cdot\varphi'(z)^{\frac{1}{2}},\quad\text{where}\ \ \etaother=\eta\cdot\overline{\varphi'(a)}\vphantom{)}^{\frac{1}{2}}.\label{eq: feta-conf-cov}
\end{equation}
\end{lem}

\begin{proof}
The claim easily follows from Lemma \ref{lem: SBC_conformal}, the
uniqueness of $f^{[\eta]}$, and the following computation: 
\[
\frac{\bar{\eta}'}{\varphi(z)-\varphi(a)}\cdot\varphi'(z)^{\frac{1}{2}}\sim\frac{\bar{\eta}'}{z-a}\cdot\varphi'(a)^{-\frac{1}{2}}=\frac{\bar{\eta}}{z-a}.\qedhere
\]
\end{proof}
With the definition above, the roles of the arguments $a$ and $z$
of $f$ are apparently quite different. We now elucidate the (anti)-symmetry
between them (cf. the anti-symmetry between the arguments of the discrete
observables $F$, see Lemma \ref{lem: anti-symmetry}). For shortness,
below we omit the subscripts $\Omega,\cvr$ from the notation if no
confusion arises.
\begin{lem}
\label{lem: antisym_cont}Given $z_{1},z_{2}\in\Omega_{\cvr}$ and
\textup{$\eta_{1},\eta_{2}\in\C$}, define 
\[
f^{[\eta_{1},\eta_{2}]}(z_{1},z_{2})\text{:}=\re[\bar{\eta}_{2}f^{[\eta_{1}]}(\zz)].
\]
Then, $f^{[\eta_{1},\eta_{2}]}(z_{1,}z_{2})=-f^{[\eta_{2},\eta_{1}]}(z_{2},z_{1}).$
\end{lem}

\begin{proof}
We assume that $\Omega$ is nice, the general case follows due to
the conformal covariance of both sides provided by Lemma \ref{lem: Conformal-covariance-observable}.
Let $\mathcal{V^{\reg}}$ denote the set of all single-point boundary
components (punctures) of $\Omega$, and let $\mathcal{V}^{\flat}$
be the set of all endpoints of free arcs. Note that $g(\cdot):=f^{[\eta_{1}]}(z_{1},\cdot)f^{[\eta_{2}]}(z_{2},\cdot)$
is a well-defined holomorphic \emph{function} in $\Omega\setminus\{z_{1},z_{2}\}$.
Moreover, it is continuous up to $\pa\Omega$, except, possibly, for
simple poles at $\mathcal{V^{\reg}}$ and $\mathcal{V}^{\flat}$,
see the conditions (\ref{eq: bc_branchpoint}) and (\ref{eq: bc_free_arc}).
The Cauchy integral formula yields the identity
\begin{multline}
\text{v.p.}\int_{\fixed\cup\free}g(\zeta)d\zeta-2\pi i\sum_{v\in\mathcal{V}^{\reg}}\res_{z=v}g(z)-\pi i\sum_{b\in\mathcal{\mathcal{V}^{\flat}}}\res_{\zeta=b}g(\zeta)\\
=2\pi i\left(\res_{z=z_{1}}g(z)+\res_{z=z_{2}}g(z)\right)\\
=2\pi i\left(\bar{\eta}_{1}f^{[\eta_{2}]}(z_{2},z_{1})+\bar{\eta}_{2}f^{[\eta_{1}]}(z_{1},z_{2})\right),\label{eq: antisymm_proof}
\end{multline}
where the principal integral in the first line is understood as the
limit of integrals over $\{\zeta\in\pa\Omega:\dist(\zeta,\mathcal{\mathcal{V}^{\flat}})>r\}$
as $r\to0$. We claim that the imaginary part of (\ref{eq: antisymm_proof})
vanishes. Indeed, note that if $\zeta\in\fixed$ or $\zeta\in\free$,
then we have $g(\zeta)\tau_{\zeta}\in\R$ by (\ref{eq: bc_fixed})
and (\ref{eq: bc_free}). If $v\in\mathcal{V}^{\sharp}$, then we
have $\res_{z=v}g(z)\in i\R$ by (\ref{eq: bc_branchpoint}). Finally,
if $b_{1},b_{2}\in\mathcal{V}^{\flat}$ are the two endpoints of a
free arc, then we have $\res_{\zeta=b_{1}}g(\zeta)=-\res_{\zeta=b_{2}}g(\zeta)$
by (\ref{eq: bc_free_arc}). Therefore, by taking the imaginary part
of the equation (\ref{eq: antisymm_proof}), we get 
\[
\re\left(\bar{\eta}_{1}f^{[\eta_{2}]}(z_{2},z_{1})+\bar{\eta}_{2}f^{[\eta_{1}]}(z_{1},z_{2})\right)=0,
\]
as required.
\end{proof}
We are now ready to formulate the first convergence result for fermionic
observables. Recall the setup for the convergence of discrete domains
$\Od$ to $\Omega$ discussed in Section \ref{subsec: domains_convergence}.
The proof of this theorem is postponed until Section \ref{sec: proofs_of_convergence_theorems}.
\begin{thm}
\label{thm: conv_bulk_bulk} Under the assumption that $\cvr=\double{\vv}$
is not ramified over purely free boundary components of $\Omega$,
one has, as $\delta\to0$, 
\begin{equation}
\delta^{-1}F_{\Od,\cvr}(\zz)=\Cpsi^{2}\cdot\eta_{z_{1}}\eta_{z_{2}}\cdot f_{\Omega,\cvr}^{[\eta_{z_{1}},\eta_{z_{2}}]}(\zz)+o(1).\label{eq: conv_1}
\end{equation}
uniformly in $z_{1},z_{2}$ and $\{\vv\}\cap\Omega$ in the bulk of
$\Omega$ and away from each other.
\end{thm}

In the left-hand side of (\ref{eq: conv_1}), $z_{1}$ and $z_{2}$
denote corners of the graph $\Ocvr$, while the arguments of $f^{[\eta_{z_{1}},\eta_{z_{2}}]}(\cdot,\cdot)$
are points in $\Ocvrc$. The information about the orientation of
the corners $z_{1},z_{2}$ is retained in $\eta_{z_{1}},\eta_{z_{2}}$,
which are defined up to signs. However, since $f^{[-\eta]}=-f^{[\eta]},$
the right-hand side of (\ref{eq: conv_1}) does not depend on the
choice of these signs. 

While Lemma \ref{lem: antisym_cont} already reveals the symmetry
between the two arguments of continuous fermionic observables $\feta(\cdot,\cdot)$,
it is convenient to make one more step and to pass from $\feta(\cdot,\cdot)$
to functions, which are not only holomorphic in the second variable
but also (anti-)holomorphic in the first.
\begin{lem}
\label{lem: spinor_hol_antyhol}In the setup described above, the
spinor $\feta(z_{1},z_{2})$, $z_{1},z_{2}\in\Omega$ can be written
as a linear combination
\begin{equation}
\feta(z_{1},z_{2})=\tfrac{1}{2}\left(\bar{\eta}f(z_{1},z_{2})+\eta\fdag(z_{1},z_{2})\right).\label{eq: f_eta_decomp}
\end{equation}
of two $\cvr$-spinors $f$ and $f^{\star}$. Both $f$ and $f^{\star}$
are holomorphic in their second argument and satisfy the following
anti-symmetry relations:
\begin{equation}
f(z_{1,}z_{2})=-f(z_{2},z_{1}),\qquad\fdag(z_{1},z_{2})=-\bar{\fdag(z_{2},z_{1})}.\label{eq: antysymm_f_fdag}
\end{equation}
In particular, $f$ is holomorphic in the first argument, and $\fdag$
is anti-holomorphic in the first argument. Moreover, $f$ and $\fdag$
satisfy the following asymptotic expansions as $z_{2}\to z_{1}$:
\begin{equation}
f(z_{1},z_{2})=\frac{2}{z_{2}-z_{1}}+O(z_{2}-z_{1}),\qquad\fdag(z_{1},z_{2})=\fdag(z_{1},z_{1})+O(z_{2}-z_{1}),\label{eq: f_fdag_expansion}
\end{equation}
where $f^{\star}(z_{1},z_{1})\in i\R.$
\end{lem}

\begin{proof}
Since $f^{[\eta]}(z_{1},z_{2})$ is a real-linear function of $\eta$,
it can be written in the form (\ref{eq: f_eta_decomp}) with some
complex numbers $f(z_{1},z_{2})$ and $\fdag(z_{1},z_{2})$. Plugging
$\eta=1$ and $\eta=i$ into (\ref{eq: f_eta_decomp}), we see that
\begin{equation}
f=f^{[1]}+if^{[i]},\qquad\fdag=f^{[1]}-if^{[i]};\label{eq: f_f_dagger_back}
\end{equation}
hence both $f(z_{1},z_{2})$ and $\fdag(z_{1},z_{2})$ are holomorphic
functions of $z_{2}$. For all $\eta_{1},\eta_{2}\in\C$, it follows
from Lemma \ref{lem: antisym_cont} that
\begin{multline*}
\re\left(\bar{\eta}_{1}\bar{\eta}_{2}f(z_{1},z_{2})+\eta_{1}\bar{\eta}_{2}\fdag(z_{1},z_{2})\right)\\
=2f^{[\eta_{1},\eta_{2}]}(z_{1},z_{2})=-2f^{[\eta_{2},\eta_{1}]}(z_{2},z_{1})\\
=-\re\left(\bar{\eta}_{1}\bar{\eta}_{2}f(z_{2},z_{1})+\bar{\eta}_{1}\eta_{2}\fdag(z_{2},z_{1})\right),
\end{multline*}
which is only possible if the relations (\ref{eq: antysymm_f_fdag})
are fulfilled (indeed, both sides are linear combinations of $\bar{\eta}_{1}\bar{\eta}_{2}$,
$\eta_{1}\bar{\eta}_{2}$, $\bar{\eta}_{1}\eta_{2}$ and $\eta_{1}\eta_{2}$,
therefore all the coefficients in front of these four terms should
match). In particular, $f(z_{1},z_{2})$ is a holomorphic function
of $z_{1}$ while $\fdag(z_{1},z_{2})$ is anti-holomorphic in $z_{1}.$
Moreover, one easily sees that (\ref{eq: antysymm_f_fdag}) are fulfilled
for each fixed $z_{1}\in\Omega$. Finally, by Hartogs' theorem the
functions

\[
f(z_{1},z_{2})-2(z_{2}-z_{1})^{-1}\quad\text{and\ensuremath{\quad f^{\star}(\bar z_{1},z_{2})}}
\]
are jointly analytic as functions of two complex variables. The anti-symmetry
relations imply $f(z_{1},z_{2})-2(z_{2}-z_{1})^{-1}=O(z_{2}-z_{1})$
and $f^{\star}(z_{1},z_{1})=-\overline{f^{\star}(z_{1},z_{1})}$.
\end{proof}
\begin{rem}
\label{rem: Hartogs}As a byproduct of the proof given above we see
that the function $f^{[\eta_{1},\eta_{2}]}(z_{1},z_{2})$ is real
analytic on $\{z_{1},z_{2}\in\Ocvrc,\ \eta_{1},\eta_{2}\in\C:\cvr(z_{1})\neq\cvr(z_{2})\}$
as a function of eight real variables $\re z_{1}$, $\im z_{1}$,
$\re\eta_{1}$, $\im\eta_{1}$, $\re z_{2}$, $\im z_{2}$, $\re\eta_{2}$,
$\im\eta_{2}.$
\end{rem}

It is useful to rewrite the standard boundary conditions (see Definition
\ref{def: Standard_BC}) for the functions $\feta$ in terms of $f$
and $\fdag$. Let $\Omega$ be a nice domain and $\cvr$ its double
cover not ramified over purely free boundary components of $\pa\Omega$.
The following is fulfilled (the proofs are straightforward, e.g.,
from (\ref{eq: f_f_dagger_back})):
\begin{itemize}
\item If $\zeta\in S\cap\fixed$, then both $f(z_{1},\cdot)$ and $\fdag(z_{1},\cdot)$
extend continuously to $\zeta,$ and 
\begin{equation}
\fdag(z_{1,}\zeta)=\tau_{\zeta}\cdot\overline{f(z_{1},\zeta)}.\label{eq: bc_fixed-1}
\end{equation}
\item If $\zeta\in S\cap\free$, then both $f(z_{1},\cdot)$ and $\fdag(z_{1},\cdot)$
extend continuously to $\zeta,$ and 
\begin{equation}
\fdag(z_{1,}\zeta)=-\tau_{\zeta}\cdot\overline{f(z_{1},\zeta)}.\label{eq: bc_free-1}
\end{equation}
\item If $\{v\}$ is a puncture, then there exists a constant $c(z_{1})$
such that
\begin{equation}
f(z_{1},z_{2})=\frac{\lambb\cdot c(z_{1})}{\sqrt{z_{2}-v}}+O(1)\ \ \text{and}\ \ \fdag(z_{1},z_{2})=\frac{\lambb\cdot\overline{c(z_{1})}}{\sqrt{z_{2}-v}}+O(1)\ \ \text{as}\ \ z_{2}\to v.\label{eq: bc_branchpoint-1}
\end{equation}
\item If $\nu$ is a free boundary arc and $b_{1},b_{2}$ are the endpoints
of $\nu$, following in this order when tracing $\nu$ in the direction
of $\tau$, then there exists a constant $c(z_{1})$ such that 
\begin{equation}
\begin{array}{ccc}
f(z_{1},z_{2})={\displaystyle \vphantom{\Biggl|}\frac{-i\cdot c(z_{1})}{\sqrt{z_{2}-b_{1}}}}+O(1), & \fdag(z_{1},z_{2})={\displaystyle \frac{-i\cdot\overline{c(z_{1})}}{\sqrt{z_{2}-b_{1}}}}+O(1), & z\to b_{1},\\
f(z_{1},z_{2})={\displaystyle \frac{c(z_{1})}{\sqrt{z_{2}-b_{2}}}}+O(1), & \fdag(z_{1},z_{2})={\displaystyle \frac{\bar{c(z_{1})}}{\sqrt{z_{2}-b_{2}}}}+O(1), & z\to b_{2},
\end{array}\label{eq: bc_free_arc-1}
\end{equation}
with the same conventions on the choice of the signs of $f,\fdag$
and $(z-b_{1,2})^{-\frac{1}{2}}$ on $\text{\ensuremath{\Omega_{\cvr}} near }$$\nu$
as in (\ref{eq: bc_free}). 
\end{itemize}
\begin{rem}
(Schwarz reflection). \label{rem: Schwarz-reflection} Let $\Omega$
and $\zeta\in\pa\Omega$ be such that $\Omega\cap B_{r}(\zeta)=\H\cap B_{r}(\zeta)$
for $r$ small enough. If $\zeta\in\fixed$, it easily follows from
(\ref{eq: bc_fixed-1}), (\ref{eq: bc_free-1}) and the anti-symmetry
relations (\ref{eq: antysymm_f_fdag}) that the function $f(z_{1},z_{2})$
can be meromorphically continued in both its arguments to $z_{1},z_{2}\in B_{r}(\zeta)$
as
\begin{equation}
f(z_{1},z_{2}):=\begin{cases}
f(z_{1},z_{2}), & \im z_{1}\ge0,\ \im z_{2}\ge0;\\
\overline{\fdag(z_{1},\bar z_{2})}, & \im z_{1}\ge0,\ \im z_{2}\le0;\\
\fdag(\bar z_{1},z_{2}), & \im z_{1}\le0,\ \im z_{2}\ge0;\\
\overline{f(\bar z_{1},\bar z_{2})}, & \im z_{1}\le0,\ \im z_{2}\le0.
\end{cases}\label{eq: schwarz-f,fdag}
\end{equation}
A similar statement holds for the function $\fdag(z_{1},z_{2}):=f(\bar z_{1},z_{2})$,
$z_{1},z_{2}\in B_{r}(\zeta)$, with the relations (\ref{eq: antysymm_f_fdag})
preserved. If $w\in\free$, a similar analytic continuation in both
variables appears. Finally, if $b$ is an endpoint of a free arc,
a continuation similar to (\ref{eq: schwarz-f,fdag}) still exists
but leads to \emph{spinors} ramified over $b$ rather than to functions,
see (\ref{eq: bc_free_arc-1}).
\end{rem}

\begin{example}
\label{exa: f_half-plane_explicit} (i) If $\Omega=\H$, equipped
with wired boundary conditions (and $\cvr$ is trivial), then 
\[
f(z_{1},z_{2})=\frac{2}{z_{2}-z_{1}},\qquad f^{\star}(z_{1},z_{2})=\frac{2}{z_{2}-\bar z_{1}}.
\]
(ii) If $\Omega=\H\setminus\{v\}$, again with wired boundary conditions
on $\R$, and $\cvr=\double v,$ then
\[
f(z_{1},z_{2})=\fdag(\bar z_{1},z_{2})=\frac{1}{z_{2}-z_{1}}\left(\sqrt{\frac{(z_{1}-v)(z_{2}-\bar v)}{(z_{2}-v)(z_{1}-\bar v)}}+\sqrt{\frac{(z_{2}-v)(z_{1}-\bar v)}{(z_{1}-v)(z_{2}-\bar v)}}\right).
\]
(iii) If $\Omega=\H$ is equipped with free boundary conditions on
the interval $(b_{1},b_{2})\subset\R$ and wired ones outside this
interval (and $\cvr$ is trivial), then

\[
f(z_{1},z_{2})=\frac{1}{z_{2}-z_{1}}\left(\sqrt{\frac{(z_{1}-b_{1})(z_{2}-b_{2})}{(z_{2}-b_{1})(z_{1}-b_{2})}}+\sqrt{\frac{(z_{2}-b_{1})(z_{1}-b_{2})}{(z_{1}-b_{1})(z_{2}-b_{2})}}\right)
\]
and $\fdag(z_{1},z_{2})=-f(\bar z_{1},z_{2})$, where $\bar z_{1}=z_{1}$
on $\R\setminus[b_{1},b_{2}]$ and $\bar z_{1}=z_{1}^{*}$ (as points
on the double cover $\C_{\double{b_{1},b_{2}}}$) on the free boundary
arc $(b_{1},b_{2})$$.$ 
\end{example}

\begin{proof}
In either case, it is enough to check that $f$ and $\fdag$ satisfy
the conditions (\ref{eq: antysymm_f_fdag}), (\ref{eq: f_fdag_expansion}),
$(\ref{eq: bc_fixed-1})$\textendash (\ref{eq: bc_free_arc-1}), and
that $f(z_{1},z_{2}),\fdag(z_{1},z_{2})=O(z_{2}^{-1})$ as $z_{2}\to\infty$;
the latter asymptotics is required when mapping $\Omega=\H$ to a
nice \emph{bounded} domain $\Oother$ and using the conformal covariance
rule (\ref{eq: feta-conf-cov}). This check is straightforward and
left to the reader.
\end{proof}
We now consider the situation when $z_{1}$, $z_{2}$ are opposite
corners of an edge in the bulk of $\Omega$. This is the case required
to handle the energy density correlations, see (\ref{eq: en_spin_disorder}).
\begin{thm}
\label{thm: Convergence_bulk_near}Assume that $z_{1,2}\in\Cgr(\Ocvr)$
are such that $z_{2}=z_{1}\pm i(z_{1}^{\bullet}-z_{1}^{\circ}).$
We have 
\[
\delta^{-1}F_{\Od,\cvr}(z_{1},z_{2})=\Cpsi^{2}\cdot\eta_{z_{1}}^{2}\cdot(-\tfrac{i}{2})\fdag_{\Omega,\cvr}(z_{1},z_{1})+o(1)
\]
as $\delta\to0$, uniformly in $z_{1}$ and $\{\vv\}\cap\Omega$ in
the bulk of $\Omega$ and away from each other.
\end{thm}

\begin{figure}
\includegraphics[width=0.5\textwidth]{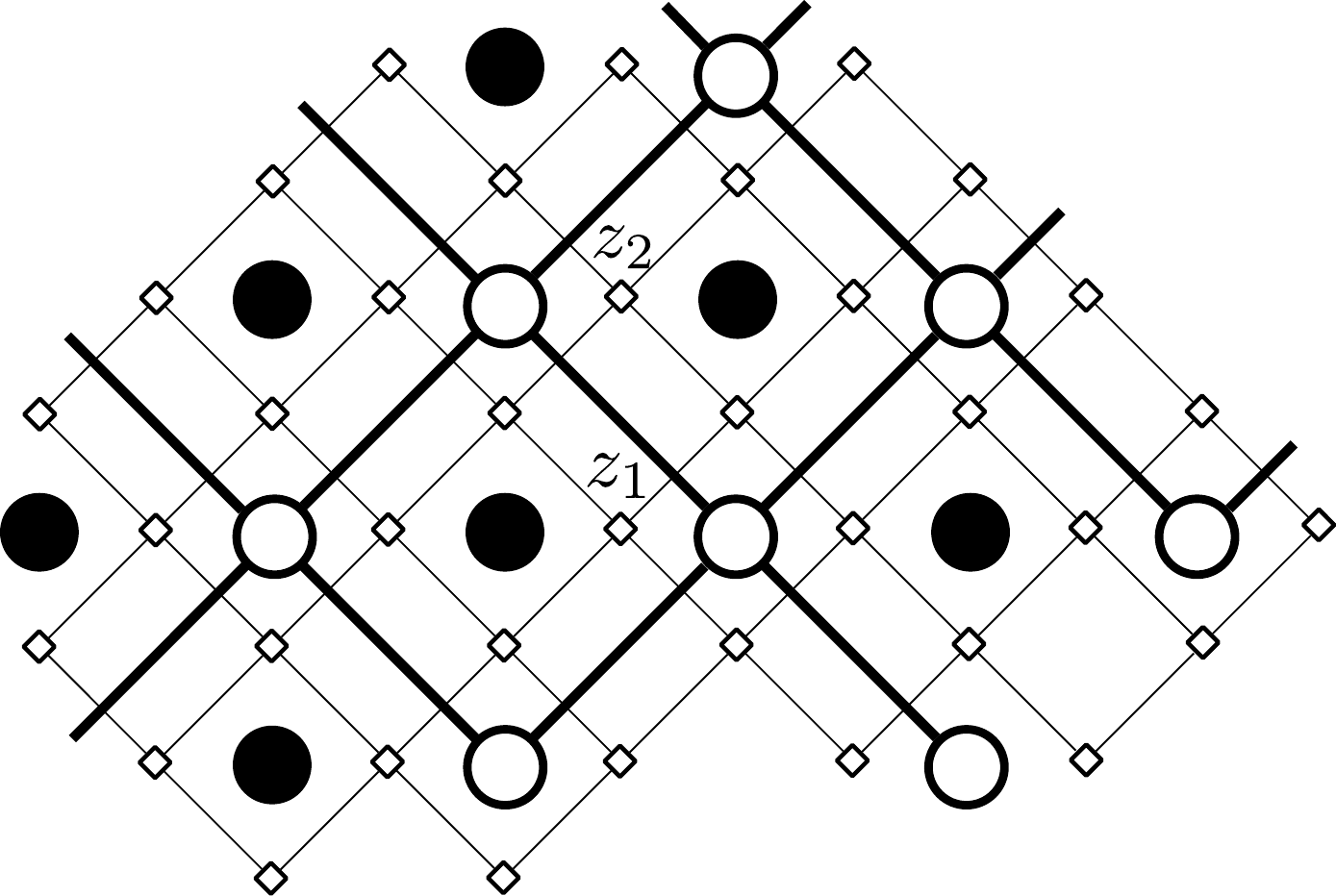}\caption{The configuration of points for Theorem \ref{thm: Convergence_bulk_near}\label{fig: energy}}
\end{figure}
Again, the proof of this theorem is postponed until Section \ref{sec: proofs_of_convergence_theorems}.
We conclude this section by the conformal covariance properties of
$f_{\Omega,\cvr}$ and $\fdag_{\Omega,\cvr}$.
\begin{lem}
\label{lem: ccov_f_fdag}Let $\varphi:\Omega\to\Oother$ be a conformal
map and $\cvrother$ be isomorphic to $\cvr$. As above, we assume
that $\cvr$ is not ramified over purely free boundary components.
Then, 
\begin{eqnarray}
f_{\Omega,\cvr}(z_{1},z_{2}) & = & f_{\Oother,\cvrother}(\varphi(z_{1}),\varphi(z_{2}))\cdot\varphi'(z_{1})^{\frac{1}{2}}\varphi'(z_{2})^{\frac{1}{2}},\label{eq: ccov_f}\\
\fdag_{\Omega,\cvr}(z_{1},z_{2}) & = & \fdag_{\Oother,\cvrother}(\varphi(z_{1}),\varphi(z_{2}))\cdot\overline{\varphi'(z_{1}}){}^{\frac{1}{2}}\varphi'(z_{2})^{\frac{1}{2}}.\label{eq: ccov_f_dag}
\end{eqnarray}
In particular, $\fdag_{\Omega,\cvr}(z,z)=\fdag_{\Oother,\cvrother}(\varphi(z),\varphi(z))\cdot|\varphi'(z)|$.
\end{lem}

\begin{proof}
This is a straightforward corollary of the conformal covariance (\ref{eq: feta-conf-cov})
of the function $\feta_{\Omega,\cvr}(z_{1},z_{2})$ and the identities
(\ref{eq: f_eta_decomp}), (\ref{eq: f_f_dagger_back}) relating the
functions $\feta$ and $f,\fdag$.
\end{proof}

\subsection{Fermionic observables at punctures and boundary points of $\Omega$}

We also need the convergence results for the cases when $z_{1}$ and/or
$z_{2}$ are allowed to be on the boundary of $\Omega$ (in particular,
next to its single-faced boundary component). To formulate such convergence
results concisely we introduce an additional notation.
\begin{defn}
\label{def: sharp}Let $a\in\Ocvrc$, and $\{v\}$ is a single-point
boundary component of $\Omega$, assume that $\cvr$ is ramified over
$v$. We define 

\begin{equation}
f^{[\eta,\reg]}(a,v):=\lim_{z\to v}\lamb(z-v)^{\frac{1}{2}}f^{[\eta]}(a,z).\label{eq: def-sharp}
\end{equation}
Due to ($\ref{eq: bc_branchpoint}$), this is a real quantity, well
defined up to the sign. This sign depends on the choice of the branch
of $\sqr{z-v}$ on $\Ocvrc$ locally near $v$, which we leave unfixed
for now.
\end{defn}

\begin{defn}
\label{def: flat} Assume that $\Omega$ is nice, $a\in\Ocvrc$, and
$b\in\pa\Ocvrc$ is an endpoint of a free boundary arc. We define
\begin{equation}
f_{\Omega,\cvr}^{[\eta,\flat]}(a,b):=\lim_{z\to b}\sflat\sqr{z-b}f_{\Omega,\cvr}^{[\eta]}(a,z)\label{eq: def-flat}
\end{equation}
where $\sflat=\sflat_{b}:=1$ if the boundary conditions change at
$b$ from free to wired when tracing $\pa\Omega$ in the direction
of $\tau$, and $\sflat_{b}:=i$ if, on the contrary, the boundary
conditions change from wired to free. Once again, $f^{[\eta,\flat]}$
is a real quantity whose sign depends on the choice of $\sqr{z-b}$
on $\Ocvrc$ locally near $b$. Moreover, if $b_{1}$ and $b_{2}$
are the two endpoints of the same free arc, then
\begin{equation}
f_{\Omega,\cvr}^{[\eta,\flat]}(a,b_{1})=f_{\Omega,\cvr}^{[\eta,\flat]}(a,b_{2})\label{eq: f_flat_b1=00003Df_flat_b2}
\end{equation}
provided that the branches of $(z-b_{1})^{\frac{1}{2}}$ and $(z-b_{2})^{\frac{1}{2}}$
are chosen consistently; see (\ref{eq: bc_free_arc}). 

Let
\[
\mathcal{V}=\mathcal{V}^{\reg}\cup\mathcal{V}^{\flat}:=\pa\Omega\setminus(\fixed\cup\free),
\]
where $\mathcal{\mathcal{V}^{\sharp}}$ denotes the set of single-point
components (punctures) of $\Omega$, recall that we assume that $\cvr$
branches over all of them, and $\mathcal{V}^{\flat}$ is the set of
all endpoints of free arcs. 
\end{defn}

\begin{rem}
\label{rem: Conf_covariance_sharp_flat} From Lemma \ref{lem: Conformal-covariance-observable},
we immediately deduce the conformal covariance of $f^{[\eta,\sharp]}$
and $f^{[\eta,\flat]}$. Namely, in the notation of (\ref{eq: feta-conf-cov})
we have 
\[
f_{\Omega,\cvr}^{[\eta,\any]}(a,w)=f_{\Oother,\cvrother}^{[\etaother,\any]}(\varphi(a),\varphi(w)),\quad\text{where}\quad\any\in\{\reg,\flat\}\ \ \text{and}\ \ w\in\mathcal{V}^{\any},
\]
recall that $\etaother=\eta\cdot\overline{\varphi'(a)}\vphantom{)}^{\frac{1}{2}}$.
In particular, this allows to extend Definition \ref{def: flat} to
the case when $\Omega$ is not necessarily nice by the conformal covariance.
\end{rem}

\begin{rem}
\label{rem: _flat_shar_residues} Let $c^{\reg}:=\lambda$. Definitions
\ref{def: sharp} and \ref{def: flat} are given in terms of limits
as $z\to w$. Alternatively, for $\any\in\{\reg,\flat\}$ and $w\in\mathcal{V}^{\any}$,
we could write 
\[
f^{[\eta,\any]}(a,w)=c^{\any}\,\res_{z=w}(f^{[\eta]}(a,z)(z-w)^{-\frac{1}{2}})=\frac{c^{\any}}{2\pi i}\oint_{(w)}\frac{f^{[\eta]}(a,z)}{(z-w)^{\frac{1}{2}}}dz,
\]
where the integral is taken over a small contour encircling the point
$w$; to make sense of this formula in the case $\any=\flat$, we
assume that $\Omega$ is as in Remark \ref{rem: Schwarz-reflection}
and $f^{[\eta]}(a,z)$ is analytically continued (as a spinor ramified
over $w=b$) to a neighborhood of $w$. These integral representations,
together with Remark \ref{rem: Hartogs}, ensure that the limits (\ref{eq: def-sharp})\textendash (\ref{eq: def-flat})
are actually uniform in $a$ over compact subsets of $\Ocvrc$, or
any bigger domain to which $f$ and $\fdag$ can be (anti-)analytically
continued.
\end{rem}

In the next lemma, we study the properties of $f^{[\eta,\reg]}(a,v)$,
$v\in\mathcal{V}^{\reg}$, and $f^{[\eta,\flat]}(a,b)$, $b\in\mathcal{V}^{\flat},$
as functions of their first argument.
\begin{lem}
\label{lem: f_star_holom} Let $\any\in\{\reg,\flat\}$ and $w\in\mathcal{V}^{\any}$.
Then, there exists a holomorphic spinor $f^{[\any]}(w,\cdot)$ such
that, for any $\eta\in\C$, 
\[
f^{[\any,\eta]}(w,\cdot):=-f^{[\eta,\any]}(\cdot,w)=\re[\bar{\eta}\cdot f^{[\any]}(w,\cdot)]
\]
Moreover, $f^{[\any]}(w,\cdot)$ satisfies the standard boundary conditions
on $\pa\Omega\setminus\{w\}$ and 
\begin{equation}
f^{[\any,\eta]}(w,z)=\frac{\sany}{\sqrt{z-w}}+O(1)\quad\text{as}\ \ z\to w,\label{eq: f_sharp_asymp}
\end{equation}
where $\ssharp:=\lambda$ and $\sflat\in\{1,i\}$ is as in Definition
\ref{def: flat}; we assume that $\Omega$ is nice if $\any=\flat$.
\end{lem}

\begin{proof}
In the case $\any=\flat,$ we will assume that $\Omega$ is as in
Remark \ref{rem: Schwarz-reflection} and that the corresponding analytic
continuations are made; the general case follows by conformal covariance.
Recall the expansions (\ref{eq: bc_branchpoint-1}), (\ref{eq: bc_free_arc-1})
of $f(z_{1},z_{2})$ , $\fdag(z_{1},z_{2})$ as $z_{2}\to w$ and
let 
\begin{equation}
\begin{aligned}f^{[\any]}(w,z_{1}):= & -\frac{\sany}{2\pi i}\oint_{(w)}\frac{f(z_{1},z_{2})}{(z_{2}-w)^{\frac{1}{2}}}dz_{2},\\
\overline{f^{[\any]}(w,z_{1})}= & -\frac{\sany}{2\pi i}\oint_{(w)}\frac{\fdag(z_{1},z_{2})}{(z_{2}-w)^{\frac{1}{2}}}dz_{2},
\end{aligned}
\label{eq: def_f_any}
\end{equation}
the contour integrals are taken over a small contour encircling $w$
(in other words, we define $f^{[\any]}(w,z_{1})$ to be the constant
$c(z_{1})$ in the asymptotics (\ref{eq: bc_branchpoint-1}), (\ref{eq: bc_free_arc-1}),
multiplied by $-c^{\any}$). Using (\ref{eq: f_eta_decomp}), it is
easy to see that
\[
\re[\bar{\eta}\cdot f^{[\any]}(w,z_{1})]=-\frac{\sany}{2\pi i}\oint_{(w)}\frac{f^{[\eta]}(z_{1},z_{2})}{(z_{2}-w)^{\frac{1}{2}}}dz_{2}=-f^{[\eta,\any]}(z_{1},w),
\]
as required. Moreover, adding the first line of (\ref{eq: def_f_any})
and the conjugation of the second one, and using the anti-symmetry
properties (\ref{eq: antysymm_f_fdag}) of $f$ and $\fdag$, we obtain
the identity 
\[
\begin{aligned}f^{[\any]}(w,z_{1})\  & =\ \frac{1}{2}\biggl[\frac{c^{\any}}{2\pi i}\oint\frac{f(z_{2},z_{1})}{(z_{2}-w)^{\frac{1}{2}}}dz_{2}\ -\ \frac{\bar c^{\any}}{2\pi i}\oint_{(w)}\frac{\fdag(z_{2},z_{1})}{(\bar z_{2}-\bar w)^{\frac{1}{2}}}d\bar z_{2}\biggr]\\
 & =\ \frac{r^{\frac{1}{2}}}{2\pi}\int_{0}^{2\pi}f^{[\eta(\theta)]}(w+re^{i\theta},z_{1})d\theta,\qquad\eta(\theta):=\bar c^{\any}\cdot e^{-\frac{i}{2}\theta},
\end{aligned}
\]
where $z_{2}=w+re^{i\theta}$ and $r>0$ can be chosen so that $|z_{1}-w|>r$
provided that $z_{1}\ne w$. As the functions $f^{[\eta(\theta)]}(w+re^{i\theta},z_{1})$
satisfy the standard boundary conditions, so does $f^{[\any]}(w,z_{1})$
except, possibly, for the asymptotic behavior at the point $z_{1}=w$
itself. 

To prove (\ref{eq: f_sharp_asymp}), let $C_{1}$, $C_{2}$ and $C'_{1}$
be the circles of small radii $r_{1}>r_{2}>r'_{1}>0$ centered at
$w$ and note that $w$ is an isolated singularity of the holomorphic
function $(z_{1}-w)^{\frac{1}{2}}f^{[\any]}(w,z_{1}).$ The $(-k)$-th
coefficient of its Laurent expansion reads as

\begin{equation}
\begin{aligned}\frac{1}{2\pi i}\oint_{C_{1}}\frac{f^{[\any]}(w,z_{1})}{(z_{1}-w)^{\frac{1}{2}-k}}dz_{1} & =\frac{\sany}{(2\pi i)^{2}}\oint_{C_{1}}\oint_{C_{2}}\frac{f(z_{2},z_{1})dz_{2}}{(z_{1}-w)^{\frac{1}{2}-k}(z_{2}-w)^{\frac{1}{2}}}dz_{1}\\
 & =\frac{\sany}{(2\pi i)^{2}}\oint_{C_{2}}\oint_{C'_{1}}\frac{f(z_{2},z_{1})dz_{1}}{(z_{1}-w)^{\frac{1}{2}-k}(z_{2}-w)^{\frac{1}{2}}}dz_{2}+2\sany\delta_{k=0}\,,
\end{aligned}
\label{eq: exchange_the_contours}
\end{equation}
where along the exchange of the contours of integrations we have computed,
using (\ref{eq: f_fdag_expansion}), 
\[
\frac{1}{2\pi i}\oint_{C_{2}}\mathop{\res}\limits _{z=z_{2}}\frac{f(z_{2},z)}{(z-w)^{\frac{1}{2}-k}(z_{2}-w)^{\frac{1}{2}}}dz_{2}=\frac{1}{2\pi i}\oint_{C_{2}}\frac{2dz_{2}}{(z_{2}-w)^{1-k}}=2\delta_{k=0}.
\]
Since we have $f(z_{2},z_{1})=O((z_{1}-w)^{-\frac{1}{2}})$ as $z_{1}\to w$
(e.g., see (\ref{eq: bc_branchpoint-1})), the contour integral 
\[
\oint_{C'_{1}}\frac{f(z_{2},z_{1})}{(z_{1}-w)^{\frac{1}{2}-k}}dz_{1}
\]
vanishes for $k>0$, and hence so does the right-hand side of (\ref{eq: exchange_the_contours}).
Finally, for $k=0$ the anti-symmetry of $f$ and the deformation
of contours $(C_{1},C_{2})$ to $(C_{2},C'_{1})$ imply that
\[
\oint_{C_{1}}\oint_{C_{2}}\frac{f(z_{2},z_{1})dz_{2}}{(z_{1}-w)^{\frac{1}{2}}(z_{2}-w)^{\frac{1}{2}}}dz_{1}=-\oint_{C_{2}}\oint_{C'_{1}}\frac{f(z_{2},z_{1})dz_{1}}{(z_{1}-w)^{\frac{1}{2}}(z_{2}-w)^{\frac{1}{2}}}dz_{2}.
\]
Therefore, (\ref{eq: exchange_the_contours}) with $k=0$ implies
(\ref{eq: f_sharp_asymp}).
\end{proof}
We use Lemma \ref{lem: f_star_holom} to define the functions $f^{[\sharp,\sharp]}(\cdot,\cdot)$,
$f^{[\sharp,\flat]}(\cdot,\cdot)$, $f^{[\flat,\sharp]}(\cdot,\cdot)$
and $f^{[\flat,\flat]}(\cdot,\cdot)$ on the corresponding subsets
of $\mathcal{V}\times\mathcal{V}$; recall that $\text{\ensuremath{\mathcal{V}=\mathcal{V}^{\reg}\cup\mathcal{V}^{\flat}}, wher e }\mathcal{V}^{\reg}$
denotes the set of punctures of $\Omega$ (and that we assume that
the double cover $\cvr$ is ramified over all these points) while
$\mathcal{V}^{\flat}$ stands for the set of endpoints of free boundary
arcs.
\begin{defn}
\label{def: f-any-any} Let $\triangleleft,\anyother\in\{\reg,\flat\}$,
$w_{1}\in\mathcal{V}^{\triangleleft},$ $w_{2}\in\mathcal{V}^{\anyother}$
and $w_{1}\ne w_{2}.$ We define the quantities $f^{[\triangleleft,\anyother]}(w_{1},w_{2})$
as in Definitions \ref{def: sharp}, \ref{def: flat}, replacing $\feta(a,\cdot)$
everywhere by $f^{[\triangleleft]}(w_{1},\cdot)$. 
\end{defn}

\begin{lem}
The functions $f^{[\triangleleft,\anyother]}$, where $\triangleleft$
and $\anyother$ stand for any of $\sharp$ and $\flat,$ are anti-symmetric:
\[
f^{[\triangleleft,\anyother]}(w_{1},w_{2})=-f^{[\anyother,\triangleleft]}(w_{2},w_{1}),\quad w_{1}\in\mathcal{V}^{\triangleleft},\ \ w_{2}\in\mathcal{\mathcal{V}}^{\anyother},\ \ w_{1}\ne w_{2}.
\]
\end{lem}

\begin{proof}
Using Remark \ref{rem: _flat_shar_residues} and (\ref{eq: def_f_any})
we can write 
\begin{equation}
f^{[\triangleleft,\anyother]}(w_{1},w_{2})=\frac{c^{\anyother}}{2\pi i}\oint_{(w_{2})}\frac{f^{[\triangleleft]}(w_{1},z_{2})}{(z_{2}-w_{2})^{\frac{1}{2}}}dz_{2}=\frac{c^{\triangleleft}c^{\anyother}}{2\pi i}\oint_{w_{2}}\oint_{w_{1}}\frac{f(z_{1},z_{2})dz_{1}}{(z_{1}-w_{1})^{\frac{1}{2}}(z_{2}-w_{2})^{\frac{1}{2}}}dz_{2},\label{eq: f_any_any_anti}
\end{equation}
and the result follows from the anti-symmetry of $f(z_{1},z_{2})$. 
\end{proof}
In the next lemma we collect the conformal covariance properties of
$f^{[\any]}$ and $f^{[\triangleleft,\anyother]}$.
\begin{lem}
\label{lem: ccov_f_any}Let $\varphi:\Omega\to\Oother$ be a conformal
map and $\cvrother$ isomorphic to $\cvr$. Then, 
\begin{eqnarray}
f_{\Omega,\cvr}^{[\any]}(w,z) & = & f_{\Oother,\cvrother}^{[\any]}(\varphi(w),\varphi(z))\cdot\varphi'(z)^{\frac{1}{2}},\label{eq: ccov_sharp}\\
f_{\Omega,\cvr}^{[\triangleleft,\anyother]}(w_{1},w_{2}) & = & f_{\Oother,\cvrother}^{[\triangleleft,\anyother]}(\varphi(w_{1}),\varphi(w_{2})),\label{eq: ccov_sharp_sharp}
\end{eqnarray}
where each of the superscripts $\any$, $\triangleleft$ and $\anyother$
stands for either $\sharp$ or $\flat$. 
\end{lem}

\begin{proof}
Both relations readily follow from Lemma \ref{lem: ccov_f_fdag} using
(\ref{eq: def_f_any}) and (\ref{eq: f_any_any_anti}). 
\end{proof}
We are now almost ready to formulate an extension of Theorem \ref{thm: conv_bulk_bulk}
from the bulk of $\Omega$ to special boundary points $w\in\mathcal{V}$.
The last definition we need is the following.
\begin{defn}
\label{def: admissible-corners}We say that $z^{\delta}\in\Cgr(\Od_{\cvr})$
is an \emph{admissible sequence} of corners of if one of the following
holds: 
\end{defn}

\begin{itemize}
\item (type $\eta$): $z^{\delta}\to z\in\Ocvrc$ as $\delta\to0$; 
\item (type $\sharp$): for each $\delta$, $(z^{\delta})^{\circ}\in\{\vv\}$
and $z^{\delta}\to v$ as $\delta\to0$, where $\{v\}$ is a puncture
(i.e., a single-point boundary component) of $\Omega$;
\item (type $\flat$): for each $\delta$, $z^{\delta}$ is a corner separating
a free boundary arc from a wired one, meaning that $(z^{\delta})^{\circ}\in\fixed^{\circ}$
and $(z^{\delta})^{\bullet}\in\free^{\bullet}$, and $z^{\delta}\to b$,
where $b$ is an endpoint of a free arc of $\Omega_{\cvr}$. 
\end{itemize}
\begin{thm}
\label{thm: Convergence_singular} In the setup of Theorem \ref{thm: conv_bulk_bulk},
assume that two admissible sequences $z_{1}^{\delta}$ and $z_{2}^{\delta}$
of corners of $\Ocvr$ approximate two distinct points $z_{1},z_{2}\in\Omega_{\cvr}\cup\mathcal{V}$
as $\delta\to0$. Then, the following asymptotic holds: 
\begin{equation}
\delta^{-\Delta_{1}-\Delta_{2}}\cdot F_{\Od,\cvr}(z_{1}^{\delta},z_{2}^{\delta})\ =\ C_{1}\cdot C_{2}\cdot\eta_{z_{1}^{\delta}}\eta_{z_{2}^{\delta}}f_{\Omega,\cvr}^{[\triangleleft,\triangleright]}(\zz)+o(1),\label{eq: thm_Convergence _singular}
\end{equation}
where each of the indices $\triangleleft,\triangleright$ stands for
either $\eta_{z_{p}^{\delta}}$ or $\sharp$ or $\flat$, depending
on the type of the corresponding sequence $z_{p}^{\delta}$; the exponents
$\Delta_{p}$, $p=1,2$, are given by 
\[
\Delta_{p}=\begin{cases}
\frac{1}{2}, & \text{if }z_{p}^{\delta}\text{ is of type \ensuremath{\eta},}\\
0, & \text{\text{if }\ensuremath{z_{p}^{\delta}\text{ is of type }\reg}\ or \ensuremath{\flat};}
\end{cases}
\]
and the constants $C_{p}$, $p=1,2$, are given by 
\[
C_{p}=\begin{cases}
\Cpsi, & \text{if }z_{p}^{\delta}\text{ is of type \ensuremath{\eta},}\\
1, & \text{if }z_{p}^{\delta}\text{ is of type \ensuremath{\reg}\ or \ensuremath{\flat}}.
\end{cases}
\]
\end{thm}

The proof of Theorem \ref{thm: Convergence_singular} is postponed
until Section \ref{sec: proofs_of_convergence_theorems}.
\begin{rem}
(Signs). \label{rem: Signs} Recall that to define $f^{[\eta,\sharp]}(z_{1},v)$,
$v=z_{2}\in\mathcal{\mathcal{V}^{\reg}}$, it was necessary to choose
the sign of $(z-v)^{\frac{1}{2}}$ in a neighborhood of $v$ on $\Ocvrc$.
Also, $\eta_{z_{2}^{\delta}}=\lambb(\delta/2){}^{\frac{1}{2}}(z_{2}^{\delta}-(z_{2}^{\delta})^{\circ})^{-\frac{1}{2}}$
is defined up to a choice of the sign in the square root. In Theorem
\ref{thm: Convergence_singular}, these signs are assumed to be chosen
in the same way provided that $(z_{2}^{\delta})^{\circ}$ is identified
with $v$. Similarly, the definition of $f^{[\eta,\flat]}(z_{1},b)$
for $b=z_{2}\in\mathcal{V}^{\flat}$ involved a choice of the sign
of $(z-b)^{\frac{1}{2}}$ in a neighborhood of $b$ on $\Ocvrc$.
In Theorem \ref{thm: Convergence_singular}, this sign is assumed
to be the same as the sign of the square root in $\eta_{z_{2}^{\delta}}=\lamb(\delta/2)^{\frac{1}{2}}(z_{2}^{\delta}-(z_{2}^{\delta})^{\bullet})^{-\frac{1}{2}}$
provided that $(z_{2}^{\delta})^{\bullet}$ is identified with $b$
and $z_{2}^{\delta}$ is viewed as a point of $\Ocvr$. The same discussion
applies if $\triangleleft\in\{\reg,\flat\}$ .
\end{rem}

In the proof of Theorem \ref{thm: Intro_3} via Lemma \ref{lem: corr_to_obs_for_Thm_3},
we will also have to allow $z_{1}$ and/or $z_{2}$ to be on the wired
part of the boundary of $\Omega$. Without additional assumptions
on the regularity of $\pa\Omega$, it is not possible to handle the
asymptotics of $F(z_{1},z_{2})$ directly in this case; if fact, even
the right-hand side of (\ref{eq: thm_Convergence _singular}) may
not be well defined for $z_{1,2}\in\pa\Omega$ when $\pa\Omega$ is
rough. However, we are able to treat these asymptotics up to an unknown
\emph{local} normalizing factor. To formulate the corresponding result,
we extend Definition \ref{def: admissible-corners} of admissible
sequences $z^{\delta}$ by adding to it the forth type:
\begin{itemize}
\item (type $\beta$): for each $\delta$, $z^{\delta}$ is a corner on
a wired boundary arc of $\Od_{\cvr}$ and $z^{\delta}\to z$, where
$z\in\fixed\subset\pa\Omega_{\cvr}$, as $\delta\to0$.
\end{itemize}
For the lemma below, we \emph{fix} a conformal isomorphism $\varphi$
from $\Omega$ to a nice domain $\Lambda$ and let $\cvr_{\Lambda}$
be the pushforward of $\cvr$.

\begin{lem}
\label{lem: Clements_clever_lemma} Assume that $z_{2}\in\fixed\subset\pa\Omega_{\cvr},$
and let $z_{2}^{\delta}\to z_{2}$ be an admissible sequence of type
$\beta$. There exists a normalizing sequence $\normLoc{z_{2}}\in\C\setminus\{0\}$
depending only on $(\Lambda,\varphi)$, on the sequence $z_{2}^{\delta}$,
and on the boundary $\pa\Od$ locally near $z_{2}$ (but not on $\cvr$
or other marked points) such that the following extension of the asymptotics
(\ref{eq: thm_Convergence _singular}) to admissible sequences $z_{2}^{\delta}\to z_{2}$
of type $\beta$ holds:

\begin{equation}
\delta^{-\Delta_{1}}\cdot\normLoc{z_{2}}\cdot F_{\Od,\cvr}(z_{1}^{\delta},z_{2}^{\delta})\ =\ C_{1}\cdot\eta_{z_{1}^{\delta}}\eta_{z_{2}^{\delta}}\lim_{w_{2}\to z_{2}}\left(\varphi'(w_{2})^{-\frac{1}{2}}f_{\Omega,\cvr}^{[\any]}(z_{1},w_{2})\right)+o(1),\label{eq: conv_bdry_rough}
\end{equation}
as $\delta\to0$, where $C_{1},\Delta_{1},$ and $\any$ are as in
Theorem \ref{thm: Convergence_singular} depending on the type of
$z_{1}^{\delta}$. Moreover, if $z_{1}^{\delta}\to z_{1}$, $z_{2}^{\delta}\to z_{2}$
are both admissible sequences of type $\beta$ for $z_{1}\ne z_{2}$,
then
\begin{equation}
\normLoc{z_{1}}\normLoc{z_{2}}\cdot F_{\Od,\cvr}(z_{1}^{\delta},z_{2}^{\delta})\ =\ \eta_{z_{1}^{\delta}}\eta_{z_{2}^{\delta}}\lim_{w_{1,2}\to z_{1,2}}\left(\varphi'(w_{1})^{-\frac{1}{2}}\varphi'(w_{2})^{-\frac{1}{2}}f_{\Omega,\cvr}(w_{1},w_{2})\right)+o(1)\label{eq: conv_bdry_bdry}
\end{equation}
as $\delta\to0$. 
\end{lem}

\begin{rem}
To see that the limits in the right-hand side exist, observe that
by conformal covariance, $\varphi'(w_{2})^{-\frac{1}{2}}f_{\Omega,\cvr}^{[\any]}(z_{1},w_{2})=f_{\Lambda,\varphi(\cvr)}^{[\any']}(\varphi(z_{1}),\varphi(w_{2})),$
where $\any'=\any$ if $z_{1}^{\delta}$ is of type $\#$ or $\flat$,
and $\any'=\bar{\varphi'(z_{1})}^{-\frac{1}{2}}\eta_{z_{1}}$ for
type $\eta_{z}.$ Since $\Lambda$ is nice, the observables in $\Lambda$
are continuous up to the boundary. Same reasoning applies to (\ref{eq: conv_bdry_bdry}).
\end{rem}

\begin{rem}
If $z^{\delta}\to z$ is an admissible sequence of type $\beta$ and
the boundaries $\pa\Od$ are straight lines near $z$, then one can
also take $\Lambda=\Omega$, $\varphi=id$ and $\normLoc z:=C_{\psi}^{-1}\delta^{-\frac{1}{2}}$.
\end{rem}

The use of this lemma will be based on the fact that in the expressions
of interest, such as (\ref{eq: Corr_to_obs_2}), each of the normalizing
factors $\normLoc{z_{i}}$ will appear exactly as many times in the
numerator as in the denominator, and thus cancel out. The proof of
Lemma \ref{lem: Clements_clever_lemma} is postponed until Section
\ref{sec: proofs_of_convergence_theorems}. 
\begin{example}
To illustrate the results provided by Theorem \ref{thm: Convergence_singular}
and Lemma \ref{lem: Clements_clever_lemma}, assume that $z_{1}$
is on the straight part of the wired boundary of $\Omega$, the discrete
approximations $\Od$ also have straight boundaries near $z_{1},$and
that $v\in\mathcal{V}^{\reg}$ is a puncture (single-point boundary
component) of $\Omega$. Then, for $z_{1}^{\delta}$ on the boundary
of $\pa\Ocvr$ approximating $z_{1}$, and $z_{2}^{\delta}\sim v$,
we have 
\[
\delta^{-\frac{1}{2}}\cdot F(z_{1}^{\delta},z_{2}^{\delta})\ =\ C_{\psi}\cdot\eta_{z_{1}^{\delta}}\eta_{z_{2}^{\delta}}f^{[\eta_{z_{1}^{\delta}},\reg]}(\zz)+o(1).
\]
\end{example}

\subsection{The case of both points close to a puncture $v$ and the coefficient
$\protect\coefA_{\protect\cvr}(v)$}

Note that the results collected so far are actually enough to prove
a big portion of Theorem \ref{thm: intro_2}. Namely, due to the Pfaffian
formula (\ref{eq: Pfaff_discrete}), the first term in the right-hand
side of (\ref{eq: corr_to_obs_2}) can be expressed in terms of $F(z_{i},z_{j})$,
where the positions of $z_{1}$ and $z_{2}$ are such that the asymptotics
of $F(z_{i},z_{j})$ is covered in Theorems \ref{thm: Convergence_singular},
\ref{thm: Convergence_bulk_near}. Thus, to prove Theorem \ref{thm: intro_2},
it essentially remains to treat the spin correlations, which we do
in the next subsection. 
\begin{defn}
\label{def: coefA}If $\{v\}$ is a single-points boundary component
of $\Omega$ (i. e., $\cvr$ is ramified at $v$), we define the quantity
\[
\coefA(v)=\coefA_{\Omega,\cvr}(v)
\]
by the expansion 
\begin{equation}
f_{\Omega,\cvr}^{[\reg]}(v,z)=\frac{\lamb}{(z-v)^{\frac{1}{2}}}\left(1+2\coefA(v)(z-a)+o(z-a)\right).\label{eq: def_A}
\end{equation}
\end{defn}

\begin{thm}
\label{thm: Conv_both_near_spin}If, for some $s=1,\dots,k,$ $z_{1}^{\circ}=v_{s},$
and $z_{2}^{\bullet}=z_{1}^{\bullet}$, then one has, as $\delta\to0$,
\[
F_{\Od,\cvr}(\zz)=\eta_{z_{1}}\eta_{z_{2}}(1+\re[(z_{2}^{\circ}-z_{1}^{\circ})\cdot\mathcal{A}_{\Omega,\cvr}(v)]+o(\delta)),
\]
uniformly over positions of $v_{1},\dots,v_{n}$ at a definite distance
from the boundary and from each other. 
\end{thm}

The relevance of this result to the spin consideration stems from
the fact that in this configuration, we have
\[
F(z_{1},z_{2})=\eta_{z_{1}}\eta_{z_{2}}\frac{\E[\sigma_{z_{1}^{\circ}}\sigma_{z_{2}^{\circ}}\mu_{z_{1}^{\bullet}}\mu_{z_{2}^{\bullet}}\sigma_{v_{1}}\dots\sigma_{v_{n}}]}{\E[\sigma_{v_{1}}\dots\sigma_{v_{n}}]}=\eta_{z_{1}}\eta_{z_{2}}\frac{\E[\sigma_{v_{1}}\dots\sigma_{v'_{s}}\dots\sigma_{v_{n}}]}{\E[\sigma_{v_{1}}\dots\sigma_{v_{n}}]},
\]
where $v_{s}'=z_{2}^{\circ}$ is a vertex adjacent or diagonally adjacent
to $v_{s}$, see Figure \ref{fig: spin_convergence}. Thus, Theorem
\ref{thm: Conv_both_near_spin} asserts that the \emph{discrete logarithmic
derivatives }of spin correlation with respect to the position of one
of the points converge to quantities governed by $\mathcal{A}_{\Omega,\cvr}(v)$.
The asymptotics of the spin correlations themselves will be recovered
by integrating, fixing the constants of integration from probabilistic
considerations, see Section \ref{sec: corr-continuum}.

\begin{figure}

\includegraphics[width=0.5\textwidth]{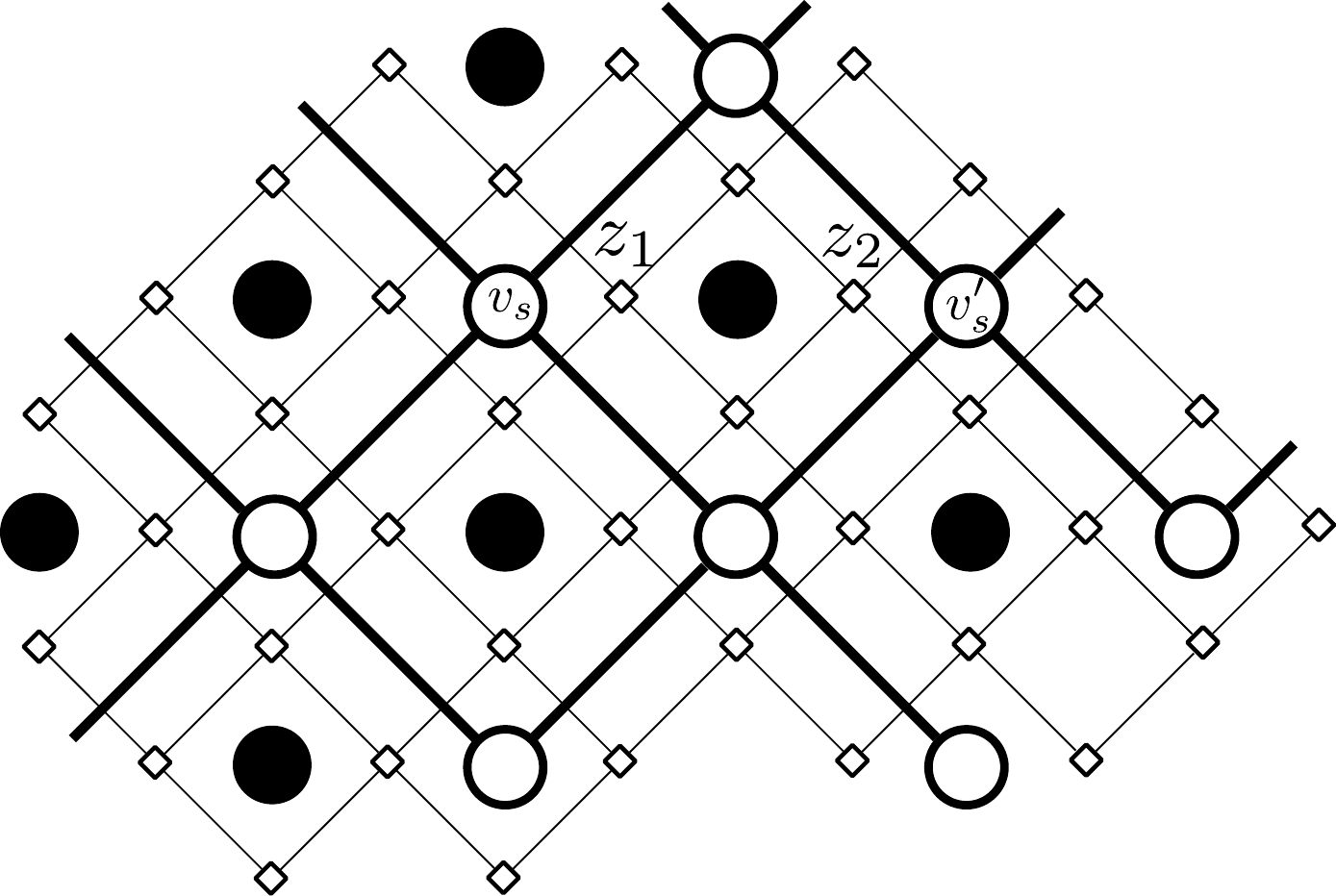}\caption{The configuration of points for Theorem \ref{thm: Conv_both_near_spin}\label{fig: spin_convergence}}

\end{figure}

\newpage{}

\section{Proofs of convergence results for fermionic observables:}

\label{sec: proofs_of_convergence_theorems}

\subsection{Preliminaries on s-holomorphicity and standard boundary conditions}

We first briefly recall a few elementary properties of s-holomorphic
functions and their relation to discrete holomorphicity. In what follows,
all the functions defined on the corner graph $\Cgr(\Od)$ will be
assumed to satisfy the phase condition (\ref{eq: phases_condition}).
The corner lattice $\Cgr:=\Cgr(\Cd)$ can be decomposed in a checkerboard
fashion into a union of two sub-lattices, 
\[
\Cgr=\wcgr\sqcup\bcgr,
\]
where $\bcgr=\{z\in\Cgr:\ds z\in\{\pm1,\pm\i\}\}$ and $\wcgr=\{z\in\Cgr:\ds z\in\{\pm\lamb,\pm\lambb\}\}.$
The discrete $\dbar$ operator acts on a function defined on $\wcgr$
and produces a function on $\bcgr$, and reads
\begin{equation}
\dbar F(z):=\frac{1}{2\delta^{2}}\cdot\sum_{w\sim z}F(w)\cdot(w-z),\label{eq: Disc_hol}
\end{equation}
where the sum is over the four vertices $w\in\wcgr$ incident to $z\in\bcgr$. 

A function $F:\wcgr\mapsto\C$ is called \emph{discrete holomorphic}
at $z\in\bcgr$ if it satisfies the condition $\dbar F(z)=0$. The
following elementary Lemma relates the s-holomorphicity and the discrete
holomorphicity. 
\begin{lem}
\label{lem: dhol_to_shol}Let $F:\wcgr\mapsto\C$ be a function satisfying
the phase condition (\ref{eq: phases_condition}), and let $K\subset\bcgr.$
Then, the following are equivalent:

\begin{enumerate}
\item $F$ is discrete holomorphic at every $z\in K$;
\item $F$ can be extended to $K$ in such a way that, for every edge $e$
incident to some $z\in K$, one can further extend $F$ to all four
corners incident to $e$ so that the s-holomorphicity condition (\ref{eq: s-hol})
holds.
\end{enumerate}
\end{lem}

\begin{proof}
If $e$ is an edge and $a,b$ are the two corners of $\wcgr$ incident
to $e$, we extend $F$ by putting 
\begin{equation}
F(w):=\proj{\eta_{w}}{F(a)+F(b)},\label{eq: extension}
\end{equation}
for $w\in\bcgr$ incident to $e$. If $z\in K$ and $e,e'$ are the
two edges incident to $z$, then a straightforward computation shows
that under the phase condition (\ref{eq: phases_condition}), the
discrete holomorphicity $\dbar F(z)=0$ is equivalent to 
\begin{equation}
\proj{\eta_{z}}{F(a)+F(b)}=\proj{\eta_{z}}{F(a')+F(b')};\label{eq: equiality projections}
\end{equation}
thus, the above extension is well-defined on $K.$ By construction,
such an extension is s-holomorphic due to Remark \ref{rem:s-hol-projections}.
Conversely, if we start with an s-holomorphic function, then we have
(\ref{eq: equiality projections}), since both sides are equal to
$F(z)$, and the discrete holomorphicity follows.
\end{proof}
\begin{rem}
\label{rem: project}The following observation will be useful: if
$F:=\wcgr\to\C$ satisfies the discrete holomorphicity $\dbar F(w)=0$
at some $w\in\bcgr$, then so does the function $z\mapsto\proj{\eta_{z}}{F(z)}.$
To see this, pick some $\wcgr\ni\hat{z}\sim w$ and project the equation
$\dbar F(w)=0$ onto the line $(\hat{z}-w)\eta_{\hat{z}}\R.$
\end{rem}

In particular, the restriction of the observable $F_{\cvr}(z_{1},z_{2})$
to $\wcgr$ is discrete holomorphic, at least away from $z_{1}$.

Suppose a function $F(\cdot)$ is defined at the corners adjacent
to an edge $e$ and is s-holomorphic at $e$. Then,
\begin{multline}
|F(z_{N})|^{2}+|F(z_{S})|^{2}=|F(z_{N})+F(z_{S})|^{2}\\
=|F(z_{E})+F(z_{W})|^{2}=|F(z_{E})|^{2}+|F(z_{W})|^{2}.\label{eq: abs_F_2}
\end{multline}
Moreover, since (\ref{eq: phases_condition}) implies $F(z)^{2}=\eta_{z}{}^{2}|F(z)|^{2}$,
we can rewrite the above relation as follows: 
\begin{multline}
F(z_{N})^{2}\cdot(z_{N}^{\bullet}-z_{N}^{\circ})+F(z_{E})^{2}\cdot(z_{E}^{\circ}-z_{E}^{\bullet})\\
+F(z_{S})^{2}\cdot(z_{S}^{\bullet}-z_{S}^{\circ})+F(z_{W})^{2}\cdot(z_{W}^{\circ}-z_{W}^{\bullet})=0.\label{eq: H_closed_form}
\end{multline}
In other words, the discrete differential form $F^{2}(z)(z^{\bullet}-z^{\circ})$,
defined on the edges of $\Od\cup(\Od)^{\star}$, is closed over the
domain of s-holomorphicity of $F$. If now $F_{1,2}$ are two s-holomorphic
functions, then, plugging $F_{1}$, $F_{2}$ and $F_{1}+F_{2}$ into
(\ref{eq: H_closed_form}), we infer that the discrete differential
form $F_{1}(z)F_{2}(z)(z^{\bullet}-z^{\circ})$ is also closed. This
leads to the following definition.
\begin{defn}
\label{def: H}Given s-holomorphic functions $F_{1,2}$ on a simply-connected
subset of $\Cd$, we define the real-valued function on $\Od\cup\Odual\cup\fixed^{\circ}$
\[
H(w):=\intFsquare{F_{1}F_{2}}w
\]
to be the integral of the discrete differential form $-2i\cdot F_{1}(z)F_{2}(z)(z^{\bullet}-z^{\circ})$.
In other words, for any neighboring $\vz\in\Omega^{\delta,\star}$
and $\fz\in\Od\cup\fixed^{\circ}$, we put
\begin{equation}
H(z^{\bullet})-H(z^{\circ})=-2i\cdot F_{1}(z)F_{2}(z)(z^{\bullet}-z^{\circ}).\label{eq: def_H}
\end{equation}
\end{defn}

\begin{rem}
Note that the function $F_{1}(z)F_{2}(z)$ is purely imaginary $\wcgr(\Od)$
and real on $\bcgr(\Od)$. If it were discrete holomorphic (which
it is not), then one could define a discrete integral of $F_{1}(z)F_{2}(z)dz$
on $\Cgr(\Od)^{\star}$, purely imaginary on $\Od\cup\Odual$ and
purely real on $\Cgr(\Od)^{\star}\setminus(\Od\cup\Odual)$. Therefore,
$H$ can be indeed be viewed as a discretization of the \emph{imaginary
part} of $\int F_{1}(z)F_{2}(z)dz$. 
\end{rem}

\begin{rem}
\label{rem: conv_f_conv_H}We will be interested in the setup when
the restriction of s-holomorphic functions $F_{1,2}$ to each sub-lattice
$\Cgr^{\eta}$ converge to a corresponding projections of holomorphic
functions, i. e., $F_{1,2}(z)=\proj{\eta_{z}}{f_{1,2}(z)}+o(1)$ uniformly
over some domain. In this case, we also have 
\[
\intFsquare{F_{1}F_{2}}w=\im\left(\int^{w}f_{1}(z)f_{2}(z)dz\right)+o(1).
\]
\end{rem}

\begin{lem}
\label{lem: H_bc}If $F$ satisfies the standard boundary conditions,
then $H(\cdot)=\intFsquare{F^{2}}{\cdot}$ is constant along each
arc (i. e., connected component) of $\free^{\bullet}$, and is constant
along the vertices of $\fixed^{\circ}$ on the same boundary component.
In particular, $H$ is a single-valued function. 
\end{lem}

\begin{proof}
Note that $-2i\cdot F(z)^{2}(z^{\bullet}-z^{\circ})=2\delta|F(z)|^{2}.$
If $z_{a}^{\bullet}=z_{b}^{\bullet}$, then 
\[
H(z_{a}^{\circ})-H(z_{b}^{\circ})=H(z_{a}^{\circ})-H(z_{a}^{\bullet})+H(z_{b}^{\bullet})-H(z_{b}^{\circ})=2\delta|F(z_{a})|^{2}-2\delta|F(z_{b})|^{2}.
\]
Hence, by (\ref{eq:bcond}), $H$ is constant along each fixed boundary
arc. Similarly, $H$ is constant along each free boundary arc. Taking
$z_{a}$ and $z_{b}$ to be as in case (3) of Definition \ref{def: Standard_BC}
shows that for each free boundary arc $\gamma$, the values of $H$
on the fixed boundary arcs adjacent to $\gamma$ coincide.
\end{proof}
\begin{rem}
\label{rem: H_extension}Suppose $z_{a,b}$ are as in cases (1) or
(2) of Definition \ref{def: Standard_BC}, and that, moreover, $z_{a}$
and $z_{b}$ are incident to the same edge $e$. If we put $F(\hat{z}_{a}):=F(\hat{z}_{b}):=0$,
where $\hat{z}_{a}$ and $\hat{z}_{b}$ are two other corners incident
to $e$, then (\ref{eq:bcond}) guarantee that (\ref{eq: abs_F_2})
still holds. With this observation, we extend $H$ to the vertices
of $\Cd\setminus\Od$ separated from $\Od$ by $\free$ (with the
same values as on the corresponding arc of $\free^{\bullet}$), and
to dual vertices that are connected to dual edges in $(\Od)^{\star}$
by an edge separating two vertices of $\fixed$, with the same constant
as on the corresponding arc of $\fixed^{\circ}$. 
\end{rem}

Let $\lapv$ denote the standard nearest-neighbor Laplacian on vertices
of $\Od$, modified at the boundary as follows:
\[
\lapv G(\v)=\sum_{w\sim\v}c_{wv}\left(G(w)-G(\v)\right),
\]
where $c_{wv}:=1$ unless the edge $(wv)$ intersects an edge of a
\emph{free} boundary arc, in which case $c_{wv}:=2\sqrt{2}-2.$ Similarly,
let $\lapf$ denote the nearest-neighbor Laplacian on vertices of
$\dual{\Od}$, defined by 
\[
\lapf G(\u)=\sum_{w\sim\u}c_{wu}\left(G(w)-G(\u)\right),
\]
where $c_{wu}:=1$ unless the edge $(w\u)$ separates two vertices
of a fixed boundary arc, in which case $c_{wv}=2\sqrt{2}-2.$ 
\begin{lem}
\label{lem: sub_super}Assume that $F$ is an s-holomorphic $\cvr$-spinor
on $\Cgr(\Od)$ satisfying the standard boundary conditions. Then,
for $H(w)=\intFsquare{F^{2}}w$, one has 
\[
\lapv H(\v)=-2|F(z_{N})+iF(z_{E})-F(z_{S})-iF(z_{W})|^{2}
\]
\[
\lapf H(\u)=2|F(z_{N})+iF(z_{E})-F(z_{W})-iF(z_{S})|^{2}
\]
for any $v\in\Od$ and $u\in\Odual\setminus\free^{\bullet},$ provided
that $\cvr$ is not ramified at $\v$ (respectively, at $\u$). Here
$z_{N,E,S,W}$ denote the four corners in $\Cgr(\Od)$ incident to
$\v$ (respectively, to $\u$). In particular, the restriction of
$H$ to the vertices of $\Od$ (respectively, to the vertices of $\Odual$)
is subharmonic (respectively, superharmonic).
\end{lem}

\begin{proof}
This is a straightforward computation, which can be found, e. g.,
in \cite{IzyurovFree}. 
\end{proof}
We denote by $H^{\bullet}$ and $H^{\circ}$ the restrictions of $H$
to the vertices of $\Cdual$ and $\Cd$, respectively. We now have
enough tools to prove Proposition \ref{prop: discrete-uniqueness}.
\begin{proof}[Proof of Proposition \ref{prop: discrete-uniqueness}]
 We follow closely the proof of Proposition \ref{prop: Uniqueness_continuous}.
Assume that $F$ is not identically zero, in which case the corresponding
function $H$ is not a constant. Let $\umax$ be a point where $H^{\bullet}$
attains its maximal value. By Lemma \ref{lem: sub_super}, and since
$\cvr$ is not ramified at any of the dual vertices, $\umax$ cannot
be a point of $\Odual\setminus\free^{\bullet}.$ If $\umax$ is a
point connected to a point $\hat{u}\in\Odual$ by a dual edge $e^{\star}$
separating two vertices $v,\hat{v}\in\fixed^{\circ}$, see Remark
\ref{rem: H_extension}, then 
\[
H(\umax)=H(v)=H(\hat{u})-\delta|F_{\cvr}(z_{1},z)|^{2}\leq H(\hat{u}),
\]
where $z$ is such that $z^{\circ}=v$, $z^{\bullet}=\hat{u}$. This
is impossible.

The only remaining possibility is $\umax\in\free^{\bullet}$. As in
the continuous case, we will get a contradiction by looking at the
sign of $\bar{\eta}_{z}F(z)$ as $z\in\Cgr(\Od)$ travels along the
free boundary arc $\nu$ containing $\umax$. If $z_{a,b}\in\Cgr(\Od),$
$z_{a}\sim z_{b}$ are such that $z_{a}^{\circ}=z_{b}^{\circ}$ and
$z_{a}^{\bullet},z_{b}^{\bullet}\in\nu$, then it follows directly
from (\ref{eq:bcond}) that $\bar{\eta}_{z_{a}}F(z_{a})=\bar{\eta}_{z_{b}}F(z_{b}).$
Now, suppose $z_{a}\sim z_{b}$ are such that $u:=z_{a}^{\bullet}=z_{b}^{\bullet}\in\nu$,
let $(u\hat{u})$ be the dual edge separating $z_{a}^{\circ}$ from
$z_{b}^{\circ}$, and assume that $\hat{u}\notin\free^{\bullet}$.
Then, taking into account Lemma \ref{lem: H_bc}, we have $H(\hat{\u})<H(u)=H(\umax)$,
which can be rewritten as 
\begin{equation}
|F(z_{a})|^{2}>|F(z_{d})|^{2};\quad|F(z_{b})|^{2}>|F(z_{c})|^{2},\label{eq: wedges}
\end{equation}
where $z_{c,d}$ are two other corners incident to $(u\hat{u})$.
The s-holomorphicity conditions (\ref{eq: phases_condition}\textendash \ref{eq: s-hol})
imply that $F(z_{a,b,c,d})$ are orthogonal projections of the complex
number $W:=F(z_{a})+F(z_{c})=F(z_{b})+F(z_{d})$ onto the lines $\eta_{z_{a,b,c,d}}\R$,
respectively. It is now easy to check that the condition (\ref{eq: wedges})
is equivalent to $\bar{\eta}_{z_{a}}F(z_{a})/\bar{\eta}_{z_{b}}F(z_{b})>0.$ 

Now, let $z_{1}\sim z_{2}\sim\dots\sim z_{N}\in\Cgr(\Od)$ be the
shortest path on $\Cgr(\Od)$ along $\nu$ (that is, $z_{i}^{\bullet}\in\nu$
for all $i$) such that either $z_{1}^{\bullet}$,$z_{2}^{\bullet}$
are two endpoints of $\nu$, and $z_{1}^{\circ},z_{N}^{\circ}\in\fixed^{\circ}$,
or $z_{1}=z_{N}$ (in the case $\nu$ comprises an entire boundary
component). We cannot have $z_{i}^{\circ}$ separated from $z_{i+1}^{\circ}$
by an edge $(u\hat{u})$ with $\hat{u}\in\free^{\bullet}$, since
in that case, by our assumption on $\Od$, one would have $\hat{u}\in\nu$,
and the path could be shortened by ``cutting the fjord'' separated
by $(u\hat{u})$. Thus, the above considerations show that the sign
of $\bar{\eta}_{z}F_{z}$ is preserved along the path, which contradicts
either Case (3) of Definition \ref{lem: bc_discrete}, or, if $z_{1}=z_{N}$,
the fact that by our assumptions on $\cvr$, $F$ does not change
the sign along the path, while $\eta_{z}$ does.

\end{proof}

\subsection{Special s-holomorphic functions. }

In this subsection, we collect some results about special s-holomorphic
functions and spinors in the full plane $\Cd$. We will need two such
functions, which serve as discrete analogs of $z^{-1}$ and $z^{-\frac{1}{2}}$.
By now, several constructions of these functions are known, see \cite{kenyon2002laplacian,ChelkakSmirnov1,HonglerViklundKytolaCFT,ChelkakHonglerIzyurov,gheissari2019ising}.
In Appendix, we will give a construction following \cite{DubedatDimers}We
first put 
\[
\delta=1,
\]
 so that $\C^{\delta}=\C^{1}$ and $\wcgr(\Cd)=\Z^{2}+\frac{1}{2}.$
Simple scaling will then allow to handle arbitrary $\delta$.
\begin{lem}
\label{lem: dzm1}(Discrete $z^{-1}$) For each $a\in\Cgr(\C^{1})$,
there exists an s-holomorphic function $\dzmone{\cdot}a$ on $\Cgr_{\{a\}}(\C^{1})$
such that 
\[
\dzmone{a^{+}}a=\eta_{a},\quad\dzmone{a^{-}}a=-\eta_{a},
\]
and, as $z\to\infty$, one has 
\begin{equation}
\dzmone za=\frac{2}{\pi}\proj{\eta_{z}}{\bar{\eta}_{a}(z-a)^{-1}}+O(|z-a|^{-2}).\label{eq: asymp_zmone}
\end{equation}
Moreover, we have 
\begin{equation}
P_{a}(a\pm i(a^{\bullet}-a^{\circ}))=0.\label{eq: P_a_close}
\end{equation}
\end{lem}

\begin{lem}
\label{lem: dzm12}(Discrete $z^{-\frac{1}{2}}$) There exists an
s-holomorphic $\double 0$-spinor $\dmsqrt{\cdot}{}$ on $\Cgr(\C^{1})$
such that 
\begin{equation}
\dmsqrt z{}=\eta_{z}=\lambb(2z)^{-\frac{1}{2}}\label{eq: q-value}
\end{equation}
 for all eight corners $z\in\Cgr(\C_{[0]}^{1})$ incident to $0$,
and 
\begin{equation}
\dmsqrt z{}=\left(\frac{2}{\pi}\right)^{\frac{1}{2}}\proj{\eta_{z}}{\lambb z^{-\frac{1}{2}}}+\text{o}(|z|^{-\frac{1}{2}}),\quad z\to\infty.\label{eq: q_asymp}
\end{equation}
\end{lem}

We now introduce the notation for scaled and shifted versions of these
function. Note that when we shift the ramification point to be a point
of $\Cdual$, additional factor of $i$ has to be introduced to to
preserve the phase condition (\ref{eq: phases_condition}).
\begin{defn}
Denote, for $a\in\Cgr(\Cd)$ and $w\in\Cd\cup\Cdual$, 
\[
\dzmone za:=\delta^{-1}\dzmone{z\delta^{-1}}{a\delta^{-1}},\quad z\in\Cgr(\Cd)[a];
\]
\[
\dmsqrt zw:=\begin{cases}
\delta^{-\frac{1}{2}}\dmsqrt{(z-w)\delta^{-1}}{}, & w\in\Cd\\
-i\delta^{-\frac{1}{2}}\dmsqrt{(z-w)\delta^{-1}}{}, & w\in\Cdual
\end{cases},\quad z\in\Cgr(\Cd_{[w]});
\]
\end{defn}

The asymptotics (\ref{eq: asymp_zmone}), (\ref{eq: q_asymp}), translate
into the following asymptotics as $\delta\to0$, uniformly over $z$
away from $a,w$: 
\begin{equation}
\dzmone za=\frac{2}{\pi}\proj{\eta_{z}}{\bar{\eta}_{a}(z-a)^{-1}}+o(1);\label{eq: asymp_zm1}
\end{equation}
\begin{equation}
\dmsqrt zw=\begin{cases}
\left(\frac{2}{\pi}\right)^{\frac{1}{2}}\proj{\eta_{z}}{\lambb(z-w)^{-\frac{1}{2}}}+o(1), & w\in\Cd,\\
\left(\frac{2}{\pi}\right)^{\frac{1}{2}}\proj{\eta_{z}}{\lamb(z-w)^{-\frac{1}{2}}}+o(1), & w\in\Cdual
\end{cases}\label{eq: asymp_zm12}
\end{equation}

\subsection{Cauchy integral formula and the maximum principle for s-holomorphic
spinors.}

The function $\dmsqrt{\cdot}w$ allows one to recover the value of
an s-holomorphic spinor near a ramification point. Recall that if
$v\in\Od$ and $u\in\Odual$ are adjacent, then the product of two
spinors on $\Cgr(\Od)$, ramified, respectively, at $u$ and $v$,
can be seen as a \emph{function} on the modified graph $\Cutz{\Od}z,$
defined up to a global choice of signs, and whose value at $z^{\pm}$
differs by sign. 
\begin{lem}
\label{lem: Cauchy_spinors}Let $v\in\Cd$ and $u\in\Cdual$ be adjacent,
and $F$ be an s-holomorphic spinor defined in a neighborhood of $v$
and ramified at $v$. Then, for $z=\frac{v+u}{2}$, one has 
\[
F(z)=\frac{1}{4}\delta^{-\frac{1}{2}}\eta_{z}\cdot\intFQ{FQ_{u}}{\gamma},
\]
Here $\gamma$ is any simple counterclockwise contour surrounding
$u$ and $v$, and the sign of the function $FQ_{u}$ on $\Cutz{\mathrm{int}(\gamma)}z$
is related to that of $F(z)$ and $\eta_{z}$ by $F(z^{-})Q_{u}(z^{-})=F(z)\eta_{z}\delta^{-\frac{1}{2}}$. 
\end{lem}

\begin{proof}
Note that $FQ_{u}$ is a well defined \emph{function} on $\Cutz{\mathrm{int}(\gamma)}z$,
where $\mathrm{int}(\gamma)$ is the domain enclosed by $\gamma$.
Therefore, by the above discussion, the discrete differential form
$F(z)Q_{u}(z)(z^{\bullet}-z^{\circ})$ is closed in $\mathrm{int}(\gamma)$,
and we can contract $\gamma$ to the contour consisting of just two
edges, one going from $z^{\bullet}$ to $z^{\circ}$ through $z^{-}$,
and another going from $z^{\circ}$ to $z^{\bullet}$ through $z^{+}$.
The discrete integral over these two edges is 
\[
-2i\cdot F(z^{-})Q_{u}(z^{-})(v-u)-2i\cdot F(z^{+})Q_{u}(z^{+})(u-v)=-4i\cdot F(z^{-})Q_{u}(z^{-})(v-u)=4F(z)\bar{\eta}_{z}\delta^{\frac{1}{2}}.
\]
\end{proof}
\begin{lem}
(Maximum principle for spinors \cite{DubedatDimers}). \label{lem: Maximum_spinors}Let
$F$ be an s-holomorphic $\cvr(v)$-spinor on $\Cgr(\C^{\delta}\cap B_{R}(v))$
such that $F(z_{0})=0$ for one of the corners $z_{0}\in\Cgr(\Cd)$
adjacent to $v$. Then, $\Delta_{i\eta}F\equiv0$, where $\eta=\eta_{z_{0}}$.
In particular, 
\[
|F(z)|\leq\sup_{w\in B_{R}(v)\setminus B_{R/2}(v)}|F(w)|,\quad z\in\Cgr^{i\eta}(\Cd\cap B_{R/2}(v)).
\]
\end{lem}

\begin{proof}
Let $z\in\Cgr^{i\eta}(\Cd)$, and denote by $\placket_{z}$ the $3\times3$
placket of $\Cgr^{\eta}(\Cd)\cup\Cgr^{i\eta}(\Cd)$ centered at $z$:
\[
\placket_{z}:=\{w\in\C:|\re w-\re z|\leq\delta,|\im w-\im z|\leq\delta\}.
\]
We can view the restriction of either of the two branches of $F$
to $\mathcal{S}_{z}$ as a discrete holomorphic function on this placket.
Indeed, if $z$ is not incident to $v$, then $\mathcal{S}_{z}$ does
not contain $v$. If $z\sim v$, that is, $z=v+(v-z_{0}),$ then,
since $F(z_{0})=0=-F(z_{0})$, we can choose a branch cut passing
through $z_{0}$. In either case, the discrete harmonicity at $z$
follows in the usual way from the factorization $\Delta_{i\eta}=4\pa\dbar$.
\end{proof}

\subsection{Proofs of the convergence theorems. }

The convergence  Theorems \ref{thm: conv_bulk_bulk}, \ref{thm: Convergence_singular},
\ref{thm: Convergence_bulk_near}, \ref{thm: Conv_both_near_spin}
are stated in the form of uniform estimates. However, by compactness,
they can be reformulated as follows: to prove, e. g. Theorem \ref{thm: conv_bulk_bulk},
it suffices to show that if, as $\delta_{k}\to0$, sequences $z_{1}^{\delta_{k}}\in\Cgr^{\eta_{1}}(\Ocvrc^{\delta_{k}})$
and $z_{2}^{\delta_{k}}\in\Cgr^{\eta_{2}}(\Ocvrc^{\delta_{k}})$ with
fixed $\eta_{1,2}$ converge to distinct points $z_{1,2}\in\Omega$,
then 
\begin{equation}
\bar{\eta}_{1}\delta_{k}^{-1}F(z_{1}^{\delta_{k}},z_{2}^{\delta_{k}})\to\proj{\eta_{2}}{f^{[\eta_{1}]}(z_{1},z_{2})}.\label{eq: conv_reduced}
\end{equation}
Indeed, if the conclusion of Theorem \ref{thm: conv_bulk_bulk} did
not hold, we could arrive at a contradiction to the above claim by
extracting a subsequence. In the following proofs, we will assume
that this and similar reductions have been made, and we will drop
the subscript $k$. 
\begin{proof}[Proof of Theorem \ref{thm: conv_bulk_bulk}]

Define the function $H^{\delta}(w):=\intFsquare{F^{2}(z)\,}w$ according
to Definition \ref{def: H} with $F(\cdot)=\bar{\eta_{1}}\delta^{-1}F(z_{1}^{\delta},\cdot)$,
$\eta_{1}:=\eta_{z_{1}},$ fixing the additive constant so that $H^{\delta}\equiv0$
at one of the boundary arcs (recall Lemma \ref{lem: H_bc}). We first
assume that 

\begin{itemize}[leftmargin=2cm]
\item[\textbf{(UBC)}]  $H^{\delta}(\cdot)$ is uniformly bounded on compact subsets of
$\Omega\setminus\{z_{1}\}$ and on $\pa\Od$.
\end{itemize}
Under this assumption, \cite[Theorem 3.12]{ChelkakSmirnov2} and Arzela-Ascoli's
theorem allows one to extract a convergent subsequence. Assuming such
an extraction has been made, there exists a holomorphic $\cvr$-spinor
$f$ on $\Omega$ such that 
\[
\delta^{-1}\bar{\eta}_{z_{1}^{\delta}}F(z_{1}^{\delta},z_{2}^{\delta})-\proj{\eta_{z_{2}}}{f(z_{2}^{\delta})}=o(1)
\]
uniformly over $z_{2}^{\delta}$ in compact subsets of $\Omega\setminus\{z_{1}\}$;
then also 
\[
H^{\delta}(\cdot)\to h(\cdot):=\im\int f^{2}
\]
uniformly on compact subsets of $\Omega\setminus\{z_{1}\}$. Our goal
is to deduce from Remark (\ref{rem: cont_obs_existence_and_uniqueness})
that $f(\cdot)\equiv\Cpsi^{2}f_{\Omega,\cvr}^{[\eta]}(z_{1},\cdot)$.

We start by proving that $f$ satisfies the standard boundary conditions.
Following Remark \ref{rem: bc_h_to_bc_f}, we will first check that
conditions (1)\textendash (5) of Proposition \ref{prop: f_to_h} are
satisfied. From convergence of the discrete harmonic measure to the
continuous one, we know that $h$ is a constant along each free or
fixed boundary arc, and its value is the limit the value of $H^{\delta}$
on the corresponding arc; ensuring (1). Moreover, \cite[Remark 6.3]{ChelkakSmirnov2}
shows that for every $a\in\fixed$, there is an $\varepsilon>0$ such
that $h(z)\geq h(a)$ for every $z\in B_{\varepsilon}(a)$. After
conformal mapping to a nice domain, this implies (2) of Proposition
\ref{prop: f_to_h}. The same argument with primal and dual lattices
exchanged implies (3) of Proposition \ref{prop: f_to_h}. Since $\cvr$
is only ramified at vertices and not dual vertices, $H^{\bullet}$
satisfies the maximum principle, which means that for each single-point
boundary component $v$, we have 
\[
h(z)\leq\sup_{w\in B_{\eps}(v)\setminus B_{\eps/2}(v)}h(w),
\]
for $\eps$ small enough. This proves (5) of Proposition \ref{prop: f_to_h}.
To prove (4), we notice that by (\ref{eq: def_H}), the jump in the
constant value of $H^{\delta}$ between a fixed arc $\nu_{1}$ and
a free arc $\nu_{2}$ is proportional to $|F(\hat{z})|^{2}$, where
$\hat{z}$ is such that $\hat{z}^{\bullet}\in\nu_{2}$ and $\hat{z}^{\circ}\in\nu_{1}$.
By the third case in the definition of standard boundary conditions
(Definition \ref{def: Standard_BC}), we see that these jumps at the
endpoints of any free arc are negatives of each other. 

It remains to check that the sign condition in (\ref{eq: bc_free_arc})
holds. Let $\nu$ be a free boundary arc of $\pa\Omega$, and let,
for a fixed small $\epsilon>0$, $\hat{\Omega}\subset\Omega$ be a
simply-connected domain whose boundary coincides with $\pa\Omega$
in the $\epsilon$-neighborhood of $\nu$. Let $\hat{\Omega}^{\delta}$
be discrete simply-connected approximations to $\hat{\Omega}$ whose
boundaries coincide with $\pa\Od$ in the $\epsilon$-neighborhood
of $\nu$. We claim that it suffices to find \emph{any} sequence $\hat{F}^{\delta}(\cdot)$
of s-holomorphic functions on $\hat{\Omega}^{\delta}$, satisfying
the standard boundary conditions in the $\epsilon$-neighborhood of
$\nu$, converging to a holomorpic function $\hat{f}$ uniformly on
compact subsets of $\hat{\Omega}$, such that the corresponding $\hat{H}^{\delta}$
are uniformly bounded in an $\epsilon$-neighborhood of $\nu$, and
such that $\hat{f}$ satisfies \ref{eq: bc_free_arc} with $\hat{c}_{1,2}\neq0$.
Indeed, given such a sequence, consider the linear combination $\check{F}(\cdot):=\delta^{-1}F(z_{1}^{\delta},\cdot)-\alpha_{\delta}\hat{F}(\cdot)$
together with the corresponding function $\check{H}^{\delta}$, with
$\alpha_{\delta}$ chosen in such a way that the jump in the boundary
value of $\check{H}^{\delta}$ at one of the endpoints of $\nu^{\delta}$
vanishes, where $\nu^{\delta}$ is the discrete counterpart of $\nu$.
By the third case in the Definition \ref{def: Standard_BC}, the jump
of $\check{H}^{\delta}$ also vanishes at the other endpoint of $\nu^{\delta}$.
Note that $|\alpha_{\delta}|\to|\alpha|:=|c_{1}/\hat{c}_{1}|$, and,
passing to a subsequence if needed, we may assume that $\alpha_{\delta}\to\alpha\in\R$.
But then the function 
\[
\check{h}:=\lim\check{H}^{\delta}=\im\int(f-\alpha f)^{2}
\]
has no jumps at the endpoint of $\nu$, which implies that (\ref{eq: bc_free_arc})
is satisfied with $c_{1,2}=\alpha\hat{c}_{1,2}$.

The desired sequence $\check{F}^{\delta}$ can been exhibited as follows:
impose fixed boundary conditions on $\hat{\Omega}^{\delta}\setminus\nu^{\delta}$
and put $\check{F}^{\delta}(\cdot):=\beta_{\delta}F_{\hat{\Omega}^{\delta}}(\hat{z}^{\delta},\cdot)$,
where $\hat{z}^{\delta}\in\pa\hat{\Omega}^{\delta}$ converges to
$z\in\pa\hat{\Omega}\setminus\nu$, and $\beta_{\delta}$ is an appropriate
normalization constant. It coincides with the observable in \cite[Proposition 1.1]{IzyurovFree},
whose limit was computed explicitly and does satisfy the sign condition
in (\ref{eq: bc_free_arc}). This concludes the proof that $f$ satisfies
the standard boundary conditions on $\pa\Omega\setminus z_{1}$. 

Next, we prove that $\Cpsi^{-2}f$ satisfies the asymptotics (\ref{eq: f_residue}).
We work with the restriction of $F(z_{1}^{\delta},\cdot)$ to a small
neighborhood $B_{\eps}(z_{1})$ of $z_{1}$ (on one of the sheets
of the double cover). We have $\bar{\eta}_{1}\cdot F(z_{1}^{\delta},(z_{1}^{\delta})^{\pm})=\pm\eta_{1}$.
Therefore, the difference 
\[
\hat{F}(\cdot):=\bar{\eta}_{1}\cdot\delta^{-1}F(z_{1}^{\delta},\cdot)-\dzmone{\cdot}{z_{1}^{\delta}}
\]
is an s-holomorphic function on $\Cgr(\Cd\cap B_{\eps}(z_{1}))[z_{1}^{\delta}]$
which vanishes at $(z_{1}^{\delta})^{\pm}$, and thus it can be viewed,
by identifying $(z_{1}^{\delta})^{\pm}$, as an s-holomorphic function
on $\Cd\cap B_{\eps}(z_{1})$. In particular, by maximum principle,
\[
|\hat{F}(z)|\leq\sup_{w\in B_{\eps}(z_{1})\setminus B_{\eps/2}(z_{1})}|\hat{F}(w)|,\quad z\in B_{\eps/2}(z_{1}).
\]
By passing to the limit, using asymptotics (\ref{eq: asymp_zm1}),
we see that 
\[
\left|f(z)-\left(\frac{2}{\pi}\right)\frac{\bar{\eta}_{1}}{z-z_{1}}\right|\leq\sup_{w\in B_{\eps}(z_{1})\setminus B_{\eps/2}(z_{1})}|f(w)|+\frac{2}{\eps},\quad z\in B_{\eps/2}(z_{1}),
\]
thus proving that $\frac{\pi}{2}f$ satisfies (\ref{eq: f_residue}).

Hence, we indeed have $f(\cdot)=\Cpsi^{2}f_{\Omega,\cvr}^{[\eta_{1}]}(z_{1},\cdot),$
and it remains to justify the assumption (UBC). Given a small fixed
$\eps>0$, denote $\Od_{\eps}:=\Od\setminus B_{\eps}$, where $B_{\eps}$
is the union of $\eps$-neighborhoods of $z_{1}$ and all single-point
boundary components. If (UBC) does not hold, we may assume, by extracting
a subsequence, that there is a number $\eps>0$ such that 
\begin{equation}
M_{\eps}^{\delta}:=\max_{\Od_{\eps}}|H^{\delta}|\to+\infty.\label{eq: M_delta_to_infty}
\end{equation}
The functions $\tilde{H}^{\delta}(\cdot):=\left(M_{\eps}^{\delta}\right)^{-1}H^{\delta}$
are uniformly bounded on $\Omega_{\eps}^{\delta}$, and therefore,
by \cite[Theorem 3.12]{ChelkakSmirnov2}, $\tilde{F}^{\delta}(\cdot):=\left(M_{\eps}^{\delta}\right)^{-\frac{1}{2}}\bar{\eta_{1}}\delta^{-1}F(z_{1}^{\delta},\cdot)$
are also uniformly bounded, say, on $\Od_{2\eps}\setminus\Od_{5\eps}$.
We claim that, in fact, $\tilde{H}^{\delta}$ are uniformly bounded
on $\Omega_{\hat{\eps}}^{\delta}$ for every $\hat{\eps}>0$. We will
prove the claim for $\tilde{F}^{\delta}$; then, it follows readily
for $\tilde{H}^{\delta}$ by integration. From the maximum principle
applied to $\tilde{F}(\cdot)-(M_{\eps}^{\delta})^{-\frac{1}{2}}\delta^{-1}P_{z_{1}}(\cdot),$
we see that $\tilde{F}^{\delta}(\text{\ensuremath{\cdot}})$ are uniformly
bounded on $B_{3\eps}(z_{1})\setminus B_{\hat{\eps}}(z_{1}),$ for
every $\hat{\eps}>0$. Now, let $\{v\}$ be a single-point boundary
component, $v^{\delta}$ the corresponding vertex of $\Omega^{\delta}$.
By Lemma \ref{lem: Cauchy_spinors}, if $z_{0}^{\delta}$ is one of
the corners adjacent to $v^{\delta}$, then $\tilde{F}(z_{0}^{\delta})=O(\delta^{-\frac{1}{2}}).$
Let $\beta_{\delta}\in\R$ be chosen in such a way that $\tilde{F}(z_{0}^{\delta})=\beta_{\delta}\dmsqrt{z_{0}^{\delta}}{v^{\delta}}$,
so that $\beta_{\delta}=O(1).$ Now, by the maximum principle (Lemma
(\ref{lem: Maximum_spinors})) applied to to $\tilde{F}(\cdot)-\beta_{\delta}\dmsqrt{\cdot}{v^{\delta}}$,
the restrictions of these functions to $\Cgr^{i\eta}$ are uniformly
bounded over the entire $B_{2\eps}(v)$. By (\ref{eq: asymp_zm12})
and the boundedness of $\beta_{\delta}$, we get that $\beta_{\delta}\dmsqrt{\cdot}{v^{\delta}},$
and hence $\tilde{F}(\cdot),$ are uniformly bounded on $B_{2\eps}(v)\setminus B_{\hat{\eps}}(v).$
Since we can take $z_{0}$ to be any of the four corners adjacent
to $v^{\delta}$, this yields that indeed, $\tilde{F}^{\delta}$ is
uniformly bounded on $\Od_{\hat{\eps}}$ for every $\hat{\eps}>0$.

We may therefore assume by another extraction that there exists a
holomorphic spinor $\tilde{f}$ such that $\tilde{F}(z_{2}^{\delta})-\proj{z_{2}^{\delta}}{\tilde{f}(z_{2})}=o(1)$
uniformly on compact subsets of $\Omega\setminus\{z_{1}\}.$ The same
arguments as above imply that $\tilde{f}$ satisfies the standard
boundary conditions on $\pa\Omega\setminus\{z_{1}\}$, and $\tilde{F}(\cdot)-(M_{\eps}^{\delta})^{-\frac{1}{2}}\delta^{-1}P_{z_{1}}(\cdot)$
is uniformly bounded in the neighborhood of $z_{1}$. Passing to the
limit, we see, by (\ref{eq: asymp_zm1}) and (\ref{eq: M_delta_to_infty}),
that $\tilde{f}$ extends holomorphically to $z_{1}$. But then, by
Proposition \ref{prop: Uniqueness_continuous}, $\tilde{f}\equiv0$,
and hence $h:=\lim_{\delta\to0}\tilde{H}^{\delta}=\im\int f^{2}\equiv0$.
This is a contradiction with the choice of $M_{\eps}^{\delta}$, which
proves (UBC) and thus concludes the proof.
\end{proof}

\begin{proof}[Proof of Theorem \ref{thm: Convergence_bulk_near}]
 Denote 
\[
\hat{F}(\cdot):=\bar{\eta}_{1}\delta^{-1}F(z_{1}^{\delta},\cdot)-\dzmone{\cdot}{z_{1}^{\delta}};\quad\hat{f}(w):=\Cpsi^{2}\left(f^{[\eta_{1}]}(z_{1},w)-\frac{\bar{\eta}_{1}}{w-z_{1}}\right),
\]
where, as above, $\eta_{1}=\eta_{z_{1}}$ is assumed to be fixed.
By Theorem $\ref{thm: conv_bulk_bulk}$ and (\ref{eq: asymp_zmone}),
we have 
\[
\hat{F}(w)-\proj{\eta_{w}}{\hat{f}(w)}=o(1),
\]
uniformly in $w\in B_{2\eps}(z_{1})\setminus B_{\eps}(z_{1}),$ for
every $\eps>0$. As explained above, $\hat{F}(\cdot)$ is s-holomorphic
in $B_{2\eps}(z_{1})$ for $\eps$ small enough. Therefore, the maximum
principle implies that for every fixed $\eps>0$, we have 
\[
\inf_{w\in\left(B_{2\eps}(z_{1})\setminus B_{\eps}(z_{1})\right)}\bar{\eta}_{2}\proj{\eta_{2}}{\hat{f}(w)}-\eps\leq\bar{\eta}_{2}\hat{F}(z_{2}^{\delta})\leq\sup_{w\in\left(B_{2\eps}(z_{1})\setminus B_{\eps}(z_{1})\right)}\bar{\eta}_{2}\proj{\eta_{2}}{\hat{f}(w)}+\eps,
\]
where $\eta_{2}=\eta_{z_{2}^{\delta}},$ provided that $\delta$ is
small enough. Since, by (\ref{eq: f_residue}), $\hat{f}$ extends
holomorphically to $w=z_{1}$, this implies, by letting $\eps\to0$,
that 
\[
\hat{F}(z_{2}^{\delta})\to\proj{\eta_{2}}{\hat{f}(z_{1})},\quad\delta\to0.
\]
Since, by (\ref{eq: P_a_close}), we have $\bar{\eta}_{1}\delta^{-1}F(z_{1}^{\delta},z_{2}^{\delta})=\hat{F}(z_{2}^{\delta}),$
it remains to notice that $\eta_{2}=i\eta_{1}$, and combine the decomposition
(\ref{eq: f_eta_decomp}) with the asymptotics (\ref{eq: f_fdag_expansion})
and the fact that $\fdag(z_{1},z_{1})\in i\R$.
\end{proof}

\begin{proof}[Proof of Theorem \ref{thm: Convergence_singular}]
. We give the proof case by case, first doing the case of $z_{1}^{\delta}$
in the bulk and $z_{2}^{\delta}$ adjacent to a spin, i. e., $z_{1}^{\delta}\to z\in\Ocvrc$
and $(z_{2}^{\delta})^{\circ}=v^{\delta}\in\{v_{1},\dots,v_{n}\}$.
We can write, using Lemma \ref{lem: Cauchy_spinors}, 
\[
\delta^{-\frac{1}{2}}F(z_{1}^{\delta},z_{2}^{\delta})=\frac{\delta^{-1}}{4}\intFQ{F(z_{1},z)Q_{z_{2}^{\bullet}}(z)}{\gamma},
\]
where $\gamma$ is a contour $B_{2\eps}(v)\setminus B_{\eps}(v)$
encircling $v$, for a small fixed $\eps.$ We know from Theorem \ref{thm: conv_bulk_bulk}
that $\delta^{-1}F(z_{1}^{\delta},z)=\frac{2}{\pi}\proj{\eta_{z}}{f^{[\eta_{z_{1}}]}(z_{1},z)}+o(1)$
uniformly on $B_{2\eps}(v)\setminus B_{\eps}(v);$ taking into account
the asymptotics (\ref{eq: asymp_zm12}), we apply Remark \ref{rem: conv_f_conv_H}
to get 
\begin{multline*}
\delta^{-\frac{1}{2}}F(z_{1}^{\delta},z_{2}^{\delta})\to\frac{1}{4}\left(\frac{2}{\pi}\right)^{\frac{3}{2}}\im\int_{\gamma}f^{[\eta_{z_{1}}]}(z_{1},z)\lamb(z-v)^{-\frac{1}{2}}dz\\
=\frac{1}{4}\left(\frac{2}{\pi}\right)^{\frac{3}{2}}\im\left(2\pi if^{[\eta_{z_{1}},\sharp]}(z_{1},v)\right)=\left(\frac{2}{\pi}\right)^{\frac{1}{2}}f^{[\eta_{z_{1}},\sharp]}(z_{1},v)
\end{multline*}
It is straightforward to see from the proof that the signs are consistent
with Remark~\ref{rem: Signs}.

We now consider the case when $z_{1}^{\delta}\to z\in\Omega_{\cvr}$
and $z_{2}^{\delta}$ is separating a fixed and a free boundary arc,
i. e., $(z_{2}^{\delta})^{\circ}\in\fixed$ and $(z_{2})^{\bullet}\in\free$.
Recall that the function 
\[
h(w):=\Cpsi^{4}\im\int^{w}(f^{[\eta_{1}]})^{2}
\]
is the limit of $H^{\delta}$ defined from $\bar{\eta}_{1}\delta^{-1}F^{\delta}$
by (\ref{eq: def_H}). The jump in the boundary values of $h$ at
$z_{2}$, which is given by $\pi\Cpsi^{4}f^{[\eta_{1},\flat]}(z_{1},z_{2})^{2}$,
is also the limit of the corresponding jump for $H^{\delta}$, which,
by (\ref{eq: def_H}), is given by $\delta^{-1}|F(z_{1},z_{2})|^{2}=\delta^{-1}(\bar{\eta}_{1}\bar{\eta}_{2}F(z_{1},z_{2}))^{2}$.
This gives the required result up to sign. Checking that the sign
is in accordance with Remark \ref{rem: Signs} is similar to checking
(\ref{eq: bc_free_arc}) in the proof of Theorem \ref{thm: conv_bulk_bulk};
we leave the details to the reader.

The above argument in particular shows that the s-holomorphic function
$\bar{\eta}_{2}\delta^{-\frac{1}{2}}F(\cdot,z_{2}^{\delta})$ is within
$o(1)$ from $C_{2}\cdot\proj{\eta_{1}}{-f^{[\any]}(z_{2},\cdot)},$
uniformly in the bulk and away from other marked points. Hence, exactly
the same proof applied to $\bar{\eta}_{2}\delta^{-\frac{1}{2}}F(\cdot,z_{2}^{\delta})$
and $-f^{[\any]}(z_{2},\cdot)$ in the place of $\bar{\eta}_{1}\delta^{-1}F(z_{1}^{\delta},\cdot)$
and $f^{[\eta_{1}]}(z_{1},\cdot)$ respectively, settles the remaining
cases.
\end{proof}

\begin{proof}[Proof of Theorem \ref{thm: Conv_both_near_spin}]
Note that by definition, $\bar{\eta}_{z_{1}}F(z_{1},\cdot)$ can
be viewed as an s-holomorphic spinor ramified at $u,v_{1},\dots\hat{v}_{s},\dots,v_{n}$,
where $u:=z_{1}^{\bullet}$. We can therefore write, by Lemma \ref{lem: Cauchy_spinors},
\[
\bar{\eta}_{z_{2}}\bar{\eta}_{z_{1}}F(z_{1},z_{2})-\bar{\eta}_{z_{1}}^{2}F(z_{1},z_{1})=\frac{1}{4}\delta^{-\frac{1}{2}}\intFQ{\bar{\eta}_{z_{1}}F(z_{1},z)\left(\dmsqrt z{z_{2}^{\circ}}-\dmsqrt z{z_{1}^{\circ}}\right)}{\gamma}
\]
for a contour $\gamma$ staying in the annulus $B_{2\eps}(v_{s})\setminus B_{\eps}(v_{s})$
for a small but fixed $\eps.$ We can rewrite $\dmsqrt z{z_{2}^{\circ}}-\dmsqrt z{z_{1}^{\circ}}=\dmsqrt{z+z_{1}^{\circ}-z_{2}^{\circ}}{z_{1}^{\circ}}-\dmsqrt z{z_{1}^{\circ}}.$
Now, note that for each $z$, $z+z_{1}^{\circ}-z_{2}^{\circ}\in\Cgr^{\eta_{z}}$
and $z\in\in\Cgr^{\eta_{z}}$; the restriction of $\dmsqrt z{z_{1}^{\circ}}$
to each of these four sub-lattices are locally harmonic functions
that converge uniformly on $B_{2\eps}(v_{s})\setminus B_{\eps}(v_{s})$
to a harmonic function. It is well known that this implies uniform
converges of finite difference to the derivatives; in other words,
\[
\dmsqrt z{z_{2}^{\circ}}-\dmsqrt z{z_{1}^{\circ}}=\left(\frac{2}{\pi}\right)^{\frac{1}{2}}\proj{\eta_{z}}{(z_{1}^{\circ}-z_{2}^{\circ})\pa_{z}\left[\lambb(z-z_{1}^{\circ})^{-\frac{1}{2}}\right]}+o(\delta),
\]
uniformly in $B_{2\eps}(v_{s})\setminus B_{\eps}(v_{s}).$ Consequently,
by Remark \ref{rem: conv_f_conv_H}, we get 
\begin{multline*}
\bar{\eta}_{z_{2}}\bar{\eta}_{z_{1}}F(z_{1},z_{2})-\bar{\eta}_{z_{1}}^{2}F(z_{1},z_{1})=\frac{1}{8}\frac{2}{\pi}\im\left[(z_{2}^{\circ}-z_{1}^{\circ})\int_{\gamma}\left(\lambb f^{[\sharp]}(v_{s},z)(z-v_{s})^{-\frac{3}{2}}dz\right)\right]+o(\delta)\\
=\frac{1}{8}\frac{2}{\pi}\im\left[(z_{2}^{\circ}-z_{1}^{\circ})4\pi i\coefA_{\Omega,\cvr}(v)\right]+o(\delta)=\re[(z_{2}^{\circ}-z_{1}^{\circ})\coefA_{\Omega,\cvr}(v)]+o(\delta).
\end{multline*}
It remains to recall that $\bar{\eta}_{z_{1}}^{2}F(z_{1},z_{1})=1.$
\end{proof}

\subsection{Boundary behavior of s-holomorphic functions.}

This section, leading to the proof of Lemma \ref{lem: Clements_clever_lemma},
suitable generalizes and simplifies the arguments in \cite{Hongler_Kytola}.
If $w\in\pa\Omega$ is a prime end, by a \emph{neighborhood} $\Neigh$
of $w$ we will mean a simply-connected domain whose boundary consists
of an arc $\ArcNeig\subset\pa\Omega$ and a cross-cut $\CCNeig$,
such that $w\in\ArcNeig$. If $w\in\fixed$, we will consider only
\emph{small neighborhoods} for which $\ArcNeig\subset\fixed$. We
can Carathéodory approximate such a neighborhood by a sequence $\tilde{\Omega}^{\delta}$
of simply-connected subdomains of $\Od$ in such a way that $\CCNeig^{\delta}\to\CCNeig$
and $\tilde{\gamma}^{\delta}\to\text{\ensuremath{\tilde{\gamma}}}.$ 

We start with a weaker form of (\ref{eq: conv_bdry_rough}), in which
we only consider $\{z_{1}^{\delta}\}$ of type $\eta$ and do not
insist that the normalizing factor is local:
\begin{lem}
\label{lem: z_2_at_bdry}Fix a conformal map $\varphi$ from $\Omega$
to a nice domain $\Lambda$. If $z_{2}^{\delta}$ is a sequence of
boundary corners converging to $z_{2}\in\fixed$, then there exists
a normalizing factors $\normAny$ such that 
\[
\normAny\bar{\eta}_{z_{2}}F_{\Od,\cvr}^{\delta}(z,z_{2}^{\delta})=\proj{\eta_{z}}{f(z)}+o(1)\quad\delta\to0,
\]
uniformly over $z$ in the bulk and away from other marked points,
where 
\[
f(z)=\varphi'(z)^{\frac{1}{2}}f_{\Lambda,\varphi(\cvr)}(\varphi(z),\varphi(z_{2})).
\]
\end{lem}

\begin{proof}
We follow closely the proof of Theorem \ref{thm: conv_bulk_bulk}.
We define the function $H^{\delta}(\cdot)$ according to Definition
\ref{def: H} with $F(\cdot)=\bar{\eta}_{z_{2}}F_{\Od,\cvr}^{\delta}(\cdot,z_{2}^{\delta})$,
fixing the additive constant so that $H^{\delta}\equiv0$ at the boundary
arc containing $z_{2}^{\delta}.$ Consider a family $\tilde{\Omega}_{\eps}$
of small neighborhoods of $z_{2}$ in $\Omega$, indexed by $\eps>0$,
and such that the corresponding cross-cuts converge to $z_{2}$ as
$\eps\to0.$ We put $\Od_{\eps}:=\Od\setminus\left(\cup B_{\eps}(v_{i})\cup\tilde{\Omega}_{\eps}^{\delta}\right)$,
where $\tilde{\Omega}_{\eps}^{\delta}$ approximates $\tilde{\Omega}_{\eps}$
as discussed above. We define 
\[
M_{\eps}^{\delta}:=\max_{\Od_{\eps}}|H^{\delta}|.
\]
We claim that $\left(M_{\eps}^{\delta}\right)^{-1}H^{\delta}$ are,
in fact, uniformly bounded on any set of the form $\Od_{\hat{\eps}}$
where $\hat{\eps}>0.$ Indeed, the boundedness near $v_{i}$ is shown
as in the proof of Theorem \ref{thm: conv_bulk_bulk}, and the boundedness
near $z_{2}$ is shown by the following argument, as in \cite{ChelkakSmirnov2}.
First, using Lemma \ref{lem: sub_super}, we have by maximum principle,
$\left(M_{\eps}^{\delta}\right)H^{\delta}(w)\geq-1$ for $w\in\Omega^{\delta,\circ}$
on the whole $\Od\setminus\left(\cup_{i}B_{\eps}(v_{i})\right)$ (recall
that $H^{\delta}\equiv0$ at the boundary arc containing $z_{2}^{\delta}$).
Also, we always have $H^{\delta}(z^{\circ})\leq H^{\delta}(z^{\bullet});$
hence it is sufficient to show that $H^{\delta}$ is bounded from
\emph{above} on $\Odual\setminus\left(\cup B_{\hat{\eps}}(v_{i})\cup\hat{\Omega}_{\hat{\eps}}^{\delta}\right)$.
By Lemma \ref{lem: sub_super}, if $w\in\tilde{\Omega}_{\eps}^{\delta,\bullet}\setminus\tilde{\Omega}_{\hat{\eps}}^{\delta,\bullet},$
then there is a point $w\in$ $w_{1}\sim w$ with $H^{\delta,\bullet}(w_{1})\geq H^{\delta,\bullet}(w).$
Iterating, we obtain a path $\gamma^{\bullet}=w\sim w_{1}\sim w_{2}\sim\dots$
with $H^{\delta}(w_{i})\geq H^{\delta}(w).$ Unless $\left(M_{\eps}^{\delta}\right)^{-1}H^{\delta}(w)\leq1,$
this path can only end at $z_{2}^{\bullet}.$ Let $\gamma^{\circ}=\{\hat{w}\in\Omega^{\delta,\circ}:\exists i:\,\hat{w}\sim w_{i}\}$
be the set of points of $\Omega^{\delta,\circ}$ adjacent to $\gamma^{\bullet}.$
The discrete $\Delta^{\circ}$-harmonic measure of $\gamma^{\circ}$
as seen from some fixed point $\tilde{w}\in\Omega\setminus\left(\cup B_{\eps}(v_{i})\cup\tilde{\Omega}_{\eps}^{\delta}\right)$
is bounded from below by a constant $c>0$ independent of $\delta$;
also, \cite[Proposition 3.8]{ChelkakSmirnov2} implies that there
are constants $C_{1,2}$ such that $H^{\delta}(\hat{w})\geq C_{1}H^{\delta}(w)-C_{2}M_{\eps}^{\delta}$
for any $\hat{w}\in\gamma^{\circ}.$ From Lemma \ref{lem: sub_super},
this means that 
\[
M_{\eps}^{\delta}\geq H^{\delta}(\tilde{w})\geq c\left(C_{1}H^{\delta}(w)-C_{2}M_{\eps}^{\delta}\right)-M_{\eps}^{\delta},
\]
which is the desired uniform upper bound on $H^{\delta}(w).$

After this, we see as in the proof of Theorem \ref{thm: conv_bulk_bulk}
that at least along sub-sequences, we have 
\[
F^{\delta}(z)=\proj{\eta_{z}}{f(z)}+o(1);\quad H^{\delta}(\cdot)=h(\cdot)+o(1),
\]
 where $h=\im\int f^{2}dz$ and $f$ satisfies the standard boundary
conditions on $\pa\Omega\setminus\{w\}.$ Moreover, we know that $h\text{\ensuremath{\geq0}}$
near $w$; hence $h$ is locally a linear combination of a Poisson
kernel and a harmonic function continuous up to the boundary (with
zero boundary values there). This identifies $h$ uniquely up to a
multiplicative constant, which can be taken care of by a choice of
the factor $\normAny$.
\end{proof}
The next key lemma allows one to ``factorize'' the value of an s-holomorphic
function near rough boundary into a factor that depends on the function
and the factor that depends on the local geometry of the boundary. 
\begin{lem}
\label{lem: rough_bdry}Let $\tilde{\Omega}$ be a simply connected
domain whose boundary is divided into two arcs $\gamma_{1,2}$, and
let $\tilde{\Omega}^{\delta},\gamma_{1,2}^{\delta}$ be their respective
approximations. Assume that $G^{\delta}$ is a sequence of s-holomorphic
functions in $\tilde{\Omega}^{\delta}$ satisfying $\fixed$ boundary
conditions on $\gamma_{1}^{\delta}$, and such that the corresponding
functions $H^{\delta}=\intFQ{G^{\delta}(z)^{2}}{}$ are uniformly
bounded in $\tilde{\Omega}^{\delta}$. Assume furthermore that 
\begin{equation}
G^{\delta}(z)=\proj{\eta_{z}}{g(z)}+o(1),\label{eq: gconv}
\end{equation}
uniformly on compact subsets of $\tilde{\Omega}$. Fix a conformal
map $\varphi$ from $\tilde{\Omega}$ to any domain, such that $\varphi(\gamma_{1})$
is smooth. Let $w\in\gamma_{1}$, and let $w^{\delta}\to w$ be a
sequence of boundary corners of $\Omega^{\delta}$. Then, there is
a sequence $\beta^{\delta}\in\C$ of normalizing factors that depends
on $\tilde{\Omega}_{\delta},w^{\delta}$, and the choice of $\varphi,$
but not on $G^{\delta}$, such that 
\begin{equation}
\beta^{\delta}G^{\delta}(w^{\delta})\to\lim_{z\to w}\varphi(z)^{-\frac{1}{2}}g(z).\label{eq: limit_rought_bdry}
\end{equation}
\end{lem}

\begin{rem}
\label{rem: limit_exists}As explained in the proof of Theorem \ref{thm: conv_bulk_bulk},
$g$ must have standard $\fixed$ boundary conditions on $\gamma_{1}.$
Therefore, $f(z)=\varphi(z)^{\frac{1}{2}}f_{\varphi(\tilde{\Omega})}(\varphi(z))$,
where $f_{\varphi(\tilde{\Omega})}$ satisfies $\fixed$ boundary
conditions on $\varphi(\gamma_{1}),$ and thus extends there continuous.
Therefore, the limit in the right-hand side of (\ref{eq: limit_rought_bdry})
exists. 
\end{rem}

\begin{proof}[Proof of Lemma \ref{lem: rough_bdry}]
. We choose a smaller neighborhood $\SOmega\subset\Neigh$ of $w$,
such that $\hat{\gamma}_{\mathrm{out}}\subset\Neigh$ is a piece-wise
smooth curve with endpoints in $\gamma_{1}$. We also choose a sequence
of neighborhoods $\SOd$ of $w^{\delta}$ in $\Od$ that Carathéodory
converges to $\SOmega$ and such that $\ArcHat_{\mathrm{out}}\to\hat{\gamma}_{\mathrm{out}}$.
We furthermore require that $\ArcHat_{\mathrm{out}}$ are of bounded
length and approach $\pa\tilde{\Omega}^{\delta}$ non-tangentially,
i. e., that there is a constant $c>0$ such that 
\begin{equation}
\text{\ensuremath{\dist}(z;\ensuremath{\pa\Od})\ensuremath{\geq}c\ensuremath{\cdot}\ensuremath{\min\{|z-b_{1}^{\delta}|,|z-b_{2}^{\delta}|\}},}\label{eq: nontangential}
\end{equation}
where $b_{1,2}^{\delta}$ are endpoints of $\hat{\gamma}_{\text{out}}^{\delta}.$

Consider a sequence of s-holomorphic functions $\tilde{F}^{\delta}(\cdot):=\tilde{\beta}^{\delta}\bar{\eta}_{w^{\delta}}F_{\tilde{\Omega}^{\delta},\fixed}(w^{\delta},\cdot)$
in $\tilde{\Omega}^{\delta}.$ The normalization constants $\tilde{\beta}^{\delta}$
are chosen as in Lemma \ref{lem: z_2_at_bdry}, so that 
\begin{equation}
\tilde{F}^{\delta}(z)=\proj{\eta_{z}}{\tilde{f}(z)}+o(1).\label{eq: rough_bdry_fconv}
\end{equation}
where $\tilde{f}(z)=\varphi(z)^{\frac{1}{2}}f_{\varphi(\tilde{\Omega})}(\varphi(w),\varphi(z)).$
Since composing the conformal map $\varphi$ with another conformal
map $\tilde{\varphi}$ changes the right-hand side of (\ref{eq: limit_rought_bdry})
by a constant factor which can be absorbed into $\beta^{\delta}$,
we may in fact assume that $\varphi(\tilde{\Omega})=\H$, $\tilde{f}(z)=\varphi(z)^{\frac{1}{2}}f_{\H}(\varphi(w),\varphi(z))=\frac{2\varphi'(z)}{\varphi(z)-\varphi(w)}.$ 

We now consider the discrete integral $\intFQ{\tilde{F}(z)G(z)}{\pa\hat{\Omega}^{\delta}}.$
Because of the standard boundary conditions, we have 
\begin{multline}
\intFQ{\tilde{F}(z)G(z)}{\pa\hat{\Omega}^{\delta}}=\intFQ{\tilde{F}(z)G(z)}{\ArcHat_{\mathrm{out}}}\\
\stackrel{\delta\to0}{\longrightarrow}\int_{\hat{\gamma}_{\mathrm{out}}}\im(\tilde{f}(z)g(z)dz.\label{eq: conv_integral_rough}
\end{multline}
Indeed, the convergence of the part of the discrete integral ``in
the bulk'' follows from Remark \ref{rem: conv_f_conv_H}, and the
parts close to the boundary are negligible because of the following
estimate: since $H^{\delta}(z)$ are uniformly bounded, Lemma \ref{lem: sub_super}
and the discrete Beurling inequality implies that $|H^{\delta}(u)|\leq C\cdot\dist(u;\gamma_{1}^{\delta})^{\frac{1}{2}}$,
and \cite[Theorem 3.12]{ChelkakSmirnov2} readily gives $|G^{\delta}(z)|\leq C\cdot\dist(z;\gamma_{1}^{\delta})^{-\frac{1}{4}}.$
Since $\tilde{H}^{\delta}=\intFQ{\tilde{F}^{2}}{}$ are also uniformly
bounded on $\ArcHat_{\mathrm{out}}$, similar estimate holds $\tilde{F}$,
combining to $|\tilde{F}\tilde{G}|\leq\dist(z;\gamma_{1}^{\delta})^{-\frac{1}{2}},$
yielding a convergent improper integral and thus justifying the convergence
in (\ref{eq: conv_integral_rough}). 

Now, as in the proof of Lemma \ref{lem: Cauchy_spinors}, we can shrink
the contour $\pa\hat{\Omega}^{\delta}$ to the contour consisting
of just two edges, one running from $w^{\delta,\bullet}$ to $w^{\delta,\circ}$
along $w^{\delta,+}$ and another running back along $w^{\delta,-};$
this yields 
\begin{multline*}
-4i\tilde{F}(w^{\delta,+})G(w^{\delta})(w^{\delta,\bullet}-w^{\delta,\circ})\to\int_{\hat{\gamma}_{\mathrm{out}}}\im(\tilde{f}(z)g(z)dz\\
=\im\int_{\varphi(\hat{\gamma}_{\mathrm{out}})}f_{\varphi(\tilde{\Omega})}(u)\frac{2}{u-\varphi(w)}du=2\pi f_{\varphi(\tilde{\Omega})}(\varphi(w))
\end{multline*}
 where $f_{\varphi(\tilde{\Omega})}$ is as in Remark \ref{rem: limit_exists}.
Hence, $\beta^{\delta}:=-4i\tilde{F}(w^{\delta,+})(w^{\delta,\bullet}-w^{\delta,\circ})(2\pi)^{-1}$
is the required normalizing factor, and the lemma is proven. 
\end{proof}
\begin{rem}
In the above proof, we used sub/superharmonicity of $H^{\delta}$
and the discrete Beurling estimates; however, what we need is in fact
$|\tilde{F}\tilde{G}|\leq\dist(z;\gamma_{1}^{\delta})^{-1+\eps}$
for some any $\eps>0$, which can be achieved by weaker tools.
\end{rem}

\begin{proof}[Proof of Lemma \ref{lem: Clements_clever_lemma}]
 For (\ref{eq: conv_bdry_rough}), we apply Lemma \ref{lem: rough_bdry}
to $G^{\delta}(\cdot):=\delta^{-\Delta_{1}}\bar{\eta}_{z_{1}^{\delta}}F(z_{1},\cdot);$
its convergence in the bulk follows from Theorem \ref{thm: Convergence_singular}
and the fact that $H^{\delta}$ is uniformly bounded (the UBC condition)
was a step in the course of the proof of Theorem \ref{thm: conv_bulk_bulk}
for the case when $z_{1}^{\delta}$ is of type $\eta$ and can be
justified similarly in other cases. For (\ref{eq: conv_bdry_bdry}),
we apply Lemma \ref{lem: rough_bdry} to $G^{\delta}(\cdot):=c\cdot\beta_{z_{1}^{\delta}}^{\delta}F(z_{1}^{\delta},\cdot)$,
with $c=\bar{\eta}_{z_{1}^{\delta}}\bar{\beta}_{z_{1}^{\delta}}^{\delta}/|\bar{\beta}_{z_{1}^{\delta}}^{\delta}|$
chosen so that the phase condition (\ref{eq: phases_condition}) is
satisfied; the fact that in this case, $H^{\delta}$ is uniformly
bounded away from $z_{1}^{\delta}$ was justified in the proof of
Lemma \ref{lem: z_2_at_bdry}.
\end{proof}
\newpage{}

\section{Correlation functions in continuum and proofs of the main theorems}

\label{sec: corr-continuum}

\subsection{Spin correlations.}

In this section, we will assume that $\cvr=\cvr(v_{1},\dots,v_{n})$,
where $\{v_{\text{i}}\}$ are the single-point boundary components.
All the statements can be modified to accommodate boundary ramification
points, but we will not need this generality.

Recall from (\ref{eq: f_sharp_asymp}) that if $v\in\{v_{1},\dots,v_{n}\}$,
then $f_{\cvr}^{[\sharp]}(v,z)\sim\lamb(z-v)^{-\frac{1}{2}}$ as $z\to v$.
The crucial role in the analysis of spin correlations is played by
the next coefficient in the expansion.

Until the end of this subsection, we will view $v_{1},\dots,v_{n}$
as points in $\Omega$. 
\begin{defn}
We define the differential form $\formL$ on the space of $n$-tuples
$(v_{1},\dots,v_{n})$ of distinct points in $\Omega$ by 
\[
\formL=\formL^{(n,\Omega)}:=\re\left(\sum_{i=1}^{n}\mathcal{A}_{\Omega,\cvr}(v_{i})dv_{i}\right),
\]
where $\cvr=\cvr(v_{1},\dots,v_{n}).$
\end{defn}

Using the definition of $F(z_{1},z_{2})$ for $z_{1},z_{2}$ as in
Theorem \ref{thm: Conv_both_near_spin}, we deduce the following Corollary. 
\begin{cor}
\label{cor: spin_ratio} Let $\gamma$ be any path in the space of
$n$-tuples of distinct points in $\Omega$ connecting $(\hat{v}_{1},\dots,\hat{v}_{n})$
to $(v_{1},\dots,v_{n})$. Then we have, as $\delta\to0$, 
\begin{equation}
\frac{\E_{\Od}(\sigma_{v_{1}}\dots\sigma_{v_{n}})}{\E_{\Od}(\sigma_{\hat{v}_{1}}\dots\sigma_{\hat{v}_{n}})}=\exp\left(\int_{\gamma}\formL\right)+o(1),\label{eq: conv_ratio}
\end{equation}
In particular, the form $\formL$ is exact, and integrates to a symmetric
function of $v_{1},\dots,v_{n}$.
\end{cor}

\begin{proof}
We first give the proof in the case when $v_{2}=\hat{v}_{2},\dots,v_{n}=\hat{v}_{n}.$
Choose a lattice path $v_{1}=v^{(0)}\sim v^{(1)}\sim\dots\sim v^{(p)}=\hat{v}_{1}$
which stays at a definite distance from $v_{2},\dots,v_{n}$; we take
$\gamma$ to be the segments $[v_{1};\hat{v}_{1}]$. Choose the corners
$z^{(2k)}$ and $z^{(2k+1)}$ such that $(z^{(2k)})^{\circ}=v^{(k)}$
and $(z^{(2k+1)})^{\circ}=v^{(k+1)}$ and $(z^{(2k)})^{\bullet}=(z^{(2k+1)})^{\bullet}$.
This way, $z^{(2k)}$ and $z^{(2k+1)}$ are as in Theorem \ref{thm: Conv_both_near_spin},
so that

\begin{multline*}
\frac{\E_{\Od}(\sigma_{v^{(k+1)}}\sigma_{v_{2}}\dots\sigma_{v_{n}})}{\E_{\Od}(\sigma_{v^{(k)}}\sigma_{v_{2}}\dots\sigma_{v_{n}})}=\bar{\eta}_{z^{(2k)}}\bar{\eta}_{z^{(2k+1)}}F\left(z^{(2k)},z^{(2k+1)}\right)\\
=1+\re[(v^{(k+1)}-v^{(k)})\mathcal{A}(v^{(k)})]+o(\delta)=e^{\re[(v^{(k+1)}-v^{(k)})\mathcal{A}(v^{(k)})]+o(\delta)}.
\end{multline*}
Multiplying over $k$ and passing to the limit $\delta\to0$ gives
the result.

Iterating the above argument and exchanging the roles of $v_{1},\dots,v_{n}$
if necessary, we can prove (\ref{eq: conv_ratio}) in the general
case. Since the left-hand side does not depends on the paths chosen,
we infer that $\formL$ is in fact exact.
\end{proof}
Motivated by Corollary \ref{cor: spin_ratio}, we define the continuous
spin correlation functions by 
\begin{equation}
\ccor{\sigma_{v_{1}}\dots\sigma_{v_{n}}}_{\Omega}=\exp\left(\int\formL^{(n,\Omega)}\right),\label{eq: def_spincorr_up_to_constant}
\end{equation}
which is a function of an $n$-tuple $v_{1},\dots,v_{n}$ of distinct
points defined up to a multiplicative constant. The following proposition
completes the definition of $\CorrO{\sigma_{v_{1}}\dots\sigma_{v_{n}}}$
when all spins are in the bulk.
\begin{prop}
\label{prop: multiplicative-normalisation}It is possible to choose
the constants of integration in (\ref{eq: def_spincorr_up_to_constant})
in such a way that 
\begin{equation}
\ccor{\sigma_{v_{1}}\sigma_{v_{2}}\dots\sigma_{v_{n}}}_{\Omega}\sim|v_{1}-v_{2}|^{-\frac{1}{4}}\ccor{\sigma_{v_{3}}\dots\sigma_{v_{n}}}_{\Omega}\label{eq: spins_coherent}
\end{equation}
as $v_{1}$ tends to $v_{2}\in\Omega$. With the convention $\CorrO{}=1$,
this choice is unique.
\end{prop}

\begin{proof}
The proof of existence of such a normalization is postponed until
Section \ref{sec: ccor_fusion}. To prove uniqueness, just observe
that iterating (\ref{eq: spins_coherent}) for $n=2,4,\dots$ uniquely
fixes the integration constants in (\ref{eq: def_spincorr_up_to_constant})
.
\end{proof}

We now define $\CorrO{\sigma_{1}\dots\sigma_{n}}$ in the case some
spins are at the boundary. Recall that, when $\Omega$ is a simply-connected
domain, the \emph{conformal radius} $\crad_{\Omega}(v)$ of $\Omega$
as seen from $v\in\Omega$ is defined as $|\varphi'(0)|$, where $\varphi$
is any conformal map from the unit disc $\D$ to $\Omega$. When $\Omega$
is multiply connected, we define $\crad_{\Omega}(v)$ to be the minimum
of conformal radii $\crad_{\Omega'}(v)$ over simply-connected domains
$\Omega'$ obtained by placing $\Omega$ on the Riemann sphere and
filling in all but one holes in $\Omega$. 
\begin{prop}
\label{prop: spin_bdry_behavior}Assume that $w_{1},\dots,w_{d}\in\fixed$.
Then, the following limit exists:
\begin{equation}
\ccor{\sigma_{v_{1}}\dots\sigma_{v_{n}}\sigma_{w_{1}}\dots\sigma_{w_{d}}}_{\Omega}:=\lim_{v_{n+1}\to w_{1},\dots,v_{n+d}\to w_{d}}\prod_{i=1}^{d}\crad_{\Omega}(v_{n+i})^{-\frac{1}{8}}\cdot\ccor{\sigma_{v_{1}}\dots\sigma_{v_{n+d}}}_{\Omega}.\label{eq: prop_spin_bdry}
\end{equation}
 
\end{prop}

\begin{proof}
See Section \ref{sec: ccor_fusion}.
\end{proof}
\begin{rem}
\emph{A priori}, the correlation $\ccor{\sigma_{v_{1}}\dots\sigma_{v_{n}}\sigma_{w_{1}}\dots\sigma_{w_{n}}}_{\Omega}$
defined in the above Lemma may depend on the positions of points $w_{1},\dots,w_{d}$.
Theorem \ref{thm: spin_convergence} below shows that in fact, it
only depends on the parity of the number of $w_{i}$ at each boundary
component. 
\end{rem}

\subsection{Definitions of $\protect\CorrO{...}$ in continuum}

\label{subsec: Def_corr_cont}

In this section, we collect the definitions of the continuous correlations
$\CorrO{\ldots}$ and $\ccor{\ldots}_{\Omega,\bcond}$, which appear
in the statements of Theorems \ref{thm: Intro_1}\textendash \ref{thm: Intro_3}.
The definitions are motivated by convergence Theorems \ref{thm: conv_bulk_bulk},
\ref{thm: Convergence_singular}, \ref{thm: Convergence_bulk_near},
\ref{thm: spin_convergence}, that, via the Pfaffian formula \ref{eq: Pfaff_discrete}
and Lemmas \ref{lem: corr_to_obs_for_THM_2} and \ref{lem: corr_to_obs_for_Thm_3},
lead to the proofs of Theorems \ref{thm: Intro_1}\textendash \ref{thm: Intro_3}.

The correlations 
\begin{equation}
\CorrO{\psi_{z_{1}}^{\any}\ldots\psi_{z_{k}}^{\any}\en_{e_{1}}\ldots\en_{e_{s}}\mu_{u_{1}}\ldots\mu_{u_{m}}\sigma_{v_{1}}\ldots\sigma_{v_{n}}}\label{eq: anycorr}
\end{equation}
that we define will be multi-valued functions of complex numbers $\eta_{1},\dots,\eta_{k}$
and the marked points $v_{1},\dots,v_{n}\in\Omega\cup\fixed$, $u_{1},\dots u_{m}\in\Omega$,
$z_{1},\dots,z_{k}\in\Omega$, $e_{1},\dots e_{s}\in\Omega$. All
the latter are assumed to be distinct, with a possible exception discussed
in Remark \ref{rem: psi_to_psibar} below. We first define the correlations
with standard boundary conditions, in which case we assume that $n+k$
and $m+k$ are both even.

Eventually, the expression (\ref{eq: anycorr}) will turn out to be
a spinor with the same ramification structure as (\ref{eq: RS_all_around_all}),
\emph{symmetric} under transpositions of $u_{i}$ and $u_{j}$, $v_{i}$
and $v_{j}$, $s_{i}$ and $s_{j}$, and \emph{anti-symmetric} under
transpositions of $z_{i}$ and $z_{j}$. However, for technical reasons,
it is more convenient at first to treat $u_{1},\dots,u_{m},v_{1},\dots,v_{n}$
as \emph{fixed} points, and the others as living on the double cover
\[
\cvr:=\cvr(u_{1},\dots,u_{m},v_{1},\dots,v_{n}).
\]
We fix an arbitrary choice of square roots $(z-u_{1})^{\frac{1}{2}},\dots,(z-u_{m})^{\frac{1}{2}}$
locally on $\Ocvrc$ near $u_{1},\dots,u_{m}$, respectively. Once
this choice is made, it is possible to view the disorder $\mu_{u_{i}}$
as a fermion $\psi_{w}$ placed ``infinitesimally close'' to a spin
$\sigma_{u_{i}}$, cf. Lemma \ref{lem: corr_to_obs_for_THM_2}, Theorem
\ref{thm: Convergence_singular}, and the definitions of $f^{[\sharp]}$
and $f^{[\sharp,\sharp]}$. In this situation, we will write $\psi_{u_{i}}^{\sharp}\sigma_{u_{i}}$
in the place of $\mu_{u_{i}}$, and the correlation (\ref{eq: anycorr})
will be \emph{anti-symmetric }under transposition of $\mu_{u_{i}}$
and $\mu_{u_{j}}$. More generally, we will allow the symbols in (\ref{eq: anycorr})
to be written \emph{in arbitrary order}, with the convention that
$\sigma_{v}$, $\mu_{u}$ and $\en_{v}$ \emph{``commute}'' with
everything, and $\psi_{z}^{\any}$ ``\emph{anti-commute}'' with
each other. 

If $u_{i}$ are allowed to move, then the choices of $(z-u_{1})^{\frac{1}{2}},\dots,(z-u_{m})^{\frac{1}{2}}$
are no longer arbitrary, and, in fact, this leads to an additional
$-1$ sign as $u_{i}$ and $u_{j}$ exchange positions, reconciling
the two points of view. \textcolor{black}{It is possible, given a
point on the Riemann surface of (\ref{eq: RS_all_around_all}), to
describe a natural way to make a choice of $(z-u_{1})^{\frac{1}{2}},\dots,(z-u_{m})^{\frac{1}{2}}$;
this choice will depend on the ordering of $u_{1},\dots,u_{m}$.}

In this chapter, we consider the domain $\Omega$ to be fixed and
drop it from the notation, writing simply 
\[
\ccor{\ldots}:=\CorrO{\ldots}.
\]
We also denote 
\[
\sigma_{\cvr}:=\sigma_{u_{1}}\dots\sigma_{u_{m}}\sigma_{v_{1}}\ldots\sigma_{v_{n}}.
\]

\begin{defn}
\label{def: spins}(Spin correlations) If $v_{1},\dots,v_{n}\in\Omega$,
we define the quantity $\ccor{\sigma_{v_{1}}...\sigma_{v_{n}}}$ by
(\ref{eq: def_spincorr_up_to_constant}) and (\ref{prop: multiplicative-normalisation}).
More generally, if $v_{1},\dots,v_{n}\in\Omega\cup\fixed$, we define
$\ccor{\sigma_{v_{1}}...\sigma_{v_{n}}}$ by the limiting procedure
as in (\ref{eq: prop_spin_bdry}).
\end{defn}

\begin{defn}
\label{def: real_fermions}(Real fermions) Given $\eta_{1},\eta_{2}\in\C$,
we define
\[
\ccor{\psi_{z_{1}}^{\eta_{1}}\psi_{z_{2}}^{\eta_{2}}\sigma_{v_{1}}...\sigma_{v_{n}}}:=f_{\cvr(v_{1},\dots,v_{n})}^{[\eta_{2},\eta_{1}]}(z_{2},z_{1})\cdot\ccor{\sigma_{v_{1}}...\sigma_{v_{n}}},
\]
Furthermore, given $\eta_{1},\dots,\eta_{k}\in\C$ and distinct points
$z_{1},\dots,z_{k}\in\Omega$ and $v_{1},\dots,v_{n}\in\Omega\cup\fixed$,
we define
\begin{equation}
\frac{\ccor{\psi_{z_{1}}^{\eta_{1}}...\psi_{z_{k}}^{\eta_{k}}\sigma_{v_{1}}...\sigma_{v_{n}}}}{\ccor{\sigma_{v_{1}}...\sigma_{v_{n}}}}:=\Pf\left[\frac{\ccor{\psi_{z_{p}}^{\eta_{p}}\psi_{z_{q}}^{\eta_{q}}\sigma_{v_{1}}...\sigma_{v_{n}}}}{\ccor{\sigma_{v_{1}}...\sigma_{v_{n}}}}\right]_{p,q=1}^{k}=\Pf\left[f_{\cvr(v_{1},\dots,v_{n})}^{[\eta_{q},\eta_{p}]}(z_{q},z_{p})\right]_{p,q=1}^{k}.\label{eq: ccor_pfaff_expansion}
\end{equation}
\end{defn}

\begin{rem}
Note that by Lemma \ref{lem: antisym_cont}, the matrix under the
Pfaffian is indeed anti-symmetric, and hence the quantity $\ccor{\psi_{z_{1}}^{[\eta_{1}]}...\psi_{z_{k}}^{[\eta_{k}]}\sigma_{v_{1}}...\sigma_{v_{n}}}$
is anti-symmetric under permutations of $\psi$'s and symmetric under
permutations of $\sigma$'s.
\end{rem}

Recall that $\ccor{\psi_{z_{1}}^{\eta_{1}}\psi_{z_{2}}^{\eta_{2}}\sigma_{v_{1}}...\sigma_{v_{n}}}$
is real linear in each $\eta_{1},\eta_{2}.$ Hence, by the linearity
of the Pfaffian, $\ccor{\psi_{z_{1}}^{\eta_{1}}\ldots\psi_{z_{k}}^{\eta_{k}}\sigma_{v_{1}}...\sigma_{v_{n}}}$
is real-linear in each $\eta_{k}$. Therefore, we can expand 
\[
\ccor{\psi_{z_{1}}^{\eta_{1}}\ldots\psi_{z_{k}}^{\eta_{k}}\sigma_{v_{1}}...\sigma_{v_{n}}}=\sum_{S\subset\{1,\dots,k\}}a_{S}\cdot\prod_{i\in S}\eta_{i}\prod_{i\notin S}\bar{\eta}_{i}
\]
with some coefficients $a_{S}=a_{S}(z_{1},\dots,v_{n})$ that do not
depend on $\eta_{1},\dots,\eta_{k}$. More generally, if $S\subset\{1,\dots,k\}$,
we can write 
\begin{equation}
\ccor{\psi_{z_{1}}^{\eta_{1}}\ldots\psi_{z_{k}}^{\eta_{k}}\sigma_{v_{1}}...\sigma_{v_{n}}}=\sum_{S'\subset S}a_{S,S'}\cdot\prod_{i\in S'}\eta_{i}\prod_{i\in S\setminus S'}\bar{\eta}_{i}\label{eq: exp_multilinear}
\end{equation}
where $a_{S,S'}=a_{S,S'}(z_{1},\dots,v_{n},\{\eta_{i}\}_{i\notin S})$
is real-linear in each $\eta_{i}$, $i\notin S$.
\begin{defn}
\label{def: any fermions}(Holomorphic and anti-holomorphic fermions)
Under conditions of Definition \ref{def: real_fermions}, assume that
for each $i=1,\dots,k$, $\psi_{z_{i}}^{\any}$ stands for $\psi_{z_{i}}$,
$\psi_{z_{i}}^{\star}$, or $\psi_{z_{i}}^{[\eta_{i}]}$. We define
\[
\ccor{\psi_{z_{1}}^{\diamond}...\psi_{z_{k}}^{\diamond}\sigma_{v_{1}}...\sigma_{v_{n}}}:=2^{|S|}a_{S,S'},
\]
where $a_{S,S'}$ is the coefficient in the expansion (\ref{eq: exp_multilinear})
with 
\[
S=\{i:\psi_{z_{i}}^{\any}=\psi_{z_{i}}\text{ or }\psi_{z_{i}}^{\any}=\psi_{z_{i}}^{\star}\}\text{ and }S'=\{i:\psi_{z_{i}}^{\any}=\psi_{z_{i}}^{\star}\}.
\]
\end{defn}

Definition \ref{def: any fermions} can be understood as the identity
\[
\psi_{z}^{\eta}=\frac{1}{2}(\bar{\eta}\psi_{z}+\eta\psi_{z}^{\star}),
\]
which can be plugged into any continuous correlation; we also declare
the latter to be linear in the symbols appearing therein. Formally,
this means that 
\begin{equation}
\ccor{\psi_{z}^{\eta}\Op}=\frac{1}{2}\bar{\eta}\ccor{\psi_{z}\Op}+\frac{1}{2}\eta\ccor{\psi_{z}^{\star}\Op},\label{eq: psi_eta_expand}
\end{equation}
where $\Op=\Op(\psi^{\any},\sigma)$. In fact, since more general
correlations will be obtained from $\ccor{\psi_{z_{1}}^{\diamond}...\psi_{z_{k}}^{\diamond}\sigma_{v_{1}}...\sigma_{v_{n}}}$
by linear procedures, (\ref{eq: psi_eta_expand}) will hold for arbitrary
$\Op=\Op(\psi^{\any},\en,\mu,\sigma).$
\begin{rem}
\label{rem: _psi_psibar_down_to_earth}It follows from the linearity
of the Pfaffian and Lemma \ref{lem: spinor_hol_antyhol} that we could
alternatively define 
\end{rem}

\begin{equation}
\frac{\ccor{\psi_{z_{1}}^{\diamond}...\psi_{z_{k}}^{\diamond}\sigma_{v_{1}}...\sigma_{v_{n}}}}{\ccor{\sigma_{v_{1}}...\sigma_{v_{n}}}}\ =\ \Pf\biggl[\frac{\ccor{\psi_{z_{p}}^{\diamond}\psi_{z_{q}}^{\diamond}\sigma_{v_{1}}...\sigma_{v_{n}}}}{\ccor{\sigma_{v_{1}}...\sigma_{v_{n}}}}\biggr]_{p,q=1}^{k}\,,\label{eq: multipsi_pfaffian}
\end{equation}
where 
\begin{align}
\ccor{\psi_{z}\psi_{a}^{[\eta]}\sigma_{v_{1}}...\sigma_{v_{n}}} & =f^{[\eta]}(a,z)\cdot\ccor{\sigma_{v_{1}}...\sigma_{v_{n}}}\label{eq: dte_psi_psieta}\\
\ccor{\psistar z\psi_{a}^{[\eta]}\sigma_{v_{1}}...\sigma_{v_{n}}} & =\overline{f^{[\eta]}(a,z)}\cdot\ccor{\sigma_{v_{1}}...\sigma_{v_{n}}}\label{eq: dte_psistar_psieta}\\
\ccor{\psi_{z_{1}}\psi_{z_{2}}\sigma_{v_{1}}...\sigma_{v_{n}}} & =f(z_{2},z_{1})\cdot\ccor{\sigma_{v_{1}}...\sigma_{v_{n}}}\label{eq: dte_psi_psi}\\
\ccor{\psi_{z_{1}}\psistar{z_{2}}\sigma_{v_{1}}...\sigma_{v_{n}}} & =f^{\star}(z_{2},z_{1})\cdot\ccor{\sigma_{v_{1}}...\sigma_{v_{n}}}\label{eq: dte_psi_psistar}\\
\ccor{\psistar{z_{1}}\psistar{z_{2}}\sigma_{v_{1}}...\sigma_{v_{n}}} & =\overline{f(z_{2},z_{1})}\cdot\ccor{\sigma_{v_{1}}...\sigma_{v_{n}}}\label{eq: dte_psistar_psistar}
\end{align}
and otherwise the correlations are anti-symmetric in $\psi^{\any}$.
\begin{rem}
\label{rem: psi_to_psibar}Observe from (\ref{eq: ccor_pfaff_expansion})
that, as $z_{k}\to z_{i}$, one has 
\[
\ccor{\psi_{z_{1}}^{\diamond}...\psi_{z_{k-1}}^{\any}\psi_{z_{k}}\sigma_{v_{1}}...\sigma_{v_{n}}}=\alpha_{i}\ccor{\psi_{z_{i}}^{\diamond}\psi_{z_{k}}\sigma_{v_{1}}...\sigma_{v_{n}}}+O(1),
\]
where $\alpha_{i}$ does not depend on $z_{k}$. In particular, in
the case $\psi_{z_{i}}^{\diamond}=\psi_{z_{i}}^{\star}$, the left-hand
side is analytic at $z_{k}=z_{i}$. Thus, the expression $\ccor{\psi_{z_{1}}^{\diamond}...\psi_{z_{k}}^{\diamond}\sigma_{v_{1}}...\sigma_{v_{n}}}$
is well-defined also for $z_{i}=z_{j}$, provided that $\psi_{z_{i}}^{\any}=\psi_{z_{i}}$,
$\psi_{z_{j}}^{\any}=\psi_{z_{j}}^{\star}$, or vice versa.
\end{rem}

\begin{defn}
\label{def: Disorders}(Disorders) We define
\begin{equation}
\ccor{\psi_{z_{1}}^{\diamond}...\psi_{z_{k}}^{\diamond}\psi_{u}^{\sharp}\sigma_{u}\sigma_{v_{1}}...\sigma_{v_{n}}}:=\lim_{w\to u}\lamb(w-u)^{\frac{1}{2}}\ccor{\psi_{z_{1}}^{\diamond}...\psi_{z_{k}}^{\diamond}\psi_{w}\sigma_{u}\sigma_{v_{1}}...\sigma_{v_{n}}}\label{eq: def_mu_1}
\end{equation}
More generally, 
\begin{equation}
\ccor{\psi_{z_{1}}^{\diamond}...\psi_{z_{k}}^{\diamond}\psi_{u_{1}}^{\sharp}\ldots\psi_{u_{m}}^{\sharp}\sigma_{\cvr}}:=\lim_{w\to u_{m}}\lambda(w-u_{m})^{\frac{1}{2}}\ccor{\psi_{z_{1}}^{\diamond}...\psi_{z_{k}}^{\diamond}\psi_{u_{1}}^{\sharp}\ldots\psi_{u_{m}}^{\sharp}\psi_{w}\sigma_{\cvr}},\label{eq: def_mu_many}
\end{equation}
where $\cvr=\cvr[u_{1},\dots,u_{m},v_{1},\dots,v_{n}].$ We write
\[
\mu_{u}:=\psi_{u}^{\sharp}\sigma_{u},
\]
 typically using this notation when we no longer view $u$ as fixed
(see the discussion at the beginning of the section).
\end{defn}

Let us explain why the limits in Definition \ref{def: Disorders}
exist. By (\ref{eq: dte_psi_psieta}\textendash \ref{eq: dte_psi_psistar}),
(\ref{eq: def-sharp}), and the discussion in the beginning of the
proof of Lemma \ref{lem: f_star_holom}, we see that in the case $m=1$,
$k=1$, we have 
\begin{equation}
\frac{\ccor{\psi_{z_{i}}^{\diamond}\psi_{u}^{\sharp}\sigma_{u}\sigma_{v_{1}}...\sigma_{v_{n}}}}{\ccor{\sigma_{u}\sigma_{v_{1}}...\sigma_{v_{n}}}}=\lim_{w\to u}\lamb(w-u)^{\frac{1}{2}}\frac{\ccor{\psi_{z_{i}}^{\diamond}\psi_{w}\sigma_{u}...\sigma_{v_{n}}}}{\ccor{\sigma_{u}...\sigma_{v_{n}}}}=\begin{cases}
f_{\cvr}^{[\sharp,\eta_{i}]}(u,z_{i}). & \any=\eta_{i}\\
\bar{f_{\cvr}^{[\sharp]}(u,z_{i})}, & \any=\star\\
f_{\cvr}^{[\sharp]}(u,z_{i}), & \text{else}.
\end{cases}\label{eq: psi_to_mu_limits}
\end{equation}
Then, in the case $m=2$, $k=0$, we have, by definition, 
\begin{equation}
\frac{\ccor{\psi_{u_{2}}^{\sharp}\psi_{u_{1}}^{\sharp}\sigma_{u_{1}}\sigma_{u_{2}}\sigma_{v_{1}}...\sigma_{v_{n}}}}{\ccor{\sigma_{u_{1}}\sigma_{u_{2}}\sigma_{v_{1}}...\sigma_{v_{n}}}}=\lim_{w\to u}\lamb(w-u_{2})^{\frac{1}{2}}\frac{\ccor{\psi_{w}\psi_{u_{1}}^{\sharp}\sigma_{u_{1}}\sigma_{u_{2}}...\sigma_{v_{n}}}}{\ccor{\sigma_{u_{1}}\sigma_{u_{2}}\sigma_{v_{1}}...\sigma_{v_{n}}}}=f_{\cvr}^{[\sharp,\sharp]}(u_{1},u_{2}).\label{eq: mu_mu_sigma}
\end{equation}
In the general case, we can use Pfaffian representation (\ref{eq: multipsi_pfaffian})
and pass to the limits in each term of the Pfaffian; this gives the
existence of the limits in (\ref{eq: def_mu_many}) and a Pfaffian
representation 
\begin{equation}
\frac{\ccor{\psi_{z_{1}}^{\diamond}...\psi_{z_{k}}^{\diamond}\psi_{u_{1}}^{\sharp}\ldots\psi_{u_{m}}^{\sharp}\sigma_{\cvr}}}{\ccor{\sigma_{\cvr}}}=\Pf\left[\frac{\ccor{\psi_{w_{i}}^{\diamond}\psi_{w_{j}}^{\diamond}\sigma_{\cvr}}}{\ccor{\sigma_{\cvr}}}\right]{}_{w_{1,2}\in\{z_{1},\dots,z_{k},u_{1},\dots,u_{m}\}}.\label{eq: psi_mu_sigma_pfaff}
\end{equation}

\begin{rem}
Alternatively, we could write 
\begin{multline}
\ccor{\psi_{w_{1}}^{\sharp}\dots\psi_{w_{m}}^{\sharp}\sigma_{u_{1}}..\sigma_{u_{m}}\Op(\psi^{\any},\sigma)}=\\
\frac{\lamb^{m}}{(2\pi i)^{m}}\oint_{u_{m}}\dots\oint_{u_{1}}\frac{\ccor{\psi_{w_{1}}\dots\psi_{w_{m}}\sigma_{u_{1}}..\sigma_{u_{m}}\Op(\psi^{\any},\sigma)}}{(w_{1}-u_{1})^{\frac{1}{2}}\dots(w_{m}-u_{m})^{\frac{1}{2}}}dw_{1}\dots dw_{m},\label{eq: mu_resudues}
\end{multline}
where the integrals are over small contours surrounding $w_{1},\dots,w_{m}.$
Since the integrals can be exchanged, we see that, actually, the limits
in (\ref{eq: def_mu_many}) can be taken in any order.
\end{rem}

\begin{rem}
It is possible to replace the holomorphic fermion $\psi_{w}$ in Definition
\ref{def: Disorders} by anti-holomorphic one $\psi_{w}^{\star}$,
as follows: 
\begin{multline*}
\ccor{\psi_{z_{1}}^{\diamond}...\psi_{z_{k}}^{\diamond}\mu_{u_{1}}\ldots\mu_{u_{m}}\sigma_{v_{1}}\ldots\sigma_{v_{n}}}\\
=\lim_{w\to u_{m}}\lambb(\overline{w}-\overline{u}_{m})^{\frac{1}{2}}\ccor{\psi_{z_{1}}^{\diamond}...\psi_{z_{k}}^{\diamond}\psistar w\mu_{u_{1}}...\mu_{u_{m-1}}\sigma_{u_{m}}\sigma_{v_{1}}...\sigma_{v_{n}}}.
\end{multline*}
Indeed, if all $\any$ are $\eta$'s, then this follows from (\ref{eq: corr_conj}),
and the general case follows by linearity.
\end{rem}

We now extend the definition of fermionic correlation by allowing
the fermions to be placed on the boundary, with a necessary renormalization. 
\begin{defn}
\textcolor{black}{\label{def: corr_flat}Assume that $\Omega$ is
a nice domain. We extend the definition of $\ccor{\psi_{z}\Op}$ to
$z\in\fixed$ by continuity. If $b$ is an endpoint of a free arc,
we define 
\begin{equation}
\ccor{\psi_{b}^{\flat}\Op}=\lim_{z\to b}\sflat(z-b)^{\frac{1}{2}}\ccor{\psi_{z}\Op}.\label{eq: def_psi_flat}
\end{equation}
where $\sflat$ is as in definition \ref{eq: def-flat}. }
\end{defn}

\begin{rem}
The entire discussion above applies verbatim to $\flat$ instead of
$\sharp$. In particular, the Pfaffian formula \ref{eq: psi_mu_sigma_pfaff}
still holds when we allow some of $\any$ to be $\flat$; the expansion
may contain, e. g., terms like $f_{\cvr}^{[\sharp,\flat]}(u,b).$
The continuity of $\ccor{\psi_{z}\Op}$ up to $\fixed$ follows from
continuity of terms in Pfaffian expansion, which, in its turn, ultimately
follows form continuity of $f^{\any}(z_{1},\cdot)$, see (\ref{eq: bc_fixed}).
Also, Theorem (\ref{thm: ccov}) below, more specifically, the fact
that the conformal weight for $\psi^{\flat}$ is zero, allows one
to extend the definition (\ref{eq: def_psi_flat}) of $\psi^{\flat}$
to rough domains.
\end{rem}

\begin{defn}
\label{def: energy}We define 
\begin{equation}
\ccor{\psi_{z_{1}}^{\diamond}...\psi_{z_{k}}^{\diamond}\en_{e_{1}}..\en_{e_{s}}\mu_{u_{1}}...\mu_{u_{m}}\sigma_{v_{1}}...\sigma_{v_{n}}}:=\left(\frac{i}{2}\right)^{s}\ccor{\psi_{z_{1}}^{\diamond}...\psi_{z_{k}}^{\diamond}\psi_{e_{1}}\psi_{e_{1}}^{\star}\dots\psi_{e_{s}}\psi_{e_{s}}^{\star}\mu_{u_{1}}...\mu_{u_{m}}\sigma_{v_{1}}...\sigma_{v_{n}}};\label{eq: def_energy}
\end{equation}
\end{defn}

Note that the right-hand side is well defined due to Remark \ref{rem: psi_to_psibar}. 

In the next Proposition, we collect the properties of continuous correlation
functions. 
\begin{prop}
\label{prop: cont_corr_conjugation}One has 
\begin{equation}
\overline{\CorrO{\Obs{\psi^{\any},\en,\mu,\sigma}}}=\CorrO{\Obs{\bar{\psi^{\any}},\en,\mu,\sigma}},\label{eq: corr_conj}
\end{equation}
where we put $\overline{\psi^{\eta}}:=\psi^{\eta}$, $\overline{\psi}:=\psistar{\,}$,
$\overline{\psistar{\,}}:=\psi$. Furthermore, each correlation function
of the form $\CorrO{\psi_{z}\,\Obs{\psi^{\any},\en,\mu,\sigma}}$
is holomorphic in $z$ (apart from other marked points), and each
correlation function of the form $\CorrO{\psi_{z}\,\Obs{\psi^{\eta},\en,\mu,\sigma}}$
satisfies standard boundary conditions on $\pa\Omega$.
\end{prop}

\begin{rem}
This can be understood as a coupling of ``holomorphic $\psi$ and
anti-holomorphic $\psi^{\star}$ free fermions'' at the boundary
of $\Omega$ via $\psistar{\zeta}=\tau_{\zeta}\psi_{\zeta}$ on $\fixed$
and $\psistar{\zeta}=-\tau_{\zeta}\psi_{\zeta}$ on $\free$.
\end{rem}

\begin{proof}
We start with (\ref{eq: corr_conj}). Taking into account that 
\[
\bar{\en_{e}}=\bar{\frac{i}{2}\psi_{e}\psi_{e}^{\star}}=-\frac{i}{2}\psi_{e}^{\star}\psi_{e}=\frac{i}{2}\psi_{e}\psi_{e}^{\star}=\en_{e},
\]
it is enough to consider $\Op:=\Obs{\psi,\mu,\sigma}$. By (\ref{eq: psi_mu_sigma_pfaff}),
the claim can be reduced to the case to the case $k=2$, where it
is observed directly from (\ref{eq: dte_psi_psieta}\textendash \ref{eq: dte_psistar_psistar}),
(\ref{eq: psi_to_mu_limits}), (\ref{eq: mu_mu_sigma}) and the fact
that $f^{[\sharp,\sharp]}$ is real. 

The holomorphicity of $\CorrO{\psi_{z}\,\Op}$ follows directly from
the Pfaffian formula (\ref{eq: psi_mu_sigma_pfaff}), since the terms
in the Pfaffian that depend on $z$ have the form 
\[
\frac{\CorrO{\psi_{z}\psi_{z_{i}}^{\eta}\sigma_{\cvr}}}{\CorrO{\sigma_{\cvr}}}=f_{\cvr}(z_{i},z)\quad\text{or}\quad\frac{\CorrO{\psi_{z}\psi_{u_{i}}^{\sharp}\sigma_{\cvr}}}{\CorrO{\sigma_{\cvr}}}=f^{[\sharp]}(u_{i},z).
\]
For the final claim, write, by (\ref{eq: psi_mu_sigma_pfaff}), 
\begin{multline*}
\CorrO{\psi_{z}\,\Op}=\sum_{p=1}^{k}\CorrO{\psi_{z}\psi_{z_{p}}^{\eta_{p}}\sigma_{\cvr}}\alpha_{p}+\sum_{p=1}^{m}\CorrO{\psi_{z}\psi_{u_{p}}^{\sharp}\sigma_{\cvr}}\hat{\alpha}_{p}\\
+\sum_{p=1}^{s}\frac{i}{2}\left(\CorrO{\psi_{z}\psi_{e_{p}}\sigma_{\cvr}}\beta_{p}-\CorrO{\psi_{z}\psi_{e_{p}}^{\star}\sigma_{\cvr}}\beta_{p}^{\star}\right),
\end{multline*}
where $\alpha_{p},\hat{\alpha}_{p},\beta_{p},\beta_{p}^{\star}$ do
not depend on $z$, and it is easy to see from (\ref{eq: corr_conj})
that $\alpha_{i}\in\R$, $\hat{\alpha}_{i}\in\R$, and $\beta_{p}^{\star}=\bar{\beta_{p}}$.
Consequently, 
\[
\frac{i}{2}\left(\CorrO{\psi_{z}\psi_{e_{p}}\sigma_{\cvr}}\beta_{i}-\CorrO{\psi_{z}\psi_{e_{p}}^{\star}\sigma_{\cvr}}\beta_{i}^{\star}\right)=\CorrO{\psi_{z}\psi_{e_{p}}^{i\beta_{p}/2}\sigma_{\cvr}},
\]
and the result follows from the standard boundary conditions for $f_{\cvr}^{[\eta]}(z_{p},z)=\frac{\CorrO{\psi_{z}\psi_{z_{p}}^{\eta}\sigma_{\cvr}}}{\CorrO{\sigma_{\cvr}}}$
and the fact that standard boundary conditions are real linear.
\end{proof}
\begin{thm}
\label{thm: ccov}The continuous correlations, defined as above, satisfy
the following conformal covariance rule: if $\varphi:\Omega\to\hat{\Omega}$
is a conformal isomorphism, then \textup{
\begin{equation}
\ccor{\Op(\sigma,\mu,\epsilon,\psi,\psi^{\star},\psi^{\flat})}_{\Omega}=\CF\cdot\ccor{\Op(\hat{\sigma},\hat{\mu},\hat{\epsilon},\hat{\psi},\hat{\psi}^{\star},\psi^{\flat})}_{\hat{\Omega}},\label{eq: ccov}
\end{equation}
}where $\Op(\hat{\sigma},\hat{\mu},\hat{\epsilon},\hat{\psi},\hat{\psi}^{\star})$
is obtained from $\Op(\sigma,\mu,\epsilon,\psi,\psi^{\star})$ by
replacing every $v_{1},\dots,z_{k}$ with their images $\varphi(v_{1}),\dots,\varphi(z_{k})$,
respectively, and the conformal factor $\CF$ is given by 
\[
\CF=\prod_{i=1}^{m}|\varphi'(u_{i})|^{\frac{1}{8}}\prod_{\substack{1\leq i\leq n\\
v_{i}\in\Omega
}
}|\varphi'(v_{i})|^{\frac{1}{8}}\prod_{i=1}^{s}|\varphi'(e_{i})|\prod_{\substack{1\leq i\leq k\\
\psi_{z_{i}}^{\any}=\psi_{z_{i}}
}
}\varphi(z_{i}){}^{\frac{1}{2}}\prod_{\substack{1\leq i\leq k\\
\psi_{z_{i}}^{\any}=\psi_{z_{i}}^{\star}
}
}\bar{\varphi(z_{i})}{}^{\frac{1}{2}}.
\]
\end{thm}

\begin{rem}
Recall that the observable $\Op$ in (\ref{eq: ccov}) may contain
spins on the $\fixed$ parts of the boundary, which have conformal
weight zero and do not enter the covariance factor $\CF$. 
\end{rem}

\begin{rem}
\label{rem: psi_eta_conf_cov}It is possible to include also $\psi^{\eta}$
into the formula (\ref{eq: ccov}). The conformal covariance rule
for those is somewhat more contrived: $\hat{\psi}_{z_{i}}^{\eta_{i}}:=\psi_{\varphi(z_{i})}^{\hat{\eta}_{i}}$,
where $\hat{\eta}_{i}=\bar{\varphi'(z_{i})}^{\frac{1}{2}}\eta_{i}.$
\end{rem}

\begin{proof}[Proof of Theorem \ref{thm: ccov}]
 We start with spin correlations. Plugging the expansion (\ref{eq: def_A})
into both sides of (\ref{eq: ccov_sharp}) with $\any=\sharp$, and
equating the coefficients at $(z-v)^{\frac{1}{2}},$ one gets 
\[
\coefA_{\Omega,\cvr}(v)=\coefA_{\hat{\Omega},\hat{\cvr}}(\varphi(v))\cdot\varphi'(v)+\frac{1}{8}\frac{\varphi''(v)}{\varphi'(v)}.
\]
Integrating the real part of the form $\coefA_{\Omega,\cvr}(v)dv$
and taking the exponential, as in (\ref{eq: def_spincorr_up_to_constant}),
we infer that 
\[
\frac{\ccor{\sigma_{v_{1}}\dots\sigma_{v_{n}}}_{\Omega}}{\ccor{\sigma_{v'_{1}}\dots\sigma_{v'_{n}}}_{\Omega}}=\frac{\prod_{i=1}^{n}|\varphi'(v_{i})|^{\frac{1}{8}}\ccor{\sigma_{\varphi(v_{1})}\dots\sigma_{\varphi(v_{n})}}_{\hat{\Omega}}}{\prod_{i=1}^{n}|\varphi'(v_{i}')|^{\frac{1}{8}}\ccor{\sigma_{\varphi(v'_{1})}\dots\sigma_{\varphi(v'_{n})}}_{\hat{\Omega}}},
\]
for any two $n$-tuples of distinct points $(v_{1},\dots,v_{n})$
and $(v_{1}',\dots,v_{n}').$ Sending $v_{1}'\to v_{2}'$ and taking
into account the asymptotics (\ref{eq: spins_coherent}), we get,
by induction in $n$,
\[
\ccor{\sigma_{v_{1}}\dots\sigma_{v_{n}}}_{\Omega}=\prod_{i=1}^{n}|\varphi'(v_{i})|^{\frac{1}{8}}\ccor{\sigma_{\varphi(v_{1})}\dots\sigma_{\varphi(v_{n})}}_{\hat{\Omega}}.
\]

Now, for the general correlation, we substitute $\en_{e}=\frac{i}{2}\psi_{e}\psi_{e}^{\star}$
(see Definition \ref{def: energy}), and then use the Pfaffian expansion
(\ref{eq: psi_mu_sigma_pfaff}); it is sufficient to prove the result
for the correlations arising in this expansion. Each term in the Pfaffian
expansion takes one of the forms (\ref{eq: mu_mu_sigma}), (\ref{eq: psi_to_mu_limits}),
$f_{\Omega,\cvr}(w_{i},w_{j})$, $\fdag_{\Omega,\cvr}(w_{i},w_{j})$
or $\bar{f_{\Omega,\cvr}(w_{i},w_{j})}$ (with $w_{i,j}\in\{z_{1},\dots,z_{k},e_{1},\dots,e_{s}\}$);
recall (\ref{eq: dte_psi_psi}\textendash \ref{eq: dte_psistar_psistar}).
Therefore, the result follows from Lemma \ref{lem: ccov_f_any}.
\end{proof}

We now move to the definition of $\ccor{\cdot}_{\Omega,\bcond}$ for
other boundary conditions. We start with the auxiliary boundary conditions
$\tilde{\bcond}$, as in Lemma \ref{lem: corr_to_obs_for_Thm_3}.
These boundary conditions are specified, in addition to the open sets
$\free$ and $\fixed,$ by a collection of points $b_{1},\dots,b_{q}\in\pa\Omega$
such that for each $p$, either $b_{p}\in\fixed$ or $b_{p}\in\pa\Omega\setminus(\fixed\cup\free)$
(i. e., $b_{p}$ is an endpoint of a free arc), and such that each
connected component of $\pa\Omega$ carries an even number of $b_{p}$. 
\begin{lem}
\textcolor{black}{\label{lem: bc_tilde_nonvanish}Assume that $\Omega$
is a nice domain, $b_{1},\dots,b_{q}$ are as above, and numbered
in such a way that for some $k$, $b_{1},\dots,b_{k}\in\fixed$ and
$b_{k+1},\dots,b_{q}$ are endpoints of free arcs. Then, 
\[
\ccor{\psi_{b_{1}}\dots\psi_{b_{k}}\psi_{b_{k+1}}^{\flat}\dots\psi_{b_{q}}^{\flat}}\neq0.
\]
}
\end{lem}

\begin{proof}
\textcolor{black}{We follow the argument in \cite[Lemma 3.1]{IzyurovMconn},
where the claim was proven in the absence of free arcs. Assume without
loss of generality that $\Omega$ is circular. We start with the case
$k>0,$ and consider the function 
\[
f(z):=\tau_{b_{2}}^{\frac{1}{2}}\dots\tau_{b_{k}}^{\frac{1}{2}}\ccor{\psi_{z}\psi_{b_{2}}\dots\psi_{b_{k}}\psi_{b_{k+1}}^{\flat}\dots\psi_{b_{q}}^{\flat}},
\]
which, according to Remark \ref{rem: Schwarz-reflection}, can be
continued analytically in the punctured neighborhoods of $b_{2},\dots,b_{k}$
and, as a spinor, to the neighborhoods of $b_{k+1},\dots,b_{k}$.
Note that $f$ satisfies the standard boundary conditions on $\pa\Omega\setminus\{b_{2},\dots,b_{q}\}$.
We can expand $f(z)$ by the Pfaffian formula:
\begin{equation}
f(z)=\sum_{p=2}^{q}(-1)^{p}\ccor{\psi_{z}\psi_{z_{p}}^{\any}}\Pf[\psi_{\alpha}^{\any}\psi_{\beta}^{\any}]_{\substack{2\leq\alpha,\beta\leq q\\
\alpha,\beta\neq p
}
},\label{eq: pfaff_for_nvlemma}
\end{equation}
where $\psi_{\alpha}^{\any}$ denotes $\tau_{bz}^{\frac{1}{2}}\psi_{b_{\alpha}}$
for $\alpha\leq k$, and $\psi_{b_{\alpha}}^{\flat}$ for $\alpha>k$.
Therefore, for $2\leq p\leq k$, one has
\begin{equation}
f(z)=\frac{\tau_{p}^{\frac{1}{2}}c_{p}}{z-b_{p}}+o(1)\quad z\to b_{p},\label{eq: f_exp_nonvanish_lemma}
\end{equation}
where $c_{p}\in\R$; indeed, the $p$-th term in (\ref{eq: pfaff_for_nvlemma})
satisfies this expansion by (\ref{eq: f_fdag_expansion}), while the
other terms are continuous at $z=b_{p}$ and combine to a Pfaffian
with $p$-th row proportional to the first one. }

\textcolor{black}{Moreover, we claim that if $\nu\subset\pa\Omega$
is a free arc, then $f(z)$ satisfies (\ref{eq: bc_free_arc}) if
both (or none) of the endpoints of $\nu$ belong to $\{b_{k+1},\dots,b_{q}\},$
and $f$ satisfies (\ref{eq: bc_free_arc}) with the sign flipped
in the right-hand side if exactly one of the endpoints of $\nu$ belongs
to $\{b_{k+1},\dots,b_{q}\}$. }

Indeed, let $a,b$ be two endpoints of a free arc $\nu$, such that
$\sflat_{a}=i$ and $\sflat_{b}=1$, and let the square roots $(z-a)^{\frac{1}{2}}$
and $(z-b)^{-\frac{1}{2}}$ be chosen as in (\ref{eq: bc_free_arc}).
It was shown in (\ref{eq: f_sharp_asymp}) that $\lim_{z\to b}(z-b)^{\frac{1}{2}}\ccor{\psi_{z}\psi_{b}^{\flat}}=\sflat_{b}=1;$
in fact, applying the calculation in (\ref{eq: exchange_the_contours})
to $\ccor{\psi_{z}\psi_{b}^{\flat}\Op}$ instead of $\ccor{\psi_{z}\psi_{b}^{\flat}},$
this readily extends to 
\[
\lim_{z\to b}(z-b)^{\frac{1}{2}}\ccor{\psi_{z}\psi_{b}^{\flat}\Op}=\ccor{\Op}.
\]
In a similar spirit, we compute 
\begin{multline}
\ccor{\psi_{a}^{\flat}\psi_{b}^{\flat}\Op}=\frac{i\cdot1}{(2\pi i)^{2}}\oint_{a}\oint_{b}\frac{\ccor{\psi_{\zeta}\psi_{z}\Op}dzd\zeta}{(\zeta-a)^{\frac{1}{2}}(z-b)^{\frac{1}{2}}}=\frac{1}{(2\pi i)^{2}}\int_{\zeta\in C_{R}}\int_{z\in C_{r}}\frac{\ccor{\psi_{\zeta}\psi_{z}\Op}dzd\zeta}{(\zeta-a)^{\frac{1}{2}}(z-a)^{\frac{1}{2}}}\\
=\frac{1}{(2\pi i)^{2}}\int_{\zeta\in C_{r}}\int_{z\in C_{R}}\frac{\ccor{\psi_{\zeta}\psi_{z}\Op}}{(\zeta-a)^{\frac{1}{2}}(z-a)^{\frac{1}{2}}}+2\ccor{\Op}=\ccor{\Op},\label{eq: bemol_bemol}
\end{multline}
where $C_{R}$ and $C_{r}$ are circles of small radii $R>r$ around
$a$; in the second identity, we have used (\ref{eq: f_f_dagger_back})
and the standard boundary conditions (\ref{eq: bc_free_arc}). That
is to say, 
\[
i\lim_{z\to z}(z-a)^{\frac{1}{2}}\ccor{\psi_{z}\psi_{b}^{\flat}}=\ccor{\Op}=\lim_{z\to b}(z-b)^{\frac{1}{2}}\ccor{\psi_{z}\psi_{b}^{\flat}\Op},
\]
which proves the above claim in the case $b\in\{b_{k+1},\dots,b_{q}\},$
$a\notin\{b_{k+1},\dots,b_{q}\}$. The symmetric case is similar,
and the case $a,b\in\{b_{k+1},\dots,b_{q}\}$ follows readily from
(\ref{eq: bemol_bemol}).

\textcolor{black}{We define the function 
\[
h(z):=\im\int^{z}f^{2}
\]
and note that it satisfies the conditions (1), (2) and (3) of Proposition
\ref{prop: f_to_h}. Moreover, it has the expansion $h(z)=-\im[\tau_{p}c_{p}^{2}(z-b_{p})^{-1}]+O(1)$
as $z\to b_{p}$, and, as follows from the discussion above, the jumps
of $h$ at the endpoints of any free arc have the same magnitude and
opposite signs. Hence, for each connected component $\nu$ of $\pa\Omega$,
$h$ is constant on $\nu\cap\fixed\setminus\{b_{2},\dots,b_{k}\};$
in particular, it follows that $h$ is a single-valued harmonic function. }

Consider the connected component $\nu$ of $\pa\Omega$ with the smallest
value of $h$ on $\nu\cap\fixed$, which must be also a global minimum
of $h$. It is easy to see by induction that $f$ is not identically
zero; hence $h$ is not constant. Therefore, $(\tau_{z})^{\frac{1}{2}}f(z)=(2\pa_{i\tau}h(z))^{\frac{1}{2}}$
does not vanish on $\nu\cap\fixed$, in particular, the coefficients
$c_{p}$ in (\ref{eq: f_exp_nonvanish_lemma}) are non-zero for $z_{p}\in\nu\cap\fixed$.
Hence, (\ref{eq: f_exp_nonvanish_lemma}) together with the discussion
of the asymptotics of $f$ at the endpoints of the free arcs imply
that the number of sign changes of $(\tau_{z})^{\frac{1}{2}}f(z)$
along $\nu$ is equal to the number of the marked points $b_{2},\dots,b_{q}$
that belong to $\nu.$ Since $(\tau_{z})^{\frac{1}{2}}$ picks a $-1$
sign as $z$ goes around $\nu$ and $f$ does not, this number must
be odd. Since each connected components carries an even number of
points $b_{1},\dots,b_{q}$, we conclude that $z_{1}\in\nu\cap\fixed$,
and so $f(z_{1})\neq0$, as required. 

When $k=0$, we define $f(z):=\ccor{\psi_{z}\psi_{b_{2}}^{\flat}\dots\psi_{b_{q}}^{\flat}}$
an apply the same argument; the conclusion follows from the fact that
since $\pa_{i\tau}h(z)\leq0$ for $z\in\free$, the jumps of $h$
at the endpoints of any free sub-arc of $\nu$ must be non-zero. 
\end{proof}
\begin{defn}
\label{def: bcondtilde}(Correlations with auxiliary boundary conditions
$\tilde{\bcond}$) We define the correlation of $\Op=\Op(\en,\sigma)$
with boundary conditions $\tilde{\bcond}$ to be the limit 
\begin{equation}
\corr{\Op}{\Omega}{\tilde{\bcond}}:=\lim_{z_{1}\to b_{1},\dots,z_{q}\to b_{q}}\frac{\ccor{\psi_{z_{1}}\dots\psi_{z_{q}}\Op}_{\Omega}}{\ccor{\psi_{z_{1}}\dots\psi_{z_{q}}}_{\Omega}}.\label{eq: bcondtilde_def}
\end{equation}
\end{defn}

Note that the limit exists due to the following considerations. First,
we fix a conformal isomorphis $\varphi$ to a nice domain $\Lambda$;
multiplying the numerator and denomenator by $\prod\varphi'(z_{i})^{\frac{1}{2}}$
and using Theorem \ref{thm: ccov}, we see that the right-hand side
is equal to 
\[
\lim_{z_{1}\to\varphi(b_{1}),\dots,z_{q}\to\varphi(b_{q})}\frac{\ccor{\psi_{z_{1}}\dots\psi_{z_{q}}\Op}_{\Lambda}}{\ccor{\psi_{z_{1}}\dots\psi_{z_{q}}}_{\Lambda}}
\]
 (up to a conformal factor coming from $\Op$ that does not depend
on $z_{1},\dots,z_{q}$). Further, we can multiply the numerator and
denomenator by $c^{\flat}(z_{p}-\varphi(b_{p}))^{\frac{1}{2}}$ for
those $p$ for which $b_{p}$ an endpoint of a free arc. After that,
limits of the numerator and denomenator exist because of the existence
of the limit in Definition \ref{def: corr_flat} and the fact that
$\psi_{z}$ are continuous up to nice parts of $\fixed$ boundary.
Finally, we use Lemma \ref{lem: bc_tilde_nonvanish} to ensure that
the limit of the denominator is non-zero.

The definition of correlations with plus-minus-free boundary conditions
$\bcond$ is naturally motivated by the discussion after Lemma \ref{lem: corr_to_obs_for_Thm_3}.
Given a subdivision of the boundary into three disjoint open subsets
$\fixed,\plus,\minus$, such that $\pa\Omega\setminus\left(\fixed\cup\plus\cup\minus\right)$
is finite, let $b_{1},\dots,b_{2q}$ be the boundary points separating
``plus'' arcs from ``minus'' or ``free'' arcs, and let $\tilde{\bcond}$
denote the corresponding auxiliary boundary conditions. For each boundary
component that has a non-empty $\fixed$ part, choose arbitrary point
on that fixed part, and denote those points by $v_{n+1},\dots,v_{n+d}$.
\begin{lem}
\label{lem: bcond_well_defined}For any boundary conditions $\bcond$
as above, one has 
\[
\sum_{S}\alpha_{S}^{\bcond}\corr{\sigma_{S}}{\Omega}{\tilde{\bcond}}\neq0,
\]
where the sum is over $S\subset\{n+1,\dots,n+d\}$, $\sigma_{S}=\prod_{i\in S}\sigma_{v_{i}}$,
and $\alpha_{S}^{\bcond}$ are defined by (\ref{eq: FW_indicator})
as the Fourier\textendash Walsh coefficients of the expansion of the
corresponding indicator function.
\end{lem}

\begin{proof}
\textcolor{black}{To prove this Lemma, we use that the quantity of
interest is the limit of ratios of partition functions:
\[
\sum_{S}\alpha_{S}^{\bcond}\corr{\sigma_{S}}{\Omega}{\tilde{\bcond}}=\lim_{\delta\to0}\frac{Z(\Od,\bcond^{\delta})}{Z(\Od,\tilde{\bcond}^{\delta})}.
\]
To bound this from below uniformly in $\delta$, we use the RSW inequality.
Consider the critical $q=2$ FK random cluster model, with boundary
conditions that are free along $\free$ and wired along each $\plus$
or $\minus$ boundary arc, with different arcs not wired with each
other. By Edwards\textendash Sokal coupling, after coloring the clusters
of the model uniformly at random into two colors, labeled as $\pm1$,
one obtains an instance of the Ising model with spins conditioned
to be the same along each $\plus$ or $\minus$ boundary arcs, and
with free boundary conditions along $\free$. Denote by $Z(\Od,\hat{\bcond}^{\delta})$
the partition function and by $\P_{\hat{\bcond}^{\delta}}(\cdot)$
the probability in this model. Then, we have 
\[
\frac{Z(\Od,\bcond^{\delta})}{Z(\Od,\tilde{\bcond}^{\delta})}=\frac{Z(\Od,\bcond^{\delta})}{Z(\Od,\hat{\bcond}^{\delta})}\cdot\frac{Z(\Omega^{\delta},\hat{\bcond}^{\delta})}{Z(\Omega^{\delta},\tilde{\bcond}^{\delta})}=\frac{\P_{\hat{\bcond}^{\delta}}(\bcond^{\delta})}{\P_{\hat{\bcond}^{\delta}}(\tilde{\bcond}^{\delta})}\geq\frac{\P_{\hat{\bcond}^{\delta}}(\bcond^{\delta})}{\P_{FK}(\tilde{\bcond}^{\delta}\text{ is possible})},
\]
where by $\bcond^{\delta}$ we denote the event that all the arcs
receive the spins prescribed by $\bcond^{\delta},$ and by $\tilde{\bcond}^{\delta}$
we denote the event that all the arcs receive spins compatible with
$\tilde{\bcond^{\delta}}.$ The FK event $\{\tilde{\bcond}^{\delta}\text{ is possible\}}$
is the event that no two $\fixed$ arcs separated by an odd number
of marked points $b_{1},\dots,b_{q}$ are in the same cluster. We
note that this event is }\textcolor{black}{\emph{decreasing. }}\textcolor{black}{Let
$E$ denote the event that for each pair of boundary components, there
is a circuit of open edges in the }\textcolor{black}{\emph{dual}}\textcolor{black}{{}
configuration that separates those arcs. This is also a }\textcolor{black}{\emph{decreasing}}\textcolor{black}{{}
event; moreover, by RSW inequality, we have the lower bound $\P_{FK}(E)\geq c(\Omega)>0.$
Applying FKG inequality, we get 
\[
\frac{Z(\Od,\bcond^{\delta})}{Z(\Od,\tilde{\bcond}^{\delta})}\geq\frac{\P_{\hat{\bcond}^{\delta}}(\bcond^{\delta})\P_{FK}(E)}{\P_{FK}(\tilde{\bcond}^{\delta}\text{ is possible})\P_{FK}(E)}\geq c(\Omega)\frac{\P_{\hat{\bcond}^{\delta}}(\bcond^{\delta})}{\P_{FK}(E\cap\{\tilde{\bcond}^{\delta}\text{ is possible}\})}\geq c(\Omega)2^{-N},
\]
where $N$ is the total number of $\plus$ and $\minus$ boundary
arcs. Indeed, on the event $E$, there are no clusters connecting
arcs on different boundary components, therefore, if $\tilde{\bcond}^{\delta}$
is possible, then every coloring of $\plus$ and $\minus$ arcs compatible
with $\tilde{\bcond}^{\delta}$ is realized with probability $2^{-\hat{N}}\geq2^{-N}$,
where $\hat{N}$ is the number of clusters adjacent to $\plus$ and
$\minus$ arcs.}
\end{proof}
This lemma allows one to define the continuous correlations as follows.
\begin{defn}
\label{def: bcond_corr_continuous}We define the correlation of $\Op(\en,\sigma)=\en_{e_{1}}\dots\en_{e_{s}}\sigma_{v_{1}}\dots\sigma_{v_{n}}$
with boundary conditions $\bcond$ by the formula 
\[
\corr{\Op(\sigma,\en)}{\Omega}{\bcond}=\frac{\sum_{S}\alpha_{S}^{\bcond}\corr{\Op(\sigma,\en)\sigma_{S}}{\Omega}{\tilde{\bcond}}}{\sum_{S}\alpha_{S}^{\bcond}\corr{\sigma_{S}}{\Omega}{\tilde{\bcond}}},
\]
 where the sums are over $S\subset\{n+1,\dots,n+d\}$, $\sigma_{S}=\prod_{i\in S}\sigma_{v_{i}}$,
and $\alpha_{S}^{\bcond}$ are defined by (\ref{eq: FW_indicator})
as the Fourier\textendash Walsh coefficients of the expansion of the
corresponding indicator function.
\end{defn}

\begin{cor}
The conformal covariance rule (\ref{eq: ccov}) holds in the same
form for $\corr{\Op(\en,\sigma)}{\Omega}{\tilde{\bcond}}$ and $\corr{\Op(\en,\sigma)}{\Omega}{\bcond}$.
\end{cor}

\begin{proof}
Straightforward from Theorem \ref{thm: ccov} and the definitions.
\end{proof}

\subsection{Proofs of the main theorems}

Throughout this section, we will use a shorthand notation 
\[
A\appe B
\]
for the two-sided estimate $(1-\eps)A\leq B\leq(1+\eps)A,$ where
$A,B>0$. We start by collecting the properties of one- and two-point
spin correlations that are used in the proof of Theorem \ref{thm: spin_convergence}.
These results are already contained in \cite{ChelkakHonglerIzyurov}.
To address the minor issue of the regularity of the boundary near
the endpoints of $\gamma$ in Lemma \ref{lem: Dobrushin}, and for
the sake of self-containedness, we include their proofs.
\begin{lem}
\label{lem: spin_2p} Let $D\subset\C$ be a disc, and let $D^{\delta}$
be a sequence of discrete domains approximating $D$. Then, for every
$\eps>0$, one can find a small number $r>0$ such that if $v_{1},v_{2}$
are at distance at most $r$ from the center of $D$ but at least
$r/2$ from each other, and $\re v_{1}=\re v_{2}$, then 
\begin{equation}
\delta^{-\frac{1}{4}}\E_{D^{\delta},\text{free}}(\sigma_{v_{1}}\sigma_{v_{2}})\appe\Csigma^{2}|v_{1}-v_{2}|^{-\frac{1}{4}};\label{eq: two_point_free}
\end{equation}
\begin{equation}
\delta^{-\frac{1}{4}}\E_{D^{\delta},+}(\sigma_{v_{1}}\sigma_{v_{2}})\appe\Csigma^{2}|v_{1}-v_{2}|^{-\frac{1}{4}},\label{eq: two_point_fixed}
\end{equation}
provided that $\delta$ is small enough. 
\end{lem}

\begin{proof}
Note that by a celebrated result of T. T. Wu and by our choice of
$\Csigma$, we do have 
\begin{equation}
\delta^{-\frac{1}{4}}\E_{\C^{\delta}}(\sigma_{v_{1}}\sigma_{v_{2}})=\Csigma^{2}|v_{1}-v_{2}|^{-\frac{1}{4}}+o(1)\label{eq: full_plane}
\end{equation}
 for the ``diagonal'' (i. e., for $\re v_{1}=\re v_{2}$) spin-spin
correlation in the full plane \cite{mccoy2014two,chelkak_Hongler_Mahfouf},
provided that $|v_{1}-v_{2}|\delta^{-1}\to\infty$. Also, by FKG inequality,
we have 
\begin{equation}
\E_{D^{\delta},\text{free}}(\sigma_{v_{1}}\sigma_{v_{2}})\leq\E_{\C^{\delta}}(\sigma_{v_{1}}\sigma_{v_{2}})\leq\E_{D^{\delta},+}(\sigma_{v_{1}}\sigma_{v_{2}}).\label{eq: FKG}
\end{equation}
Consider the critical $q=2$ FK random cluster model in $D^{\delta}$.
Then, 
\[
\E_{D^{\delta},\text{free}}(\sigma_{v_{1}}\sigma_{v_{2}})=\P_{D^{\delta},\text{free}}^{\FK}(v_{1}\leftrightsquigarrow v_{2})\text{ and }\E_{D^{\delta},+}(\sigma_{v_{1}}\sigma_{v_{2}})=\P_{D^{\delta},\text{wired}}^{\FK}(v_{1}\leftrightsquigarrow v_{2}).
\]
Let $R$ denote the radius of $D$. Given $0<r<\hat{r}\leq R$, let
$A_{r,\hat{r}}$ by the annulus of radii $\hat{r},r$ concentric with
$D$, and let $E_{r,\hat{r}}^{\text{open}}$ (respectively, $E_{r,\hat{r}}^{\text{closed}}$)
be the event that there is an open (respecvively, dual-open) circuit
$\gamma$ in $A_{r,\hat{r}}$ separating its boundary components.
The event $\{v_{1}\leftrightsquigarrow v_{2}\text{ and }E_{r,\hat{r}}^{\text{open}}\}$
is determined by the state of edges at distance $\leq\hat{r}$ from
the center of $D$. Therefore, conditioning on these states and using
FKG, we obtain

\begin{align*}
\P_{D^{\delta},\text{wired}}^{\FK}(v_{1}\leftrightsquigarrow v_{2}\text{ and }E_{\hat{r},R}^{\text{closed}}) & \geq\P_{D^{\delta},\text{wired}}^{\FK}(v_{1}\leftrightsquigarrow v_{2}\text{ and }E_{\hat{r},R}^{\text{closed}}\text{ and }E_{r,\hat{r}}^{\text{open}})\\
 & \geq\P_{A_{\hat{r},R},\text{wired}}^{\FK}(E_{\hat{r},R}^{\text{closed}})\P_{D^{\delta},\text{wired}}^{\FK}(v_{1}\leftrightsquigarrow v_{2}\text{ and }E_{r,\hat{r}}^{\text{open}})\\
 & \geq\P_{A_{\hat{r},R},\text{wired}}^{\FK}(E_{\hat{r},R}^{\text{closed}})\P_{D^{\delta},\text{wired}}^{\FK}(E_{r,\hat{r}}^{\text{open}})\P_{D^{\delta},\text{wired}}^{\FK}(v_{1}\leftrightsquigarrow v_{2}).
\end{align*}
By RSW inequality for the $q=2$ FK model \cite{RSW_Ising}, for any
$\eps>0$, we can find an $\hat{r}>r>0$ such that $\P_{D^{\delta},\text{wired}}^{\FK}(E_{r,\hat{r}}^{\text{open}})\P_{A_{\hat{r},R},\text{free}}^{\FK}(E_{\hat{r},R}^{\text{closed}})\geq(1-\eps).$
On the other hand, conditionally on the outermost circuit $\gamma$
realizing $E_{\hat{r},R}^{\text{closed}}$, the model in the interior
$D_{\gamma}^{\delta}$ of $\gamma$ is the critical $q=2$ FK model
with free boundary conditions. Therefore, by monotonicity, we have
\[
\P_{D^{\delta},\text{wired}}^{\FK}(v_{1}\leftrightsquigarrow v_{2}\text{ and }E_{\hat{r},R}^{\text{closed}})=\E_{D^{\delta},\text{wired }}^{\FK}\left[\ind_{E_{\hat{r},R}^{\text{closed}}}\P_{D_{\gamma}^{\delta},\text{free}}^{\FK}(v_{1}\leftrightsquigarrow v_{2})\right]\leq\P_{D^{\delta},\text{free}}^{\FK}(v_{1}\leftrightsquigarrow v_{2}).
\]
We conclude that $(1-\eps)\E_{D^{\delta},+}(\sigma_{v_{1}}\sigma_{v_{2}})\leq\E_{D^{\delta},\text{free}}(\sigma_{v_{1}}\sigma_{v_{2}})$
provided that $v_{1},v_{2}\in D_{r}^{\delta}$ and $r$ is small enough.
Combining this with (\ref{eq: full_plane}) and (\ref{eq: FKG}),
we get the result.
\end{proof}
\begin{lem}
\label{lem: Dobrushin}If discrete domains $(\Od,\bcond^{\delta})$
approximate a simply-connected domain $(\Omega,\bcond)$ with Dobrushin
boundary conditions, that is, $\plus$ on a boundary arc $\gamma^{\delta}$approximating
an arc $\gamma\subset\pa\Omega$ and $\minus$ on $\pa\Od\setminus\gamma^{\delta}$,
then 
\begin{equation}
\delta^{-\frac{1}{8}}\E_{\Od,\bcond^{\delta}}(\sigma_{v})=-C_{\sigma}\cdot\cos(\pi\text{hm}_{D_{i}}(v,\gamma_{i}))\cdot\crad_{D_{1}}^{-\frac{1}{8}}(v)+o(1),\label{eq: one_spin_Dobrushin}
\end{equation}
uniformly in $v$ in the bulk.
\end{lem}

\begin{proof}
We first briefly outline the proof in the case $\gamma^{\delta}=\pa\Od$,
i. e., $\plus$ boundary conditions. From the case $d=0$ of Theorem
\ref{thm: spin_convergence} below, we know that
\begin{align*}
\delta^{-\frac{1}{4}}\E_{\Od,+} & (\sigma_{v_{1}}\sigma_{v_{2}})=\Csigma^{2}\ccor{\sigma_{v_{1}}\sigma_{v_{2}}}+o(1),\\
\delta^{-\frac{1}{4}}\E_{\Od,\text{free}} & (\sigma_{v_{1}}\sigma_{v_{2}})=\Csigma^{2}\ccor{\sigma_{v_{1}}\sigma_{v_{2}}}_{\text{free}}+o(1).
\end{align*}
(There is no circular reasoning, for the present Lemma is only used
in the $d>0$ case.). We then write, as $\delta\to0$, 
\begin{align*}
\left(\delta^{-\frac{1}{8}}\E_{\Od,+}(\sigma_{v})\right)^{2} & =\frac{\E_{\Od,+}(\sigma_{v})}{\E_{\Od,+}(\sigma_{\hat{v}})}\cdot\delta^{-\frac{1}{4}}\E_{\Od,+}(\sigma_{v}\sigma_{\hat{v}})\cdot\frac{\E_{\Od,+}(\sigma_{v})\E_{\Od,+}(\sigma_{\hat{v}})}{\E_{\Od,+}(\sigma_{v}\sigma_{\hat{v}})}\\
 & =\left(\Csigma^{2}\ccor{\sigma_{v}\sigma_{\hat{v}}}\frac{\ccor{\sigma_{v}}_{+}}{\ccor{\sigma_{\hat{v}}}_{+}}+o(1)\right)\cdot\frac{\E_{\Od,+}(\sigma_{v})\E_{\Od,+}(\sigma_{\hat{v}})}{\E_{\Od,+}(\sigma_{v}\sigma_{\hat{v}})}
\end{align*}
uniformly over $v$ and the auxiliary point $\hat{v}$ in the bulk.
Note that, as $\hat{v}\to\pa\Omega,$ we have $\ccor{\sigma_{v}\sigma_{\hat{v}}}\sim\ccor{\sigma_{v}}_{+}\ccor{\sigma_{\hat{v}}}_{+}.$
The FKG and the GHS inequalities give 
\[
1\geq\frac{\E_{\Od,+}(\sigma_{v})\E_{\Od,+}(\sigma_{\hat{v}})}{\E_{\Od,+}(\sigma_{v}\sigma_{\hat{v}})}\geq1+\frac{\E_{\Od,+}(\sigma_{v})\E_{\Od,+}(\sigma_{\hat{v}})-\E_{\Od,+}(\sigma_{v}\sigma_{\hat{v}})}{\E_{\Od,+}(\sigma_{v}\sigma_{\hat{v}})}\geq1-\frac{\E_{\Od,\text{free}}(\sigma_{v}\sigma_{\hat{v}})}{\E_{\Od,+}(\sigma_{v}\sigma_{\hat{v}})}.
\]
The scaling limit of the latter expression tends to $1$ as $\hat{v}\to\pa\Omega$.
Summarizing, for every $\eps>0$, we can choose $\hat{v}\in\Omega$
close to the boundary so that the above estimates give 
\[
\left(\delta^{-\frac{1}{8}}\E_{\Od,+}(\sigma_{v})\right)^{2}\approx_{\eps}\Csigma^{2}\ccor{\sigma_{v}}_{+}^{2},
\]
provided that $\delta$ is small enough. Together with the positivity
of magnetization, this proves the result for $\plus$ boundary conditions.

When neither $\gamma$ nor $\pa\Omega\setminus\text{\ensuremath{\gamma}}$
are empty, we invoke the representation (\ref{eq: FW_correlation})
and Lemma \ref{lem: corr_to_obs_for_THM_2}. In this case, we have
$n=1$ and $d=1$, we rename $v_{1},v_{2}:=v,\hat{v}$, where the
spin at $\hat{v}$ is placed on $\gamma^{\delta}$. We simply have
$\ind(\sigma_{\hat{v}}=1)=(\sigma_{\hat{v}}+1)/2,$ so that 
\[
\E_{\Od,\bcd}(\sigma_{v})=\E_{\Od,\tilde{\bcond}^{\delta}}(\sigma_{v}\sigma_{\hat{v}})=\E_{\Od}(\sigma_{v}\sigma_{\hat{v}})\frac{F_{[v,\hat{v}]}(z_{1},z_{2})}{F(z_{1},z_{2})}=\E_{\Od,+}(v)\cdot\frac{F_{[v,\hat{v}]}(z_{1},z_{2})}{F(z_{1},z_{2})},
\]
where $z_{1,2}$ are corners outside of $\Od$ adjacent to two different
endpoints of $\gamma^{\delta}$. We already know the scaling limit
of the first factor, and it remains to apply Lemma \ref{lem: Clements_clever_lemma}
to the numerator and denominator: if $\varphi$ is a conformal map
from $\Omega$ to $\H$, then 
\[
\frac{F_{[v,\hat{v}]}(z_{1},z_{2})}{F(z_{1},z_{2})}=\frac{\normLoc{z_{1}}\normLoc{z_{2}}F_{[v,\hat{v}]}(z_{1},z_{2})}{\normLoc{z_{1}}\normLoc{z_{2}}F(z_{1},z_{2})}=\frac{f_{\H,\cvr(v)}(\varphi(z_{1}),\varphi(z_{2}))}{f_{\H}(\varphi(z_{1}),\varphi(z_{2}))}+o(1),
\]
where $f_{\H}(z_{1},z_{2})=2/(z_{2}-z_{1})\neq0$, and $f_{\H,\cvr(v)}(z_{1},z_{2})$
are given in Example \ref{exa: f_half-plane_explicit}. We leave it
to the reader to check that the resulting expression can be succinctly
represented as (\ref{eq: one_spin_Dobrushin}).
\end{proof}
\begin{thm}
\label{thm: spin_convergence}One has, as $\delta\to0$, 
\[
\delta^{-\frac{n}{8}}\E_{\Od}(\sigma_{v_{1}}\dots\sigma_{v_{n+d}})=C_{\sigma}^{n}\cdot\CorrO{\sigma_{v_{1}}\ldots\sigma_{v_{n+d}}}+o(1)
\]
uniformly over $v_{1},\dots,v_{n}$ in the bulk and away from each
other. As in the previous chapter, the spins $v_{n+1},\dots,v_{n+d}$
are on the boundary.
\end{thm}

\begin{proof}
The strategy of the proof is as follows. Observe that Corollary \ref{cor: spin_ratio}
already gives the result with $\delta^{-\frac{n}{8}}$ replaced by
an unknown normalizing factor 
\[
\rho(\delta)=\rho(\Od,n,v_{n+1},\dots,v_{n+d}).
\]
Thus, it suffices to prove that $\delta^{-\frac{n}{8}}\E_{\Od}(\sigma_{v_{1}}\dots\sigma_{v_{n+d}})\approx C_{\sigma}^{n}\cdot\CorrO{\sigma_{v_{1}}\ldots\sigma_{v_{n+d}}}$
for \emph{some} positions of the inner spins $v_{1},\dots,v_{n}$.
If $d=0$, we will take these spins pairwise close to each other,
and use FKG inequality to show that the correlation $\E_{\Od}(\sigma_{v_{1}}\dots\sigma_{v_{n+d}})$
splits into a product of two-point correlations in simple domains
that we control well. If $d\neq0$, we will start with the correlation
of $d+n$ spins in the bulk, and then send $d$ of them to the boundary
and show that an analog of (\ref{eq: prop_spin_bdry}) holds in the
discrete. 

We start with the case $d=0$, in which case $n$ is even. Let $D_{1},\dots,D_{n/2}$
be discs in $\Omega$ with non-intersecting closures. Assume that
$v_{1},v_{2}\in D_{1},\dots,v_{n-1},v_{n}\in D_{n/2}.$ By FKG inequality
and monotonicity in the boundary conditions, we have
\begin{multline}
\E_{\Od}(\sigma_{v_{1}}\dots\sigma_{v_{n}})\geq\E_{\Od}(\sigma_{v_{1}}\sigma_{v_{2}})\ldots\E_{\Od}(\sigma_{v_{n-1}}\sigma_{v_{n}})\\
\geq\E_{D_{1},\text{free}}(\sigma_{v_{1}}\sigma_{v_{2}})\ldots\E_{D_{n/2},\text{free}}(\sigma_{v_{n-1}}\sigma_{v_{n}}),\label{eq: FKG_lower_bound}
\end{multline}
were $"\text{free}"$ stands for free boundary conditions. On the
other hand, passing to the FK model and wiring the boundaries of $D_{1},\dots,D_{n/2}$,
we obtain that 
\begin{equation}
\E_{\Od}(\sigma_{v_{1}}\dots\sigma_{v_{n}})\leq\E_{D_{1},+}(\sigma_{v_{1}}\sigma_{v_{2}})\ldots\E_{D_{n/2},+}(\sigma_{v_{2n-1}}\sigma_{v_{2n}}),\label{eq: FKG_Upper_bound}
\end{equation}
where the expectations in the right-hand side are with fixed boundary
conditions. By Proposition \ref{prop: multiplicative-normalisation},
for each $\eps>0$, there exists a number $r>0$ such that if, for
every $i=1,\dots,n/2$, both $v_{2i-1}$ and $v_{2i}$ are at distance
at most $r$ from the center of $D_{i}$, then 
\begin{equation}
\ccor{\sigma_{v_{1}}\dots\sigma_{v_{n}}}_{\Omega}\appe\prod_{i=1}^{n/2}|v_{2i-1}-v_{2i}|^{\frac{1}{4}}.\label{eq: _ccor_versus_twopoint}
\end{equation}
Combining (\ref{eq: FKG_lower_bound}\textendash \ref{eq: FKG_Upper_bound}),
(\ref{eq: two_point_free}\textendash \ref{eq: two_point_fixed})
and (\ref{eq: _ccor_versus_twopoint}), we see that if, for each $i$,
$v_{2i-1},v_{2i}$ are at distance at most $r$ from the center of
$D_{i}$, but at least $r/2$ from each other, and $\re v_{2i-1}=\re v_{2i},$
and $\delta$ is small enough, then 
\[
\delta^{-\frac{n}{8}}\E_{\Od}(\sigma_{v_{1}}\dots\sigma_{v_{n}})\approx_{\theta(\eps)}\Csigma^{n}\cdot\CorrO{\sigma_{v_{1}}\ldots\sigma_{v_{n}}},
\]
where $\theta(\eps)\to0$ as $\eps\to0.$

Combining this result with (\ref{eq: conv_ratio}) and the definition
(\ref{eq: def_spincorr_up_to_constant}) of $\ccor{\sigma_{v_{1}}\dots\sigma_{v_{n}}}_{\Omega}$,
we infer that 
\[
\delta^{-\frac{n}{8}}\E_{\Od}(\sigma_{v_{1}}\dots\sigma_{v_{n}})\approx_{2\theta(\eps)}\Csigma^{n}\cdot\CorrO{\sigma_{v_{1}}\ldots\sigma_{v_{n}}},
\]
provided that $v_{1},\dots,v_{n}$ are at a definite distance from
each other and from the boundary, and $\delta$ is small enough. Since
$\eps$ is arbitrarily small, this concludes the proof in the case
$d=0$.

Now we turn to the case $d\neq0$. Note that the correlations do not
depend on the positions of $v_{n+1},\dots,v_{n+d}$ within the fixed
part of their boundary component. We fix those positions, and fix
small radii $R>r>0$ (to be specified later) such that $\pa\Omega\cap B_{R}(v_{i})$
is a subset of the fixed arc containing $v_{i},$ $i=n+1,\dots,n+d,$
and that this arc crosses the annulus $B_{R}(v_{i})\setminus B_{r}(v_{i}).$
Denote by $\P_{\FK}$ the probability measure on the $q=2$ critical
FK percolation in $\Od$ with all $\fixed$ arcs belonging to the
same boundary component wired together, and free boundary conditions
on $\free.$ We have, by Edwards\textendash Sokal coupling, $\E_{\Od}(\sigma_{v_{1}}\dots\sigma_{v_{n+d}})=\P_{FK}(A),$
where $A$ is the event that the points $v_{1},\dots,v_{n+d}$ are
connected into clusters so that no cluster contains an odd number
of them. 

We introduce auxiliary \emph{interior} points $\hat{v}_{n+1}\in B_{r}(v_{n+1})\cap\Od$,
..., $\hat{v}_{n+d}\in B_{r}(v_{d})\cap\Od$, and write $\E_{\Od}(\sigma_{v_{1}}\ldots\sigma_{v_{n}}\sigma_{\hat{v}_{n+1}}\dots\sigma_{\hat{v}_{n+d}})=\P_{FK}(\hat{A}),$
where $\hat{A}$ is similar to $A$ with $v_{n+1},\dots,v_{n+d}$
replaced with $\hat{v}_{n+1},\dots,\hat{v}_{n+d}.$ We also denote
by $\eventCut$ the event that in the FK percolation, there is an
open path in $B_{R}(v_{i})\setminus B_{r}(v_{i})$ which is a cross-cut
separating $v_{i}$ from $\pa B_{R}(v_{i});$ we denote by $\gamma_{i}^{\delta}$
the ``outermost'' such path. Note that by RSW inequality \cite{RSW_Ising},
for every $\eps>0$, we can choose $r,R$ so that $\P(\eventCut)>1-\eps$
for all $\delta$ small enough. 

We now observe that $A,\hat{A}$ and $\eventCut$ are increasing events,
therefore, by FKG, we have 
\begin{gather*}
(1-\eps)\P(A)\leq\P(A)\P(\eventCut)\leq\P(A\cap\eventCut)\leq\P(A);\\
(1-\eps)\P(\hat{A})\leq\P(\hat{A})\P(\eventCut)\leq\P(\hat{A}\cap\eventCut)\leq\P(\hat{A}).\\
(1-\eps)\P(A\cap\hat{A})\leq\P(A\cap\hat{A})\P(\eventCut)\leq\P(A\cap\hat{A}\cap\eventCut)\leq\P(A\cap\hat{A})
\end{gather*}
Clearly, the event $\hat{A}\cap\eventCut$ implies $A$, therefore,
$\P(A|\hat{A})\geq(1-\eps).$ In view of the above inequalities, we
also have 
\[
\P(\hat{A}|A)\appe\P(\hat{A}|A\cap\eventCut).
\]
However, given $A\cap\eventCut$, the event $\hat{A}$ is simply the
event that that for each $i$, $\hat{v}_{i}$ is connected to the
boundary of the simply-connected domain $D_{i}^{\delta}$ cut out
by the cross-cut $\gamma_{i}^{\delta}.$ Conditionally on $D_{i}^{\delta},$
the configurations inside $D_{i}^{\delta}$ are that of the FK percolation
with wired boundary conditions, independent of each other. Therefore,
\[
\P(\hat{A}|A\cap\eventCut)=\prod_{i=n+1}^{n+d}\E\left[\left.\E_{D_{i}^{\delta},+}[\sigma_{\hat{v}_{i}}]\right|A\cap\eventCut\right].
\]
We infer that for any $\eps>0,$ we have 
\[
\P(A)=\P(\hat{A})\cdot\frac{\P(A|\hat{A})}{\P(\hat{A}|A)}\approx_{\theta(\eps)}\P(\hat{A})\cdot\left(\prod_{i=n+1}^{n+d}\E_{FK}\left[\left.\E_{D_{i}^{\delta},+}[\sigma_{\hat{v}_{i}}]\right|A\cap\eventCut\right]\right)^{-1},
\]
 where $\theta(\eps)\to0$ as $\eps\to0.$ We therefore have 
\begin{multline*}
\delta^{-\frac{n}{8}}\cdot\E_{\Od}[\sigma_{v_{1}}\ldots\sigma_{v_{n}}\sigma_{v_{n+1}}\dots\sigma_{v_{n+d}}]\\
\approx_{\eps}\delta^{-\frac{n+d}{8}}\E[\sigma_{v_{1}}\ldots\sigma_{v_{n}}\sigma_{\hat{v}_{1}}\dots\sigma_{\hat{v}_{d}}]\left(\prod_{i=n+1}^{n+d}\E_{FK}\left[\delta^{-\frac{1}{8}}\left.\E_{D_{i}^{\delta},+}[\sigma_{\hat{v}_{i}}]\right|A\cap\eventCut\right]\right)^{-1}\\
\approx_{2\eps}\Csigma^{n+d}\CorrO{\sigma_{v_{1}}\ldots\sigma_{v_{n}}\sigma_{\hat{v}_{1}}\dots\sigma_{\hat{v}_{d}}}\left(\prod_{i=n+1}^{n+d}\E_{FK}\left[\delta^{-\frac{1}{8}}\left.\E_{D_{i}^{\delta},+}[\sigma_{\hat{v}_{i}}]\right|A\cap\eventCut\right]\right)^{-1}.
\end{multline*}
provided that $\delta$ is small enough. (Here ``small enough''
is uniform in the positions of $v_{1},\dots,v_{n}$ away from each
other and the boundary, but it may depend on $\eps$ and $\hat{v}_{n+1},\dots,\hat{v}_{n+d}$.)
To handle the last term, we note that by monotonicity, we can write
\[
\E_{\check{D}_{i}^{\delta},+}[\sigma_{\hat{v}_{i}}]\leq\E_{D_{i}^{\delta},+}[\sigma_{\hat{v}_{i}}]\leq\E_{\hat{D}_{i}^{\delta},+}[\sigma_{\hat{v}_{i}}],
\]
where $\hat{D}_{i}^{\delta}$ and $\check{D}_{i}^{\delta}$ are approximations
to fixed (non-random) neighborhoods $\check{D}_{i}\subset\hat{D}_{i}$
of $v_{i}$. Therefore, by Lemma \ref{lem: Dobrushin}, we have for
any $\eps>0,$ 
\[
(1-\eps)\Csigma\crad_{\check{D}_{i}}^{-\frac{1}{8}}(\hat{v}_{i})\leq\delta^{-\frac{1}{8}}\E_{D_{i}^{\delta},+}[\sigma_{\hat{v}_{i}}]\leq(1+\eps)\Csigma\crad_{\hat{D}_{i}}^{-\frac{1}{8}}(\hat{v}_{i}),
\]
provided that $\delta$ is small enough. Now, as $\hat{v}_{i}\to v_{i},$
we have $\crad_{\check{D}_{i}}(\hat{v}_{i})\sim\crad_{\hat{D}_{i}}(\hat{v}_{i})\sim\crad_{\Omega}(\hat{v}_{i}).$
Hence, for any $\eps>0$, we can choose $\hat{v}_{i}$ in such a way
that 
\[
\prod_{i=n+1}^{n+d}\E\left[\delta^{-\frac{1}{8}}\left.\E_{D_{i}^{\delta},+}[\sigma_{\hat{v}_{i}}]\right|A\cap\eventCut\right]\approx_{\eps}\Csigma^{d}\prod_{i=n+1}^{n+d}\crad_{\Omega}^{-\frac{1}{8}}(\hat{v}_{i})
\]
provided that $\delta$ is small enough. We conclude that for any
$\eps>0,$ we have 
\begin{multline*}
\delta^{-\frac{n}{8}}\cdot\E_{\Od}[\sigma_{v_{1}}\ldots\sigma_{v_{n}}\sigma_{v_{n+1}}\dots\sigma_{v_{n+d}}]\approx_{\eps}\Csigma^{n}\CorrO{\sigma_{v_{1}}\ldots\sigma_{v_{n}}\sigma_{\hat{v}_{1}}\dots\sigma_{\hat{v}_{d}}}\prod_{i=n+1}^{n+d}\crad_{\Omega}^{\frac{1}{8}}(\hat{v}_{i})
\end{multline*}
provided that $\hat{v}_{i}$ are close enough to $v_{i}$ and $\delta$
is small enough (depending on $\hat{v}_{i}$). Together with the asymptotics
(\ref{prop: multiplicative-normalisation}), this implies that for
every $\eps>0$, one has 
\[
\delta^{-\frac{n}{8}}\cdot\E_{\Od}[\sigma_{v_{1}}\dots\sigma_{v_{n+d}}]\appe\Csigma^{n}\CorrO{\sigma_{v_{1}}\dots\sigma_{v_{n+d}}}
\]
provided that $\delta$ is small enough. This completes the proof.
\end{proof}

\begin{proof}[Proof of Theorem \ref{thm: intro_2}]
 We can rewrite the left-hand side of (\ref{eq: thm_2_intro}) according
to Lemma \ref{lem: corr_to_obs_for_THM_2}:
\[
\delta^{-\frac{n+m}{8}-s-\frac{k}{2}}\E(\Op(\sigma,\mu,\psi,\en))=\delta^{-\frac{n+m}{8}}\E(\sigma_{v_{1}}\dots\sigma_{u_{m}})\cdot\delta^{-\frac{k}{2}-s}F_{[v_{1}\dots v_{n+m}]}(z_{1},\dots,z_{2N}).
\]
By Theorem \ref{thm: spin_convergence}, the first factors is within
$o(1)$ from $\Csigma^{n+m}\ccor{\sigma_{v_{1}}\dots\sigma_{v_{n}}\sigma_{u_{1}}\dots\sigma_{u_{m}}}$,
uniformly in the marked points in the bulk and away from each other.
To treat the second factor, invoke Proposition \ref{prop: pfaff_discrete}
to write 
\[
\delta^{-\frac{k}{2}-\frac{2s}{2}}F_{[v_{1}\dots v_{n+m}]}(z_{1},\dots,z_{2N})=\Pf[\delta^{-\Delta_{p}-\Delta_{q}}F_{[v_{1}\dots v_{n+m}]}(z_{p},z_{q})]=:\Pf[F_{pq}],
\]
where $\Delta_{p}=\frac{1}{2}$ for $1\leq p\leq k+2s$ , and $\Delta_{p}=0$
for $k+2s<p\leq2N$. By Theorems \ref{thm: Convergence_singular}
and \ref{thm: Convergence_bulk_near}, each term in the Pfaffian converges
to the corresponding continuous limit; concretely, we have, for $p<q$,
\[
F_{pq}=\begin{cases}
-C_{p}C_{q}\frac{i}{2}\fdag(z_{p},z_{q})+o(1), & (p,q)=(k+2l-1;k+2l),\quad1\leq l\leq s\\
\eta_{z_{p}}\eta_{z_{q}}C_{p}C_{q}f^{[\any_{p},\any_{q}]}(z_{p},z_{q})+o(1), & \text{else, }
\end{cases}
\]
where 
\[
\any_{p}=\begin{cases}
\eta_{p}, & p\leq k+2s,\\
\sharp, & p>k+2s,
\end{cases}\quad\text{and}\quad C_{p}=\begin{cases}
\Cpsi, & p\leq k+2s,\\
1, & p>k+2s.
\end{cases}
\]
Taking the constants out of the Pfaffian, and rewriting everything
in terms of correlations rather than observables (see (\ref{eq: dte_psi_psieta}\textendash \ref{eq: dte_psistar_psistar})
and (\ref{eq: psi_to_mu_limits})), we get that 
\[
\Pf[F_{pq}]=\Cpsi^{k+2s}\eta_{z_{1}}\dots\eta_{z_{k+2s}}\Pf[f_{pq}]+o(1),
\]
where, for $p<q$, one has, with the notation $M:=k+2s$ and $\sigma_{\cvr}:=\sigma_{v_{1}}\dots\sigma_{v_{n}}\sigma_{u_{1}}\dots\sigma_{u_{m}},$
\[
\ccor{\sigma_{\cvr}}\cdot f_{pq}=\begin{cases}
\frac{i}{2}\bar{\eta}_{z_{k+2l-1}}\bar{\eta}_{z_{k+2l}}\ccor{\psistar{e_{l}}\psi_{e_{l}}\sigma_{\cvr}}, & (p,q)=(k+2l-1;k+2l),\quad1\leq l\leq s;\\
\ccor{\psi_{\hat{z}_{p}}^{\any}\psi_{\hat{z}_{q}}^{\any}\sigma_{\cvr}}, & \text{else},
\end{cases}
\]
where $\hat{z}_{p}=z_{p}$ if $p\le k$, $\hat{z}_{p}=e_{l}$ if $p\in\{k+2l-1;k+2l\}$
for some $0<l\leq s$, and $\hat{z}_{p}=u_{l}$ for $p=k+2s+l$, $ $

Note that if $s=0$, then we are already done, since the matrix $f_{pq}$
is exactly the one featuring in the Pfaffian expansion (\ref{eq: psi_mu_sigma_pfaff}).
Assuming $s>0,$ note that $\eta_{z_{k+1}}=\lambb$ and $\eta_{z_{k+2}}=-\lamb$
and hence $\eta_{z_{k+1}}\eta_{z_{k+2}}=-1$. Hence, we can write
\[
\ccor{\Op_{1}\psi_{e_{1}}^{\eta_{z_{k+1}}}\Op_{2}}=\frac{\lambda}{2}\left(\ccor{\Op_{1}\psi_{e_{1}}\Op_{2}}+i\ccor{\Op_{1}\psistar{e_{1}}\Op_{2}}\right);
\]
 
\[
\ccor{\Op_{1}\psi_{e_{1}}^{\eta_{z_{k+2}}}\Op_{2}}=-\frac{\lambb}{2}\left(\ccor{\Op_{1}\psi_{e_{1}}\Op_{2}}-i\ccor{\Op_{1}\psistar{e_{1}}\Op_{2}}\right);
\]
and also 
\[
\bar{\eta}_{z_{k+1}}\bar{\eta}_{z_{k+2}}\frac{i}{2}\ccor{\psistar{e_{1}}\psi_{e_{1}}\sigma_{\cvr}}=-\frac{\lambb\lamb}{4}\left(0+i\ccor{\psistar{e_{1}}\psi_{e_{1}}\sigma_{\cvr}}-i\ccor{\psi_{e_{1}}\psistar{e_{1}}\sigma_{\cvr}}+0\right).
\]
Therefore, by linearity, we can decompose $\Pf f_{pq}$ into a sum
of four Pfaffians, which have the same form as $\hat{F}_{pq}$, but
with $\psi_{\hat{z}_{k+1}}^{\eta_{k+1}}$ (respectively, $\psi_{\hat{z}_{k+2}}^{\eta_{z_{k+2}}}$)
replaced everywhere with $\frac{\lamb}{2}\psi_{e_{1}}$ or $\frac{\lamb}{2}i\psistar{e_{1}}$,
(respectively, with $-\frac{\lambb}{2}\psi_{e_{1}}$ or $\frac{\lambb}{2}i\psistar{e_{1}}$).
Note that two of those Pfaffians have proportional $(k+1)$-th and
$(k+2)$-th rows and thus vanish, and two others, by anti-symmetry,
are equal to each other; their sum equals the Pfaffian of the matrix
$\hat{F}_{pq}$ with $\psi_{\hat{z}_{k+1}}^{\eta_{z_{k+1}}}$ (respectively,
$\psi_{\hat{z}_{k+2}}^{\eta_{z_{k+2}}}$) replaced everywhere by $\lamb2^{-\frac{1}{2}}i\psistar{e_{1}}$
(respectively, by $-\lambb2^{-\frac{1}{2}}\psi_{e_{1}}$). Applying
a similar transformation to other egdes $e_{2},\dots,e_{s}$, we recover
the Pfaffian (\ref{eq: psi_mu_sigma_pfaff}) with $\frac{i}{2}\psistar{e_{i}}\psi_{e_{i}}$
inserted for each $i=1,\dots,s$, in agreement with Definition \ref{def: energy}.
Taking into account that $\Ceps=\Cpsi^{2}$ and $\Csigma=\Cmu$, this
completes the proof.
\end{proof}

\begin{proof}[Proof of Theorem \ref{thm: Intro_3}]
We start by proving the convergence of correlations with boundary
conditions $\tilde{\bcond}$. To this end, we invoke the decomposition
(\ref{eq: Lem_Thm_3_intermediate}), so that we can write 
\begin{equation}
\delta^{-\frac{n}{8}-s}\E_{\tilde{\bcond}^{\delta}}(\Op(\sigma,\en))=\delta^{-\frac{n}{8}}\E(\sigma_{v_{1}}\dots\sigma_{v_{n}})\cdot\frac{\delta^{-s}F_{[v_{1}\dots v_{n}]}(z_{1},\dots,z_{N})}{F(z_{1},\dots,z_{q})}.\label{eq: thm_3_proof}
\end{equation}
By Theorem \ref{thm: spin_convergence}, the first term is equal to
$\Csigma^{n}\ccor{\sigma_{v_{1}}\dots\sigma_{v_{n}}}+o(1),$ uniformly
over $v_{i}$ in the bulk and away from each other. We re-number the
corners the corners $z_{1},\dots,z_{q}$ so that \textcolor{black}{that
for some $\hat{q}$, $b_{1},\dots,b_{\hat{q}}\in\fixed$ and $b_{\hat{q}+1},\dots,b_{q}$
are endpoints of free arcs (this will introduce the same sign in the
numerator and in the denominator in (\ref{eq: thm_3_proof})).} 

Fix a conformal isomorphism $\varphi$ from $\Omega$ to a nice domain
$\Lambda$; recall that this entails fixing the normalizing factors
$\normLoc{z_{i}}$ in Lemma \ref{lem: Clements_clever_lemma}. By
the Pfaffian formula (\ref{eq: Pfaff_discrete}) and the linearity
of the Pfaffian, we can write 
\[
\frac{\delta^{-s}F_{[v_{1}\dots v_{n}]}(z_{1},\dots,z_{N})}{F(z_{1},\dots,z_{q})}=\frac{\Pf[\hat{F}_{ij}]_{1\leq i,j\leq N}}{\Pf[F_{ij}]_{1\leq i,j\leq q}},
\]
where 
\begin{align*}
\hat{F}_{ij}= & \normPf i\normPf jF_{[v_{1}\dots v_{n}]}(z_{i},z_{j}),\\
F_{ij}= & \normPf i\normPf jF(z_{i},z_{j}),
\end{align*}
and the normalizing factors $\normPf i$ are given by 
\[
\normPf i=\begin{cases}
\normLoc{z_{i}}, & i\leq\hat{q};\\
1, & \hat{q}<i\leq q;\\
\delta^{-\frac{1}{2}} & q<i\leq N,
\end{cases}
\]
Now we can apply Theorems \ref{eq: thm_Convergence _singular}, \ref{thm: Convergence_bulk_near}
and Lemma \ref{lem: Clements_clever_lemma} to each term in the numerator
and in the denominator. We get 
\begin{equation}
\frac{\delta^{-s}F_{[v_{1}\dots v_{n}]}(z_{1},\dots,z_{N})}{F(z_{1},\dots,z_{q})}=\frac{\lim_{\hat{z}_{1}\to z_{1},\dots\hat{z}_{\hat{q}}\to z_{q}}\Pf[\hat{f}_{ij}]_{1\leq i,j\leq N}+o(1)}{\lim_{\hat{z}_{1}\to z_{1},\dots\hat{z}_{\hat{q}}\to z_{q}}\Pf[f_{ij}]_{1\leq i,j\leq q}+o(1)},\label{eq: thm_3_proof_apply_conv}
\end{equation}
where $f_{ij}^{\Lambda}=\eta_{z_{i}}\eta_{z_{j}}\ccor{\Op_{i}\Op_{j}}_{\Omega}$,
and 
\[
\hat{f}_{ij}=\begin{cases}
\i f_{\Omega,\cvr}^{\star}(z_{i},z_{i})=|\varphi'(z_{i})|\frac{\i}{2}\frac{\ccor{\psistar{z_{i}}\psi_{z_{i}}\sigma_{\cvr}}_{\Omega}}{\ccor{\sigma_{\cvr}}_{\Omega}}, & (i;j)=(2k-1;2k)\text{ for }2k-1>q;\\
\eta_{z_{i}}\eta_{z_{j}}\ccor{\Op_{i}\Op_{j}\sigma_{\cvr}}_{\Omega}/\ccor{\sigma_{\cvr}}_{\Omega}, & \text{ else,}
\end{cases}
\]
where $\Op_{i}=\varphi'(\hat{z}_{i})^{-\frac{1}{2}}\psi(\hat{z}_{i})$
for $i\leq\hat{q}$, $\Op_{i}=\psi^{\flat}(z_{i})$ for $\hat{q}+1\leq i\leq q$
and $\Op_{i}=\psi_{z_{i}}^{\eta_{z_{i}}}$, $i>q.$ By Pfaffian expansion,
Lemma \ref{lem: bc_tilde_nonvanish}, and conformal covariance, we
have 
\begin{align*}
\Pf[f_{ij}]_{1\leq i,j\leq q}= & \ccor{\psi_{\varphi(z_{1})}\dots\psi_{\varphi(z_{\hat{q}})}\psi_{\varphi(z_{\hat{q}+1})}^{\flat}\dots\psi_{\varphi(z_{q})}^{\flat}}_{\Lambda}\neq0,
\end{align*}
 hence the right-hand side of (\ref{eq: thm_3_proof_apply_conv})
is in fact equal to 
\[
\lim_{\hat{z}_{1}\to z_{1},\dots\hat{z}_{\hat{q}}\to z_{q}}\frac{\Pf[\hat{f}_{ij}]_{1\leq i,j\leq N}}{\Pf[f_{ij}]_{1\leq i,j\leq q}}+o(1).
\]
It remains to identify this expression with $\ccor{\en_{e_{1}}\dots\en_{e_{k}}\sigma_{v_{1}}\dots\sigma_{v_{n}}}_{\Omega,\tilde{\bcond}}/\ccor{\sigma_{v_{1}}\dots\sigma_{v_{n}}}_{\Omega,}$.
To this end, we use the linearity of the Pfaffian as in the proof
of Theorem \ref{thm: intro_2}, can replace in $\Pf[\hat{f}_{ij}]_{1\leq i,j\leq N}$,
for each $l=1,\dots,s$, every instance of $\eta_{z_{2k-1}}\psi_{z_{q+2l-1}}^{\eta_{z_{q+2l-1}}}$
and $\eta_{z_{2k}}\psi_{z_{q+2l}}^{\eta_{z_{q+2l}}}$ with $\lambb2^{-\frac{1}{2}}\psistar{e_{l}}$
and $\lambb2^{-\frac{1}{2}}\psi_{e_{l}},$ respectively. After this
transformation, using the Pfaffian formula (\ref{eq: multipsi_pfaffian})
we get 
\[
\Pf[\hat{f}_{ij}]_{1\leq i,j\leq N}=\prod_{i=1}^{\hat{q}}\varphi'(\hat{z}_{i})^{-\frac{1}{2}}\left(\frac{i}{2}\right)^{s}\frac{\ccor{\psi_{\hat{z}_{1}}\dots\psi_{\hat{z}_{\hat{q}}}\psi_{\hat{z}_{\hat{q}+1}}^{\flat}\dots\psi_{\hat{z}_{q}}^{\flat}\psistar{e_{1}}\psi_{e_{1}}\dots\psistar{e_{s}}\psi_{e_{s}}\sigma_{\cvr}}_{\Omega}}{\ccor{\sigma_{\cvr}}}.
\]
Similarly, by Pfaffian formula (\ref{eq: multipsi_pfaffian}) and
the linearity of the Pfaffian, we have 
\[
\Pf[f_{ij}]_{1\leq i,j\leq q}=\prod_{i=1}^{\hat{q}}\varphi'(\hat{z}_{i})^{-\frac{1}{2}}\ccor{\psi_{\hat{z}_{1}}\dots\psi_{\hat{z}_{\hat{q}}}\psi_{\varphi(z_{\hat{q}+1})}^{\flat}\dots\psi_{\varphi(z_{q})}^{\flat}}_{\Omega}.
\]
 It remains to recall Definitions \ref{def: energy} and \ref{def: bcond_corr_continuous}
to conclude the proof.

To extend the result to the boundary conditions $\bcond$, we use
(\ref{eq: FW_correlation}) and notice that by the above result, we
have 
\begin{multline*}
\E_{\Od,\bcd}(\Op(\sigma,\en))=\frac{\sum_{S}\alpha_{S}\E_{\Od,\bcdd}(\Op(\sigma,\en)\sigma_{S})}{\sum_{S}\alpha_{S}\E_{\Od,\bcdd}(\sigma_{S})}=\frac{\sum_{S}\alpha_{S}\ccor{\Op(\sigma,\en)\sigma_{S}}_{\Omega,\tilde{\bcond}}+o(1)}{\sum_{S}\alpha_{S}\ccor{\sigma_{S}}_{\Omega,\tilde{\bcond}}+o(1)}\\
=\ccor{\Op(\sigma,\en)}_{\Omega,\tilde{\bcond}}+o(1),
\end{multline*}
where we used Lemma \ref{lem: bcond_well_defined} to ensure that
the sum in the denominator does not vanish, and Definition \ref{def: bcond_corr_continuous}
to conclude.

Similarly, the ratios of partition functions can be rewritten as 
\[
\frac{Z(\bcond_{1})}{Z(\bcond_{2})}=\frac{\sum_{S}\alpha_{S}^{\bcond_{1}}\E_{\tilde{\bcond}}(\sigma_{S})}{\sum_{S}\alpha_{S}^{\bcond_{2}}\E_{\tilde{\bcond}}(\sigma_{S})}=\frac{\sum_{S}\alpha_{S}^{\bcond_{1}}\ccor{\sigma_{S}}_{\tilde{\bcond}}+o(1)}{\sum_{S}\alpha_{S}^{\bcond_{2}}\ccor{\sigma_{S}}_{\tilde{\bcond}}+o(1)}=\frac{\sum_{S}\alpha_{S}^{\bcond_{1}}\ccor{\sigma_{S}}_{\tilde{\bcond}}}{\sum_{S}\alpha_{S}^{\bcond_{2}}\ccor{\sigma_{S}}_{\tilde{\bcond}}}+o(1).
\]
\end{proof}
\newpage{}

\section{Fusion rules}

\label{sec: ccor_fusion}

\subsection{Statements of the fusion rules}

\label{subsec: fusion-statements}

Our next topic is \emph{fusion rules, }also known as \emph{operator
product expansions} (OPE), for the (continuous) correlation functions
$\ccor{\Op}_{\Omega}.$ These are asymptotic expansions of the correlation
functions as some of the marked points tend to each other. To formulate
the results concisely, we will adopt a short-hand notation and write
the expansions in terms of the fields, or symbols,\emph{ $\psi_{z}^{\any}$,
$\sigma_{v}$, $\mu_{u}$, $\en_{e}$ }etc. The meaning of these expansions
is that they hold \emph{inside correlations} with any other fields,
with a postulated linearity of correlations: 
\[
\CorrO{(\alpha\Op_{1}+\beta\Op_{2})\Op}\stackrel{\text{def}}{=}\alpha\CorrO{\Op_{1}\Op}+\beta\CorrO{\Op_{2}\Op}.
\]
For example, the fusion rule 
\begin{equation}
\sigma_{\wone}\sigma_{\wtwo}=|\wone-\wtwo|^{-\frac{1}{4}}\left(1+\tfrac{1}{2}\en_{\wtwo}\cdot|\wone-\wtwo|+o(|\wone-\wtwo|)\right),\quad\wone\to\wtwo,\label{eq: sigma-sigma-OPE-example}
\end{equation}
by definition, means that 
\[
\CorrO{\Op\sigma_{\wone}\sigma_{\wtwo}}=|\wone-\wtwo|^{-\frac{1}{4}}\left(\CorrO{\Op}+\tfrac{1}{2}\CorrO{\Op\en_{\wtwo}}\cdot|\wone-\wtwo|+o(|\wone-\wtwo|)\right),
\]
where $\Omega$ is an arbitrary finitely connected domain, and $\Op=\Op(\psi^{\any},\en,\mu,\sigma)$,
as before, stands for \emph{any} expression of the form 
\[
\psi_{z_{1}}^{\any}\ldots\psi_{z_{k}}^{\any}\en_{e_{1}}\ldots\en_{e_{s}}\mu_{u_{1}}\ldots\mu_{u_{m}}\sigma_{v_{1}}\ldots\sigma_{v_{n}},
\]
where the points $z_{1}\ldots,z_{k},e_{1},\ldots,e_{s},\ldots,u_{1},\ldots,u_{m},v_{1},\ldots,v_{n}$
are distinct from each other and from $\wt$. Formally, to translate
a fusion rule for fields into a statement about correlations, we consider
the correlations of the two sides of the rule with an arbitrary $\Op$
and expand everything by linearity, taking the scalars and the operators
of derivation out of the correlations. Note that the factors $|\wo-\wt|^{-\frac{1}{4}}$
and $|\wo-\wt|$ in (\ref{eq: sigma-sigma-OPE-example}) can be easily
reconstructed from the scaling exponents $\frac{1}{8},0,1$ of the
fields $\sigma,1,\en$. This allows to shorten (\ref{eq: sigma-sigma-OPE-example})
even further to 
\[
\sigma_{\wo}\sigma_{\wt}=1+\tfrac{1}{2}\en_{\wt}+\ldots\,,\quad\wo\to\wt,
\]
the notation commonly used in the physics literature. The goal of
this section is to prove the fusion rules listed in Table~\ref{tab: Fusion-rules(OPE)},
where the latter notation is used for the shortness.

\begin{table}
\begin{tabular}{|c||c||c||c||c||c|}
\hline 
 & $\psi_{\wtwo}$  & $\psistar{\wtwo}$  & $\en_{\wtwo}$  & $\mu_{\wtwo}$  & $\sigma_{\wtwo}$\tabularnewline
\hline 
\hline 
$\psi_{\wone}$  & $2+\ldots$  & $-2i\en_{\wtwo}+\ldots$  & $i\psistar{\wtwo}+\ldots$  & $\lamb(\sigma_{\wtwo}+4\pa_{\wtwo}\sigma_{\wtwo}+\ldots)$  & $\lambb(\mu_{\wtwo}+4\pa_{\wtwo}\mu_{\wtwo}+\ldots)$\tabularnewline
\hline 
\hline 
$\psistar{\wone}$  & $2i\en_{\wtwo}+\ldots$  & $2+\ldots$  & $-i\psi_{\wtwo}+\ldots$  & $\lambb(\sigma_{\wtwo}+4\dbar_{\wtwo}\sigma_{\wtwo}+\ldots)$  & $\lambda(\mu_{\wtwo}+4\dbar_{\wtwo}\mu_{\wtwo}+\ldots)$\tabularnewline
\hline 
\hline 
$\en_{\wone}$  &  &  & $1+\ldots$  & $-\tfrac{1}{2}\mu_{\wtwo}+\ldots$  & $\tfrac{1}{2}\sigma_{\wtwo}+\ldots$\tabularnewline
\hline 
\hline 
$\mu_{\wone}$  &  &  &  & $1-\tfrac{1}{2}\en_{\wtwo}+\ldots$  & $\psi_{w}^{\eta_{\wone\wtwo}}+\ldots$\tabularnewline
\hline 
\hline 
$\sigma_{\wone}$  &  &  &  &  & $1+\tfrac{1}{2}\en_{\wtwo}+\ldots$\tabularnewline
\hline 
\hline 
 &  &  &  &  & \tabularnewline
\hline 
\end{tabular}

\bigskip{}

\caption{Fusion rules (OPE) for the fields $\psi,\protect\psistar{},\protect\en,\mu,\sigma$.
The empty cells should be filled symmetrically (recall that $\protect\en,\mu,\sigma$
``commute'' with each other and with ``anti-commuting'' fermions
$\psi,\protect\psistar{\ }$). We refer the reader to Theorems~\ref{thm: fusion rules-1}\textendash \ref{thm: fusion rules-3}
for the exact statements of these asymptotics as $\protect\wone\to\protect\wtwo$.\label{tab: Fusion-rules(OPE)}}
\end{table}

\medskip{}

We split these statements into three separate theorems as follows. 
\begin{thm}
\label{thm: fusion rules-1}The following asymptotics hold, as $\wone\to\wtwo$
: 
\begin{align}
\psi_{\wone}\psi_{\wtwo} & \ =\ 2(\wone-\wtwo)^{-1}+O(\wone-\wtwo),\label{eq: fuse_psi_psi}\\
\psi_{\wone}\psistar{\wtwo} & \ =\ -2i\en_{\wtwo}+O(\wone-\wtwo),\label{eq: fuse_psi_psistar}\\
\psi_{\wone}\en_{\wtwo} & \ =\ i\psistar{\wtwo}\cdot(\wone-\wtwo)^{-1}+O(1),\label{eq: fuse_psi_en}\\
\en_{\wone}\en_{\wtwo} & \ =\ |\wone-\wtwo|^{-2}+O(1).\label{eq: fuse_en_en}
\end{align}
The similar rules for $\psistar{\wtwo}\psistar{\wone}$, $\psistar{\wtwo}\psi_{\wone}$
and $\psistar{\wtwo}\en$ can be obtained from (\ref{eq: fuse_psi_psi})\textendash (\ref{eq: fuse_psi_en})
by the formal complex conjugation of the right-hand sides (which includes
replacing $\overline{\psistar{\wtwo}}$ by $\psi_{\wtwo}$ in~(\ref{eq: fuse_psi_en})). 
\end{thm}

\begin{thm}
\label{thm: fusion rules-2} The following asymptotics hold, as $\wone\to\wtwo$:
\begin{align}
\psi_{\wone}\mu_{\wt} & \ =\ \lamb(\wone-\wtwo)^{-\frac{1}{2}}\left(\sigma_{\wt}+4\partial_{\wt}\sigma_{\wt}\cdot(\wo-\wt)+O((\wone-\wt)^{2})\right),\label{eq: fuse_psi_mu}\\
\psi_{\wone}\sigma_{\wtwo} & \ =\ \lambb(\wone-\wtwo)^{-\frac{1}{2}}\left(\mu_{\wtwo}+4\partial_{\wt}\mu_{\wt}\cdot(\wo-\wt)+O((\wone-\wtwo)^{2})\right),\label{eq: fuse_psi_sigma}\\
\en_{\wone}\mu_{\wtwo} & \ =\ -\tfrac{1}{2}\mu_{\wt}\cdot|\wo-\wt|^{-1}+O(1),\label{eq: fuse_en_mu}\\
\en_{\wone}\sigma_{\wtwo} & \ =\ \tfrac{1}{2}\sigma_{\wt}\cdot|\wo-\wt|^{-1}+O(1).\label{eq: fuse_en_sigma}
\end{align}
The similar rules for $\psistar{\wtwo}\mu_{\wone}$ and $\psistar{\wtwo}\sigma_{\wone}$
can be obtained from (\ref{eq: fuse_psi_mu})\textendash (\ref{eq: fuse_psi_sigma})
by the formal complex conjugation of the right-hand sides (which includes
replacing $\partial_{\wt}$ by $\dbar_{\wt}$). 
\end{thm}

Let 
\begin{equation}
\eta_{uv}=\lambb|u-v|^{\frac{1}{2}}(u-v)^{-\frac{1}{2}},\label{eq: eta_uv}
\end{equation}
compare with the discrete notation (\ref{eq: def_eta_discrete}).

\begin{thm} \label{thm: fusion rules-3} The following asymptotics
hold, as $\wone\to\wtwo$: 
\begin{align}
\mu_{\wone}\sigma_{\wtwo} & \ =\ |\wone-\wtwo|^{\frac{1}{4}}(\psi_{\wtwo}^{\eta_{\wone w}}+O(|\wone-\wtwo|)),\label{eq: fuse_mu_sigma}\\
\sigma_{\wone}\sigma_{\wtwo} & \ =\ |\wone-\wtwo|^{-\frac{1}{4}}\left(1+\tfrac{1}{2}\en_{w}\cdot|\wone-\wtwo|+o(|\wone-\wtwo|)\right),\label{eq: fuse_sigma_sigma}\\
\mu_{\wone}\mu_{\wtwo} & \ =\ |\wone-\wtwo|^{-\frac{1}{4}}\left(1-\tfrac{1}{2}\en_{\wtwo}\cdot|\wone-\wtwo|+o(|\wone-\wtwo|)\right).\label{eq: fuse_mu_mu}
\end{align}
\end{thm}

\begin{rem} The proofs given below ensure that, for a fixed number
of $\psi,\en,\mu,\sigma$ involved into $\Op$, all estimates in Theorems~\ref{thm: fusion rules-1}\textendash \ref{thm: fusion rules-3}
are uniform provided that $w$ and all marked points in $\Op$ are
in the bulk of $\Omega$ and away from each other. \end{rem}

Theorem~\ref{thm: fusion rules-1} covers the ``free fermion''
sector of the theory and is a straightforward corollary of our definitions
of the correlations functions, see Section~\ref{subsec: fusion-proofs-1}
for details. Theorem~\ref{thm: fusion rules-2} relates fermions
with disorders and spins, its proof is discussed in Section~\ref{subsec: fusion-proofs-2}.
Note that the first term in~(\ref{eq: fuse_psi_mu}) is essentially
Definition~\ref{def: Disorders} of correlations containing~$\mu_{u}$.
The second term in~(\ref{eq: fuse_psi_mu}), in absence of~$\Op$,
was used in~Section~\ref{sec: ccor} to define the spin correlations
$\CorrO{\sigma_{v_{1}}\ldots\sigma_{v_{n}}}$; see~(\ref{eq: def_spincorr_up_to_constant}).
To prove~(\ref{eq: fuse_psi_mu}) in presence of~$\Op$ we also
refer to discrete observables while the rest of Theorem~\ref{thm: fusion rules-2}
follows from standard considerations via the so-called satellite integral
computation. The most important statements (fusion of spins and/or
disorders leads back to the ``free fermion'' sector) are collected
in Theorem~\ref{thm: fusion rules-3}. In sharp contrast with Theorems~\ref{thm: fusion rules-1}
and~\ref{thm: fusion rules-2}, we do not have holomorphic (or anti-holomorphic)
fields in the left-hand side, thus more involved arguments should
be used. Note that we prove Theorem~\ref{thm: fusion rules-3} directly
``in continuum'' (see Sections~\ref{subsec: asymptotics-as-spins-collide}\textendash \ref{subsec: fusion-proofs-3}),
in particular justifying the existence of the coherent normalization
of spin correlations claimed in Proposition~\ref{prop: multiplicative-normalisation}.

\subsection{Proof of Theorem~\ref{thm: fusion rules-1}}

\label{subsec: fusion-proofs-1}

We start with general remarks about proofs of the fusion rules given
by Theorems~\ref{thm: fusion rules-1}\textendash \ref{thm: fusion rules-3}.
Below, we assume the domain $\Omega$ to be fixed and drop it from
the notation, so that $\ccor{\ldots}$ stands for $\CorrO{\ldots}$.
Further, in order to prove any of the asymptotics (\ref{eq: fuse_psi_psi})\textendash (\ref{eq: fuse_mu_sigma})
in full generality, it suffices to check that it holds in the correlation
with arbitrary $\Op=\Op(\psi,\psi^{\star},\sigma),$ rather than $\Op(\psi^{\any},\en,\mu,\sigma)$.
Indeed, to include $\psi^{\eta}$ and $\en$, just recall that $\psi_{z}^{\eta}=\bar{\eta}\psi_{z}+\eta\psi_{z}^{\star}$
and $\en_{e}=\frac{i}{2}\psi_{e}\psi_{e}^{\star}$, and to include
$\mu$, use the contour integral representation (\ref{eq: mu_resudues}).

We will denote 
\[
\sigma_{\cvr}:=\sigma_{v_{1}}\ldots\sigma_{v_{n}};\quad\Op(\psi^{\any}):=\psi_{z_{1}}^{\any}\dots\psi_{z_{k}}^{\any}.
\]
To prove a fusion rule for $\Op:=\Op(\psi^{\any})\scvr,$ we will
typically use the Pfaffian expansion~(\ref{eq: ccor_pfaff_expansion})
to reduce the problem to the case of $\Op$ comprised of $n$ spins
and zero, one, or two fermions. Note also that, since $\CorrO{\psi_{\wone}\Op}$
is holomorphic in $\wone$, it suffices to prove~%
\mbox{%
(\ref{eq: fuse_psi_psi})\textendash (\ref{eq: fuse_psi_en})%
} and~%
\mbox{%
(\ref{eq: fuse_psi_mu})\textendash (\ref{eq: fuse_psi_sigma})%
} with $O(\cdot)$ replaced with $o(\cdot)$ of the last term of the
expansion, e.g., o(1) in (\ref{eq: fuse_psi_psi}). Finally, note
that the fusion rules with~$\psistar{\wone}$ instead of $\psi_{\wone}$
can be obtained from the original ones using Proposition~\ref{prop: cont_corr_conjugation}.
\begin{proof}[Proof of (\ref{eq: fuse_psi_psi})]
 Taking into account the above remarks and the anti-symmetry of $\ccor{\psi_{\wone}\psi_{\wtwo}\Op}$
in $\wone,\wtwo$, it suffices to prove (\ref{eq: fuse_psi_psi})
with the error term in the form $O(1)$. We write $\ccor{\psi_{\wone}\psi_{\wtwo}\Op}$
as a $(k+2)\times(k+2)$ Pfaffian by (\ref{eq: multipsi_pfaffian}),
and expand the Pfaffian in the row corresponding to $\wone$. This
yields 
\[
\ccor{\psi_{\wone}\psi_{\wtwo}\Op}=\frac{\ccor{\psi_{\wone}\psi_{\wtwo}\scvr}}{\ccor{\scvr}}\cdot\ccor{\Op}+\sum_{l=1}^{k}\alpha_{l}(w)\cdot\ccor{\psi_{\wone}\psi_{z_{l}}^{\any}\scvr},
\]
where $\alpha_{l}(w)$ is a fermion\textendash spin correlation not
involving $\psi_{\wone}$. Hence, the sum in the right-hand side is
holomorphic at $\wone=\wtwo,$ and we conclude that 
\[
\ccor{\psi_{\wone}\psi_{\wtwo}\Op}=f_{\varpi}(\wt,\wo)\cdot\ccor{\Op}+O(1)=2(\wo-\wt)^{-1}\cdot\ccor{\Op}+O(1),
\]
see (\ref{eq: dte_psi_psi}) and (\ref{eq: f_fdag_expansion}). 
\end{proof}
\begin{proof}[Proof of (\ref{eq: fuse_psi_psistar})]
 This is just the definition of $\en_{w}$; see Definition~\ref{def: energy}
and Remark~\ref{rem: psi_to_psibar}. 
\end{proof}
\begin{proof}[Proof of (\ref{eq: fuse_psi_en})]
 Write $\ccor{\psi_{\wone}\en_{\wtwo}\Op}=\frac{i}{2}\ccor{\psi_{\wone}\psi_{\wtwo}\psi_{w}^{\star}\Op}$
and observe that the argument in the proof of (\ref{eq: fuse_psi_psi})
extends verbatim, since $\ccor{\psi_{\wone}\psi_{w}^{\star}\Op}$
is holomorphic at $\wone=\wtwo$. 
\end{proof}
\begin{proof}[Proof of (\ref{eq: fuse_en_en})]
 Using the definition $\en_{\wo}=\frac{i}{2}\psi_{\wo}\psistar{\wo}$
and $\en_{w}=\frac{i}{2}\psi_{w}\psistar w$, we can write 
\[
\frac{\ccor{\en_{\wo}\en_{\wt}\Op(\psi^{\any})\scvr}}{\ccor{\scvr}}=
-\frac{1}{4}\Pf\biggl[\frac{\ccor{\psi^{\triangleleft}\psi^{\triangleright}\scvr}}{\ccor{\scvr}}\biggr]_{\psi^{\triangleleft},\psi^{\triangleright}\in\{\psi_{\wo},\psistar{\wo},\psi_{w},\psistar w,\psi_{z_{1}}^{\any},\dots,\psi_{z_{k}}^{\any}\}}\,.
\]
The entries of this Pfaffian have finite limits as $\wo\to\wt$, except
$\frac{\ccor{\psi_{\wo}\psi_{\wt}\scvr}}{\ccor{\scvr}}$ and $\frac{\ccor{\psi_{\wo}^{\star}\psi_{\wt}^{\star}\scvr}}{\ccor{\scvr}}$.
Expanding the Pfaffian we obtain 
\begin{multline*}
\frac{\ccor{\en_{\wo}\en_{\wt}\Op(\psi^{\any})\scvr}}{\ccor{\scvr}}=\frac{1}{4}\frac{\ccor{\psi_{\wo}\psi_{\wt}\scvr}}{\ccor{\scvr}}\frac{\ccor{\psi_{\wo}^{\star}\psi_{\wt}^{\star}\Op(\psi^{\any})\scvr}}{\ccor{\scvr}}\\
+\frac{1}{4}\frac{\ccor{\psi_{\wo}^{\star}\psi_{\wt}^{\star}\scvr}}{\ccor{\scvr}}\frac{\ccor{\psi_{\wo}\psi_{\wt}\Op(\psi^{\any})\scvr}}{\ccor{\scvr}}\\
-\frac{1}{4}\frac{\ccor{\psi_{\wo}\psi_{\wt}\scvr}}{\ccor{\scvr}}\frac{\ccor{\psi_{\wo}^{\star}\psi_{\wt}^{\star}\scvr}}{\ccor{\scvr}}\frac{\ccor{\Op(\psi^{\any})\scvr}}{\ccor{\scvr}}+O(1).
\end{multline*}
The desired asymptotics readily follows from (\ref{eq: fuse_psi_psi})
and the similar rule for~$\psistar{\wo}\text{\ensuremath{\psistar w}}$.
\end{proof}

\subsection{Proof of Theorem~\ref{thm: fusion rules-2}}

\label{subsec: fusion-proofs-2}
\begin{proof}[Proof of (\ref{eq: fuse_psi_mu})]
 Note, taking into account (\ref{eq: psi_to_mu_limits}), that the
first term of the asymptotics has been established in a particular
case $\Op=\Op(\sigma)$ in Lemma \ref{lem: f_star_holom}, see (\ref{eq: f_sharp_asymp}),
and the second term of the asymptotics, in this case, is simply the
\emph{definition }of the spin correlation, see (\ref{eq: def_A})
and (\ref{eq: def_spincorr_up_to_constant}).

In the general case, the first term of the asymptotics can be derived
in the same way as in Lemma \ref{lem: f_star_holom} using $(\ref{eq: fuse_psi_psi})$
and the fact that in general, $\ccor{\psi_{\hat{w}}\mu_{u}\Op}$ is
obtained from $\ccor{\psi_{\hat{w}}\psi_{w}\sigma_{u}\Op}$ by the
same procedure as $f_{\Omega,\cvr}^{[\sharp]}(u,\hat{w})$ is obtained
from $f_{\Omega,\cvr}(w,\hat{w}).$ The fastest way to derive the
second term is via a discrete argument. To wit, it is clear by linearity
that we may assume that the operator $\Op$ contains only $\psi^{\eta},\sigma,\mu,\en$
with $\eta\in\{1,i,\lamb,\lambb\}$. In this case, the correlation
$\ccor{\psi_{\hat{w}}\mu_{u}\Op}$ can be obtained as a scaling limit
of the corresponding discrete correlation, 
\[
G^{\delta}(\hat{w})=(\Cmu C_{\Op})^{-1}\delta^{-1-\Delta}\E_{\Omega^{\delta}}\left[\eta_{\hat{w}}\psi_{\hat{w}}^{[\eta_{\hat{w}}]}\mu_{u}\Op\right]=\Cpsi\proj{\eta}{\ccor{\psi_{\hat{w}}\mu_{u}\Op}}+o(1).
\]
Locally near $u$, $G^{\delta}$ is a s-holomorphic spinor in $\hat{w}$
ramified at $u$. If $\hat{w}_{1,2}$ are two corners such that $\hat{w}_{1,2}^{\bullet}=u,$
then, repeating the argument as in the proof of Theorem \ref{thm: Conv_both_near_spin},
which is entirely local, yields \label{eq: fuse_spin_spin_general}
\begin{multline*}
(\Cmu C_{\Op})^{-1}\delta^{-\Delta+\frac{1}{2}}\left(\E_{\Omega^{\delta}}\left[\sigma_{w_{2}^{\circ}}\Op\right]-\E_{\Omega^{\delta}}\left[\sigma_{w_{1}^{\circ}}\Op\right]\right)=\delta^{\frac{1}{2}}\left(\bar{\eta}_{\hat{w}_{1}}G^{\delta}(\hat{w}_{2})-\bar{\eta}_{\hat{w}_{1}}G^{\delta}(\hat{w}_{1})\right)\\
=\re\left((\hat{w}_{2}^{\circ}-\hat{w}_{1}^{\circ})\hat{\mathcal{A}}_{u,\Op}\right)+o(\delta),
\end{multline*}
where $\hat{\mathcal{A}}_{u,\Op}$ is defined by the expansion 
\[
\ccor{\psi_{\hat{w}}\mu_{u}\Op}=\lamb(\wone-\wtwo)^{-\frac{1}{2}}\left(\ccor{\sigma_{u}\Op}+2\hat{\mathcal{A}}_{u,\Op}\cdot(\wo-\wt)+O((\wone-\wt)^{2})\right).
\]
 Summing the above identity over a path from $\hat{w}_{2}$ to $\hat{w}_{1}$
and passing to the limit (taking into account $\Cmu=\Csigma$) yields
\[
\ccor{\sigma_{w_{2}}\Op}-\ccor{\sigma_{w_{1}}\Op}=\re\left(\int_{w_{1}}^{w_{2}}\hat{\mathcal{A}}_{w,\Op}dw\right),
\]
that is, $d\ccor{\sigma_{w}\Op}=\frac{1}{2}\hat{\mathcal{A}}_{w,\Op}dw+\frac{1}{2}\bar{\hat{\mathcal{A}}_{w,\Op}}d\bar w,$
or $\hat{\mathcal{A}}_{w,\Op}=2\pa_{w}\ccor{\sigma_{w}\Op},$ as required.
\end{proof}
\begin{proof}[Proof of (\ref{eq: fuse_psi_sigma})]
 The first term of the asymptotics~(\ref{eq: fuse_psi_sigma}) is
just the definition of the disorder variables $\mu_{\wt}$ in continuum;
see Definition~\ref{def: Disorders}. In order to derive the second
term, note that the first term of~(\ref{eq: fuse_psi_mu}) gives
\begin{multline*}
\oint_{(\wt)}\frac{\ccor{\psi_{\wo}\sigma_{\wt}\Op}}{(\wo-\wt)^{\frac{3}{2}}}d\wo\ =\ \frac{\lambb}{2\pi i}\oint_{(\wt,z)}\oint_{(\wt)}\frac{\ccor{\psi_{\wo}\psi_{z}\mu_{\wt}\Op}dz}{(z-\wt)^{\frac{1}{2}}(\wo-\wt)^{\frac{3}{2}}}d\wo\\
\ =\ \frac{\lambb}{2\pi i}\oint_{(\wt,\wo)}\oint_{(\wt)}\frac{\ccor{\psi_{\wo}\psi_{z}\mu_{\wt}\Op}d\wo}{(z-\wt)^{\frac{1}{2}}(\wo-\wt)^{\frac{3}{2}}}dz\ -\ \frac{\lambb}{2\pi i}\oint_{(\wt)}\oint_{(\wo)}\frac{\ccor{\psi_{\wo}\psi_{z}\mu_{\wt}\Op}dz}{(z-\wt)^{\frac{1}{2}}(\wo-\wt)^{\frac{3}{2}}}d\wo
\end{multline*}
and that the satellite integral vanishes: 
\[
\frac{1}{2\pi i}\oint_{(\wt)}\oint_{(\wo)}\frac{\ccor{\psi_{\wo}\psi_{z}\mu_{\wt}\Op}dz}{(z-\wt)^{\frac{1}{2}}(\wo-\wt)^{\frac{3}{2}}}d\wo\ =\ -2\oint_{(\wt)}\frac{\ccor{\mu_{\wt}\Op}}{(\wo-\wt)^{2}}d\wo\ =\ 0.
\]
Using the second term of the asymptotics~(\ref{eq: fuse_psi_mu})
to compute the integral over $\wo$ (and the anti-commutativity $\psi_{\wo}\psi_{z}=-\psi_{z}\psi_{\wo}$),
we obtain the identity 
\begin{align*}
\oint_{(\wt)}\frac{\ccor{\psi_{\wo}\sigma_{\wt}\Op}}{(\wo-\wt)^{\frac{3}{2}}}d\wo & \ =\-4\oint_{(\wt)}\frac{\partial_{\wt}\ccor{\psi_{z}\sigma_{\wt}\Op}}{(z-\wt)^{\frac{1}{2}}}dz\\
 & \ =\ -4\cdot\biggl(\partial_{\wt}\oint_{(\wt)}\frac{\ccor{\psi_{z}\sigma_{\wt}\Op}}{(z-\wt)^{\frac{1}{2}}}dz-\frac{1}{2}\oint_{(\wt)}\frac{\ccor{\psi_{z}\sigma_{\wt}\Op}}{(z-\wt)^{\frac{3}{2}}}dz\biggr).
\end{align*}
Therefore, we have 
\begin{multline*}
\oint_{(\wt)}\frac{\ccor{\psi_{\wo}\sigma_{\wt}\Op}}{(\wo-\wt)^{\frac{3}{2}}}d\wo\ =\%-\oint_{(\wt,\wo)}\oint_{(\wt)}\frac{\ccor{\psi_{z}\psi_{\wo}\sigma_{\wt}\Op}d\wo}{(z-\wt)^{\frac{1}{2}}(\wo-\wt)^{\frac{3}{2}}}dz\ \\
=\ 4\partial_{\wt}\oint_{(\wt)}\frac{\ccor{\psi_{z}\sigma_{\wt}\Op}}{(z-\wt)^{\frac{1}{2}}}dz\ =\ 2\pi\cdot4\lambb\partial_{\wt}\ccor{\mu_{\wt}\Op},
\end{multline*}
which gives the second term in~(\ref{eq: fuse_psi_sigma}).
\end{proof}
\begin{proof}[Proofs of (\ref{eq: fuse_en_mu}) and (\ref{eq: fuse_en_sigma})]
 Due to the definition of $\en_{\wt}$ we have 
\[
\ccor{\en_{\wo}\mu_{\wt}\Op}=\frac{i}{2}\ccor{\psi_{\wo}\psi_{\wo}^{\star}\mu_{\wt}\Op}=\frac{1}{4\pi}\oint_{(\wo)}\frac{(z-\wt)^{\frac{1}{2}}\ccor{\psi_{z}\psistar{\wo}\mu_{\wt}\Op}}{(\wo-\wt)^{\frac{1}{2}}(z-\wo)}dz.
\]
Moreover, the contour of integration can be moved away from the collapsing
points $\wo,\wt$ since the function $(z-\wt)^{\frac{1}{2}}\ccor{\psi_{z}\psistar{\wo}\mu_{\wt}\Op}$
is holomorphic at $z=w$. Using the fusion rules $\psistar{\wo}\mu_{\wt}=\lambb\sigma_{\wt}+\ldots$
and $\psi_{\wo}\sigma_{\wt}=\lambb\mu_{\wt}+\ldots$ (see~(\ref{eq: fuse_psi_mu})
and~(\ref{eq: fuse_psi_sigma}), respectively) we conclude that 
\begin{align*}
\ccor{\en_{\wo}\sigma_{\wt}\Op} & =\frac{1}{4\pi}(\wo-\wt)^{-\frac{1}{2}}\oint_{(\wt)}(z-\wt)^{-\frac{1}{2}}(1+O(\wo-\wt))\ccor{\psi_{z}\psistar{\wo}\mu_{\wt}\Op}dz\\
 & =\frac{\lambb}{4\pi}|\wo-\wt|^{-1}\cdot\biggl(\oint_{(w)}(z-w)^{-\frac{1}{2}}\ccor{\psi_{z}\sigma_{\wt}\Op}dz+O(\wo-\wt)\biggr)\\
 & =-\frac{1}{2}|\wo-\wt|^{-1}\cdot\left(\ccor{\mu_{\wt}\Op}+O(\wo-\wt)\right)
\end{align*}
since $i\lambb{}^{2}=-1$. The fusion rule (\ref{eq: fuse_en_sigma})
can be obtained by the same arguments. 
\end{proof}

\subsection{Asymptotic expansions of continuous fermionic observables}

\label{subsec: asymptotics-as-spins-collide}

We now move to the analysis of the continuous fermionic observables
$f_{\double{\vv}}^{[\eta]}(a,z)$ as some of the branching points
$v_{1},\ldots,v_{n}$ collide or approach $\fixed\subset\partial\Omega$.
This analysis is the crucial ingredient of the proof of Theorem~\ref{thm: fusion rules-3}
as well as of the proofs of Proposition~\ref{prop: multiplicative-normalisation}
and Proposition~\ref{prop: spin_bdry_behavior}, which we used to
fix the multiplicative normalization of the correlations $\ccor{\sigma_{v_{1}}\ldots\sigma_{v_{n}}}_{\Omega}$
in a coherent way.

Throughout this section, we assume $\Omega$ to be a \emph{circular}
domain; all the results extend to general domains by conformal covariance,
see Lemma~\ref{lem: Conformal-covariance-observable}. This assumption
allows one to extend the observables by Schwarz reflection to a neighborhood
of any compact subset of $\Omega\cup\fixed$, cf. Remark~\ref{rem: Schwarz-reflection};
we will always assume that such an extension has been made. We will
rely upon Proposition~\ref{prop: unelegant_but_useful}. Informally
speaking, this proposition claims that, in order to find a limit of
a sequence of holomorphic spinors, it suffices to analyze their behavior
at, or near, singularities. \begin{prop} \label{prop: unelegant_but_useful}Let
$w\in\Omega\cup\fixed$. Let $f_{\ell}(\cdot)$, $\ell=1,2,\ldots,$
be a sequence of holomorphic $\cvr$-spinors in $\Ocvrc\setminus B_{r_{\ell}}(w)$
and $r_{\ell}\to0$ as $\ell\to\infty$. Assume that $f_{\ell}$ satisfy
the standard boundary conditions outside $B_{r_{\ell}}(w)$. Given
a small enough $R>0$, assume that $f_{\ell}$ admit the following
representations near $w$: 
\begin{equation}
f_{\ell}(z)=\chi_{\ell}(z)\cdot(\rho_{\ell}(z)+g_{\ell}(z)),\qquad z\in B_{R}(w)\setminus B_{r_{\ell}}(w),\label{eq: spinor_decomp}
\end{equation}
where $\chi_{\ell}$ and $\rho_{\ell}$ are holomorphic functions
in $B_{R}(w)\setminus\{w\}$ such that, as $\ell\to\infty$, 
\[
\chi_{\ell}(z)\to1,\quad\rho_{\ell}(z)\to\rho(z),\qquad z\in B_{R}(w)\setminus\{w\},
\]
uniformly on compact subsets, while $g_{\ell}$ does not have singularities
at $w$ and thus can be extended to holomorphic functions in $B_{R}(w)$.
Then, 
\[
f_{\ell}(z)\to f(z),\qquad z\in(\Ocvrc\cup\fixed)\setminus\{w\},
\]
uniformly on compact subsets, where~$f(\cdot)$ satisfies the standard
boundary conditions on $\pa\Ocvrc\setminus\{w\}$ and the asymptotics
$f(z)=\rho(z)+O(1)$ as $z\to w$. \end{prop}

In the applications, we will take either $\chi_{\ell}=1$ or $\chi_{\ell}=\chi_{u^{(\ell)},v^{(\ell)}}$,
where 
\begin{equation}
\chi_{u,v}(z):=\left(\frac{z-u}{z-v}\right)^{\frac{1}{2}}\label{eq: chi_u_v_def}
\end{equation}
and the sequences $u^{(\ell)},v^{(\ell)}$ both tend to $w$ as $l\to\infty$.
In the latter case note that the branchings over $u^{(\ell)}$ and
$v^{(\ell)}$ compensate each other outside a very small vicinity
of~$w$, which allows us to view $\chi_{\ell}$ as a \emph{function}
in the annulus $B_{R}(w)\setminus B_{r_{\ell}}(w)$.

\smallskip{}

Before giving a proof of Proposition~\ref{prop: unelegant_but_useful},
we illustrate it by the following example.
\begin{example}
\label{exa: diff_obs_by_a} For a sequence $a_{\ell}\to a\in\Omega$
such that $\nu_{\ell}:=(a_{\ell}-a)/|a_{\ell}-a|\to\nu$, put 
\[
f_{\ell}(z):=\frac{f^{[1]}(a_{\ell},z)-f^{[1]}(a,z)}{|a_{\ell}-a|},\qquad\chi_{\ell}(z):=1,\qquad\rho_{\ell}(z):=\frac{\nu_{\ell}}{(z-a_{\ell})(z-a)}.
\]
Proposition \ref{prop: unelegant_but_useful} shows that $f_{\ell}$
converge to a spinor $f$ satisfying the standard boundary conditions
and such that $f(z)=\nu(z-a)^{-2}+O(1)$ as $z\to a$. Iterating this
procedure, we see that for any polynomial $P$ there exists a $\cvr$-spinor
$f_{\cvr}^{[\overline{P}]}(a,\cdot)$ satisfying the standard boundary
conditions and such that 
\[
f_{\cvr}^{[\overline{P}]}(a,z)=(z-a)^{-1}P((z-a)^{-1})+O(1),\quad z\to a.
\]
Note that $f_{\cvr}^{[\overline{P}]}$ is uniquely determined by these
conditions due to Proposition~\ref{prop: Uniqueness_continuous}. 
\end{example}

\begin{proof}[Proof of Proposition~\ref{prop: unelegant_but_useful}]
 Assume first that the sequence $f_{\ell}$ eventually remains uniformly
bounded on each compact subset of $(\Ocvrc\cup\fixed)\setminus\{w\}.$
Then, we can extract a convergent subsequence, denote the sub-sequential
limit of $f_{\ell}$ by $f$. This limit satisfies the standard boundary
conditions (except at the point~$w$ in the case $w\in\fixed$).
From the convergence of $f_{\ell}$, $\chi_{\ell}$ and $\rho_{\ell}$,
we deduce that $g_{\ell}$ also converges uniformly on compact subsets
of $B_{R}\setminus\{w\}$. Since $g_{\ell}$ satisfies the maximum
principle, we must therefore have $f(z)=\rho(z)+O(1)$ as~$z\to w$.
The uniqueness in the boundary value problem shows that the convergence
holds without extraction.

If spinors $f_{\ell}$ are not uniformly bounded on compact subsets
of $\Omega\cup\fixed\setminus\{w\}$, define $h_{\ell}(z):=\im\int_{z_{0}}^{z}f_{\ell}^{2}$,
where the point $z_{0}\in\Omega$ is chosen arbitrarily. Proposition
\ref{prop: f_to_h} and Remark \ref{rem: h_well_defined} guarantee
that $h_{\ell}$ is a well defined harmonic function with locally
constant Dirichlet boundary conditions on $\fixed$; we can also extend
it to a neighborhood of any compact set $K\subset(\Omega\cup\fixed)\setminus\{B_{r_{\ell}}(w)\}$.
Note that, for some $r$ such that $0<r<R/4$, we have 
\[
M_{\ell}:=(\sup\nolimits _{\Omega\setminus B_{r}(w)}|h_{\ell}|)^{\frac{1}{2}}\to+\infty\quad\text{as}\ \ \ell\to\infty
\]
(otherwise, $h_{\ell}$, and hence $f_{\ell}$, would be uniformly
bounded in a neighborhood of any compact subset of $\Omega\cup\fixed\setminus\{w\}$).
Then, the re-scaled spinors $M_{\ell}^{-1}f_{\ell}$ are uniformly
bounded on compact subsets of $(\Omega\cup\fixed)\setminus B_{R/2}(w)$.
Using the decomposition (\ref{eq: spinor_decomp}) and the fact that
$g_{\ell}$ satisfies the maximum principle, we infer that these re-scaled
spinors remain uniformly bounded on any compact subset of $(\Omega\cup\fixed)\setminus\{w\}$.
Moreover, each sub-sequential limit of the sequence~$M_{\ell}^{-1}f_{\ell}$
is \emph{bounded at $w$} since $M_{\ell}^{-1}\chi_{\ell}\rho_{\ell}\to0$.
As these sub-sequential limits also satisfies the standard boundary
conditions, they must vanish, which means that $M_{\ell}^{-1}f_{\ell}\to0$
as $\ell\to\infty$, uniformly on compact subsets of $(\Omega\cup\fixed)\setminus\{w\}$.
Therefore, the functions $M_{\ell}^{-2}h_{\ell}=\im\int(M_{\ell}^{-1}f_{\ell})^{2}$
also converge to zero as $\ell\to\infty$, which contradicts to the
definition of $M_{\ell}$. 
\end{proof}
\begin{rem}
\label{rem: unelegant_many_points} In the proof of Proposition~\ref{prop: spin_bdry_behavior}
(see Section~\ref{subsec: proof-normalization} below) we will also
need the following straightforward generalization of Proposition~\ref{prop: unelegant_but_useful}.
Let $w_{1},\dots,w_{p}$ be distinct points of $\Omega\cup\fixed$,
and assume that $\cvr$-spinors $f_{\ell}(\cdot)$ are holomorphic
and satisfy the standard boundary conditions outside $\cup_{q=1}^{p}B_{r_{\ell}}(w_{q}),$
where $r_{\ell}\to0$. Assume that for each $q=1,\dots,p$, we can
write 
\[
f_{\ell}(z)=\chi_{\ell}^{(q)}(z)(\rho_{\ell}^{(q)}(z)+g_{\ell}^{(q)}(z)),\quad z\in B_{R}(w_{q})\setminus B_{r_{\ell}}(w_{q}),
\]
where $\chi^{(q)}$, $\rho^{(q)}$ and $g^{(q)}$ are as above. Then,
$f_{\ell}\to f$ as $\ell\to\infty$, uniformly on compact subsets
of $(\Omega\cup\fixed)\setminus\{w_{1},\dots,w_{p}\}$, where $f$
satisfies the standard boundary conditions on $\pa\Omega\setminus\{w_{1},\dots,w_{p}\}$
and the asymptotics expansions $f(z)=\rho^{(q)}(z)+O(1)$ near each
of the points $w_{q}$, $q=1,\ldots,p$. (Note that, as usual, these
conditions define $f$ uniquely.)

The proof of Proposition~\ref{prop: unelegant_but_useful} extends
verbatim, except that $h_{\ell}$ does not have to be a single-valued
function anymore. To overcome this complication, we can consider the
restrictions of $h_{\ell}$ to a simply-connected sheet of the universal
cover of $\Omega$, and define $M_{\ell}$ to be the maximum of this
restriction. Actually, in the proof of Proposition~\ref{prop: spin_bdry_behavior}
we do know that $h_{\ell}$ is a single-valued function from the setup
under consideration. 
\end{rem}

We now give several applications of Proposition~\ref{prop: unelegant_but_useful}
needed below. The first one, Proposition~\ref{prop: u_to_v-1}, is
the key ingredient of the proof of Proposition~\ref{prop: multiplicative-normalisation}
and the fusion rule~\eqref{eq: fuse_mu_sigma}.

\begin{prop} \label{prop: u_to_v-1} Let $v\in\Omega$ and $\cvr=\double{\vv}$.
The following asymptotics hold: 
\begin{equation}
f_{\double{u,v}\cdot\cvr}^{[\sharp]}(u,z)=|u-v|^{\frac{1}{2}}f_{\cvr}^{[\eta_{uv}]}(u,z)+o(|u-v|^{\frac{1}{2}})\quad\text{as}\ \ u\to v,\label{eq: f_u_to_v}
\end{equation}
where $\eta_{uv}$ is defined by~(\ref{eq: eta_uv}). The asymptotics
are uniform for $v$ lying in the bulk of $\Omega$ (in particular,
away from $v_{1},\ldots,v_{n}$), and for $z$ away from $v,v_{1},\ldots,v_{n}$.
\end{prop}
\begin{proof}
Let 
\[
f_{u,v}(z):=|u-v|{}^{-\frac{1}{2}}f_{\double{u,v}\cdot\cvr}^{[\sharp]}(u,z),\qquad\rho_{u,v}(z):=\frac{\bar{\eta_{uv}}}{z-u},
\]
and 
\[
g_{u,v}(z):=(\chi_{u,v}(z))^{-1}f_{u,v}(z)-\rho_{u,v}(z),
\]
where the function~$\chi_{(u,v)}(z)$ is given by~(\ref{eq: chi_u_v_def}).
The asymptotics (\ref{eq: bc_branchpoint}) and (\ref{eq: f_sharp_asymp})
imply that $g_{u,v}(\cdot)$ does not have singularities neither at
$v$ nor at $u$, respectively. If $u^{(\ell)},v^{(\ell)}$ are two
sequences converging to the same limit in the bulk of $\Omega$ and
away from $v_{1},\ldots,v_{n}$, in such a way that $\eta_{u^{(\ell)}v^{(\ell)}}\to\eta$
as $\ell\to\infty$, then Proposition \ref{prop: unelegant_but_useful}
implies the convergence %
\mbox{%
$f_{u^{(\ell)},v^{(\ell)}}\to f_{\cvr}^{[\eta]}(u,\cdot)$%
}. The claim follows from a standard compactness argument.
\end{proof}
\begin{rem}
\label{rem: f_u_to_v_higher_terms}It is also possible to compute
higher terms in the asymptotics~(\ref{eq: f_u_to_v}) and, in particular,
to justify that the error term~$o(|u-v|)$ can be replaced by $O(|u-v|^{2})$.
Namely, the next term in~(\ref{eq: f_u_to_v}) can be explicitly
written as 
\[
f_{\double{u,v}\cdot\cvr}^{[\sharp]}(u,z)=|u-v|^{\frac{1}{2}}f_{\cvr}^{[\eta_{uv}]}(u,z)+|u-v|^{\frac{3}{2}}\hat{f}_{\cvr}^{[\eta_{uv}]}(u,z)+o(|u-v|^{\frac{3}{2}}),
\]
where $\hat{f}_{\cvr}^{[\eta]}(u,z)$ denotes the unique $\cvr$-spinor
satisfying the standard boundary conditions and such that 
\[
\hat{f}_{\cvr}^{[\eta]}(u,z)=-\frac{i\bar{\eta}^{2}}{2(z-u)}f_{\cvr}^{[\eta]}(u,z)+O((z-u)^{\frac{1}{2}}).
\]
as $z\to u$. To prove this refined asymptotic expansion, take 
\begin{align*}
f_{u,v}(z) & :=|u-v|^{-1}(|u-v|^{-\frac{1}{2}}f_{\double{u,v}\cdot\cvr}^{[\sharp]}(u,z)-f_{\cvr}^{[\eta_{uv}]}(u,z)),\\
g_{u,v}(z) & :=|u-v|^{-1}((\chi_{u,v}(z))^{-1}|u-v|^{-\frac{1}{2}}f_{\double{u,v}\cdot\cvr}^{[\sharp]}(u,z)-f_{\cvr}^{[\eta_{uv}]}(u,z))
\end{align*}
(note that the functions $g_{u,v}$ do not have singularities neither
at $v$ nor at $u$), and 
\begin{align*}
\rho_{u,v}(z) & :=(\chi_{u,v}(z))^{-1}f_{u,v}(z)-g_{u,v}(z)=|u-v|^{-1}(1-(\chi_{u,v}(z))^{-1})f_{\cvr}^{[\eta_{uv}]}(u,z).
\end{align*}
Applying Proposition~\ref{prop: unelegant_but_useful} as above we
obtain the desired result since 
\begin{equation}
|u-v|^{-1}(1-(\chi_{u,v}(z))^{-1})=-\frac{i\bar{\eta}_{uv}^{2}}{2(z-u)}+o(1),\quad u\to v,\label{eq: chi_expansion}
\end{equation}
uniformly on compact subsets of $\C\setminus\{v\}$. 
\end{rem}

The following Proposition~\ref{prop: v_to_v} is crucial for both
the proof of Proposition~\ref{prop: multiplicative-normalisation}
and the proof of the fusion rule~(\ref{eq: fuse_sigma_sigma}) in
Theorem~\ref{thm: fusion rules-3}.
\begin{prop}
\label{prop: v_to_v} Let $a,v\in\Omega$, $a\ne v$, and $\cvr=\double{\vv}$.
We have 
\[
f_{\double{u,v}\cdot\cvr}^{[\eta]}(a,z)=f_{\cvr}^{[\eta]}(a,z)+|u-v|f_{\cvr}^{[\theta]}(u,z)+o(|u-v|)\quad\text{as}\ \ u\to v,
\]
where $\theta:=-\frac{i}{2}f_{\cvr}^{[\eta]}(a,u)$. The asymptotics
are uniform for $v$ lying in the bulk of~$\Omega$ (in particular,
away from $v_{1},\ldots,v_{n}$), and for $a,z$ away from $v,v_{1},\ldots,v_{n}$.
\end{prop}

\begin{proof}
Define 
\begin{align*}
f_{u,v}(z) & :=|u-v|^{-1}\big(f_{\double{u,v}\cdot\cvr}^{[\eta]}(a,z)-f_{\cvr}^{[\eta]}(a,z)\big),\\
g_{u,v}(z) & :=|u-v|^{-1}\biggl((\chi_{u,v}(z))^{-1}f_{\double{u,v}\cdot\cvr}^{[\eta]}(a,z)-\frac{\lambb(u-v)^{\frac{1}{2}}}{z-u}f_{\double{u,v}\cdot\cvr}^{[\eta,\sharp]}(a,u)-f_{\cvr}^{[\eta]}(a,z)\biggr),
\end{align*}
note that the function $g_{u,v}$ indeed does not have singularities
neither at $v$ nor at $u$ (the latter follows from Definition~\ref{def: flat}).
Therefore, to apply Proposition~\ref{prop: unelegant_but_useful}
as above, it is enough to compute the limit, as $u\to v$, of the
functions 
\begin{align*}
\rho_{u,v}(z):= & (\chi_{u,v}(z))^{-1}f_{u,v}(z)-g_{u,v}(z)\\
= & |u-v|^{-1}\biggl(\frac{\lambb(u-v)^{\frac{1}{2}}}{z-u}f_{\double{u,v}\cdot\cvr}^{[\eta,\sharp]}(a,u)+(1-(\chi_{u,v}(z))^{-1})f_{\cvr}^{[\eta]}(a,z))\biggr).
\end{align*}
Due to Proposition \ref{prop: u_to_v-1} we have 
\begin{equation}
f_{\double{u,v}\cdot\cvr}^{[\eta,\sharp]}(a,u)=-f_{\double{u,v}\cdot\cvr}^{[\sharp,\eta]}(u,a)=-|u-v|^{\frac{1}{2}}(f_{\cvr}^{[\eta_{uv},\eta]}(u,a)+o(1)),\label{eq: f_uv_usedintheproof}
\end{equation}
where $\eta_{uv}$ is given by (\ref{eq: eta_uv}). Taking into account
(\ref{eq: chi_expansion}), we conclude that, as $u\to v$, 
\[
\rho_{u,v}(z)=\rho(z)+o(1),\quad\rho(z):=\frac{1}{z-u}\cdot\bigl(-i\bar{\eta}_{uv}f_{\cvr}^{[\eta_{uv},\eta]}(u,a)-\tfrac{i}{2}\bar{\eta}_{uv}^{2}f_{\cvr}^{[\eta]}(a,z)\bigr)
\]
uniformly on compact subsets of $B_{R}(v)\setminus\{v\}$. Finally,
observe that the function $\rho(z)$ has a simple pole at $u$ and
\begin{align*}
\res_{z=u}\rho(z) & \ =\ i\bar{\eta}_{uv}f_{\cvr}^{[\eta,\eta_{uv}]}(a,u)-\tfrac{i}{2}\bar{\eta}_{uv}^{2}f_{\cvr}^{[\eta]}(a,u)\\
 & \ =\ \tfrac{i}{2}\bigl(\bar{\eta}_{uv}\big(\bar{\eta}_{uv}f_{\cvr}^{[\eta]}(a,u)+\eta_{uv}\bar f{}_{\cvr}^{[\eta]}(a,u)\big)-\bar{\eta}_{uv}^{2}f_{\cvr}^{[\eta]}(a,u)\bigl)\ =\ \tfrac{i}{2}\bar f{}_{\cvr}^{[\eta]}(a,u)\ =\ \overline{\theta}.
\end{align*}
Proposition~\ref{prop: unelegant_but_useful} yields $f_{u,v}(z)=f_{\cvr}^{[\theta]}(v,z)+o(1)=f_{\cvr}^{[\theta]}(u,z)+o(1)$.
The claim follows. 
\end{proof}
We also derive similar asymptotics as one of the points $v_{1},\ldots,v_{n}$
approaches the wired part of $\partial\Omega$, this is crucial for
the proof of Proposition~\ref{prop: spin_bdry_behavior}. Let $w\in\fixed$
and $r>0$ be small enough so that $\Omega\cap B_{r}(w)=\H\cap B_{r}(w)$
as in Remark \ref{rem: Schwarz-reflection}. The similar asymptotics
in the general setup follows from this particular case by conformal
covariance. \begin{prop} \label{prop: _u_v_to_bdry} Let $\cvr=\double{\vv}$.
Under the above conditions, we have the following asymptotic expansions
as $u\to w\in\fixed\subset\R$: 
\begin{align}
f_{\double u\cdot\cvr}^{[\sharp]}(u,z) & =(2\im u){}^{\frac{1}{2}}f_{\double w\cdot\cvr}^{[1]}(w,z)+o((\im u){}^{\frac{1}{2}}),\label{eq: f_u_to_bdry}\\
f_{\double u\cdot\cvr}^{[\eta]}(a,z) & =f_{\double w\cdot\cvr}^{[\eta]}(a,z)+(2\im u)f_{\double w\cdot\cvr}^{[\theta]}(w,z)+o(\im u),\label{eq: f_v_to_bdry}
\end{align}
uniformly for $z$ (and $a$) away from $w,v_{1},\ldots,v_{n}$, and
for $w$ in compact subsets of $\fixed$, where $\theta:=-\frac{i}{2}f_{\cvr}^{[\eta]}(a,w)$.
More generally, the asymptotics (\ref{eq: f_u_to_bdry}) can be generalized
to the situation when each of the points $v_{q}$, $q=1,\ldots,n$,
either stays put or tends to a point $w_{q}\in\fixed$, where all
$w_{q}$ are distinct from $w$ and from each other. \end{prop}

\begin{rem} Recall that $\cvr(w)$, where $w\in\fixed$, denotes
the double cover of~$\Omega$ that branches over the \emph{boundary
component} to which $w$ belongs. In particular, the functions $f_{\double w\cdot\cvr}^{[\theta]}(w,\cdot)$,
being extended into a punctured neighborhood of $w$ by the Schwarz
reflection, do not branch over the point $w$ itself. \end{rem}

\begin{proof} Since $f_{\double u\cdot\cvr}^{[\sharp]}(u,z)\in\R$
and $f_{\double u\cdot\cvr}^{[\eta]}(a,z)\in\R$ for $z\in\R\cap\fixed$,
we can continue these spinors analytically to $B_{r}(w)\setminus\{u,\bar u\}$.
To prove (\ref{eq: f_u_to_bdry}), we follow the proof of Proposition
\ref{prop: u_to_v-1} substituting $v:=\bar u$; note that $\eta_{u\bar u}=1$
and that $f_{\double w\cdot\cvr}^{[1]}(w,z)=f_{\double w\cdot\cvr}(w,z)$
for $w\in\fixed\subset\R$, see Remark~\ref{rem: Schwarz-reflection}.
Similarly, to prove (\ref{eq: f_v_to_bdry}), we follow the proof
of Proposition~\ref{prop: v_to_v} with $v:=\bar u$; the only change
being that we use (\ref{eq: f_u_to_bdry}) instead of Proposition~\ref{prop: u_to_v-1}
to derive the proper analogue of (\ref{eq: f_uv_usedintheproof}).

To prove the final claim, i.e., to justify the asymptotics~(\ref{eq: f_u_to_bdry})
as $u\to w$ and (some of) $v_{q}\to w_{q}$, we use Remark \ref{rem: unelegant_many_points}.
The singularities near the points $w_{q}$ can be handled as in the
proof of (\ref{eq: f_v_to_bdry}); here we only need to know that
the leading term of~(\ref{eq: f_u_to_bdry}) satisfies the standard
boundary conditions everywhere in a vicinity of $w_{q}$ and remains
bounded at~$w_{q}$. Hence, the limit of $(2\im u)^{-\frac{1}{2}}f_{\double u\cdot\cvr}^{[\sharp]}(u,z)$
satisfies standard boundary conditions everywhere except at the point
$w$. The singularity near $w$ can be handled exactly in the same
way as in the proof of (\ref{eq: f_u_to_bdry}). The further details
are left to the reader. \end{proof}

\subsection{Proof of Proposition \ref{prop: multiplicative-normalisation} and
Proposition~\ref{prop: spin_bdry_behavior}.}

\label{subsec: proof-normalization}

Recall that we used Proposition~\ref{prop: multiplicative-normalisation}
and~\ref{prop: spin_bdry_behavior} in Section~\ref{subsec: Def_corr_cont}
to fix (and to justify the existence of) the coherent multiplicative
normalization of the spin correlations~$\CorrO{\sigma_{v_{1}}\ldots\sigma_{v_{n}}}$.
The proofs given below are based upon the results of Section~\ref{subsec: asymptotics-as-spins-collide}.
We start with a relatively straightforward corollary of Propositions~\ref{prop: u_to_v-1}\textendash \ref{prop: _u_v_to_bdry}
on the asymptotics of the relevant coefficients $\coefA$ (see Definition~\ref{def: coefA})
when some of the points $v_{1},\ldots,v_{n}$ collide or approach
$\fixed\subset\partial\Omega$.

\begin{prop} \label{prop: coefA_vary} Let $v\in\Omega$ and $\cvr:=\double{\vv}$.
The coefficients $\coefA$ satisfy the following asymptotic expansions
as $u\to v$: 
\begin{align}
\coefA_{\double{u,v}\cdot\cvr}(u) & \ =\ -\frac{1}{4(u-v)}+\frac{\eta_{uv}^{2}}{4}f_{\cvr}^{\star}(u,u)+o(1),\label{eq: A_u_to_v}\\
\coefA_{\double{u,v}\cdot\cvr}(v_{q}) & \ =\ \coefA_{\cvr}(v_{q})+o(1),\quad q=1,\dots,n,\label{eq: A_v_to_v}
\end{align}
uniformly over $u,v$ in the bulk of~$\Omega$ and away from $v_{1},\dots,v_{n}$.
Similarly, in the setup of Proposition \ref{prop: _u_v_to_bdry} we
have 
\begin{equation}
\coefA_{\double u\cdot\cvr}(u)=-\frac{1}{4(u-\bar u)}+o(1),\label{eq: A_u_to_bdry}
\end{equation}
as $u\to w\in\fixed\subset\R$ and each of the points $v_{q}$, $q=1,\ldots,n$,
either stays put or tends to a point $w_{q}\in\fixed$, where all
$w_{q}$ are distinct from $w$ and from each other. \end{prop}

\begin{proof} Note that $(z-v)^{\frac{1}{2}}/(u-v)^{\frac{1}{2}}=1+\tfrac{1}{2}(z-u)/(u-v)+O((z-u)^{2})$
as $z\to u$. Therefore, the asymptotic expansion (\ref{eq: def_A})
of the spinor~$f_{\double{u,v}\cdot\cvr}^{[\sharp]}(u,z)$ near $u$
yields 
\[
\mathop{\res}\limits _{z=u}\frac{(z-v)^{\frac{1}{2}}f_{\double{u,v}\cdot\cvr}^{[\sharp]}(u,z)}{(u-v)^{\frac{1}{2}}(z-u)^{\frac{3}{2}}}\ =\ \lambda\cdot\biggl(2\coefA_{\double{u,v}\cdot\cvr}(u)+\frac{1}{2(u-v)}\biggr)
\]
and thus 
\[
\coefA_{\double{u,v}\cdot\cvr}(u)+\frac{1}{4(u-v)}\ =\ \frac{\lambb}{4\pi i}\oint_{(u)}\frac{(z-v)^{\frac{1}{2}}f_{\double{u,v}\cdot\cvr}^{[\sharp]}(u,z)}{(u-v)^{\frac{1}{2}}(z-u)^{\frac{3}{2}}}dz.
\]
Since the integrand does not have a pole at $v$, the contour of integration
can be moved at a definite distance from $u$ and $v$. Applying Proposition~\ref{prop: u_to_v-1}
for $z$ on this contour and using the fact that $(z-v)/(z-u)=1+o(1)$
as $u\to v$, we see that 
\begin{align*}
\coefA_{\double{u,v}\cdot\cvr}(u)+\frac{1}{4(u-v)} & \ =\ \frac{\eta_{uv}}{4\pi i}\oint_{(u,v)}\frac{(z-v)^{\frac{1}{2}}(f_{\cvr}^{[\eta_{uv}]}(u,z)+o(1))}{(z-u)^{\frac{3}{2}}}dz\\
 & \ =\ \frac{\eta_{uv}}{4\pi i}\oint_{(u,v)}\frac{f_{\cvr}^{[\eta_{uv}]}(u,z)}{z-u}dz+o(1)\ =\ \frac{\eta_{uv}^{2}}{4}f_{\cvr}^{\star}(u,u)+o(1),
\end{align*}
the last identity follows from~(\ref{eq: f_eta_decomp}) and~(\ref{eq: f_fdag_expansion}).
This proves~(\ref{eq: A_u_to_v}).


\smallskip{}

To prove~(\ref{eq: A_v_to_v}), note that 
\begin{align*}
\coefA_{\double{u,v}\cdot\cvr}(v_{q})-\coefA_{\cvr}(v_{q})\  & =\ \frac{\lambb}{4\pi i}\oint_{(v_{q})}\frac{f_{\double{u,v}\cdot\cvr}^{[\sharp]}(v_{q},z)-f_{\cvr}^{[\sharp]}(v_{q},z)}{(z-v_{q})^{\frac{3}{2}}}dz\\
 & =\ -\frac{1}{8\pi^{2}}\oint_{(v_{q},\zeta)}\oint_{(v_{q})}\frac{(f_{\double{u,v}\cdot\cvr}(\zeta,z)-f_{\cvr}(\zeta,z))d\zeta}{(\zeta-v_{q})^{\frac{1}{2}}(z-v_{q})^{\frac{3}{2}}}dz
\end{align*}
where the contours of integration can be moved at a definite distance
from $v_{q}$ and from each other. Due to Proposition~\ref{prop: v_to_v}
(and (\ref{eq: f_f_dagger_back})), the right-hand side vanishes as
$u\to v$.

\smallskip{}

To prove~(\ref{eq: A_u_to_bdry}), recall that both spinors $f_{\double u\cdot\cvr}^{[\sharp]}(u,z)$
and $f_{\double w\cdot\cvr}^{[1]}(w,z)=f_{\double w\cdot\cvr}(w,z)$
satisfy the standard boundary conditions and thus can be analytically
continued (as functions of $z$) into a vicinity of the point $w\in\fixed\subset\R$
using the Schwarz reflection principle; see Remark~\ref{rem: Schwarz-reflection}.
(Note that along this procedure, $f_{\double u\cdot\cvr}^{[\sharp]}(u,\cdot)$
gains another branching point at $v:=\bar u$.) Repeating the proof
of~(\ref{eq: A_u_to_v}) with the asymptotics~(\ref{eq: f_u_to_bdry})
instead of~(\ref{eq: f_u_to_v}) we obtain 
\[
\coefA_{\double u\cdot\cvr}(u)+\frac{1}{4(u-\bar u)}\ =\ \frac{1}{4\pi i}\oint_{(w)}\frac{f_{\double w\cdot\cvr}(w,z)}{z-w}dz+o(1)\ =\ o(1),
\]
the last identity follows from~(\ref{eq: f_fdag_expansion}) and
Remark~\ref{rem: Schwarz-reflection}. \end{proof}
\begin{proof}[Proof of Proposition~\ref{prop: multiplicative-normalisation}]
 Note that we actually need to prove the following statement: for
each even $n$, there exists a constant $C_{n}$ such that for each
$v,v_{1},\ldots,v_{n}\in\Omega$ one has 
\begin{equation}
\log\CorrO{\sigma_{u}\sigma_{v}\sigma_{\cvr}}-\log\CorrO{\sigma_{\cvr}}=-\tfrac{1}{4}\log|u-v|+C_{n}+o(1)\label{eq: x_spin_to_spin}
\end{equation}
as $u\to v$, where $\sigma_{\cvr}=\sigma_{v_{1}}\ldots\sigma_{v_{n}}$
and the functions $\log\CorrO{\sigma_{u}\sigma_{v}\sigma_{\cvr}}$
and $\log\CorrO{\sigma_{\cvr}}$ are defined by~(\ref{eq: def_spincorr_up_to_constant})
up to unknown integration constants. Indeed, these asymptotics allows
one to choose the integration constants in~(\ref{eq: def_spincorr_up_to_constant})
recursively so that (\ref{eq: x_spin_to_spin}) hold with $C_{n}=0$
for all $n$. With this choice, we have (\ref{eq: spins_coherent})
for all $n$. Actually, Lemma~\ref{lem: spin_to_spin} given below
provides even sharper asymptotics~(\ref{eq: lem_spin_to_spin}),
thus concluding the proof. 
\end{proof}
\begin{lem}
\label{lem: spin_to_spin} Let $n$ be even, $v,v_{1},\dots,v_{n}\in\Omega$
and $\sigma_{\cvr}:=\sigma_{v_{1}}\ldots\sigma_{v_{n}}$. Then, the
following asymptotics hold as $u\to v$: 
\begin{equation}
\log\CorrO{\sigma_{u}\sigma_{v}\sigma_{\cvr}}-\log\CorrO{\sigma_{\cvr}}=-\tfrac{1}{4}\log|u-v|+C_{n}+\tfrac{i}{4}f_{\cvr}^{\star}(v,v)|u-v|+o(|u-v|)\label{eq: lem_spin_to_spin}
\end{equation}
uniformly over $v,v_{1},\ldots,v_{n}$ in the bulk of $\Omega$ and
away from each other, where $\log\CorrO{\sigma_{u}\sigma_{v}\sigma_{\cvr}}$
and $\log\CorrO{\sigma_{\cvr}}$ are defined by~(\ref{eq: def_spincorr_up_to_constant})
up to unknown integration constants and the constant $C_{n}$ in the
right-hand side of~(\ref{eq: lem_spin_to_spin}) is independent of
$u,v,v_{1},\dots,v_{n}$. 
\end{lem}

\begin{proof}
First, consider the correlation~$\CorrO{\sigma_{u}\sigma_{v}\sigma_{\varpi}}$
as a function of $u$ only. Definition~(\ref{eq: def_spincorr_up_to_constant})
and the asymptotics (\ref{eq: A_u_to_v}) imply that, as $u\to v$,
\begin{align*}
\log\CorrO{\sigma_{u}\sigma_{v}\sigma_{\varpi}}\  & =\ \int^{u}\re[\coefA_{\double{u,v}\cdot\cvr}(u)du]\\
 & =\ \int^{u}\re\biggl[\biggl(-\frac{1}{4(u-v)}+\frac{i}{4}\frac{\bar{u-v}}{|u-v|}f_{\cvr}^{\star}(v,v)+o(1)\biggr)du\biggr]\\
 & =\ -\frac{1}{4}\log|u-v|+C(v,\cvr)+\frac{i}{4}|u-v|f_{\cvr}^{\star}(v,v)+o(1),
\end{align*}
uniformly over $v$ in the bulk and away from other marked points,
where $C(v,\cvr)$ is a real constant. Allowing now both~$u$ and~$v$
to move and using~(\ref{eq: A_u_to_v}) again we see that 
\[
(\pa_{u}+\pa_{v})\log\CorrO{\sigma_{u}\sigma_{v}\sigma_{\varpi}}\ =\ \coefA_{\double{u,v}\cdot\cvr}(u)+\coefA_{\double{u,v}\cdot\cvr}(v)%
\ =\ o(1)\quad\text{as}\ \ u\to v,
\]
since $\eta_{uv}^{2}=-\eta_{vu}^{2}$ and $f_{\cvr}^{\star}(u,u)=f_{\cvr}^{\star}(v,v)+o(1)$.
Therefore, 
$C(v,\cvr)=:C(\cvr)$ in fact does not depend on $v$. Due to (\ref{eq: A_v_to_v})
we also know that, for each $q=1,\dots,n$, one has 
\[
\pa_{v_{q}}\log\CorrO{\sigma_{u}\sigma_{v}\sigma_{\varpi}}\ =\ \coefA_{\double{u,v}\cdot\cvr}(v_{q})\ =\ \coefA_{\cvr}(v_{q})+o(1)\ =\ \pa_{v_{q}}\log\CorrO{\sigma_{\varpi}}+o(1)
\]
as $u\to v$, which yields $C(\cvr)=\log\CorrO{\sigma_{\cvr}}+C_{n}$,
with $C_{n}$ independent of $u,v,v_{1},\dots,v_{n}$. 
\end{proof}
\begin{proof}[Proof of Proposition \ref{prop: spin_bdry_behavior}]
 By conformal covariance, we may assume that 
\[
\Omega\cap B_{r}(w_{q})=(\re w_{q}+\H)\cap B_{r}(w_{q})
\]
for $r>0$ small enough. Write $x_{q}+iy_{q}:=v_{n+q}-w_{q}$ for
$q=1,\ldots,d$. It follows from (\ref{eq: A_u_to_bdry}) that, as
$v_{n+1}\to w_{1}$,...,$v_{n+d}\to w_{d}$, we have 
\[
\re\biggl[\sum_{q=1}^{d}\coefA_{\cvr}(v_{n+q})dv_{n+q}\biggr]=\re\biggl[\sum_{q=1}^{d}\frac{i\cdot dv_{n+q}}{8y_{q}}\biggr]+o(1)=-\sum_{q=1}^{d}\frac{dy_{q}}{8y_{q}}+o(1).
\]
Viewing $v_{1},\dots,v_{n}$ as fixed parameters and $v_{n+1},\dots,v_{n+d}$
as variables, we obtain 
\[
\log\CorrO{v_{1}\ldots v_{n+d}}=-\int\biggl[\sum_{q=1}^{d}\frac{dy_{q}}{8y_{q}}+o(1)\biggr]=-\frac{1}{8}\sum_{q=1}^{d}\log y_{q}+C+o(1),
\]
where the constant $C$ might depend on $v_{1},\ldots,v_{n}$. Since
$\crad_{\Omega}(v_{n+q})=\log(2y_{q})+o(1)$ as $v_{n+q}\to w_{q}$,
this justifies the existence of the limit in~(\ref{eq: prop_spin_bdry}). 
\end{proof}

\subsection{Proof of Theorem~\ref{thm: fusion rules-3}}

\label{subsec: fusion-proofs-3}

In this section we move back to the fusion rules and prove the less
trivial part of them \textendash{} Theorem~\ref{thm: fusion rules-3}.
Actually, most of the work was already done above and we can now benefit
from the results of Section~\ref{subsec: asymptotics-as-spins-collide}
and Lemma~\ref{lem: spin_to_spin}.
\begin{proof}[Proof of (\ref{eq: fuse_mu_sigma})]
 According to the discussion given in Section~\ref{subsec: fusion-proofs-1},
it is enough to consider the case $\Op=\Op(\psi^{\any})\sigma_{\cvr}$.
The Pfaffian formula for fermions (which extends to $\mu_{\wo}=\psi_{\wo}^{\sharp}\sigma_{\wo}$,
see Definition~\ref{def: Disorders}) allows us to write 
\begin{equation}
\frac{\ccor{\mu_{\wo}\sigma_{\wt}\Op(\psi^{\any})\scvr}}{\ccor{\sigma_{\wo}\sigma_{\wt}\scvr}}\ =\ \sum_{\widehat{p}=1}^{k}(-1)^{\widehat{p}+1}\frac{\ccor{\mu_{\wo}\sigma_{\wt}\psi_{z_{\widehat{p}}}^{\any}\scvr}}{\ccor{\sigma_{\wo}\sigma_{\wt}\scvr}}\Pf\left[\frac{\ccor{\sigma_{\wo}\sigma_{\wt}\psi_{z_{p}}^{\any}\psi_{z_{q}}^{\any}\scvr}}{\ccor{\sigma_{\wo}\sigma_{\wt}\scvr}}\right]_{p,q\neq\widehat{p}}.\label{eq: mu_sigma_expansion}
\end{equation}
It follows from Proposition \ref{prop: v_to_v} that 
\[
\frac{\ccor{\sigma_{\wo}\sigma_{\wt}\psi_{z_{p}}^{\any}\psi_{z_{q}}^{\any}\scvr}}{\ccor{\sigma_{\wo}\sigma_{\wt}\scvr}}=\frac{\ccor{\psi_{z_{p}}^{\any}\psi_{z_{q}}^{\any}\scvr}}{\ccor{\scvr}}+O(|\wo-\wt|)
\]
and from Proposition \ref{prop: u_to_v-1} and Remark \ref{rem: f_u_to_v_higher_terms}
(note that these asymptotics easily extends from $\psi_{z_{\widehat{p}}}$
to $\psistar{z_{\widehat{p}}}$ and $\psi_{z_{\widehat{p}}}^{\any}$
by taking the complex conjugation and expanding in $\eta$) that 
\[
\frac{\ccor{\mu_{\wo}\sigma_{\wt}\psi_{z_{\widehat{p}}}^{\any}\scvr}}{\ccor{\sigma_{\wo}\sigma_{\wt}\scvr}}=|\wo-\wt|^{\frac{1}{2}}\biggl(\frac{\ccor{\psi_{\wt}^{\eta_{\wo\wt}}\psi_{z_{\widehat{p}}}^{\any}\scvr}}{\ccor{\scvr}}+O(|\wo-\wt|)\biggr).
\]
Plugging these asymptotics into (\ref{eq: mu_sigma_expansion})and
wrapping the sum back into a Pfaffian gives 
\[
\frac{\ccor{\mu_{\wo}\sigma_{\wt}\Op(\psi^{\any})\scvr}}{\ccor{\sigma_{\wo}\sigma_{\wt}\scvr}}\ =\ |\wo-\wt|^{\frac{1}{2}}\biggl(\frac{\ccor{\psi_{\wt}^{\eta_{\wo\wt}}\Op(\psi^{\any})\scvr}}{\ccor{\scvr}}+O(|\wo-\wt|)\biggr).
\]
Due to Lemma~\ref{lem: spin_to_spin} we also know that $\ccor{\sigma_{\wt}\sigma_{\wo}\scvr}=|\wo-\wt|^{-\frac{1}{4}}(\ccor{\scvr}+O(|\wo-\wt|))$.
Thus, 
\[
\ccor{\mu_{\wo}\sigma_{\wt}\Op(\psi^{\any})\scvr}\ =\ |\wo-\wt|^{\frac{1}{4}}\big(\ccor{\psi_{\wt}^{\eta_{\wo\wt}}\Op(\psi^{\any})\scvr}+O(|\wo-\wt|)\big).\qedhere
\]
\end{proof}
\smallskip{}

\begin{proof}[Proof of (\ref{eq: fuse_sigma_sigma})]
 First, note that Lemma~\ref{lem: spin_to_spin} provides a proof
of~(\ref{eq: fuse_sigma_sigma}) in the case $\Op=\scvr$. Indeed,
the normalization~(\ref{eq: spins_coherent}) of $\CorrO{\Op(\sigma)}$
was chosen so that all the constants $C_{n}$ in the asymptotics~(\ref{eq: lem_spin_to_spin})
are equal to zero. Exponentiating, we see that 
\[
\frac{\ccor{\sigma_{\wo}\sigma_{\wt}\scvr}}{\ccor{\scvr}}\ =\ |\wo-\wt|^{-\frac{1}{4}}\exp\left[\frac{1}{2}\frac{\ccor{\en_{\wt}\scvr}}{\ccor{\scvr}}|\wo-\wt|+o(|\wo-\wt|)\right]\quad\text{as}\ \ \wo\to\wt,
\]
which, by the Taylor expansion of the exponential, is the same as
(\ref{eq: fuse_sigma_sigma}).

We now handle the general case~$\Op=\Op(\psi^{\any})\scvr$. To this
end, we apply Proposition~\ref{prop: v_to_v} with $a=z_{1}$ and~$u=\wt$.
As $f_{\cvr}^{[\theta]}(\wt,z)=\frac{1}{2}(\bar{\theta}f_{\cvr}(\wt,z)+\theta f_{\cvr}^{\star}(\wt,z))$
and $\theta=-\frac{i}{2}f_{\cvr}^{[\eta]}(z_{1},\wt)$, the conclusion
can be written as follows: 
\begin{multline*}
\frac{\ccor{\sigma_{\wo}\sigma_{\wt}\psi_{z_{2}}\psi_{z_{1}}^{\eta}\scvr}}{\ccor{\sigma_{\wo}\sigma_{\wt}\scvr}}\ =\ \frac{\ccor{\psi_{z_{2}}\psi_{z_{1}}^{\eta}\scvr}}{\ccor{\scvr}}\\
+\frac{i|\wo-\wt|}{4}\left(\frac{\ccor{\psistar w\psi_{z_{1}}^{\eta}\scvr}}{\ccor{\scvr}}\cdot\frac{\ccor{\psi_{z_{2}}\psi_{w}\scvr}}{\ccor{\scvr}}\right.-\left.\frac{\ccor{\psi_{w}\psi_{z_{1}}^{\eta}\scvr}}{\ccor{\scvr}}\cdot\frac{\ccor{\psi_{z_{2}}\psistar w\scvr}}{\ccor{\scvr}}\right)+o(|\wo-\wt|).
\end{multline*}
By taking complex conjugates and expanding in $\eta$, the same identity
holds true with $\psi_{z_{2}}$ replaced by $\psistar{z_{2}}$ and/or
$\psi_{z_{1}}^{\eta}$ replaced by any of $\psi_{z_{1}}$, $\psistar{z_{1}}$.
Therefore, 
\begin{multline*}
\frac{\ccor{\sigma_{\wo}\sigma_{\wt}\Op(\psi^{\any})\scvr}}{\ccor{\sigma_{\wo}\sigma_{\wt}\scvr}}\ =\ \Pf\left[\frac{\ccor{\sigma_{\wo}\sigma_{\wt}\psi_{z_{p}}^{\any}\psi_{z_{q}}^{\any}\scvr}}{\ccor{\sigma_{\wo}\sigma_{\wt}\scvr}}\right]\ =\ \\
Pf\left[\frac{\ccor{\psi_{z_{p}}^{\any}\psi_{z_{q}}^{\any}\scvr}}{\ccor{\scvr}}\right.\\
\left.+\frac{i|\wo-\wt|}{4}\left(\frac{\ccor{\psi_{w}\psi_{z_{q}}^{\any}\scvr}}{\ccor{\scvr}}\frac{\ccor{\psistar w\psi_{z_{p}}^{\any}\scvr}}{\ccor{\scvr}}-\frac{\ccor{\psi_{w}\psi_{z_{p}}^{\any}\scvr}}{\ccor{\scvr}}\frac{\ccor{\psistar w\psi_{z_{q}}^{\any}\scvr}}{\ccor{\scvr}}\right)+o(|\wo-\wt|)\right].
\end{multline*}
This Pfaffian can be further expanded as follows: 
\begin{multline*}
\frac{\ccor{\sigma_{\wo}\sigma_{\wt}\Op(\psi^{\any})\scvr}}{\ccor{\sigma_{\wo}\sigma_{\wt}\scvr}}\ =\ \Pf\biggl[\frac{\ccor{\psi_{z_{p}}^{\any}\psi_{z_{q}}^{\any}\scvr}}{\ccor{\scvr}}\biggr]\\
+\ \frac{i|\wo-\wt|}{4}\!\!\sum_{1\le\widehat{p}<\widehat{q}\le k}(-1)^{\widehat{p}+\widehat{q}+1}\biggl(\frac{\ccor{\psi_{w}\psi_{z_{\widehat{q}}}^{\any}\scvr}}{\ccor{\scvr}}\frac{\ccor{\psistar w\psi_{z_{\widehat{p}}}^{\any}\scvr}}{\ccor{\scvr}}\\
-\frac{\ccor{\psi_{w}\psi_{z_{\widehat{p}}}^{\any}\scvr}}{\ccor{\scvr}}\frac{\ccor{\psistar w\psi_{z_{\widehat{q}}}^{\any}\scvr}}{\ccor{\scvr}}\biggr)\cdot\Pf\biggl[\frac{\ccor{\psi_{z_{\widehat{p}}}^{\any}\psi_{z_{\hat{q}}}^{\any}\scvr}}{\ccor{\scvr}}\biggr]_{p,q\not\in\{\widehat{p},\widehat{q}\}}\!\!+\ o(|\wo-\wt|)\,.
\end{multline*}
The sum is one term short from being the expansion of the Pfaffian
\[
\Pf\biggl[\frac{\ccor{\psi_{\alpha}\psi_{\beta}\scvr}}{\ccor{\scvr}}\biggr]_{\psi_{\alpha},\psi_{\beta}\in\{\psi_{w},\psistar w,\psi_{z_{1}}^{\any},\dots,\psi_{z_{k}}^{\any}\}}=\ \frac{\ccor{\psi_{w},\psistar w\Op(\psi^{\any})\scvr}}{\ccor{\scvr}}\ =\ -2i\frac{\ccor{\en_{w}\Op(\psi^{\any})\scvr}}{\ccor{\scvr}}
\]
in the rows corresponding to $\psi_{w},\psistar w$; the missing term
is 
\[
\frac{\ccor{\psi_{w}\psistar w\scvr}}{\ccor{\scvr}}\cdot\Pf\left[\frac{\ccor{\psi_{z_{p}}^{\any}\psi_{z_{q}}^{\any}\scvr}}{\ccor{\scvr}}\right]=-2i\frac{\ccor{\en_{w}\scvr}}{\ccor{\scvr}}\cdot\frac{\ccor{\Op(\psi^{\any})\scvr}}{\ccor{\scvr}}.
\]
Therefore, we have 
\begin{multline*}
\frac{\ccor{\sigma_{\wo}\sigma_{\wt}\Op(\psi^{\any})\scvr}}{\ccor{\sigma_{\wo}\sigma_{\wt}\scvr}}\ =\ \frac{\ccor{\Op(\psi^{\any})\scvr}}{\ccor{\scvr}}\left(1-\frac{1}{2}|\wo-\wt|\frac{\ccor{\en_{w}\scvr}}{\ccor{\scvr}}\right)\\
+\frac{1}{2}|\wo-\wt|\frac{\ccor{\en_{w}\Op(\psi^{\any})\scvr}}{\ccor{\scvr}}+o(|\wo-\wt|).
\end{multline*}
Since we already know that ${\ccor{\sigma_{\wo}\sigma_{\wt}\scvr}}=|\wo-\wt|^{-\frac{1}{4}}({\ccor{\scvr}}+\frac{1}{2}|\wo-\wt|{\ccor{\en_{w}\scvr}}+o(|\wo-\wt|))$,
we finally get the same result for ${\ccor{\sigma_{\wo}\sigma_{\wt}\Op(\psi^{\any})\scvr}}$. 
\end{proof}
\global\long\def\zo{\widehat{z}}
 \global\long\def\zt{z}

\begin{proof}[Proof of (\ref{eq: fuse_mu_mu})]
 This is a straightforward computation based upon the fusion rule~(\ref{eq: fuse_sigma_sigma})
discussed above and the definition of the correlation $\ccor{\mu_{\wo}\mu_{\wt}\Op}$.
Note that 
\[
\ccor{\mu_{\wo}\mu_{\wt}\Ob}=\frac{\lamb}{2\pi i}\oint_{(\wt)}\frac{(z-\wo)^{\frac{1}{2}}\ccor{\mu_{\wo}\psi_{z}\sigma_{\wt}\Ob}}{(\wt-\wo)^{\frac{1}{2}}(z-\wt)^{\frac{1}{2}}}dz.
\]
Since the integrand does not have a pole at $z=\wo$, we can move
the contour of integration to a definite distance from $\wo,\wt$.
Treating $\mu_{\wo}$ similarly, we obtain the formula 
\[
\ccor{\mu_{\wo}\mu_{\wt}\Ob}=\frac{1}{(2\pi i)^{2}}\oint_{(\wo,\wt,\zo)}\oint_{(\wo,\wt)}G_{\wo,\wt}(\zo,\zt)\ccor{\psi_{\zo}\psi_{\zt}\sigma_{\wo}\sigma_{\wt}\Ob}d\zo d\zt,
\]
where 
\[
G_{\wo,\wt}(\zo,\zt):=(\wt-\wo)^{-1}\frac{(\zt-\wo)^{\frac{1}{2}}(\zo-\wt)^{\frac{1}{2}}}{(\zt-\wt)^{\frac{1}{2}}(\zo-\wo)^{\frac{1}{2}}}\,.
\]
This expression is already suitable for the asymptotic analysis but
it is convenient to anti-symmetrize in $\zo,\zt$. It follows from
(\ref{eq: fuse_psi_psi}) that 
\[
\res_{\zo=\zt}G_{\wo,\wt}(\zo,\zt)\ccor{\psi_{\zo}\psi_{\zt}\sigma_{\wo}\sigma_{\wt}\Ob}\ =\ 2(\wt-\wo)^{-1}\ccor{\sigma_{\wo}\sigma_{\wt}\Ob}
\]
does not depend on~$\zt$. Therefore, the satellite integral vanishes
and one can exchange the contours of integration over $\zo$ and $\zt$.
Together with the anti-symmetry of~$\ccor{\psi_{\zo}\psi_{\zt}\sigma_{\wo}\sigma_{\wt}\Ob}$
in $\zo$, $\zt$, this observation implies that

\[
\ccor{\mu_{\wo}\mu_{\wt}\Ob}=\frac{1}{2(2\pi i)^{2}}\oint_{(\wo,\wt,\zo)}\oint_{(\wo,\wt)}(G_{\wo,\wt}(\zo,\zt)-G_{\wo,\wt}(\zt,\zo))\ccor{\psi_{\zo}\psi_{\zt}\sigma_{\wo}\sigma_{\wt}\Ob}d\zo d\zt.
\]
A simple computation gives the following asymptotics, as $\wo\to\wt$:
\[
G_{\wo,\wt}(\zo,\zt)-G_{\wo,\wt}(\zt,\zo)=\frac{1}{\zt-\wt}-\frac{1}{\zo-\wt}+\frac{\wo-\wt}{2(\zt-\wt)^{2}}-\frac{\wo-\wt}{2(\zo-\wt)^{2}}+O((\wo-\wt)^{2}).
\]
Using the fusion rule~(\ref{eq: fuse_sigma_sigma}), we arrive at
the expansion 
\begin{align*}
\ccor{\mu_{\wo}\mu_{\wt}\Ob}\  & =\ \frac{|\wo-\wt|^{-\frac{1}{4}}}{2(2\pi i)^{2}}\left[\,\oint_{(\wt,\zo)}\oint_{(\wt)}\left(\frac{1}{\zt-\wt}-\frac{1}{\zo-\wt}\right)\ccor{\psi_{\zo}\psi_{\zt}\Ob}d\zo d\zt\right.\\
 & +\ \frac{1}{2}(\wo-\wt)\oint_{(\wt,\zo)}\oint_{(\wt)}\left(\frac{1}{(\zt-\wt)^{2}}-\frac{1}{(\zo-\wt)^{2}}\right)\ccor{\psi_{z}\psi_{\zeta}\Ob}d\zo d\zt\\
 & +\ \left.\frac{1}{2}|\wo-\wt|\oint_{(\wt,\zo)}\oint_{(\wt)}\left(\frac{1}{\zt-\wt}-\frac{1}{\zo-\wt}\right)\ccor{\psi_{\zo}\psi_{\zt}\en_{w}\Ob}d\zo d\zt+o(|\wo-\wt|)\right].
\end{align*}
It remains to compute all the integrals using the fusion rules~(\ref{eq: fuse_psi_psi})\textendash (\ref{eq: fuse_psi_en}):
\begin{align*}
\oint_{(\wt,\zo)}\oint_{(\wt)}\frac{\ccor{\psi_{\zo}\psi_{\zt}\Ob}}{\zt-\wt}d\zo d\zt\  & =\ \oint_{(\wt)}0\,d\zt\ =\ 0\,,\\
\oint_{(\wt,\zo)}\oint_{(\wt)}\frac{\ccor{\psi_{\zo}\psi_{\zt}\Ob}}{\zo-\wt}d\zo d\zt\  & =\ 2\pi i\oint_{(\wt)}\ccor{\psi_{\wt}\psi_{\zt}\Ob}d\zt\ =\ -2(2\pi i)^{2}\ccor{\Ob}\,,\\
\oint_{(\wt,\zo)}\oint_{(\wt)}\frac{\ccor{\psi_{\zo}\psi_{\zt}\Ob}}{(\zt-\wt)^{2}}d\zo d\zt\  & =\ \oint_{(\wt)}0\,d\zt\ =\ 0\,,\\
\oint_{(\wt,\zo)}\oint_{(\wt)}\frac{\ccor{\psi_{\zo}\psi_{\zt}\Ob}}{(\zo-\wt)^{2}}d\zo d\zt\  & =\ 2\pi i\oint_{(\wt)}\frac{-2\ccor{\Ob}}{(\zo-\wt)^{2}}d\zo\ =\ 0\,,\\
\oint_{(\wt,\zo)}\oint_{(\wt)}\frac{\ccor{\psi_{\zo}\psi_{\zt}\en_{\wt}\Ob}}{\zt-\wt}d\zo d\zt\  & =\ 2\pi i\oint_{(w)}\frac{i\ccor{\psistar{\wt}\psi_{\zt}\Op}}{\zt-\wt}d\zt\ =\ -2(2\pi i)^{2}\ccor{\en_{\wt}\Op}\,,\\
\oint_{(\wt,\zo)}\oint_{(\wt)}\frac{\ccor{\psi_{\zo}\psi_{\zt}\en_{\wt}\Ob}}{\zo-\wt}d\zo d\zt\  & =\ 2\pi i\oint_{(\wt)}\frac{(-2\ccor{\en_{\wt}\Ob}+i\ccor{\psi_{\zo}\psistar w\Ob})}{\zo-\wt}d\zo\\
 & =\ (2\pi i)^{2}(-2\ccor{\en_{\wt}\Ob}+2\ccor{\en_{\wt}\Ob})\ =\ 0\,.
\end{align*}
Putting everything together gives the result. 
\end{proof}
\newpage{}

\section{Explicit formulae for continuous correlations}

\label{sec:Explicit-formulae}

\subsection{Simply connected domains.}

In this subsection, we give explicit formulae for the continuous correlation
functions 
\[
\ccor{\Op(\psi,\mu,\sigma,\en)}_{\Omega}\text{ or }\ccor{\Op(\sigma,\en)}_{\Omega,\bcond}
\]
in the case of a simply connected $\Omega$. By Theorem \ref{thm: ccov},
it suffices to give explicit formulae in the case $\Omega=\H$, the
upper half-plane. Moreover, it suffices to compute all the correlations
of the form 
\begin{equation}
\ccor{\sigma_{v_{1}}\dots\sigma_{v_{n}}}_{\H}\text{ and }\ccor{\sigma_{v_{1}}\dots\sigma_{v_{n}}\mu_{u_{1}}\dots\mu_{u_{m}}}_{\H};\label{eq: ccor_enough_to_compute}
\end{equation}
the rest, in principle, can be obtained therefrom by a straightforward
algorithmic procedure; in fact, even the $m=2$ case is enough. Indeed,
once we know explicit formulae for the correlations in (\ref{eq: ccor_enough_to_compute}),
the fusion rule (\ref{eq: fuse_mu_sigma}) yields an explicit formula
for 
\[
\ccor{\psi_{z}^{\eta_{1}}\psi_{w}^{\eta_{2}}\sigma_{v_{1}}\dots\sigma_{v_{n-2}}}_{\H};
\]
by Definition \ref{def: any fermions} and Proposition \ref{prop: cont_corr_conjugation},
this leads to an explicit formula for $\ccor{\Op(\psi,\sigma)}_{\H};$
by Definitions \ref{eq: def_mu_many} and \ref{def: energy}, this
extends to $\ccor{\Op(\psi,\mu,\sigma,\en)}_{\H}$, and, finally,
Definitions \ref{def: bcond_corr_continuous} and \ref{def: bcondtilde}
allow to treat $\ccor{\Op(\sigma,\en)}_{\H,\bcond}$.

We note that although, in principle, the procedure is straightforward,
in practice it involves tedious computations, and sometimes different
ways to compute the same correlation lead to different formulae that
are not obviously equivalent. For example, the multi-point fermion
correlation can be obtained by fusing spins and disorders in pairs
starting from (\ref{eq: spindisord}) below; when obtained this way,
it is not immediately obvious that the resulting formula has Pfaffian
structure.

Note that the domain $\H$, as usual, is endowed with a subdivision
of $\pa\H=\R$ into two subsets $\fixed$ and $\free$, implicit in
our notation. Let $-\infty<b_{1}<b_{2}<\dots<b_{2q-1}<b_{2q}<+\infty$
be the endpoints of the free arcs. (We assume that $\infty$ belongs
to $\fixed$; the case of $\fixed=\emptyset$ can be treated similarly
or by a limiting procedure $\free=(b_{1},b_{2}),$ $b_{1}\to-\infty$,
$b_{2}\to+\infty$.) Let $u_{1},\dots,u_{m},v_{1},\dots,v_{n}$ be
the points in the bulk. We intorduce the following notation: denote
by $N=n+q+m$ the total number of marked points, and for $1\leq i\leq N$,
put 
\begin{equation}
(a_{i},\hat{a}_{i}):=\begin{cases}
(b_{2i-1},b_{2i}), & 1\leq i\leq q;\\
(v_{i-q},\bar v_{i-q}), & q<i\leq q+n\\
(u_{i-q-n},\bar u_{i-q-n}), & q+n<i\leq N
\end{cases}.\label{eq: def_a_i}
\end{equation}
We denote by $\chi_{ij}$ the cross-ratio
\[
\chi_{ij}=\frac{(a_{i}-a_{j})(\hat{a}_{i}-\hat{a}_{j})}{(a_{i}-\hat{a}_{j})(\hat{a}_{i}-a_{j})},\quad1\leq i\leq N.
\]
Note that $\chi_{ij}>0$ if $i,j\leq q$ or $i,j\ge q+1;$ otherwise,
$|\chi_{ij}|=1$.
\begin{thm}
\label{thm: simply_conn_explicit}Let $n$ be even. In the above notation,
we have 
\begin{equation}
\ccor{\sigma_{v_{1}}\dots\sigma_{v_{n}}}_{\H}=\prod_{i=1}^{n}(\im v_{i})^{-\frac{1}{8}}\left(K^{n}\cdot\frac{\sum_{s\in\{\pm1\}^{n+q}}\prod_{1\leq i<j\leq q+n}\chi_{ij}^{\frac{s_{i}s_{j}}{4}}}{\sum_{s\in\{\pm1\}^{q}}\prod_{1\leq i<j\leq q}\chi_{ij}^{\frac{s_{i}s_{j}}{4}}}\right)^{\frac{1}{2}};\label{eq: spins_explicit}
\end{equation}
\begin{multline}
\ccor{\sigma_{v_{1}}\dots\sigma_{v_{n}}\mu_{u_{1}},\dots,\mu_{u_{m}}}_{\H}\\
=\prod_{i=1}^{n}(\im v_{i})^{-\frac{1}{8}}\prod_{i=1}^{m}(\im u_{i})^{-\frac{1}{8}}\left(-K^{n+m}\cdot\frac{\sum_{s\in\{\pm1\}^{N}}\disSign(s)\prod_{1\leq i<j\leq N}\chi_{ij}^{\frac{s_{i}s_{j}}{4}}}{\sum_{s\in\{\pm1\}^{q}}\prod_{1\leq i<j\leq q}\chi_{ij}^{\frac{s_{i}s_{j}}{4}}}\right)^{\frac{1}{2}},\label{eq: spindisord}
\end{multline}
where $K=2^{-\frac{3}{4}}$ and $\disSign(s)=\left(\prod_{p=q+n+1}^{N}s_{p}\right)$.
The same formulae hold for odd $n$ with plus boundary conditions
on the fixed boundary arcs.
\end{thm}

\begin{proof}[Proof of (\ref{eq: spins_explicit})]
We have to check that the logarithmic partial derivatives of the
expression in (\ref{eq: spins_explicit}) with respect to $\re v_{1}$
and $\im v_{1}$ match the real and imaginary part of the coefficient
$\coefA_{\H,v_{1},\dots,v_{n}},$ as defined in (\ref{eq: def_A}).
To compute the coefficient $\coefA_{\H,v_{1},\dots,v_{n}}$, we need
to start with the observable $f_{\H,\cvr}^{\sharp}(v_{1},z),$ where
$\cvr=\cvr(v_{1},\dots,v_{n}).$ Recall that it is the unique holomorphic
$\cvr$-spinor that is real on the real line, satisfies $f_{\H,\cvr}^{\sharp}(v_{1},z)=O(\frac{1}{z})$
at infinity, the conditions (\ref{eq: bc_fixed}\textendash \ref{eq: bc_free},
\ref{eq: bc_free_arc}), the conditions (\ref{eq: bc_branchpoint})
at $v_{2},\dots,v_{n}$, and the condition (\ref{eq: f_sharp_asymp})
at $v_{1}$. We will look for $f_{\H,\cvr}^{\sharp}(v_{1},z)$ in
the form 
\[
f_{\H,\cvr}^{\sharp}(v_{1},z)=\frac{Q(z)}{\prod_{i=1}^{q+n}\sqrt{(z-a_{i})(z-\hat{a}_{i})}},
\]
where $Q$ is a polynomial of degree less than $q+n$. Note that for
any $Q$ with real coefficients, this expression readily satisfies
(\ref{eq: bc_fixed}\textendash \ref{eq: bc_free}) and the condition
at infinity; the conditions (\ref{eq: bc_free_arc}), (\ref{eq: bc_branchpoint})
at $v_{2},\dots,v_{n}$, and (\ref{eq: f_sharp_asymp}) thus give
$q+n$ linear conditions for $q+n$ coefficients. Note that under
(\ref{eq: bc_free_arc}), $2\pi i\sum_{z\in\H}\res_{z}(f_{\H,\cvr}^{\sharp}(v_{1},\cdot))^{2}$
must be real, and under (\ref{eq: bc_branchpoint}), we have $2\pi i\res_{v_{j}}(f_{\H,\cvr}^{\sharp}(v_{1},\cdot))^{2}\in\R,$
$j>2$, so that also $2\pi i\res_{v_{1}}(f_{\H,\cvr}^{\sharp}(v_{1},\cdot))^{2}\in\R.$
This means that the conditions (\ref{eq: f_sharp_asymp}) can be replaced
with $\im(\lim_{z\to v_{1}}\lambda(z-v_{1})^{\frac{1}{2}}f_{\H,\cvr}^{\sharp}(v_{1},\cdot))=1.$
The condition (\ref{eq: bc_branchpoint}) can be written as $\im\lim(\lambda(z-v_{i})^{\frac{1}{2}}f_{\H,\cvr}^{\sharp}(v_{1},z))=0,$
and the condition (\ref{eq: bc_free_arc}) reads $\lim_{z\to b_{2i-1}}(z-b_{2i-1})^{\frac{1}{2}}f_{\H,\cvr}^{\sharp}(v_{1},z)=\lim_{z\to b_{2i}}(z-b_{2i})^{\frac{1}{2}}f_{\H,\cvr}^{\sharp}(v_{1},z).$
Recalling (\ref{eq: def_a_i}), the resulting linear system can be
written as 
\begin{equation}
\frac{Q(a_{i})}{\prod_{j\ne i}\sqrt{(a_{i}-a_{j})(a_{i}-\hat{a}_{j})}}+\frac{Q(\hat{a}_{i})}{\prod_{j\ne i}\sqrt{(\hat{a}_{i}-a_{j})(\hat{a}_{i}-\hat{a}_{j})}}=\sqrt{2\im v_{1}}\cdot\ind_{i=q+1},\,1\leq i\leq q+n.\label{eq: spin-lin}
\end{equation}
This is formally the same system as the one arising when $q=0,$ i.
e., for spin correlations with $+$ boundary condition. We refer the
reader to \cite[ Appendix]{ChelkakHonglerIzyurov}, where this system
has been solved by Cramer's rule, and it was shown that the resulting
coefficient $\coefA_{\H,v_{1},\dots,v_{n}}$ indeed matches the logarithmic
derivative of (\ref{eq: spins_explicit}). (Note that the denominator
in (\ref{eq: spins_explicit}) does not depend on $v_{1}$.)

It remains to check that the expression in (\ref{eq: spins_explicit})
satisfies the normalization (\ref{eq: spins_coherent}). This is straightforward
from the facts that as $v_{1}\to v_{2}$, we have $\chi_{q+1,q+2}\sim|v_{1}-v_{2}|^{2}/(2\im v_{1}\cdot2\im v_{2})$
and $\chi_{q+1,j}^{s_{q+1}s_{j}/4}\chi_{q+2,j}^{s_{q+2}s_{j}/4}=(\chi_{q+1,j}\chi_{q+2,j}^{-1})^{\pm s_{j}/4}\to1$
whenever $s_{q+1}=-s_{q+2}$ and $j\notin\{q+1,q+2\}$. 
\end{proof}
\begin{proof}[Proof of (\ref{eq: spindisord})]
 In this proof, we put $\cvr=\cvr(v_{1},\dots,v_{n},u_{1},\dots,u_{m})$.
We note that in the case $m=2$, the desired correlation could be
in principle read off the previous computation. Indeed, by definition
\[
\frac{\ccor{\sigma_{v_{1}}\dots\sigma_{v_{n}}\mu_{u_{1}}\mu_{u_{2}}}_{\H}}{\ccor{\sigma_{v_{1}}\dots\sigma_{v_{n}}\sigma_{u_{1}}\sigma_{u_{2}}}_{\H}}=\lim_{z\to u_{1}}\lambda(z-u_{2})^{\frac{1}{2}}f_{\Omega,\cvr}^{\sharp}(u_{2},z),
\]
so that the left-hand side is given by an expression with determinants
obtained by solving (\ref{eq: spin-lin}) by Cramer's rule; since
the system is always well posed by Proposition \ref{prop: Uniqueness_continuous},
we never divide by zero. We were unable to find a direct proof that
this expression matches the ratio of (\ref{eq: spindisord}) to (\ref{eq: spins_explicit}).
Note, however, that the existence of such an expression implies that
$\ccor{\sigma_{v_{1}}\dots\sigma_{v_{n}}\mu_{u_{1}}\mu_{u_{2}}}_{\H}$
is a \emph{real-analytic }function of $\re u_{1},\im u_{1},\re u_{2},\im u_{2}$;
so is the right-hand side of (\ref{eq: spindisord}). Hence, to prove
(\ref{eq: spindisord}) for $m=2$, it is sufficient to consider the
case of $u_{1}$ and $u_{2}$ close to each other. In particular,
because of the fusion rule (\ref{eq: fuse_mu_mu}), we may assume
that $\ccor{\sigma_{v_{1}}\dots\sigma_{v_{n}}\mu_{u_{1}}\mu_{u_{2}}}_{\H}\neq0$.
It is easy to check that the right-hand side of (\ref{eq: spindisord})
also satisfies the analog of (\ref{eq: fuse_mu_mu}) in the leading
order; hence, we may assume that both sides of (\ref{eq: spindisord})
are non-zero.

In the general case, we know by (\ref{eq: psi_mu_sigma_pfaff}) that
the left-hand side of (\ref{eq: spindisord}) can be expressed as
a Pfaffian of the correlations involvig two disorders. This means
that it is also a real analytic function of the parameters, and, by
the same argument, it suffices to consider the case where all $u$'s
are pairwise close to each other, in particular, as above, we may
assume that both sides of (\ref{eq: spindisord}) are non-zero.

We recall from (\ref{eq: fuse_psi_sigma}) that $\pa_{u_{m}}\log\ccor{\sigma_{v_{1}}\dots\mu_{u_{m}}}=\frac{1}{2}\coefA,$
where $\coefA$ is given by the expansion as $z\to u_{m}$,
\[
f(z):=\frac{\ccor{\sigma_{v_{1}}\dots\sigma_{v_{n}}\mu_{u_{1}}\dots\mu_{u_{m-1}}\sigma_{u_{m}}\psi_{z}}_{\H}}{\ccor{\sigma_{v_{1}}\dots\sigma_{v_{n}}\mu_{u_{1}}\dots\mu_{u_{m}}}_{\H}}=\bar{\lambda}(z-u_{m})^{-\frac{1}{2}}(1+2\coefA(z-u_{m})+o(z-u_{m})),
\]
where we drop the dependence of $f$ and $\coefA$ on their parameters
$v_{1},\dots,u_{m}.$ To finish the proof, we compare $\coefA$ with
the logarithmic derivative of the right-hand side of (\ref{eq: spindisord})
with respect to $a_{m}$, adapting the computation from \cite[Appendix]{ChelkakHonglerIzyurov}.
With that in hand, it only remains to match the normalization, which
is easily done by induction by sending $u_{m}\to u_{m-1}$ and using
(\ref{eq: fuse_mu_mu}). 

First, note that by by Proposition \ref{prop: cont_corr_conjugation},
$f(z)$ is holomorphic in the upper half-plane and satisfies the standard
boundary conditions. Hence, $f^{2}(z)$ can be extended to a meromorphic
function in $\C$; by (\ref{eq: bc_free_arc}), (\ref{eq: fuse_psi_mu}\textendash \ref{eq: fuse_psi_sigma}),
its poles at simple and located at $a_{1},\dots,a_{N}.$ Taking into
account that $f(z)=O(z^{-1})$ at $\infty,$ we conclude that $f$
must have the form 
\begin{equation}
f(z)=\frac{Q(z)}{\prod_{i=1}^{N}\sqrt{(z-a_{i})(z-\hat{a}_{i})}},\label{eq: spin-disorder-f}
\end{equation}
with $Q$ a polynomial with real coefficients of degree less than
$N$. We now write down the linear system for the coefficient of $Q$
that arises from (\ref{eq: fuse_psi_mu}\textendash \ref{eq: fuse_psi_sigma}):
\begin{align}
\re(\lim_{z\to u_{m}}\lambda(z-u_{m})^{\frac{1}{2}}f(z)) & =1;\label{eq: sp_dis_aux_1}\\
\re(\lim_{z\to\u_{j}}\lambda(z-u_{j})^{\frac{1}{2}}f(z)) & =0,\quad j=1,\dots m-1,\label{eq: sp_dis_aux_2}\\
\im(\lim_{z\to v_{j}}\lambda(z-v_{j})^{\frac{1}{2}}f(z)) & =0,\quad j=1,\dots n,\label{eq: spin-dis-aux-3}
\end{align}
which can be rewritten as 
\begin{align*}
\frac{Q(a_{i})}{\prod_{j\ne i}\sqrt{(a_{i}-a_{j})(a_{i}-\hat{a}_{j})}}+\frac{Q(\hat{a}_{i})}{\prod_{j\ne i}\sqrt{(\hat{a}_{i}-a_{j})(\hat{a}_{i}-\hat{a}_{j})}}= & 0, & i=1,\dots,q+n;\\
\frac{Q(a_{i})}{\prod_{j\ne i}\sqrt{(a_{i}-a_{j})(a_{i}-\hat{a}_{j})}}-\frac{Q(\hat{a}_{i})}{\prod_{j\ne i}\sqrt{(\hat{a}_{i}-a_{j})(\hat{a}_{i}-\hat{a}_{j})}}= & 0, & q+n+1\leq i<N;\\
\frac{Q(a_{i})}{\prod_{j\ne i}\sqrt{(a_{i}-a_{j})(a_{i}-\hat{a}_{j})}}-\frac{Q(\hat{a}_{i})}{\prod_{j\ne i}\sqrt{(\hat{a}_{i}-a_{j})(\hat{a}_{i}-\hat{a}_{j})}}= & 1, & i=N.
\end{align*}
If we now reorder the marked points by moving $(a_{N},\hat{a}_{N})$
to the beginning of the list, we get the same linear system as (A4)
in \cite{ChelkakHonglerIzyurov}, with the only difference that the
size of the system being is now $N\times N$, and there are extra
``minus'' signs in front of $D_{jj}$ with $j=1$ and $q+n+2\leq j\leq N$.
It is solved by the same computation; the only new aspect is that
it is \emph{not true in general }that the system is non-degenerate;
this is because Proposition \ref{prop: Uniqueness_continuous} only
holds if the behavior at punctures is as for the \emph{spin} insertions.
So, we compute the determinant of the system, encoding a subset $S\subset\{1,\dots,N\}$
by $\mu\in\{\pm1\}^{N},$ and denoting $\disSign'(\mu)=\mu_{q+n+2}\dots\mu_{N}$: 

\begin{multline*}
\det(D+C)=\sum_{S\subset\{1,\dots,N\}}\det C_{S}\det D_{S^{c}}\\
=\prod_{k=1}^{N}\frac{1}{\hat{a}_{k}-a_{k}}\sum_{\mu\in\{\pm1\}^{N}}\prod_{\substack{k\neq m\\
\mu_{k}=\mu_{m}=1
}
}\chi_{km}^{\frac{1}{2}}(\mu_{1}\disSign'(\mu))\prod_{\mu_{k}=-1}\prod_{j\neq k}(\chi_{kj}\chi_{jk})^{\frac{1}{4}}\\
=\prod_{k=1}^{N}\frac{1}{\hat{a}_{k}-a_{k}}\prod_{k\neq m}\chi_{km}^{\frac{3}{8}}\sum_{\mu\in\{\pm1\}^{N}}\prod_{k\neq m}(\mu_{1}\disSign'(\mu))\chi_{km}^{\frac{\mu_{k}\mu_{m}}{8}}.
\end{multline*}
The pre-factor does not vanish, and the sum, after re-ordering the
points back, is the the numerator of the fraction in the right-hand
side of (\ref{eq: spindisord}). Therefore, it does not vanish when
$u_{i}$ are pairwise close enough to each other, and as explained
above, it suffices to consider this case.

The rest of the computation is as in in \cite{ChelkakHonglerIzyurov};
rather than repeating it, we will only trace how the signs affect
it. First, we would now have 
\[
\det A_{[11]}=\prod_{k=2}^{N}\frac{1}{\hat{a}_{k}-a_{k}}\prod_{\substack{k,m=2\\
k\neq m
}
}^{N}\chi_{km}^{\frac{3}{8}}\prod_{k=2}^{N}\chi_{1k}^{\frac{1}{4}}\cdot\sum_{\substack{\mu\in\{\pm1\}^{N}\\
\mu_{1}=-1
}
}\disSign'(\mu)\prod_{k,m=1}\chi_{km}^{\frac{\mu_{k}\mu_{m}}{8}}.
\]
with the extra factor of $\disSign'(\mu)$ coming from the new $-1$
factors in $\det D_{\bar S}$. Similarly,
\begin{multline*}
\det(\tilde{D}+\tilde{C})=\prod_{k=2}^{N}\frac{1}{\hat{a}_{k}-a_{k}}\prod_{\substack{k,m=2\\
k\neq m
}
}^{N}\chi_{km}^{\frac{3}{8}}\prod_{k=2}^{N}\chi_{1k}^{\frac{1}{4}}\\
\times\sum_{\substack{\mu\in\{\pm1\}^{N}\\
\mu_{1}=-1
}
}\left(\disSign'(\mu)\prod_{k,m=1}\chi_{km}^{\frac{\mu_{k}\mu_{m}}{8}}\cdot\sum_{s=2}^{N}\frac{1+\mu_{s}}{2}(-\pa_{a_{1}}\log\chi_{1s})\right).
\end{multline*}
Since we have 
\[
\sum_{\substack{\mu\in\{\pm1\}^{N}\\
\mu_{1}=-1
}
}\disSign'(\mu)\prod_{k,m=1}\chi_{km}^{\frac{\mu_{k}\mu_{m}}{8}}=\frac{1}{2}\sum_{\substack{\mu\in\{\pm1\}^{N}}
}\mu_{1}\disSign'(\mu)\prod_{k,m=1}\chi_{km}^{\frac{\mu_{k}\mu_{m}}{8}},
\]
and similarly for the sum in $\det(\tilde{D}+\tilde{C}),$ we arrive
at 
\[
\coefA=-\frac{1}{4(a_{1}-\hat{a}_{1})}+\pa_{a_{1}}\log\left(\sum_{\substack{\mu\in\{\pm1\}^{N}}
}\mu_{1}\disSign'(\mu)\prod_{k,m=1}\chi_{km}^{\frac{\mu_{k}\mu_{m}}{8}}\right),
\]
which after reordering back $1\longleftrightarrow N$ is the desired
conclusion. 
\end{proof}
\begin{rem}
After having proven (\ref{eq: spindisord}), we can use it to write
down a formula for $\ccor{\sigma_{v_{1}}\dots\mu_{u_{m-1}}\sigma_{u_{m}}\psi_{z}}_{\H},$
or equivalently, for the polynomial $Q$ in (\ref{eq: spin-disorder-f}),
which we illustrate in the case of homogeneous fixed or plus boundary
conditions. Expressing $Q(z)=\sum_{i=1}^{N}\alpha_{i}Q_{i}(z),$ where
$Q_{i}(z)=\prod_{j\neq i}\frac{z-a_{j}}{a_{i}-a_{j}},$ we can find
$\alpha_{i}$ from fusion rules (\ref{eq: fuse_psi_mu}\textendash \ref{eq: fuse_psi_sigma}):
\[
\alpha_{i}=\frac{\ccor{\sigma_{a_{1}}\dots\Op(a_{i})\dots\mu_{a_{N-1}}\sigma_{a_{N}}}}{\ccor{\sigma_{a_{1}}\dots\mu_{a_{N}}}}\sqrt{2\im a_{i}}\prod_{j\neq i}\sqrt{(a_{j}-a_{i})(a_{j}-\bar a_{i})},
\]
where $\Op(a_{i})=i\mu_{a_{i},}$ if $i\leq n$ or $i=N,$ and $\Op(a_{i})=\sigma_{a_{i}}$
else. 
\end{rem}

\begin{rem}
For the fixed or plus boundary conditions, the right-hand side of
(\ref{eq: spindisord}) is the square root of $\ccor{:\cos\frac{\phi_{v_{1}}}{\sqrt{2}}:\ldots:\cos\frac{\phi_{v_{n}}}{\sqrt{2}}::\sin\frac{\phi_{u_{1}}}{\sqrt{2}}:\ldots:\sin\frac{\phi_{u_{m}}}{\sqrt{2}}:},$
where $\phi$ is the Gaussian free field in the upper half-plane,
i.e., the centered Gaussian field with covariance $\E(\phi(z_{1})\phi(z_{2}))=\log\left|z_{1}-\bar z_{2}\right|-\log|z_{1}-z_{2}|$
and the ``normal ordered'' exponentials (and by linearity, sines
and cosines) are defined as limits of $:e^{\alpha\phi}:=\lim_{\eps\to0}\eps^{\frac{\alpha^{2}}{2}}e^{\alpha\phi_{\eps}}$
with $\phi_{\eps}$ circle averages of $\phi,$ see e.g. \cite{kang2011gaussian,junnila2020imaginary}.
This recovers the bosonization prescriptions well-known from the physics
literature, see e.g. \cite[Section 12.3]{Yellow_book} which also
have been explained combinatorially by Dubédat \cite{dubedat2011exact}
who proved the convergence end the explicit formulae (\ref{eq: spindisord})
in the full plane case. 
\end{rem}

\subsection{Doubly connected domains. }

In this subsections, we compute one-point spin and energy functions
in doubly-connected domains, with homogeneous ($\plus,\minus,$ $\free,$
or $\fixed$) boundary conditions on each arc. Denote 
\[
\Ann_{p}:=\{z\in\C:e^{-p}<|z|<1\},
\]
the annulus of modulus $p$. Every doubly connected domain can be
conformally mapped to $\Ann_{p}$ for some $p>0$. Throughout this
section, we fix $p$, consider the correlations in $\Ann_{p}$ and
omit it from notation.

We will denote by $\ccor{\Op}_{\text{wired},\text{free}}$ the correlation
of $\Op$ in $\Ann_{p}$ with fixed boundary condition on $|z|=1$
and free boundary conditions on $|z|=e^{-p}.$ The notation $\ccor{\Op}:=\ccor{\Op}_{\tfixed,\tfixed}$,
$\ccor{\Op}_{+,-}$,$\ccor{\Op}_{+,+}$ $\ccor{\Op}_{+,\tfree}$,
$\ccor{\Op}_{\tfree,\tfree}$ is understood similarly. We denote by
$\spIn$ and $\spOut$ the spins on $|z|=e^{-p}$ and $|z|=1$ respectively,
which are meaningful observables whenever the corresponding boundary
arc belongs to $\fixed$.
\begin{thm}
\label{thm: ann_explicit}We have the following explicit formulae
for the scaling limits of magnetization in the annulus $\Ann_{p}$:
\begin{equation}
\ccor{\sigma_{v}}_{\tfree,+}=\left(2|v|\frac{\wp'(\log|v|)}{\wp''(p)}\right)^{-\frac{1}{8}};\quad\ccor{\sigma_{v}}_{\tfixed,+}=\left(2|v|\frac{\wp'(\log|v|+i\pi)}{\wp''(p+i\pi)}\right)^{-\frac{1}{8}};\label{eq: magn_free_plus}
\end{equation}
\begin{equation}
\ccor{\spIn\spOut}=\wp''(i\pi)^{-\frac{1}{8}}\wp''(p+i\pi)^{\frac{1}{8}};\label{eq: magn_in_out}
\end{equation}
\begin{equation}
\ccor{\sigma_{v}}_{+,\pm}=\frac{\left(2|v|\frac{\wp'(\log|v|+i\pi)}{\wp''(p+i\pi)}\right)^{-\frac{1}{8}}\pm\left(2|v|\frac{\wp'(-p-\log|v|+i\pi)}{\wp''(p+i\pi)}\right)^{-\frac{1}{8}}}{1\pm\wp''(i\pi)^{-\frac{1}{8}}\wp''(p+i\pi)^{\frac{1}{8}}}\label{eq: magn_pm}
\end{equation}
Here $\wp(\cdot):=\wp(\cdot|p,\pi i)$ is the Weierstrass elliptic
function with half-periods $p$ and $\pi i$.
\end{thm}

\begin{proof}
We first remark that (\ref{eq: magn_in_out}\textendash \ref{eq: magn_pm})
are easy consequences of (\ref{eq: magn_free_plus}). Indeed, by definition,
\begin{align*}
\ccor{\sigma_{v}}_{\tfree,+} & =\ccor{\sigma_{v}\spIn}_{\tfree,\tfixed},\\
\ccor{\sigma_{v}}_{\tfixed,+} & =\ccor{\sigma_{v}\spIn}.
\end{align*}
Also, by normalization convention (\ref{eq: prop_spin_bdry}), we
have 
\[
\ccor{\spIn\spOut}=\lim_{|v|\to1}\crad_{\Ann_{p}}(v)^{\frac{1}{8}}\ccor{\spIn\sigma_{v}}=\lim_{|v|\to1}(2(1-|v|))^{\frac{1}{8}}\ccor{\spIn\sigma_{v}}.
\]
Since 
\begin{align}
\ind_{\spIn=+1,\spOut=-1} & =\frac{1}{4}(1+\spIn-\spOut-\spIn\spOut),\label{eq: FW_ind_ann_1}\\
\ind_{\spIn=+1,\spOut=+1} & =\frac{1}{4}(1+\spIn+\spOut+\spIn\spOut),\label{eq: FW_ind_ann_2}
\end{align}
we readily obtain 
\begin{align*}
\ccor{\sigma_{v}}_{+,-} & =\frac{\ccor{\sigma_{v}\spIn}-\ccor{\sigma_{v}\spOut}}{1-\ccor{\spIn\spOut}};\\
\ccor{\sigma_{v}}_{+,+} & =\frac{\ccor{\sigma_{v}\spIn}+\ccor{\sigma_{v}\spOut}}{1+\ccor{\spIn\spOut}}.
\end{align*}
Therefore, it suffices to compute $\ccor{\sigma_{v}\spIn}=\ccor{\sigma_{v}\spIn}_{wired,wired}$
and $\ccor{\sigma_{v}\spIn}_{\tfree,\tfixed}.$ To this end, consider
the functions 
\[
f_{1}(z):=\frac{\ccor{\psi_{z}\spIn\mu_{v}}}{\ccor{\spIn\sigma_{v}}};\quad f_{2}(z):=\frac{\ccor{\psi_{z}\spIn\mu_{v}}_{\tfree,\tfixed}}{\ccor{\spIn\sigma_{v}}_{\tfree,\tfixed}}.
\]
By Proposition \ref{prop: cont_corr_conjugation}, $f_{1,2}$ are
holomorphic spinors ramified at the inner component and at $v$ in
$\Ann_{p}$ and satisfying the standard boundary conditions. Moreover,
they satisfy the expansion: 
\[
f_{1,2}(z)=\frac{\lamb}{(z-v)^{\frac{1}{2}}}(1+2\coefA_{1,2}(v)(z-v)+o(z-v)),\quad z\to v,
\]
where 
\[
\re[\coefA_{1}(v)dv]=d\log\ccor{\sigma_{v}\spIn};\quad\re[\coefA_{2}(v)dv]=d\log\ccor{\sigma_{v}\spIn}_{\tfree,\tfixed}.
\]
This follows from the definition of $\ccor{\sigma_{v}\sigma_{w}}$
and the fact that extraction of the coefficient $\coefA$ commutes
with the limit $|w|\to e^{-p}$ (cf. Proposition \ref{prop: coefA_vary}).
Therefore, in order to compute $\ccor{\sigma_{v}\spIn}$ and $\ccor{\sigma_{v}\spIn}_{\tfree,\tfixed}$,
it suffices to compute $f_{1,2},$ extract the coefficients $\coefA_{1,2},$
and integrate.

Consider $h_{1,2}:=\im\int f_{1,2}^{2}.$ It follows from Proposition
\ref{prop: f_to_h} that $h_{1,2}$ are single-valued harmonic functions
in $\Ann_{p}$ that are constant on each boundary component and satisfy
the asymptotics $h_{1,2}(z)\sim-\log|z-v|$ as $z\to v$. In other
words, 
\begin{equation}
h_{1,2}(z)=2\pi G(v,z)+\alpha_{1,2}(v)\log|z|,\label{eq: Green_plus_log}
\end{equation}
where $G(v,z)$ is the Green's function in $\Ann_{p}$ for zero-Dirichlet
boundary conditions. 

We claim that $f_{1,2}(z)(\tau_{z})^{\frac{1}{2}}=(\pa_{i\tau}h_{1,2}(z))^{\frac{1}{2}}$
must vanish for some $z$ with $|z|=1$. Indeed, if it did not, then
the real quantity $f_{1}(z)(\tau_{z})^{\frac{1}{2}}$ (respectively,
$if_{2}(z)(\tau_{z})^{\frac{1}{2}}$) could not change its sign along
$|z|=1$, and hence $f_{1}(z)$ (respectively, $f_{2}(z)$) would
be ramified along the outer boundary of the annulus, which it is not.
Since we know from Proposition \ref{prop: f_to_h} that $\pa_{i\tau}h_{1}(z)\geq0$
and $\pa_{i\tau}h_{2}(z)\leq0$ for $|z|=1$, these quantities vanish
at the point of minimum (respectively, maximum) of $\pa_{i\tau}G(v,\z)$
on the circle $|z|=1$. It is easy to see that these are, respectively,
$z=-v/|v|$ and $z=v/|v|.$ Hence, $f_{1}(-v/|v|)=f_{2}(v/|v|)=0.$

Assume without loss of generality that $v\in(e^{-p};1)\subset\R$.
Let $\varphi:z\mapsto e^{z}$ be the covering map of $\Ann_{p}$ by
the vertical strip $\S_{p}:=\{z:-p<\re z<0\}$, and denote $\hat{f}_{1,2}(z)=\varphi'(z)^{\frac{1}{2}}f_{1,2}(\mp\varphi(z)).$
Note that $\hat{f}_{1,2}^{2}(\cdot)$ are meromorphic functions in
$\S_{p}$ with the following properties:

\begin{enumerate}
\item they are $2\pi i$-periodic;
\item one has $\hat{f}_{1,2}^{2}(z)\in i\R$ for $\re z=0$ or $\re z=-p$;
\item $\hat{f}_{1,2}^{2}(\cdot)$ has a double zero at $0$;
\item $\hat{f}_{1}^{2}(\cdot)$ has a simple pole with residue $-i$ at
$\log|v|+i\pi$
\item $\hat{f}_{2}^{2}(\cdot)$ has a simple pole with residue $-i$ at
$\log|v|;$
\end{enumerate}
The second property above allows one to Schwarz reflect $\hat{f}_{1,2}(z)$
across $\re z=0,\re z=\pm p,\re z=\pm2p,\dots$ Since two such reflections
amount to a shift, this yields a doubly periodic function with periods
$2\pi i,2p$ with two simple poles per fundamental domain, and hence
also two zeros per fundamental domain, counted with multiplicity.
Therefore, from condition (3) above, we have $\hat{f}_{1,2}^{-2}=\beta_{1,2}\wp(z|p,i\pi)+\gamma_{1,2}.$
The constants $\beta_{1,2}$ and $\gamma_{1,2}$ are determined from
(4\textendash 5) above, yielding
\[
\hat{f}_{1,2}(z)=\left(\frac{-i\wp'(w_{1,2}|p,i\pi)}{\wp(z|p,i\pi)-\wp(w_{1,2}|p,i\pi)}\right)^{\frac{1}{2}};
\]
where $w_{1}=\log|v|+i\pi$ and $w_{2}=\log|v|.$ The corresponding
coefficients $\hat{\coefA}_{1,2}$ of the expansion of $\hat{f}_{1,2}$
at $w_{1,2}$ read 
\[
\hat{\coefA}_{1,2}(w_{1,2})=-\frac{1}{8}\cdot\frac{\wp''(w_{1,2})}{\wp'(w_{1,2})}=\left.(\log\wp'(w)^{-\frac{1}{8}})'\right|_{w=w_{1,2}}.
\]
Note that $\wp(z|p,\pi i)\in\R$ for $z\in\R$ or $z\in\R+i\pi.$
Therefore, the functions 
\[
\ccor{\sigma_{w}\spIn}_{\S_{p}}:=2^{-\frac{1}{8}}\frac{\wp'(\re w+i\pi)^{-\frac{1}{8}}}{\wp''(-p+i\pi)^{-\frac{1}{8}}}\quad\text{and }\quad\ccor{\sigma_{w}\spIn}_{\S_{p},\tfixed,\tfree}:=2^{-\frac{1}{8}}\frac{\wp'(\re w)^{-\frac{1}{8}}}{\wp''(-p)^{\frac{1}{8}}}
\]
have the logarithmic derivative given by $\re[\hat{\coefA}_{1,2}dw]$
and satisfy the asymptotics $\ccor{\sigma_{w}\spIn}_{\S_{p}}\sim(2(\re w+p))^{-\frac{1}{8}},$
$\ccor{\sigma_{w}\spIn}_{\S_{p},\tfixed,\tfree}\sim(2(\re w+p))^{-\frac{1}{8}}$
as $\re w\to-p$. Using conformal transformation properties of the
coefficient $\coefA$ (as in the proof of Theorem \ref{thm: ccov}),
and the fact that $\wp$ is even, we get 
\[
\ccor{\sigma_{v}\spIn}_{\Ann_{p}}=|v|^{-\frac{1}{8}}\ccor{\sigma_{\log v}\spIn}_{\S_{p}}=\left(2|v|\frac{\wp'(\log|v|+i\pi)}{\wp''(p+i\pi)}\right)^{-\frac{1}{8}};
\]
\[
\ccor{\sigma_{w}\spIn}_{\Ann_{p},\tfixed,\tfree}=|v|^{-\frac{1}{8}}\ccor{\sigma_{\log v}\spIn}_{\S_{p},\tfixed,\tfree}=\left(2|v|\frac{\wp'(\log|v|)}{\wp''(p)}\right)^{-\frac{1}{8}},
\]
as claimed.
\end{proof}
\begin{thm}
\label{thm: ann_explicit_fermions}We have the following explicit
formulae for fermionic correlations in the annulus $\Ann_{p}$: 
\begin{align}
\ccor{\psi_{z}\psi_{w}} & =(2/\sqrt{zw})\cdot ds(\log(z/w)); & \ccor{\psi_{z}\psi_{w}^{\star}} & =(2i/\sqrt{z\bar w})\cdot ds(\log(z\bar w));\label{eq: ann_ferm_fixed}\\
\ccor{\psi_{z}\psi_{w}}_{\tfixed,\tfree} & =(2/\sqrt{zw})\cdot cs(\log(z/w)); & \ccor{\psi_{z}\psi_{w}^{\star}}_{\tfixed,\tfree} & =(2i/\sqrt{z\bar w})\cdot cs(\log(z\bar w));\label{eq: ann_ferm_fixed_free}\\
\frac{\ccor{\psi_{z}\psi_{w}\spIn\spOut}}{\ccor{\spIn\spOut}} & =(2/\sqrt{zw})\cdot ns(\log(z/w)); & \frac{\ccor{\psi_{z}\psi_{w}^{\star}\spIn\spOut}}{\ccor{\spIn\spOut}} & =(2i/\sqrt{z\bar w})\cdot ns(\log(z\bar w)),\label{eq: ann_ferm_spins}
\end{align}
where $ns,ds,cs$ are the Jacobian elliptic functions with half-periods
$2K=2p$ and $2K':=2\pi i$.
\end{thm}

\begin{rem}
\label{rem: def_Jacobi}The Jacobian elliptic functions $ns,ds,cs$
are the unique meromorphic function on $\C$ whose poles are simple
poles at $2\pi\Z+(2pi)\Z$, with $\res_{z=0}ns(z)=\res_{z=0}ds(z)=\res_{z=0}cs(z)=1$,
and satisfying the (anti-)periodicity conditions 
\begin{align*}
ns(z+2p) & =-ns(z); & ns(z+2\pi i) & =ns(z);\\
ds(z+2p) & =-ds(z); & ds(z+2\pi i) & =-ds(z);\\
cs(z+2p) & =cs(z); & cs(z+2\pi i) & =-cs(z).
\end{align*}
\end{rem}

\begin{proof}[Proof of Theorem \ref{thm: ann_explicit_fermions}.]
 The proof is based on the fact that if we replace $\psi_{w}$ by
$\psi_{w}^{\eta}$ in the left-hand side of (\ref{eq: ann_ferm_fixed}\textendash \ref{eq: ann_ferm_spins}),
then, in view of Proposition \ref{prop: Uniqueness_continuous}, the
resulting correlation functions (respectively, spinor for (\ref{eq: ann_ferm_spins}))
are uniquely characterized by the standard boundary conditions and
the asymptotic expansion 
\begin{equation}
\bar{\eta}/(z-w)+O(1),\quad z\to w.\label{eq: Ann_z__to_w}
\end{equation}

We note that $\varphi(z):=\log z$ is the inverse of the conformal
covering map from the strip $\S_{p}=\{-p<\re z<0\}$, and $1/\sqrt{zw}=\varphi'(z)^{\frac{1}{2}}\varphi'(w)^{\frac{1}{2}},$
1/$\sqrt{z\bar w}=\varphi'(z)^{\frac{1}{2}}\bar{\varphi'(w)}^{\frac{1}{2}}.$
Note also that for $u=cs,ds,$ or $ns$, we have $u(\bar z)=\bar{u(z)}$
and $u(-z)=-u(z)$. We can therefore write, for $\eta\in\C$ and $z\in i\R,$
\[
\bar{\eta}\cdot u(z-w)+\eta\cdot i\cdot u(z+\bar w)=\lambda(\bar{\eta}\bar{\lambda}\cdot u(z-w)+\eta\lambda\cdot\bar{u(z-w)})\in\lambda\R,\quad u=cs,ds,ns,
\]
that is to say, these functions satisfies the standard boundary condition
(\ref{eq: bc_fixed}) on the right vertical boundary of $\S_{p}$.
Similarly, for $z\in-p+i\R,$ we have $\bar z=-z-2p$, and therefore
\[
\bar{\eta}\cdot u(z-w)+\eta\cdot i\cdot u(z+\bar w)=\begin{cases}
\lambda(\bar{\eta}\bar{\lambda}\cdot u(z-w)-\eta\lambda\cdot\bar{u(z-w)})\in\bar{\lambda}\R, & u=ns,ds;\\
\lambda(\bar{\eta}\bar{\lambda}\cdot u(z-w)+\eta\lambda\cdot\bar{u(z-w)})\in\lambda\R, & u=cs,
\end{cases}
\]
so that these functions satisfy (\ref{eq: bc_fixed}) (respectively,
(\ref{eq: bc_free})) on the left vertical boundary of $\S_{p}$.
Therefore, 
\[
\bar{\eta}\varphi'(\cdot)^{\frac{1}{2}}\varphi'(w)^{\frac{1}{2}}u(\varphi(\cdot)-\varphi(w))+\eta i(\varphi'(\cdot)^{\frac{1}{2}}\bar{\varphi'(w)}^{\frac{1}{2}}u(\varphi(\cdot)+\varphi(\bar w)))
\]
is a function (for $u=ds,cs$) or a spinor (for $u=ns$) in $\Ann_{p}$
satisfying the standard boundary conditions and the expansion (\ref{eq: Ann_z__to_w})
as $z\to w$. Therefore, by Proposition \ref{prop: Uniqueness_continuous},
for $u=ds,cs,ns$, it coincides with $\ccor{\psi_{z}\psi_{w}^{\eta}}_{\tfixed,\tfixed},\ccor{\psi_{z}\psi_{w}^{\eta}}_{\tfixed,\tfree},$
and $\ccor{\psi_{z}\psi_{w}^{\eta}\spIn\spOut}/\ccor{\spIn\spOut},$
respectively. Since $\eta$ is an arbitrary complex number, and in
view of (\ref{eq: psi_eta_expand}), this concludes the proof.
\end{proof}
Of course, Theorem \ref{thm: ann_explicit_fermions}, combined with
the Pfaffian formula (\ref{eq: multipsi_pfaffian}), Definition \ref{def: bcond_corr_continuous}
of the energy observable, and (\ref{eq: magn_in_out}), readily gives
the energy correlations in the annulus with $+/\free$, $+/+$ and
$+/-$ boundary conditions. As an example, let us write down the one-point
functions:
\begin{example}
One has 
\begin{align*}
\ccor{\en_{e}}_{\Ann_{p},+/\tfree} & =-cs(2\log(|e|))/|e|;\\
\ccor{\en_{e}}_{\Ann_{p},+/\pm} & =\frac{-ds(2\log(|e|))\mp ns(2\log(|e|))\cdot\ccor{\spIn\spOut}_{\Ann_{p}}}{1\pm\ccor{\spIn\spOut}_{\Ann_{p}}},
\end{align*}
where $\ccor{\spIn\spOut}_{\Ann_{p}}$ is given by (\ref{eq: magn_in_out}).
\end{example}

\begin{proof}
The first identity follows directly from Definition \ref{def: energy}
and (\ref{eq: ann_ferm_fixed_free}), while the second one is obtained
by writing $\ccor{\en_{e}}_{\Ann_{p},+/\pm}=\ccor{\en_{e}\ind_{\spIn=1,\spOut=\pm1}}/\ccor{\ind_{\spIn=1,\spOut=\pm1}}$
and using (\ref{eq: FW_ind_ann_1}\textendash \ref{eq: FW_ind_ann_2})
and (\ref{eq: ann_ferm_fixed}), (\ref{eq: ann_ferm_spins}).
\end{proof}
\begin{example}
\label{exa: annulus-SLE-PF} The SLE partition function in the annulus
$\Ann_{p}$ with Dobrushin boundary condition on the outer part of
the boundary , with marked points $b_{1,2}\in\{z:|z|=1\}$ and free
or plus boundary conditions on the outer part of the boundary are
given by, respectively, 
\[
Z_{+/-}^{\text{free}}(b_{1,}b_{2},p)=\left|(2/\sqrt{b_{1}b_{2}})\cdot cs(\log(b_{1}/b_{2})|2\pi,2pi)\right|.
\]
\begin{multline*}
Z_{+/-}^{+}(b_{1,}b_{2},p)=\frac{1}{2}\left|(2/\sqrt{b_{1}b_{2}})\cdot ns(\log(b_{1}/b_{2})|2\pi,2pi)\wp''(i\pi)^{-\frac{1}{8}}\wp''(p+i\pi)^{\frac{1}{8}}\right.\\
\left.+(2/\sqrt{b_{1}b_{2}})\cdot ds(\log(b_{1}/b_{2})|2\pi,2pi)\right|
\end{multline*}
Indeed, these are special cases of Corollary \ref{cor: SLE} obtained
using Theorem \ref{thm: ann_explicit_fermions}.
\end{example}

\section{Appendix}

\subsection{Construction of special discrete holomorphic functions}

In this section, we prove Lemmas \ref{lem: dzm1}, \ref{lem: dzm12}.
We use a systematic approach with discrete exponentials; thanks to
this approach, the results readily extend to isoradial graphs. 

On the square lattice, the discrete exponentials are just discrete
holomorphic functions on $\wcgr(\C^{1})$ of the form
\[
e(x+iy)=p^{x}q^{y}.
\]
This function satisfies (\ref{eq: Disc_hol}) if and only if 
\[
p-q+i(pq-1)=0.
\]
It is convenient to parametrize the solutions to this equation as
follows: 
\[
p(\zeta)=\frac{1+\frac{\zeta}{2}}{1-\frac{\zeta}{2}};\quad q(\zeta)=\frac{1+i\frac{\zeta}{2}}{1-i\frac{\zeta}{2}};\quad\zeta\in\C.
\]
We thus denote 
\[
e_{\zeta}(z):=p(\zeta)^{x}q(\zeta)^{y},\quad z=x+iy.
\]

We will also use a slightly different normalization of the discrete
exponentials, corresponding to ``shifting the origin'' to a vertex
$v\in\C^{1}$ or a corner $w\in\bcgr(\C^{1})$. Namely, we put 
\[
e_{\zeta}(v;z):=\frac{e_{\zeta}(z)}{e_{\zeta}(\hat{z})}\cdot\frac{1}{1-\frac{\zeta}{2}},
\]
where $\hat{z}=v+\frac{1}{2}\in\wcgr(\C^{1})$, and 

\[
e_{\zeta}(w;z):=\frac{e_{\zeta}(z)}{e_{\zeta}(\hat{z})}\cdot\frac{1}{(1-\frac{\zeta}{2})(1-i\frac{\zeta}{2})},
\]
where $\hat{z}=a+\frac{1}{2}+i\frac{1}{2}\in\wcgr(\C^{1})$ is the
``upper-right'' neighbor of $a\in\bcgr(\Cd)$.

We note that 
\begin{equation}
e_{\zeta}(0;z)=2\cdot(-1)^{\hat{s}(z)}\cdot\zeta^{-1}\cdot\bar{e_{4\bar{\zeta}^{-1}}(0;z)},\label{eq: exp_symmetry}
\end{equation}
where $\hat{s}(z)=\re z+\frac{1}{2}+\im z.$ We will need the following
asymptotic $e_{\zeta}(0;z)$ as $z\to\infty$:
\begin{equation}
e_{\zeta}(0\text{;}z)=e^{\zeta z+O(\zeta^{3}z)+O(\text{\ensuremath{\zeta}}^{2})};\quad e_{\zeta}(0\text{;}z)=2(-1)^{\hat{s}(z)}\zeta^{-1}e^{4\zeta^{-1}\bar z+O(\zeta^{-3}z)+O(\zeta^{-2})},\label{eq: asymp_e_v_z}
\end{equation}
useful when $\zeta$ is small and large respectively.
\begin{proof}[Proof of Lemma \ref{lem: dzm1}]
Label the lattice so that $a\in\bcgr(\C^{1})$ and put, for $z\in\wcgr(\C^{1}),$
\begin{equation}
\dzmone za:=\bar{\eta}_{a}K(a,z):=\bar{\eta}_{a}\pi^{-1}\int_{\zeta\in-\bar{(z-a)}\R_{\geq0}}e_{\zeta}(a,z)d\zeta.\label{eq: def_K}
\end{equation}
Here $K(a,z)$ is the discrete Cauchy kernel (see \cite[Theorem 2.21]{ChelkakSmirnov1});
that is, it satisfies $\dbar_{z}K(a,w)=\delta(w-a)$ and the asymptotics
\[
K(a,z)=\frac{2}{\pi}\eta_{z}\eta_{a}\re[\bar{\eta}_{z}\bar{\eta}_{a}(z-a)^{-1}]+O(|z-a|^{-2}).
\]
where $\eta_{z}\eta_{a}=\bar{\tau}$ in the notation of \cite{ChelkakSmirnov1}.
Specifically, if we denote 
\[
s(z):=\re(z-a)+\im(z-a)-1\in\Z,\quad z\in\wcgr(\C^{1}),
\]
then we have 
\begin{equation}
\eta_{z}\eta_{a}=\begin{cases}
\pm\lamb, & s(z)\text{ is even;}\\
\pm\lambb, & s(z)\text{ is odd.}
\end{cases}\label{eq: eta_a_eta_z}
\end{equation}

Note that $e_{\zeta}(a,z)$ is a rational function of $\zeta$ with
poles at $2\sgn(\re[z-a])$ and $2i\sgn(\im[z-a])$, satisfying $e_{\zeta}(a,z)=O(\zeta^{-2})$
as $\zeta\to\infty$. Hence, we can always rotate the integration
line to be either $\R_{\geq0}$ or $\R_{\leq0}$. For $\zeta\in\R$,
we have 
\[
e_{\zeta}(a,z)=-i\cdot(-1)^{s(z)}\cdot\zeta^{-2}\cdot\bar{e_{4\zeta^{-1}}(a,z)}.
\]
Therefore, the change of variable $\zeta\mapsto4\zeta^{-1}$ in (\ref{eq: def_K})
implies 
\[
K(a,z)=-i(-1)^{s(z)}\bar{K(a,z)},
\]
which together with (\ref{eq: eta_a_eta_z}) gives $\dzmone za\in\eta_{z}\R$,
$z\in\wcgr(\C^{1})$.

By Lemma \ref{lem: dhol_to_shol}, we can extend $P_{a}$ to $\Cgr(\C^{1})\setminus\{a\}$
so that it is s-holomorphic at all edges $e$ except for the two edges
$e_{\pm}$ incident to $a$; moreover, there are two (different) extensions
$\dzmone{a_{+}}a$ and $\dzmone{a_{-}}a$ of $\dzmone{\cdot}a$ to
$a$ that make the resulting function s-holomorphic at $e_{+}$ and
$e_{-}$, respectively. From the explicit values 
\[
K\left(a;a\pm\frac{1+i}{2}\right)=\pm\frac{\lamb}{\sqrt{2}};\quad K\left(a;a\pm\frac{1-i}{2}\right)=\pm\frac{\lambb}{\sqrt{2}}
\]
we readily get $\dzmone{a_{+}}a=\eta_{a}$ and $\dzmone{a_{-}}a=-\eta_{a}$,
as well as (\ref{eq: P_a_close}).
\end{proof}

\begin{proof}[Proof of Lemma \ref{lem: dzm12}]
We label the lattices so that $\frac{1}{2}\in\wcgr(\C^{1})$. Following
\cite{Dubedat}, we put, for $z\in\wcgr(\C^{1})$, 
\begin{equation}
\dmsqrt z{}:=\lamb2^{-\frac{1}{2}}\pi^{-1}\int_{\zeta\in-\bar z\R_{\geq0}}\zeta^{-\frac{1}{2}}e_{\zeta}(0;z)d\zeta.\label{eq: def_Q}
\end{equation}
This is clearly an $[0]$-spinor, and for each $w\in\bcgr(\C^{1})$,
we can rotate the line of integration in (\ref{eq: def_Q}) to be
the same for every $z\sim w$; thus, $\dmsqrt{\cdot}{}$ is discrete
holomorphic. If $\re z<0$ (respectively, $\re z>0$), we can rotate
the line of integration to be $\R_{\geq0}$ (respectively, $\R_{\leq0}$),
and then (\ref{eq: exp_symmetry}) and the change of variable $\zeta\mapsto4\zeta^{-1}$
implies that 
\[
\lambb\dmsqrt z{}=\lamb(-1)^{\hat{s}(z)}\bar{Q(z)}.
\]
From this, it follows that $\dmsqrt z{}\in\eta_{z}\R$ for all $z\in\wcgr(\C^{1})$.
The asymptotics (\ref{eq: q_asymp}) is computed by a version of Laplace's
method: letting $I(r,R):=\{\zeta\in-\bar z\R_{\geq0}:r\leq|z|\leq R\},$
we can write 
\begin{multline*}
\int_{\zeta\in-\bar z\R_{\geq0}}\zeta^{-\frac{1}{2}}e_{\zeta}(0;z)d\zeta\\
=\int_{I(0;|z|^{-\frac{1}{2}})}\zeta^{-\frac{1}{2}}e^{\zeta z+O(\zeta^{3}|z|)}d\zeta+2(-1)^{\hat{s}(z)}\int_{I(|z|^{\frac{1}{2}},\infty)}\zeta^{-\frac{3}{2}}e^{4\zeta^{-1}\bar z+O(\zeta^{-3}z)}d\zeta\\
+O(e^{-c|z|^{\frac{1}{4}}}).
\end{multline*}
Taking $r:=|z|$, by the change of variable $u:=-z\zeta$ in the first
integral and $\zeta:=-4\bar zu$ in the second one, we obtain
\begin{multline*}
\int_{\zeta\in-\bar z\R_{\geq0}}\zeta^{-\frac{1}{2}}e_{\zeta}(0;z)d\zeta=(-z)^{-\frac{1}{2}}\int_{0}^{\infty}u^{-\frac{1}{2}}e^{-|u|}du+(-1)^{\hat{s}(z)}(-\bar z)^{-\frac{1}{2}}\int_{0}^{\infty}u^{-\frac{3}{2}}e^{-|u|^{-1}}du+o(1)\\
=i\Gamma\left(\frac{1}{2}\right)\left(z^{-\frac{1}{2}}-(-1)^{\hat{s}(z)}\bar z^{-\frac{1}{2}}\right)+o(|z|^{-\frac{1}{2}})=i\pi^{\frac{1}{2}}\left(z^{-\frac{1}{2}}-(-1)^{\hat{s}(z)}\bar z^{-\frac{1}{2}}\right)+o(|z|^{-\frac{1}{2}}).
\end{multline*}
Note that for $z\in\wcgr(\Cd)$, we have 
\[
z^{-\frac{1}{2}}-(-1)^{\hat{s}(z)}\bar z^{-\frac{1}{2}}=2\proj{\lamb\eta_{z}}{z^{-\frac{1}{2}}}=2\lamb\eta_{z}\re\left[\bar{\eta}_{z}\lambb z^{-\frac{1}{2}}\right]=2\lamb\proj{\eta_{z}}{\lambb z{}^{-\frac{1}{2}}},
\]
hence the asymptotics (\ref{eq: asymp_zm12}) follows. The values
of $\dmsqrt z{}$ at $z$ near $0$ can be easily computed from the
definition: 
\[
\dmsqrt{\frac{1}{2}}{}=\lambb;\quad\dmsqrt{\frac{1}{2}\pm i}{}=\pm(\sqrt{2}-1)\lamb;
\]

\[
\dmsqrt{-\frac{1}{2}}{}=\lamb;\quad\dmsqrt{-\frac{1}{2}\pm i}{}=\pm(\sqrt{2}-1)\lambb;
\]
here all the values are given on the same sheet of $\Cd_{[0]}$ defined
by the cut along $-i\R_{\geq0}.$ We now can now extend $\dmsqrt{\cdot}{}$,
to an $s$-holomorphic function on the whole $\Cgr(\Cd_{[0]})$ by
Lemma \ref{lem: dhol_to_shol}; the values at $\frac{i}{2}$ (on the
same sheet as above) is then given by $\dmsqrt{\frac{i}{2}}{}=\proj 1{\dmsqrt{\frac{1}{2}}{}+\dmsqrt{\frac{1}{2}+i}{}}=\re[\lambb-\lamb+\sqrt{2}\lamb]=1,$
and the value at $-\frac{i}{2}$ is handled similarly.
\end{proof}

\subsection{Proof of Corollary \ref{cor: SLE}.}

Along with the boundary conditions $\bcond$, we consider the corresponding
boundary conditions $\tilde{\bcond}$, as introduced before Lemma
\ref{lem: corr_to_obs_for_Thm_3}, and first discuss how to prove
Corollary \ref{cor: SLE} for those boundary conditions. This setup
is similar to that in \cite{IzyurovMconn}, apart from two differences: 
\begin{enumerate}
\item we do allow free boundary arcs;
\item we do not impose any boundary regularity assumptions near the marked
points $b_{2},\dots,b_{\hat{q}}$.
\end{enumerate}
Denote by $b_{1}^{\delta},\dots,b_{q}^{\delta}\in\pa\Od$ the points
separating different boundary conditions in $\Od$ and approximating
$b_{1},\dots,b_{q}$ (see the discussion before Corollary \ref{cor: SLE});
more precisely we take $b_{i}^{\delta}$ to be \emph{corners} outside
of the domain, incident to the dual vertices where the boundary condition
change. Let $\beta$ be a cross-cut in $\Omega$ separating $b_{1}$
from other marked points and other boundary components, let $\beta^{\delta}$
be a sequence of cross-cuts in $\Od$ approximating $\beta$. Let
$\gamma_{t}$ be the first $t$ steps of an interface growing from
a boundary point $b_{1}^{\delta}\in\partial\Od$, and let $T^{\delta}$
be first time that $\gamma_{t}$ hits $\beta^{\delta}$. Denote by
$b_{1}^{\delta}(t)$ one of the \emph{corners} incident to the tip
of $\gamma_{t}$ and to the last edge in $\gamma_{t}$, and lying
\emph{outside} of $\Od\setminus\gamma_{t}$. (This corner lives the
extension of $\Od\setminus\gamma_{t}$ by one layer of faces along
the boundary, placed on a Riemann surface if necessary, cf. the discussion
in the beginning of Section \ref{subsec: Ising_definitions}.) For
$t<T^{\delta}$ and $z$ a corner in $\Od$ separated from $b_{1}^{\delta}$
by $\beta^{\delta}$ , we introduce the following observable:
\begin{equation}
M_{t}^{\delta}(z):=\frac{\E_{\Od\setminus\gamma_{t}}(\psi_{b_{1}^{\delta}(t)}\psi_{z})}{\E_{\Od\setminus\gamma_{t}}(\psi_{b_{1}^{\delta}(t)}\psi_{b_{2}^{\delta}}\dots\psi_{b_{q}^{\delta}})}\cdot\frac{\E_{\Od}(\psi_{b_{1}^{\delta}}\psi_{b_{2}^{\delta}}\dots\psi_{b_{q}^{\delta}})}{\E_{\Od}(\psi_{b_{1}^{\delta}}\psi_{z})}.\label{eq: mart_obs}
\end{equation}
Note that the correlations are with \emph{standard boundary conditions,
}with the same set $\free$ as $\bcond^{\delta}$. We make the following
observations:
\begin{enumerate}
\item Up to $t=T^{\delta},$ $M_{t}^{\delta}(z)$ is a martingale with respect
to the filtration generated by $\gamma_{t}$ under the Ising probability
measure with boundary conditions $\tilde{\bcond}^{\delta}$;
\item As $\delta\to0$, we have $M_{t}^{\delta}(z)=M_{t}(z)+o(1)$, uniformly
over $z$ in compact subsets of $\Omega$ separate from $b_{1}$ by
$\beta$, over all $t$, and over all possible realizations of $\gamma_{t}$,
where 
\[
M_{t}(z)=\lim_{w_{i}\to(b_{i}^{\delta})^{\bullet},w\to(b_{1}^{\delta}(t))^{\bullet}}\frac{\ccor{\psi_{w}\psi_{z}}_{\Omega^{\delta}\setminus\gamma_{t}}}{\ccor{\psi_{w}\psi_{w_{2}}\dots\psi_{w_{q}}}_{\Omega^{\delta}\setminus\gamma_{t}}}\cdot\frac{\ccor{\psi_{w_{1}}\dots\psi_{w_{q}}}_{\Omega^{\delta}}}{\ccor{\psi_{w_{1}}\psi_{z}}_{\Omega^{\delta}}}.
\]
\end{enumerate}
To check the first statement, notice that we have an even number of
points $b_{q}^{\delta}$ on each boundary component. Therefore, under
the standard boundary conditions, 
\[
\frac{\E_{\Od\setminus\gamma_{t}}(\psi_{b_{1}^{\delta}(t)}\psi_{z})}{\E_{\Od\setminus\gamma_{t}}(\psi_{b_{1}^{\delta}(t)}\psi_{b_{2}^{\delta}}\dots\psi_{b_{q}^{\delta}})}=\eta\frac{\E_{\Od\setminus\gamma_{t}}((\psi_{z}\sigma\mu)\cdot\mu_{b_{1}^{\delta}(t)^{\bullet}}\mu_{b_{2}^{\delta,\bullet}}\dots\mu_{b_{q}^{\delta,\bullet}})}{\E_{\Od\setminus\gamma_{t}}(\mu_{b_{1}^{\delta}(t)^{\bullet}}\mu_{b_{2}^{\delta,\bullet}}\dots\mu_{b_{q}^{\delta,\bullet}})}=\E_{\Od\setminus\gamma_{t},\tilde{\bcond}^{\delta}}(\psi_{z}\sigma\mu),
\]
where $\eta=\eta_{b_{2}^{\delta}}^{-1}\dots\eta_{b_{q}^{\delta}}^{-1}$,
$\sigma$ is the spin on the boundary component adjacent to $b_{1}$
and $\mu:=\mu_{b_{2}^{\delta,\bullet}}\dots\mu_{b_{q}^{\delta,\bullet}}$.
Since the random variable $\psi_{z}\sigma\mu$ can be measured outside
of $\beta^{\delta}$, we have 
\[
\E_{\Od\setminus\gamma_{t},\tilde{\bcond}^{\delta}}(\psi_{z}\sigma\mu)=\E_{\Od,\tilde{\bcond}^{\delta}}(\psi_{z}\sigma\mu|\gamma_{t}),
\]
which is the Lévy martingale.

The second statement above can be proven by first replacing the points
$b_{1}^{\delta}(t),b_{1}^{\delta},\dots,b_{q}^{\delta}$ by bulk points,
invoking convergence of correlations in the bulk, and then repeatedly
applying Lemma \ref{lem: Clements_clever_lemma}, as in the end of
the proof of Theorem \ref{thm: Intro_3}.

Let $\varphi:\Omega\to\Lambda$ be a conformal map from $\Omega$
to a nice domain $\Lambda\subset\H$ such that the boundary component
containing $b_{1}$ is mapped to $\R$. We can pick a sequence $\varphi^{\delta}$
be conformal maps $\Omega^{\delta}\to\text{\ensuremath{\Lambda^{\delta}}}$to
nice domains converging to $\varphi$ on compact subsets of $\Omega$.
We notice that by conformal covariance, and Carathéodory continuity
of continuous correlations (which follows from the fact that these
correlations are limits of discrete correlations under Carathéodory
convergence), we have 
\begin{align}
M_{t}(z) & =\lim_{w_{i}\to\hat{b}_{i},w\to\varphi(b_{1}^{\delta}(t))}\frac{\ccor{\psi_{w}\psi_{\varphi(z)}}_{\Lambda\setminus\gamma_{t}}}{\ccor{\psi_{w}\psi_{w_{2}}\dots\psi_{w_{q}}}_{\Lambda\setminus\gamma_{t}}}\cdot\frac{\ccor{\psi_{w_{1}}\dots\psi_{w_{q}}}_{\Lambda}}{\ccor{\psi_{w_{1}}\psi_{\varphi(z)}}_{\Lambda}}+o(1)\nonumber \\
 & =\lim_{w\to\varphi(b_{1}^{\delta}(t))}\frac{\ccor{\psi_{w}\psi_{\varphi(z)}}_{\Lambda\setminus\gamma_{t}}}{\ccor{\psi_{w}\psi_{\hat{b}_{2}}\dots\psi_{\hat{b}_{q}}^{\flat}}_{\Lambda\setminus\gamma_{t}}}\cdot\frac{\ccor{\psi_{\hat{b}_{1}}\dots\psi_{\hat{b}_{q}}^{\flat}}_{\Lambda}}{\ccor{\psi_{\hat{b}_{1}}\psi_{\varphi(z)}}_{\Lambda}}+o(1)\nonumber \\
 & =\frac{g_{t}'(\varphi(z))^{\frac{1}{2}}\ccor{\psi_{W_{t}}\psi_{g_{t}(\varphi(z))}}_{\Lambda_{t}}}{\ccor{\psi_{W_{t}}\psi_{\hat{b}_{2}(t)}\dots\psi_{\hat{b}_{q}(t)}^{\flat}}_{\Lambda_{t}}}\cdot\frac{\ccor{\psi_{\hat{b}_{1}}\dots\psi_{\hat{b}_{q}}^{\flat}}_{\Lambda}}{\ccor{\psi_{\hat{b}_{1}}\psi_{\varphi(z)}}_{\Lambda}}+o(1);\label{eq: M_t_final}
\end{align}
note that the denominator does not vanish due to Lemma \ref{lem: bc_tilde_nonvanish}.
From this point on, the same proof as in \cite{IzyurovMconn} readily
applies and shows that the interface $\gamma_{t}$ converges to the
chordal Loewner chain driven by the solution of SDE 
\[
dW_{t}=\sqrt{3}dB_{t}+3\partial\log\ccor{\psi_{W_{t}}\psi_{\hat{b}_{2}(t)}\dots\psi_{\hat{b}_{q}(t)}^{\flat}}_{\Lambda_{t}}dt,
\]
where $\partial$ denotes the partial derivative in the first argument.
Indeed, the setup in \cite{IzyurovMconn} is exactly as above, the
only missing ingredients are the asymptotics $\ccor{\psi_{W_{t}}\psi_{w}}=(w-W_{t})^{-1}+o(1)$
as $w\to W_{t}$, the non-vanishing of $\ccor{\psi_{W_{t}}\psi_{\hat{b}_{2}(t)}\dots\psi_{\hat{b}_{q}(t)}^{\flat}}$
(which follows from Lemma \ref{lem: bc_tilde_nonvanish}), and the
asymptotic expansion $\partial_{W_{t}}^{(n)}\ccor{\psi_{W_{t}}\psi_{w}}=n!(w-W_{t})^{-(n+1)}+o(1)$
as $w\to W_{t}$, which is readily available from Example \ref{exa: diff_obs_by_a}.

To pass from the boundary conditions $\tilde{\bcond}$ to the boundary
conditions $\bcond,$ notice the Radon\textendash Nikodym derivative
of the latter measure on the interfaces with respect to the former
is given by the ratio of the partition functions $Z(\Omega^{\delta}\setminus\gamma_{t},\bcond^{\delta})/Z(\Od\setminus\gamma_{t},\tilde{\bcond}^{\delta}).$
By Theorem \ref{thm: Intro_3}, this ratio converges to a conformally
invariant limit given by the second factor in (\ref{eq: SLE_partition_function}).
From here, there are two ways to conclude: one is to invoke Girsanov's
transform, and another one is to observe that multiplying the martingale
observable (\ref{eq: mart_obs}) by the Radon\textendash Nikodym derivative
yields a martingale with respect to the other measure. Then, the same
proof as above applies, with the only difference that (\ref{eq: M_t_final})
gets multipled by a factor independent of $z$.

\bibliographystyle{amsplain}
\bibliography{Mixedcorr}

\end{document}